\documentclass[11pt]{article}

\usepackage{amssymb}
\usepackage{graphicx}
\usepackage{amsmath}
\usepackage{setspace}
\usepackage{textpos}
\usepackage{changepage}
\usepackage{url}

\usepackage{hyperref}

\usepackage{amsthm}

\usepackage[round]{natbib}

\newtheorem{assumption}{Assumption}
\newtheorem{proposition}{Proposition}
\newtheorem{remark}{Remark}
\newtheorem{lemma}{Lemma}

\newtheorem{corollary}{Corollary}

\begin{document}

\bibliographystyle{asa}

{
\singlespacing
\title{\textbf{Testing the Number of Regimes in Markov Regime Switching Models}
\author{Hiroyuki Kasahara\thanks{This research is supported by the Natural Science and Engineering Research Council of Canada, JSPS Grant-in-Aid for Scientific Research (C) No.\ 17K03653, and the Institute of Statistical Mathematics Cooperative Use Registration (2017 ISM CUR-171). The authors thank the seminar participants at Indiana University, LSE, and Vanderbilt University for their helpful comments. The authors also thank Chiyoung Ahn for excellent research assistance and Marine Carrasco, Liang Hu, and Werner Ploberger for making their code available.}\\
Vancouver School of Economics\\
University of British Columbia\\
hkasahar@mail.ubc.ca \and Katsumi Shimotsu\\
Faculty of Economics \\
University of Tokyo\\
shimotsu@e.u-tokyo.ac.jp
}}
\date{January 28, 2018} 
\maketitle
}
\vspace{-0.2in}
\begin{abstract}
Markov regime switching models have been used in numerous empirical studies in economics and finance. However, the asymptotic distribution of the likelihood ratio test statistic for testing the number of regimes in Markov regime switching models has been an unresolved problem. This paper derives the asymptotic distribution of the likelihood ratio test statistic for testing the null hypothesis of $M_0$ regimes against the alternative hypothesis of $M_0 + 1$ regimes for any $M_0\geq 1$ both under the null hypothesis and under local alternatives. We show that the contiguous alternatives converge to the null hypothesis at a rate of $n^{-1/8}$ in regime switching models with normal density. The asymptotic validity of the parametric bootstrap is also established.
\end{abstract}


\section{Introduction}
The Markov regime switching model has been a popular framework for empirical work in economics and finance. Following the seminal contribution by \citet{hamilton89em}, it has been used in numerous empirical studies to model, for example, the business cycle \citep{hamilton05stlouis, morleypiger12restat}, stock market volatility \citep{hamiltonsusmel94joe}, international equity markets \citep{angbekaert02rfs, okimoto08jfqa}, monetary policy \citep{schorfheide05red, simszha06aer, bianchi13restud}, and economic growth \citep{kahnrich07jme}. Comprehensive theoretical accounts and surveys of applications are provided by \citet{hamilton08palgrave, hamilton16hdbk} and \citet{angtimmermann12annual}.

The number of regimes is an important parameter in applications of Markov regime switching models. Despite its importance, however, testing for the number of regimes in Markov regime switching models has been an unsolved problem because the standard asymptotic analysis of the likelihood ratio test statistic (LRTS) breaks down because of problems such as unidentifiable parameters, the true parameter being on the boundary of the parameter space, and the degeneracy of the Fisher information matrix. Testing the number of regimes for Markov regime switching models with normal density, which are popular in empirical applications, poses a further difficulty because normal density has the undesirable mathematical property that the second-order derivative with respect to the mean parameter is linearly dependent on the first derivative with respect to the variance parameter, leading to further singularity.

This paper proposes the likelihood ratio test of the null hypothesis of $M_0$ regimes against the alternative hypothesis of $M_0 + 1$ regimes for any $M_0\geq 1$ and derives its asymptotic distribution. To the best of our knowledge, the asymptotic distribution of the LRTS has not been derived for testing the null hypothesis of $M_0$ regimes with $M_0 \geq 2$. To test the null hypothesis of no regime switching, namely $M_0=1$, \citet{hansen92jae} derives an upper bound of the asymptotic distribution of the LRTS, and \citet{garcia98ier} also studies this problem. \citet{carrasco14em} propose an information matrix-type test for parameter constancy in general dynamic models including regime switching models. \citet{chowhite07em} derive the asymptotic distribution of the quasi-LRTS for testing the single regime against two regimes by rewriting the model as a two-component mixture model, thereby ignoring the temporal dependence of the regimes.\footnote{\citet{cartersteigerwald12em} show that ignoring temporal dependence may render the quasi-maximum likelihood estimator inconsistent.} \citet{quzhuo17wp} extend the analysis of \citet{chowhite07em} and derive the asymptotic distribution of the LRTS that properly takes into account the temporal dependence of the regimes under some restrictions on the transition probabilities of latent regimes. \citet{marmer08empirical} and \citet{dufour17emreviews} develop tests for the null hypothesis of no regime switching by using different approaches from the LRTS. The studies discussed above focus on testing the single regime against two regimes. To the best of our knowledge, however, the asymptotic distribution of the LRTS for testing the null hypothesis of $M_0$ regimes with $M_0 \geq 2$ remains unknown.

Several papers in the literature consider tests when some parameters are not identified under the null hypothesis. These include \citet{davies77bm, davies87bm}, \citet{andrewsploberger94em, andrewsploberger95as}, \citet{hansen96em}, \citet{andrews01em}, and \citet{liushao03as}, among others. Estimation and testing with a degenerate Fisher information matrix are investigated in an iid setting by \citet{chesher84em}, \citet{leechesher86joe}, \citet{rotnitzky00bernoulli}, and \citet{gukoenkervolgushev17wp}, among others. \citet{chen14joe} examine uniform inference on the mixing probability in mixture models.

To facilitate the analysis herein, we develop a version of Le Cam's differentiable in quadratic mean (DQM) expansion that expands the likelihood ratio under the loss of identifiability, while adopting the reparameterization and higher-order expansion of \citet{kasaharashimotsu15jasa}. In an iid setting, \citet{liushao03as} develop a DQM expansion under the loss of identifiability in terms of the generalized score function. We extend \citet{liushao03as} to accommodate dependent and heterogeneous data as well as modify them to fit our context of parametric regime switching models. Using a DQM-type expansion has an advantage over the ``classical'' approach based on the Taylor expansion up to the Hessian term because deriving a higher-order expansion becomes tedious as the order of expansion increases in a Markov regime switching model. Furthermore, regime switching models with normal components are not covered by \citet{liushao03as} because their Theorem 4.1 assumes that the generalized score function is obtained by expanding the likelihood ratio twice, whereas our Section \ref{subsec:hetero_normal} shows that the score function is a function of the fourth derivative of the likelihood ratio in the normal case.

Our approach follows \citet{dmr04as} [DMR hereafter], who derive the asymptotic distribution of the maximum likelihood estimator (MLE) of regime switching models. We express the higher-order derivatives of the period density ratios in terms of the conditional expectation of the derivatives of the period \textit{complete-data} log-density, i.e., the log-density when the state variable is observable, by applying the missing information principle \citep[][]{woodbury71biometrics,louis82jrssb} and extending the analysis of DMR. We then show that these derivatives of the period density ratios can be approximated by a stationary, ergodic, and square integrable martingale difference sequence by conditioning on the infinite past, and this approximation is shown to satisfy the regularity conditions for our DQM expansion.

We first derive the asymptotic null distribution of the LRTS for testing $H_0: M=1$ against $H_A: M=2$. When the regime-specific density is not normal, the log-likelihood function is locally approximated by a quadratic function of the \textit{second-order} polynomials of the reparameterized parameters. When the density is normal, the degree of deficiency of the Fisher information matrix and required order of expansion depend on the value of the unidentified parameter; in particular, when the latent regime variables are serially uncorrelated, the model reduces to a finite mixture normal model in which the fourth-order DQM expansion is necessary to derive a quadratic approximation of the log-likelihood function. We expand the log-likelihood with respect to the judiciously chosen polynomials of the reparameterized parameters---which involves the \textit{fourth-order} polynomials---to obtain a uniform approximation of the log-likelihood function in quadratic form and derive the asymptotic null distribution of the LRTS by maximizing the quadratic form under a set of cone constraints building on the results of \citet{andrews99em,andrews01em}.

To derive the asymptotic null distribution of the LRTS for testing $H_0: M=M_0$ against $H_A: M=M_0+1$ for $M_0\geq 2$, we partition a set of parameters that describes the true null model in the alternative model into $M_0$ subsets, each of which corresponds to a specific way of generating the null model. We show that the asymptotic distribution of the LRTS is characterized by the maximum of the $M_0$ random variables, each of which represents the LRTS for testing each of the $M_0$ subsets.

We also derive the asymptotic distribution of the LRTS under local alternatives. \citet{carrasco14em} show that the contiguous local alternatives of their tests are of order $n^{-1/4}$, where $n$ is the sample size. In a related problem of testing the number of components in finite mixture normal regression models, \citet{kasaharashimotsu15jasa} show that the contiguous local alternatives are of order $n^{-1/8}$ \citep[see also][]{chenli09as, chenlifu12jasa, honguyen16as}. We show that the value of the unidentified parameter affects the convergence rate of the contiguous local alternatives. When the regime-specific density is normal, some contiguous local alternatives are of the order $n^{-1/8}$, and the LRT is shown to have non-trivial power against them. The tests of \citet{carrasco14em} do not have power against such alternatives, whereas the test of \citet{quzhuo17wp} rules out such alternatives because of their restriction on the parameter space.

The asymptotic validity of the parametric bootstrap is also established both under the null hypothesis and under local alternatives. The simulations show that our bootstrap LRT has good finite sample properties. Our results also imply that the bootstrap LRT is valid for testing the number of hidden states in hidden Markov models because this paper's model includes the hidden Markov model as a special case. Although several papers have analyzed the asymptotic property of the MLE of the hidden Markov model,\footnote{See, for example, \citet{leroux92spa}, \citet{francq98stat}, \citet{krishnamurthy98jtsa}, \citet{brr98as}, \citet{jp99as}, \citet{legrandmevel00math}, and \citet{doucmatias01bernouiil}.} the asymptotic distribution of the LRTS for testing the number of hidden states has been an open question.%
\footnote{\citet{gassiat00esaim} show that the LRTS for testing $H_0: M=1$ against $H_A: M=2$ diverges when state-specific densities have \emph{known and distinct} parameter values. \citet{dannemann08cjstat} analyze the modified quasi-LRTS for testing the null of two states against three.}

The remainder of this paper is organized as follows. After introducing the notation and assumptions in Section 2, we discuss the degeneracy of the Fisher information matrix and loss of identifiability in regime switching models in Section 3. Section 4 establishes the DQM-type expansion. Section 5 presents the uniform convergence of the derivatives of the density ratios. Sections 6 and 7 derive the asymptotic null distribution of the LRTS. Section 8 derives the asymptotic distribution under local alternatives. Section 9 establishes the consistency of the parametric bootstrap. Section 10 reports the results from the simulations and an empirical application, using U.S. GDP per capita quarterly growth rate data. Section 11 collects the proofs and auxiliary results.

\section{Notation and assumptions}

Let $:=$ denote ``equals by definition.'' Let $\Rightarrow$ denote the weak convergence of a sequence of stochastic processes indexed by $\pi$ for some space $\Pi$. For a matrix $B$, let $\lambda_{\min}(B) $ and $\lambda_{\max}(B)$ be the smallest and largest eigenvalues of $B$, respectively. For a $k$-dimensional vector $x = (x_1,\ldots,x_k)'$ and a matrix $B$, define $|x| := \sqrt{x'x}$ and $|B| := \sqrt{\lambda_{\max}(B'B)}$. For a $k\times 1$ vector $a=(a_1,\ldots,a_k)'$ and a function $f(a)$, let $\nabla_{a} f(a):=(\partial f(a)/\partial a_{1},\ldots,\partial f(a)/\partial a_{k})'$, and let $\nabla_{a}^j f(a)$ denote a collection of derivatives of the form $(\partial^j/\partial a_{i_1}\partial a_{i_2} \ldots \partial a_{i_j})f(a)$. Let $\mathbb{I}\{A\}$ denote an indicator function that takes the value 1 when $A$ is true and 0 otherwise. $\mathcal{C}$ denotes a generic non-negative finite constant whose value may change from one expression to another. Let $a \vee b :=\max\{a,b\}$ and $a \wedge b :=\min\{a,b\}$. Let $\lfloor x \rfloor$ denote the largest integer less than or equal to $x$, and define $(x)_+ := \max\{x,0\}$. Given a sequence $\{f_k\}_{k=1}^n$, let $\nu_n(f_k) := n^{-1/2} \sum_{k=1}^n [f_k - \mathbb{E}_{\vartheta^*}(f_k)]$. For a sequence $X_{n\varepsilon}$ indexed by $n=1,2,\ldots$ and $\varepsilon$, we write $X_{n\varepsilon} = O_{p\varepsilon}(a_n)$ if, for any $\delta>0$, there exist $\varepsilon>0$ and $M, n_0 <\infty$ such that $\mathbb{P}(|X_{n\varepsilon}/a_n| \leq M) \geq 1- \delta$ for all $n > n_0$, and we write $X_{n\varepsilon} = o_{p\varepsilon}(a_n)$ if, for any $\delta_1,\delta_2>0$, there exist $\varepsilon>0$ and $n_0$ such that $\mathbb{P}(|X_{n\varepsilon}/a_n| \leq \delta_1) \geq 1- \delta_2$ for all $n > n_0$. Loosely speaking, $X_{n\varepsilon} = O_{p\varepsilon}(a_n)$ and $X_{n\varepsilon} = o_{p\varepsilon}(a_n)$ mean that $X_{n\varepsilon} = O_{p}(a_n)$ and $X_{n\varepsilon} = o_{p}(a_n)$ when $\varepsilon$ is sufficiently small, respectively. All limits are taken as $n \to \infty$ unless stated otherwise. The proofs of all the propositions and lemmas are presented in the appendix.
 
Consider the Markov regime switching process defined by a discrete-time stochastic process $\{(X_k,Y_k,W_k)\}$, where $(X_k,Y_k,W_k)$ takes values in a set $\mathcal{X}_M\times \mathcal{Y}\times \mathcal{W}$ with $\mathcal{Y}\subset \mathbb{R}^{q_y}$ and $\mathcal{W}\subset \mathbb{R}^{q_w}$, and let $\mathcal{B}(\mathcal{X}_M\times\mathcal{Y}\times\mathcal{W})$ denote the associated Borel $\sigma$-field. For a stochastic process $\{Z_k\}$ and $a<b$, define ${\bf Z}_{a}^b:=(Z_a,Z_{a+1},\ldots,Z_b)$. Denote $\overline {\bf Y}_{k-1}:=(Y_{k-1},\ldots,Y_{k-s})$ for a fixed integer $s$ and $\overline {\bf Y}^b_{a}:=(\overline {\bf Y}_{a},\overline {\bf Y}_{a+1},\ldots,\overline {\bf Y}_{b})$. Here, $Y_k$ is an observable variable, $X_k$ is an unobservable state variable, $\overline {\bf Y}_{k-1}$ is the lagged $Y_k$'s used as a covariate, and $W_k$ is a weakly exogenous covariate. DMR's model does not include $W_k$. 
\begin{assumption} \label{assn_a1}
(a) $\{X_k\}_{k=0}^\infty$ is a first-order Markov chain with the state space $\mathcal{X}_M:=\{1,2,\ldots,M\}$. (b) For each $k\geq 1$, $X_k$ is independent of $({\bf X}_{0}^{k-2},\overline {\bf Y}_{0}^{k-1},{\bf W}_{0}^\infty)$ given $X_{k-1}$. (c) For each $k \geq 1$, $Y_k$ is conditionally independent of $( {\bf Y}_{-s+1}^{k-s-1}, {\bf X}_{0}^{k-1}, {\bf W}_{0}^{k-1},{\bf W}_{k+1}^{\infty})$ given $(\overline {\bf Y}_{k-1},X_k,W_k)$. (d) $ {\bf W}_{1}^{\infty}$ is conditionally independent of $(\overline {\bf Y}_{0},X_{0})$ given $W_0$.\footnote{Assumptions \ref{assn_a1}(a)--(d) imply that $W_k$ is conditionally independent of $({\bf X}_{0}^{k-1},\overline{\bf Y}_{0}^{k-1})$ given ${\bf W}_{0}^{k-1}$.} (e) $\{(X_k,Y_k,W_k)\}_{k=0}^{\infty}$ is a strictly stationary ergodic process.
\end{assumption} 
When $W_k$ is absent, DMR provide a sufficient condition for the ergodicity of $(X_k,Y_k)$ in their Assumption (A2). We assume the ergodicity of $(X_k,Y_k,W_k)$ for brevity.

The unobservable Markov chain $\{X_k\}$ is called the \textit{regime}. The integer $M$ represents the number of regimes specified in the model. The parameter $\vartheta_M=(\vartheta_{M,y}',\vartheta_{M,x}')'$ belongs to $\Theta_M= \Theta_{M,y}\times \Theta_{M,x}$, a compact subset of $\mathbb{R}^{q_M}$. $\vartheta_{M,x}$ contains the parameter of the transition probability of $X_k$, which we denote by $q_{\vartheta_{M,x}}(x_{k-1},x_k):= \mathbb{P}(X_{k}=x_k|X_{k-1}=x_{k-1})$. Let $p_{ij}:= q_{\vartheta_{M,x}}(i,j)$ for $i=1,\ldots,M$ and $j=1,\ldots,M-1$, and $q_{\vartheta_{M,x}}(i,M)$ is determined by $q_{\vartheta_{M,x}}(i,M) = 1-\sum_{j = 1}^{M - 1} p_{ij}$. $\vartheta_{M,y}=(\theta_1',\ldots,\theta_M',\gamma')'$ contains the parameter of the conditional density of $Y_k$ given $(\overline {\bf Y}_{k-1},X_k,W_k)$, which is given by $g_{\vartheta_{M,y}}(y_k|\overline{\bf y}_{k-1},x_k,w_k) := \sum_{j \in \mathcal{X}_M} \mathbb{I}\{x_k=j\}f(y_k|\overline{\bf{y}}_{k-1},w_k;\gamma,\theta_j)$. Here, $\gamma$ is the structural parameter that does not vary across regimes, $\theta_j$ is the regime-specific parameter that varies across regimes, and $f(y_k|\overline{\bf{y}}_{k-1},w_k;\gamma,\theta_j)$ is the conditional density of $y_k$ given $(\overline{\bf{y}}_{k-1},w_k)$ when $x_k=j$. 
Let 
\begin{align*}
p_{\vartheta} (y_k,x_k| \overline{\bf{y}}_{k-1},x_{k-1},w_k) &:= q_{\vartheta_x}(x_{k-1},x_k)g_{\vartheta_y}(y_k|\overline{\bf{y}}_{k-1},x_k,w_k) \\
 & = q_{\vartheta_x}(x_{k-1},x_k)\sum_{j \in \mathcal{X}_M} \mathbb{I}\{x_k=j\}f(y_k|\overline{\bf{y}}_{k-1},w_k;\gamma,\theta_j).
\end{align*} 
We assume $\Theta_{M,y}=\Theta_{\theta} \times \cdots \times \Theta_\theta\times \Theta_\gamma$, and the true parameter value is denoted by $\vartheta_M^*$. 

We make the following assumptions that correspond to (A1)--(A3) in DMR.
\begin{assumption}\label{assn_a2}
(a) $0<\sigma_{-}:=\inf_{\vartheta_{M,x}\in\Theta_{M,x}} \min_{x,x'\in \mathcal{X}_M} q_{\vartheta_{M,x}}(x,x')$ and \\$\sigma_+:=\sup_{\vartheta_{M,x}\in\Theta_{M,x}}\max_{x,x'\in\mathcal{X}_M}q_{\vartheta_{M,x}}(x,x')<1$ for each $M$. (b) For all $y'\in \mathcal{Y}$, $\overline y\in \mathcal{Y}^s$, and $w\in \mathcal{W}$, $0<\inf_{\vartheta_y\in\Theta_{M,y}} \sum_{x\in \mathcal{X}_M} g_{\vartheta_{M,y}}(y'|\overline y,x,w)$ and $\sup_{\vartheta_{M,y}\in\Theta_{M,y}} \sum_{x\in \mathcal{X}_M} g_{\vartheta_{M,y}}(y'|\overline y,x,w) <\infty$. (c) $b_+:=\sup_{\vartheta_{M,y}\in\Theta_y} \sup_{\overline {\bf y}_0,y_1,w,x} g_{\vartheta_{M,y}}(y_1|\overline {\bf y}_0,x,w)<\infty$ and $\mathbb{E}_{\vartheta^*}(|\log b_{-}(\overline{\bf{Y}}_0^1,W_1)|)<\infty$, where $b_{-}(\overline{\bf y}_0^1,w_1):=\inf_{\vartheta_{M,y}\in\Theta_{M,y}} \sum_{x\in\mathcal{X}_M} g_{\vartheta_{M,y}}(y_1|\overline{\bf y}_0,x,w_1)$.
 \end{assumption} 

As discussed on p. 2260 of DMR, Assumption \ref{assn_a2}(a) implies that the Markov chain $\{X_k\}$ has a unique invariant distribution and is uniformly ergodic for all $\theta_{M,x} \in \Theta_{M,x}$.\footnote{Assumptions \ref{assn_a1}(c) and \ref{assn_a2}(a) are also employed in DMR. As discussed in \citet{kasaharashimotsu17hamilton}, these assumptions together rule out models in which the conditional density $Y_k$ depends on both current and lagged regimes. \citet{kasaharashimotsu17hamilton} show the asymptotic normality of the MLE while relaxing Assumption \ref{assn_a2}(a) to allow for $\inf_{\vartheta_{M,x}\in\Theta_{M,x}} \min_{x,x'\in \mathcal{X}_M} q_{\vartheta_{M,x}}(x,x')=0$. It is possible to derive the asymptotic distribution of the LRT under similar assumptions to \citet{kasaharashimotsu17hamilton}, albeit with a tedious derivation.} For notational brevity, we drop the subscript $M$ from $\mathcal{X}_M$, $\vartheta_{M}$, $\Theta_{M}$, etc., unless it is important to clarify the specific value of $M$. Assumptions \ref{assn_a1}(b) and (c) imply that $\{Z_k\}_{k=0}^{\infty}:=\{(X_k,\overline{\bf Y}_{k})\}_{k=0}^{\infty}$ is a Markov chain on $\mathcal{Z}:=\mathcal{X}\times\mathcal{Y}^s$ given $\{W_k\}_{k=0}^{\infty}$, and $Z_k$ is conditionally independent of $({\bf Z}_0^{k-2},{\bf W}_{0}^{k-1},{\bf W}_{k+1}^{\infty})$ given $(Z_{k-1},W_k)$. Consequently, Lemma 1, Corollary 1, and Lemma 9 of DMR go through even in the presence of $\{W_k\}_{k=0}^{\infty}$. Because $\{(Z_k,W_k)\}_{k=0}^\infty$ is stationary, we can and will extend $\{(Z_k,W_k)\}_{k=0}^\infty$ to a stationary process $\{(Z_k,W_k)\}_{k=-\infty}^\infty$ with doubly infinite time. We denote the probability measure and associated expectation of $\{(Z_k,W_k)\}_{k=\infty}^\infty$ under stationarity by $\mathbb{P}_\vartheta$ and $\mathbb{E}_\vartheta$, respectively.\footnote{DMR use $\overline{\mathbb{P}}_\vartheta$ and $\overline{\mathbb{E}}_\vartheta$ to denote probability and expectation under stationarity on $\{Z_k\}_{k=\infty}^\infty$, because their Section 7 deals with the case when $Z_0$ is drawn from an arbitrary distribution. Because we assume $\{(Z_k,W_k)\}_{k=\infty}^\infty$ is stationary, we use notations such as $\mathbb{P}_\vartheta$ and $\mathbb{E}_\vartheta$ without an overline for simplicity.}

Under Assumptions \ref{assn_a1}(a)--(d), the density of ${\bf Y}_{1}^n$ given $X_0=x_0$, $\overline{\bf Y}_{0}$ and ${\bf W}_{0}^n$ is given by
\begin{equation} \label{cond_density}
p_{\vartheta_{M}}({\bf Y}_1^n| \overline{\bf Y}_{0},{\bf W}_{0}^n,x_0) = \sum_{{\bf x}_1^n\in \mathcal{X}_{M}^n}\prod_{k=1}^n p_{\vartheta_M}(Y_k,x_k|\overline{\bf Y}_{k-1},x_{k-1},W_k).
\end{equation}
Define the conditional log-likelihood function and stationary log-likelihood function as
\begin{align*}
\ell_n(\vartheta,x_0) & := \log p_\vartheta({\bf Y}_1^n| \overline{\bf Y}_{0},{\bf W}_{0}^n,x_0)=\sum_{k=1}^n \log p_\vartheta (Y_k| \overline{\bf Y}_{0}^{k-1},{\bf W}_{0}^{k},x_0) ,\\
\ell_n(\vartheta) & := \log p_\vartheta({\bf Y}_1^n| \overline{\bf Y}_{0},{\bf W}_{0}^n)=\sum_{k=1}^n \log p_\vartheta (Y_k| \overline{\bf Y}_{0}^{k-1},{\bf W}_{0}^{k}), 
\end{align*}
where we use the fact that $p_\vartheta (Y_k| \overline{\bf Y}_{0}^{k-1},{\bf W}_{0}^{n},x_0) = p_\vartheta (Y_k| \overline{\bf Y}_{0}^{k-1},{\bf W}_{0}^{k},x_0)$ and $p_\vartheta (Y_k| \overline{\bf Y}_{0}^{k-1},{\bf W}_{0}^{n}) = p_\vartheta (Y_k| \overline{\bf Y}_{0}^{k-1},{\bf W}_{0}^{k})$, which follows from Assumption \ref{assn_a1}. Note that
\begin{align*}
& p_\vartheta (Y_k| \overline{\bf Y}_{0}^{k-1},{\bf W}_{0}^{k},x_0) - p_\vartheta (Y_k| \overline{\bf Y}_{0}^{k-1},{\bf W}_{0}^{k}) \\
& = \sum_{(x_{k-1},x_k)\in\mathcal{X}^2} p_\vartheta(Y_k,x_k|\overline{\bf Y}_{k-1},x_{k-1},W_k) \times \left(\mathbb{P}_\vartheta(x_{k-1}|\overline{\bf Y}_{0}^{k-1},{\bf W}_{0}^{k-1},x_0)- \mathbb{P}_\vartheta(x_{k-1}|\overline{\bf Y}_{0}^{k-1},{\bf W}_{0}^{k-1}) \right), 
\end{align*}
and $\mathbb{P}_\vartheta(x_{k-1}|\overline{\bf Y}_{0}^{k-1},{\bf W}_{0}^{k-1}) = \sum_{x_0 \in \mathcal{X}}\mathbb{P}_\vartheta(x_{k-1}|\overline{\bf Y}_{0}^{k-1},{\bf W}_{0}^{k-1},x_0)\mathbb{P}_\vartheta(x_0|\overline{\bf Y}_{0}^{k-1},{\bf W}_{0}^{k-1})$. Let $\rho:=1-\sigma_-/\sigma_+\in [0,1)$. Lemma \ref{x_diff}(a) in the appendix shows that, for all probability measures $\mu_1$ and $\mu_2$ on $\mathcal{B}(\mathcal{X})$ and all $(\overline{\bf y}_{0}^{k-1},{\bf w}_{0}^{k-1})$,
\begin{equation}\label{DMR-C1}
\sup_{A}\left| \sum_{x_0\in \mathcal{X}} \mathbb{P}_\vartheta(X_{k-1}\in A|\overline{\bf y}_{0}^{k-1},{\bf w}_{0}^{k-1},x_0)\mu_1(x_0)- \sum_{x_0\in\mathcal{X}} \mathbb{P}_\vartheta(X_{k-1}\in A|\overline{\bf y}_{0}^{k-1},{\bf w}_{0}^{k-1},x_0)\mu_2(x_0) \right| \leq \rho^{k-1}.
\end{equation} 
Consequently, $ p_\vartheta (Y_k| \overline{\bf Y}_{0}^{k-1},{\bf W}_{0}^{k-1},x_0) - p_\vartheta (Y_k| \overline{\bf Y}_{0}^{k-1},{\bf W}_{0}^{k-1})$ goes to zero at an exponential rate as $k\rightarrow \infty$. Therefore, as shown in the following proposition, the difference between $\ell_n(\theta,x_0)$ and $\ell_n(\theta)$ is bounded by a deterministic constant, and the maximum of $\ell_n(\vartheta,x_0)$ and the maximum of $\ell_n(\vartheta)$ are asymptotically equivalent.
\begin{proposition}\label{uconlike}
Under Assumptions \ref{assn_a1} and \ref{assn_a2}, for all $x_0\in \mathcal{X}$,
\[
\sup_{\vartheta\in \Theta} |\ell_n(\vartheta,x_0)- \ell_n(\vartheta)| \leq 1/(1-\rho)^2 \quad \mathbb{P}_{\theta^*}\text{-a.s}. 
\]
\end{proposition} 
As discussed on p. 2263 of DMR, the stationary density $p_\vartheta (Y_k| \overline{\bf Y}_{0}^{k-1},{\bf W}_{0}^{k})$ is not available in closed form for some models with autoregression. For this reason, we consider the log-likelihood function when the initial distribution of $X_0$ follows some distribution \\$\xi_M \in \Xi_M:=\{\xi(x_0)_{x_0 \in \mathcal{X}_M} :\xi(x_0)\geq 0$ and $\sum_{x_0 \in \mathcal{X}_M}\xi(x_0)=1\}$.

Define the MLE, $\hat\vartheta_{M,\xi_M}$, by the maximizer of the log-likelihood:
\begin{equation} \label{density}
\ell_n(\vartheta_M,\xi_M):=\log \left( \sum_{x_0=1}^M p_{\vartheta_M}({\bf Y}_1^n| \overline{\bf Y}_{0},{\bf W}_{0}^n,x_0) \xi_M(x_0) \right),
\end{equation}
where $p_{\vartheta_{M}}({\bf Y}_1^n| \overline{\bf Y}_{0},{\bf W}_{0}^n,x_0)$ is given in (\ref{cond_density}). We define the \textit{number of regimes} by the smallest number $M$ such that the data density admits the representation (\ref{density}). Our objective is to test $H_0: M=M_0$ against $H_A: M=M_0+1$. 

\section{Degeneracy of the Fisher information matrix and non-identifiability under the null hypothesis} 
 
Consider testing $H_0:M=1$ against $H_A:M=2$ in a two-regime model. The null hypothesis can be written as $H_0:\theta_1^*=\theta_2^*$.\footnote{The null hypothesis of $H_0: M=1$ also holds when $(p_{11},p_{22}) = (1,0)$. We impose Assumption \ref{assn_a2}(a) to bound $p_{jj}$ away from 0 and 1 because the log-likelihood function may become unbounded as $p_{11}$ or $p_{22}$ tends to 1 in view of \citet{gassiat00esaim}.} When $\theta_1=\theta_2$, the parameter $\vartheta_{2,x}$ is not identified because $Y_k$ has the same distribution across regimes. Furthermore, Section \ref{sec: testing-1} shows that when $\theta_1=\theta_2$, the scores with respect to $\theta_1$ and $\theta_2$ are linearly dependent, so that the Fisher information matrix is degenerate.

The log-likelihood function of Markov regime switching models with normal density has further degeneracy. In a two-regime model where $Y_k$ in the $j$-th regime follows $N(\mu_j,\sigma_j^2)$, the model reduces to a heteroscedastic normal mixture model when the $X_k$'s are serially uncorrelated, i.e., $p_{11}=1-p_{22}$. \citet{kasaharashimotsu15jasa} show that in a heteroscedastic normal mixture model, the first and second derivatives of the log-likelihood function are linearly dependent and the score function is a function of the fourth-order derivative. Consequently, one needs to expand the log-likelihood function four times to derive the score function.

\section{Quadratic expansion under the loss of identifiability}\label{sec:expansion}

When testing the number of regimes by the LRT, a part of $\vartheta$ is not identified under the null hypothesis. Let $\pi$ denote the part of $\vartheta$ that is not identified under the null, and split $\vartheta$ as $\vartheta = (\psi',\pi')'$. For example, in testing $H_0: M=1$ against $H_A: M=2$, we have $\psi=\vartheta_{2,y}$ and $\pi=\vartheta_{2,x}$. We denote the conditional log-likelihood function as $\ell_n(\psi,\pi, x_0) := \ell_n(\vartheta, x_0)$ and use $p_\vartheta$ and $p_{\psi\pi}$ interchangeably. 

Denote the true parameter value of $\psi$ by $\psi^*$, and denote the set of $(\psi,\pi)$ corresponding to the null hypothesis by $\Gamma^*= \{(\psi,\pi)\in\Theta: \psi=\psi^*\}$. 
Let $t_\vartheta$ be a continuous function of $\vartheta$ such that $t_{\vartheta}=0$ if and only if $\psi=\psi^*$. For $\varepsilon>0$, define the neighborhood of $\Gamma^*$ by
\[
\mathcal{N}_{\varepsilon} := \{ \vartheta \in \Theta: |t_\vartheta|< \varepsilon\}.
\]
When the MLE is consistent, the asymptotic distribution of the LRTS is determined by the local properties of the likelihood functions in $\mathcal{N}_{\varepsilon}$.

We establish a general quadratic expansion of the log-likelihood function $\ell_n(\psi,\pi,\xi):= \ell_n(\vartheta,\xi)$ defined in (\ref{density}) around $\ell_n(\psi^*,\pi,\xi)$ that expresses $\ell_n(\psi,\pi,\xi)-\ell_n(\psi^*,\pi,\xi)$ as a quadratic function of $t_\vartheta$. Once we derive a quadratic expansion, the asymptotic distribution of the LRTS can be characterized by taking its supremum with respect to $t_\vartheta$ under an appropriate constraint and using the results of \citet{andrews99em,andrews01em}. 
 
Denote the conditional density ratio by
\begin{equation} \label{density_ratio}
l_{\vartheta k x_0} := \frac{ p_{\psi\pi} (Y_k| \overline{{\bf Y}}_{0}^{k-1},{\bf W}_{0}^k,x_0)}{ p_{\psi^*\pi} (Y_k| \overline{{\bf Y}}_{0}^{k-1},{\bf W}_{0}^k,x_0)},
\end{equation} 
so that $\ell_n(\psi,\pi, x_0) - \ell_n(\psi^*,\pi, x_0) = \sum_{k=1}^n\log l_{\vartheta k x_0}$. We assume that $l_{\vartheta k x_0}$ can be expanded around $l_{\vartheta^* k x_0}=1$ as follows. With a slight abuse of the notation, let $P_n(f_k) := n^{-1} \sum_{k=1}^n f_k$ and recall $\nu_n(f_k) := n^{-1/2} \sum_{k=1}^n [f_k - \mathbb{E}_{\vartheta^*}(f_k)]$. 
\begin{assumption}\label{assn_a3}
For all $k=1,\ldots,n$, $l_{\vartheta k x_0} -1$ admits an expansion
\begin{equation} \label{lkx0_expand}
l_{\vartheta k x_0} -1 = t_\vartheta' s_{\pi k} + r_{\vartheta k} + u_{\vartheta k x_0},
\end{equation}
where $t_\vartheta$ satisfies $\psi \to \psi^*$ if $t_\vartheta \to 0$ and $(s_{\pi k}, r_{\vartheta k}, u_{\vartheta k x_0})$ satisfy, for some $C \in (0,\infty)$, $\delta >0$, $\varepsilon >0$, and $\rho \in (0,1)$, (a) $\mathbb{E}_{\vartheta^*}\sup_{\pi \in \Theta_\pi} \left|s_{\pi k}\right|^{2+\delta} < C$, (b) $\sup_{\pi \in \Theta_\pi}| P_n(s_{\pi k}s_{\pi k}') - \mathcal{I}_\pi| = o_p(1)$ with $0<\inf_{\pi\in \Theta_{\pi} }\lambda_{\min}(\mathcal{I}_\pi) \leq \sup_{\pi\in \Theta_{\pi} }\lambda_{\max}(\mathcal{I}_\pi)<\infty$, (c) $\mathbb{E}_{\vartheta^*}[\sup_{\vartheta \in \mathcal{N}_\varepsilon} |r_{\vartheta k}/(|t_\vartheta||\psi-\psi^*|)|^2] < \infty$, (d) $\sup_{\vartheta \in \mathcal{N}_\varepsilon} [ \nu_n(r_{\vartheta k})/(|t_\vartheta||\psi-\psi^*|)] = O_p(1)$, (e) $\sup_{x_0 \in \mathcal{X}}\sup_{\vartheta \in \mathcal{N}_\varepsilon} P_n (|u_{\vartheta k x_0}|/|\psi-\psi^*|)^j = O_p(n^{-1})$ for $j=1,2,3$, (f) $\sup_{x_0 \in \mathcal{X}}\sup_{\vartheta \in \mathcal{N}_\varepsilon} P_n (|s_{\pi k}||u_{\vartheta k x_0}|/|\psi-\psi^*|) = O_p(n^{-1})$, (g) $\sup_{\vartheta \in \mathcal{N}_{\varepsilon}} |\nu_n(s_{\pi k})| =O_p(1)$. 
\end{assumption} 
We first establish an expansion of $\ell_n(\psi,\pi, x_0)$ in the neighborhood of $\mathcal{N}_{c/\sqrt{n}}$ for any $c>0$. 
\begin{proposition} \label{Ln_thm1}
Suppose that Assumptions \ref{assn_a3} (a)--(f) hold. Then, for any $c>0$,
\begin{equation*} 
 \sup_{x_0 \in \mathcal{X}} \sup_{\vartheta \in \mathcal{N}_{c/\sqrt{n}}} \left| \ell_n(\psi,\pi, x_0) - \ell_n(\psi^*,\pi, x_0) - \sqrt{n}t_\vartheta' \nu_n (s_{\pi k}) + n t_\vartheta' \mathcal{I}_\pi t_\vartheta/2 \right| = o_p(1).
\end{equation*}
\end{proposition} 
The following proposition expands $\ell_n(\psi,\pi, x_0)$ in $A_{n\varepsilon}(x_0,\eta) := \{\vartheta \in \mathcal{N}_{\varepsilon}: \ell_n(\psi,\pi, x_0)-\ell_n(\psi^*,\pi, x_0) \geq -\eta \}$ for some $\eta \in [0,\infty)$. This proposition is useful for deriving the asymptotic distribution of the LRTS because a consistent MLE is in $ A_{n\varepsilon}(x_0,\eta)$ by definition. Let $A_{n\varepsilon c}(x_0,\eta) := A_{n\varepsilon }(x_0,\eta) \cup \mathcal{N}_{c/\sqrt{n}}$.
\begin{proposition} \label{Ln_thm2}
Suppose that Assumption \ref{assn_a3} holds. Then, for any $\eta>0$, (a) $\sup_{x_0 \in \mathcal{X}} \sup_{\vartheta \in A_{n\varepsilon}(x_0,\eta)} |t_\vartheta| = O_{p \varepsilon }(n^{-1/2})$, and (b) for any $c>0$,
\begin{equation*} 
 \sup_{x_0 \in \mathcal{X}} \sup_{\vartheta \in A_{n\varepsilon c}(x_0,\eta) }\left|\ell_n(\psi,\pi, x_0)-\ell_n(\psi^*,\pi, x_0) - \sqrt{n} t_\vartheta' \nu_n (s_{\pi k}) + n t_\vartheta' \mathcal{I}_\pi t_\vartheta/2 \right| = o_{p \varepsilon }(1).
\end{equation*}
\end{proposition}

The following corollary of Propositions \ref{Ln_thm1} and \ref{Ln_thm2} shows that $\ell_n(\vartheta,\xi)$ defined in (\ref{density}) admits a similar expansion to $\ell_n(\vartheta,x_0)$ for all $\xi$. Consequently, the asymptotic distribution of the LRTS does not depend on $\xi$, and $\ell_n(\vartheta,\xi)$ may be maximized in $\vartheta$ while fixing $\xi$ or jointly in $\vartheta$ and $\xi$. Let $A_{n\varepsilon } (\xi,\eta) := \{\vartheta \in \mathcal{N}_{\varepsilon }: \ell_n(\psi,\pi,\xi) - \ell_n(\psi^*,\pi,\xi) \geq -\eta \}$ and $A_{n\varepsilon c} (\xi,\eta):= A_{n\varepsilon} (\xi,\eta)\cup \mathcal{N}_{c/\sqrt{n}}$, which includes a consistent MLE with any $\xi$.
\begin{corollary} \label{cor_appn}
(a) Under the assumptions of Proposition \ref{Ln_thm1}, we have \\
$\sup_{\xi \in \Xi}\sup_{\vartheta \in \mathcal{N}_{c/\sqrt{n}}} \left| \ell_n(\psi,\pi,\xi) - \ell_n(\psi^*,\pi,\xi) - \sqrt{n}t_\vartheta' \nu_n (s_{\pi k}) + n t_\vartheta' \mathcal{I}_\pi t_\vartheta/2 \right| = o_p(1)$ for any $c>0$.
(b) Under the assumptions of Proposition \ref{Ln_thm2}, for any $\eta>0$ and $c>0$, $\sup_{\xi\in \Xi}\sup_{\vartheta \in A_{n\varepsilon} (\xi,\eta) } |t_\vartheta| = O_{p \varepsilon }(n^{-1/2})$ and $\sup_{\xi\in \Xi}\sup_{\vartheta \in A_{n\varepsilon c} (\xi,\eta) } | \ell_n(\psi,\pi,\xi) - \ell_n(\psi^*,\pi,\xi) - \sqrt{n}t_\vartheta' \nu_n (s_{\pi k}) + n t_\vartheta' \mathcal{I}_\pi t_\vartheta/2 | = o_{p \varepsilon }(1)$.
\end{corollary}

\section{Uniform convergence of the derivatives of the log-density and density ratios} \label{sec:uniform_convergence}

In this section, we establish approximations that enable us to apply the results in Section \ref{sec:expansion} to the log-likelihood function of regime switching models. Because of the presence of singularity, the expansion (\ref{lkx0_expand}) of the density ratio $l_{\vartheta k x_0}$ involves the higher-order derivatives of the density ratios $\nabla_{\psi}^j l_{\vartheta k x_0}$ with $j\geq 2$. Note that $\nabla_\psi^j l_{\vartheta k x_0}$ can be expressed in terms of the derivatives of log-densities, $\nabla_\psi^j \log p_{\psi\pi} (Y_k| \overline{{\bf Y}}_{0}^{k-1},{\bf W}_{0}^k,x_0)$. We show that these derivatives are approximated by their stationary counterpart that condition on the infinite past $(\overline{\bf{Y}}_{-\infty}^{k-1},{\bf W}_{-\infty}^{k})$ in place of $(\overline{\bf{Y}}_{0}^{k-1},{\bf W}_{0}^{k})$. Consequently, the sequence $\{\nabla_\psi^j l_{\vartheta k x_0}\}_{k=0}^{\infty}$ is approximated by a stationary martingale difference sequence. 
 
For $1 \leq k \leq n$ and $m \geq 0$, let 
\begin{equation} \label{p_bar}
\overline p_{\vartheta}({\bf Y}_{-m +1 }^k|\overline{\bf{Y}}_{-m},{\bf W}_{-m}^{k}) := \sum_{{\bf x}_{-m}^k\in \mathcal{X}^{k+m+1}}\prod_{t=-m+1}^k p_{\vartheta}(Y_t,x_t|\overline{\bf Y}_{t-1},W_t,x_{t-1})\mathbb{P}_{\vartheta^*}(x_{-m}|\overline{\bf{Y}}_{-m},{\bf W}_{-m}^{k}),
\end{equation}
denote the stationary density of ${\bf Y}_{-m +1 }^k$ associated with $\vartheta$ conditional on $\{\overline{\bf{Y}}_{-m},{\bf W}_{-m}^{k}\}$, where $X_{-m}$ is drawn from its \emph{true} conditional stationary distribution $\mathbb{P}_{\vartheta^*}(x_{-m}|\overline{\bf{Y}}_{-m}^{k-1},{\bf W}_{-m}^{k})$. Let $\overline p_{\vartheta}(Y_k|\overline{\bf{Y}}_{-m}^{k-1},{\bf W}_{-m}^{k}) : = \overline p_{\vartheta}({\bf Y}_{-m+1}^k|\overline{\bf{Y}}_{-m},{\bf W}_{-m}^{k})/\overline p_{\vartheta}({\bf Y}_{-m+1}^{k-1}|\overline{\bf{Y}}_{-m},{\bf W}_{-m}^{k-1})$ denote the associated conditional density of $Y_k$ given $(\overline{\bf{Y}}_{-m}^{k-1},{\bf W}_{-m}^{k})$.\footnote{Note that DMR use the same notation $\overline p_{\vartheta}(\cdot|\overline{\bf{Y}}_{-m}^{k-1})$ for a different purpose. On p.\ 2263 and in some other (but not all) places, DMR use $\overline p_{\vartheta}(Y_k|\overline{\bf{Y}}_{0}^{k-1})$ to denote an (ordinary) stationary conditional distribution of $Y_k$.
}

Define the density ratio as $l_{k,m,x}(\vartheta) := p_{\vartheta}(Y_k|\overline{\bf{Y}}_{-m}^{k-1},{\bf W}_{-m}^{k},X_{-m}=x)/p_{\vartheta^*}(Y_k|\overline{\bf{Y}}_{-m}^{k-1},{\bf W}_{-m}^{k},X_{-m}=x)$. For $j=1,2,\ldots,6$, $1 \leq k \leq n$, $m \geq 0$, and $x \in \mathcal{X}$, define the derivatives of the log-densities and density ratios by, with suppressing the subscript $\vartheta$ from $\nabla_\vartheta^j$ for brevity,
\begin{align*} 
& \nabla^j \ell_{k,m,x}(\vartheta) := \nabla^j \log p_{\vartheta}(Y_k|\overline{\bf{Y}}_{-m}^{k-1},{\bf W}_{-m}^{k},X_{-m}=x),\quad \nabla^j l_{k,m,x}(\vartheta) := \frac{\nabla^j p_{\vartheta}(Y_k|\overline{\bf{Y}}_{-m}^{k-1},{\bf W}_{-m}^{k},X_{-m}=x)}{ p_{\vartheta^*}(Y_k|\overline{\bf{Y}}_{-m}^{k-1},{\bf W}_{-m}^{k},X_{-m}=x)},\\
& \nabla^j \overline \ell_{k,m}(\vartheta):= \nabla^j \log \overline p_{\vartheta}(Y_k|\overline{\bf{Y}}_{-m}^{k-1},{\bf W}_{-m}^{k}), \quad\text{and}\quad \nabla^j \overline l_{k,m}(\vartheta):= \frac{ \nabla^j \overline p_{\vartheta}(Y_k|\overline{\bf{Y}}_{-m}^{k-1},{\bf W}_{-m}^{k})}{\overline p_{\vartheta^*}(Y_k|\overline{\bf{Y}}_{-m}^{k-1},{\bf W}_{-m}^{k})}.
\end{align*} 
The following assumption corresponds to (A6)--(A8) in DMR and is tailored to our setting where some elements of $\vartheta_x^*$ are not identified and $\mathcal{X}$ is finite. Note that Assumptions (A6) and (A7) in DMR pertaining to $q_{\vartheta_x}(x,x')$ hold in our case because the $p_{ij}$'s are bounded away from 0 and 1. Let $G_{\vartheta k} := \sum_{x_k\in\mathcal{X}}g_{\vartheta_y}(Y_k|\overline{\bf{Y}}_{k-1}, x_k, W_k)$. $G_{\vartheta k}$ satisfies Assumption \ref{assn_a4}(b) in general when $\mathcal{N}^*$ is sufficiently small.
\begin{assumption}\label{assn_a4}
There exists a positive real $\delta$ such that on $\mathcal{N}^*:= \{\vartheta\in\Theta: |\vartheta_y-\vartheta_y^*| < \delta\}$ the following conditions hold: (a) For all $(\overline{\bf{y}},y',x,w)\in \mathcal{Y}^s \times \mathcal{Y} \times \mathcal{X} \times \mathcal{W}$, $g_{\vartheta_y}(y'|\overline{\bf{y}},x,w)$ is six times continuously differentiable on $\mathcal{N}^*$. (b) $\mathbb{E}_{\vartheta^*}[\sup_{\vartheta\in \mathcal{N}^*} \sup_{x\in \mathcal{X}} | \nabla^j\log g_{\vartheta_y}(Y_1|\overline{\bf{Y}}_0,x,W) |^{2q_j}]< \infty$ for $j=1,2,\ldots,6$ and $\mathbb{E}_{\vartheta^*}\sup_{\vartheta \in \mathcal{N}^*} |G_{\vartheta k}/G_{\vartheta^* k}|^{q_g} < \infty$ with $q_1=6q_0, q_2=5q_0, \ldots, q_6=q_0$, where $q_0=(1+\varepsilon) q_\vartheta$ and $q_g=(1+\varepsilon)q_\vartheta/\varepsilon$ for some $\varepsilon>0$ and $q_\vartheta >\max\{3,\text{dim}(\vartheta)\}$. (c) For almost all $(\overline{\bf{y}},y',w) \in \mathcal{Y}^s \times \mathcal{Y}\times \mathcal{W}$, $\sup_{\vartheta\in \mathcal{N}^*} g_{\vartheta_y}(y'|\overline{\bf{y}},x,w) < \infty$ and, for almost all $(\overline{\bf{y}},x,w)\in \mathcal{Y}^s\times \mathcal{X} \times \mathcal{W}$, for $j=1,2,\ldots,6$, there exist functions $f^j_{\overline{\bf{y}},w,x}: \mathcal{Y} \rightarrow \mathbb{R}^+$ in $L^1$ such that $|\nabla^j g_{\vartheta_y}(y'|\overline{\bf{y}},x,w) |\leq f^j_{\overline{\bf{y}}, x, w}(y')$ for all $\vartheta\in \mathcal{N}^*$.
\end{assumption} 

The following proposition shows that $\{\nabla^j \ell_{k,m,x}(\vartheta)\}_{m\geq 0}$ and $\{\nabla^j \overline \ell_{k,m}(\vartheta)\}_{m\geq 0}$ are $L^{r_j}(\mathbb{P}_{\vartheta^*})$-Cauchy sequences that converge to $\nabla^j \ell_{k,\infty}(\vartheta)$ $\mathbb{P}_{\vartheta^*}$-a.s.\ and in $L^{r_j}(\mathbb{P}_{\vartheta^*})$ uniformly in $\vartheta\in \mathcal{N}^*$ and $x \in \mathcal{X}$. 
\begin{proposition} \label{proposition-ell}
Under Assumptions \ref{assn_a1}, \ref{assn_a2}, and \ref{assn_a4}, for $j=1,\ldots,6$, there exist random variables $(K_j, \{M_{j,k} \}_{k=1}^n )\in L^{r_j}(\mathbb{P}_{\vartheta^*})$ and $\rho_* \in (0,1)$ such that, for all $1 \leq k \leq n$ and $m' \geq m \geq 0$,
\begin{align*}
(a) &\quad \sup_{x\in\mathcal{X}}\sup_{\vartheta\in \mathcal{N}^*}|\nabla^j\ell_{k,m,x}(\vartheta) - \nabla^j\overline\ell_{k,m}(\vartheta) | \leq K_j (k+m)^7\rho_*^{k+m-1}\quad \text{$\mathbb{P}_{\vartheta^*}$-a.s.},\\
(b) &\quad \sup_{x\in\mathcal{X}}\sup_{\vartheta\in \mathcal{N}^*}|\nabla^j\ell_{k,m,x}(\vartheta) - \nabla^j \ell_{k,m',x}(\vartheta) | \leq K_j(k+m)^7\rho_*^{k+m-1}\quad \text{$\mathbb{P}_{\vartheta^*}$-a.s.},\\
(c) & \quad \sup_{m \geq 0}\sup_{x\in\mathcal{X}}\sup_{\vartheta\in \mathcal{N}^*}|\nabla^j\ell_{k,m,x}(\vartheta)| + \sup_{m \geq 0}\sup_{\vartheta\in \mathcal{N}^*}|\nabla^j\overline\ell_{k,m}(\vartheta)|\leq M_{j, k} \quad \text{$\mathbb{P}_{\vartheta^*}$-a.s.},
\end{align*}
where $r_1 = 6q_0$, $r_2 = 3q_0$, $r_3 = 2q_0$, $r_4 =3q_0/2$, $r_5 =6q_0/5$, and $r_6 =q_0$. (d) Uniformly in $\vartheta\in \mathcal{N}^*$ and $x \in \mathcal{X}$, $\nabla^j\ell_{k,m,x}(\vartheta)$ and $\nabla^j\overline\ell_{k,m}(\vartheta)$ converge $\mathbb{P}_{\vartheta^*}$-a.s.\ and in $L^{r_j}(\mathbb{P}_{\vartheta^*})$ to $\nabla^j \ell_{k,\infty}(\vartheta) \in L^{r_j}(\mathbb{P}_{\vartheta^*})$ as $m\rightarrow \infty$.
\end{proposition} 
As shown by the following proposition, we may prove the uniform convergence of the derivatives of the density ratios by expressing them as polynomials of the derivatives of the log-density and applying Proposition \ref{proposition-ell} and H\"older's inequality.
\begin{proposition}\label{lemma-omega}
Under Assumptions \ref{assn_a1}, \ref{assn_a2}, and \ref{assn_a4}, for $j=1,\ldots,6$, there exist random variables $\{K_{j,k}\}_{k=1}^n \in L^{q_\vartheta}(\mathbb{P}_{\vartheta^*})$ and $\rho_* \in (0,1)$ such that, for all $1 \leq k \leq n$ and $m' \geq m \geq 0$,
\begin{align*}
(a) &\quad \sup_{x\in\mathcal{X}}\sup_{\vartheta\in \mathcal{N}^*}|\nabla^j l_{k,m,x}(\vartheta) - \nabla^j \overline l_{k,m}(\vartheta) | \leq K_{j,k} (k+m)^{7}\rho_*^{k+m-1} \quad \mathbb{P}_{\vartheta^*}\text{-a.s.},\\
(b) &\quad \sup_{x\in\mathcal{X}}\sup_{\vartheta\in \mathcal{N}^*}|\nabla^j l_{k,m,x}(\vartheta) - \nabla^j l_{k,m',x}(\vartheta) | \leq K_{j,k} (k+m)^{7}\rho_*^{k+m-1} \quad \text{$\mathbb{P}_{\vartheta^*}$-a.s.}, \\ 
(c) &\quad \sup_{m\geq 0}\sup_{x\in\mathcal{X}}\sup_{\vartheta\in \mathcal{N}^*}|\nabla^j l_{k,m,x}(\vartheta)| + \sup_{m\geq 0}\sup_{\vartheta\in \mathcal{N}^*}|\nabla^j\overline{l}_{k,m}(\vartheta)| \leq K_{j,k} \quad \mathbb{P}_{\vartheta^*}\text{-a.s.}
\end{align*}
(d) Uniformly in $\vartheta\in \mathcal{N}^*$ and $x \in \mathcal{X}$, $\nabla^j l_{k,m,x}(\vartheta)$ and $\nabla^j\overline l_{k,m}(\vartheta)$ converge $\mathbb{P}_{\vartheta^*}$-a.s.\ and in $L^{q_\vartheta}(\mathbb{P}_{\vartheta^*})$ to $\nabla^j l_{k,\infty}(\vartheta) \in L^{q_\vartheta}(\mathbb{P}_{\vartheta^*})$ as $m\rightarrow \infty$. (e) $\sup_{\vartheta\in \mathcal{N}^*}|\nabla^j \overline l_{k,0}(\vartheta) - \nabla^j l_{k,\infty}(\vartheta) | \leq K_{j,k} k^{7}\rho_*^{k-1}$ $\mathbb{P}_{\vartheta^*}$-a.s. 
\end{proposition}
When we apply the results in Section \ref{sec:expansion} to regime switching models, $l_{k,0,x}(\vartheta)$ corresponds to $l_{\vartheta k x_0}$ on the left-hand side of (\ref{lkx0_expand}), and $s_{\pi k}$ in (\ref{lkx0_expand}) is a function of the $\nabla^j \overline l_{k, 0}(\vartheta)$'s. Proposition \ref{lemma-omega} and the dominated convergence theorem for conditional expectations \citep[][Theorem 5.5.9]{durrett10book} imply that $\mathbb{E}_{\vartheta^*}[\nabla^j l_{k,\infty}(\vartheta)|\overline{\bf{Y}}_{-\infty}^{k-1}]=0$ for all $\vartheta\in \mathcal{N}^*$. Therefore, $\{\nabla^j l_{k,\infty}(\vartheta)\}_{k=-\infty}^{\infty}$ is a stationary, ergodic, and square integrable martingale difference sequence, and $\{\nabla^j l_{k,\infty}(\vartheta)\}_{j=1}^5$ satisfies Assumption \ref{assn_a3}(a)(b)(g).

\section{Testing homogeneity} \label{sec: testing-1}

Before developing the LRT of $M_0$ components, we analyze a simpler case of testing the null hypothesis $H_0: M=1$ against $H_A:M=2$ when the data are from $H_0$. We assume that the parameter space for $\vartheta_{2,x}=(p_{11},p_{22})'$ takes the form $[\epsilon, 1-\epsilon]^2$ for a small $\epsilon \in (0,1/2)$. Denote the true parameter in the one-regime model by $\vartheta_1^*:= ((\theta^*)',(\gamma^*)')'$. The two-regime model gives rise to the true density $p_{\vartheta_1^*}({\bf Y}_1^n| \overline{{\bf Y}}_{0}, x_0)$ if the parameter $\vartheta_2=(\theta_1,\theta_2,\gamma,p_{11},p_{22})'$ lies in a subset of the parameter space
\[
\Gamma^* := \left\{ (\theta_1,\theta_2,\gamma,p_{11},p_{22})\in \Theta_{ 2}: \theta_1=\theta_2=\theta^*\ \text{and}\ \gamma=\gamma^* \right\}.
\]
Note that $(p_{11},p_{22})$ is not identified under $H_0$.

Let $\ell_n(\vartheta_2,\xi_2) := \log \left( \sum_{x_0=1}^2 p_{\vartheta_2}({\bf Y}_1^n| \overline{\bf Y}_{0},{\bf W}_{0}^n,x_0) \xi_2(x_0) \right)$ denote the two-regime log-likelihood for a given initial distribution $\xi_2(x_0) \in \Xi_2$, and let $\hat\vartheta_2:= \arg\max_{\vartheta_2 \in \Theta_{2}} \ell_n(\vartheta_2,\xi_2)$ denote the MLE of $\vartheta_2$ given $\xi_2$, where $\xi_2$ is suppressed from $\hat \vartheta_2$ because $\xi_2$ does not matter asymptotically. Let $\hat \vartheta_1$ denote the one-regime MLE that maximizes the one-regime log-likelihood function $\ell_{0,n}(\vartheta_1) := \sum_{k=1}^n \log f (Y_k|\overline{\bf Y}_{k-1},W_k;\gamma,\theta)$ under the constraint $\vartheta_1 = (\theta',\gamma')' \in \Theta_1$. 

We introduce the following assumption for the consistency of $\hat \vartheta_1$ and $\hat \vartheta_2$. Assumption \ref{A-consist}(b) corresponds to Assumption (A4) of DMR. Assumption \ref{A-consist}(c) is a standard identification condition for the one-regime model. Assumption \ref{A-consist}(d) implies that the Kullback--Leibler divergence between $ p_{\vartheta_{1}^*}(Y_1 |\overline{\bf{Y}}_{-m}^0,{\bf W}_{-m}^0)$ and $ p_{\vartheta_{2}}(Y_1 |\overline{\bf{Y}}_{-m}^0,{\bf W}_{-m}^0)$ is 0 if and only if $\vartheta_2 \in \Gamma^*$.
\begin{assumption} \label{A-consist} 
(a) $\Theta_1$ and $\Theta_2$ are compact, and $\vartheta_1^*$ is in the interior of $\Theta_1$. (b) For all $(x,x') \in \mathcal{X}$ and all $(\overline{{\bf y}},y',w)\in \mathcal{Y}^s\times \mathcal{Y}\times \mathcal{W}$, $f(y'|\overline{\bf{y}}_0,w;\gamma,\theta)$ is continuous in $(\gamma,\theta)$. (c) If $\vartheta_{1}\neq \vartheta_1^*$, then $\mathbb{P}_{\vartheta^*_{1}}\left(f(Y_1|\overline{\bf{Y}}_0,W_1;\gamma,\theta) \neq f(Y_1|\overline{\bf{Y}}_0,W_1;\gamma^*,\theta^*) \right)>0$. (d) $\mathbb{E}_{\vartheta^*_{1}}[ \log p_{\vartheta_{2}}(Y_1 |\overline{\bf{Y}}_{-m}^0,{\bf W}_{-m}^1)] = \mathbb{E}_{\vartheta^*_{1}}[\log p_{\vartheta_{1}^*}(Y_1 |\overline{\bf{Y}}_{-m}^0,{\bf W}_{-m}^1) ]$ for all $m \geq 0$ if and only if $\vartheta_{2}\in \Gamma^*$. 
\end{assumption} 
The following proposition shows the consistency of the MLEs of $\vartheta_1^*$ and $\vartheta_{2,y}^*$. 
\begin{proposition} \label{P-consist} Suppose that Assumptions \ref{assn_a1}, \ref{assn_a2}, and \ref{A-consist} hold. Then, under the null hypothesis of $M=1$, $\hat\vartheta_1 \overset{p}{\rightarrow} \vartheta_1^*$ and $\inf_{\vartheta_2\in\Gamma^*}|\hat\vartheta_2-\vartheta_2|\overset{p}{\rightarrow} 0$.
\end{proposition}
 
Let $LR_n:=2[\ell_n(\hat\vartheta_2,\xi_2) - \ell_{0,n}(\hat\vartheta_1)]$ denote the LRTS for testing $H_0:M_0=1$ against $H_A:M_0=2$. Following the notation of Section \ref{sec:expansion}, we split $\vartheta_2$ into $\vartheta_2=(\psi,\pi)$, where $\pi$ is the part of $\vartheta$ not identified under the null hypothesis; the elements of $\psi$ are delineated later. In the current setting, $\pi$ corresponds to $\vartheta_{2,x}=(p_{11},p_{22})'$. Define $\varrho := \text{corr}_{\vartheta_{2,x}} (X_k,X_{k+1}) = p_{11}+p_{22}-1$ and $\alpha := \mathbb{P}_{\vartheta_{2,x}}(X_k=1)= (1-p_{22})/(2-p_{11}-p_{22})$. The parameter spaces for $\varrho$ and $\alpha$ under restriction $p_{11},p_{22} \in [\epsilon, 1-\epsilon]$ are given by $\Theta_{\varrho}:= [-1+2\epsilon,1-2\epsilon]$ and $\Theta_{\alpha}:= [\epsilon, 1-\epsilon]$, respectively. Because the mapping from $(p_{11},p_{22})$ to $(\varrho,\alpha)$ is one-to-one, we reparameterize $\pi$ as $\pi := (\varrho,\alpha)' \in \Theta_{\pi }:=\Theta_{\varrho }\times \Theta_{\alpha }$, and let $p_{\psi\pi}(\cdot|\cdot) := p_{\vartheta_2} (\cdot|\cdot)$. Henceforth, we suppress ${\bf W}_{0}^n$ for notational brevity and write, for example, $p_{\psi\pi}({\bf Y}_1^n| \overline{{\bf Y}}_{0}, {\bf W}_{0}^n,x_0)$ as $p_{\psi\pi}({\bf Y}_1^n|\overline{{\bf Y}}_{0},x_0)$ and $p_{\psi\pi} (y_k,x_k| \overline{\bf{y}}_{k-1},x_{k-1},w_k)$ as $p_{\psi\pi} (y_k,x_k| \overline{\bf{y}}_{k-1},x_{k-1})$ when doing so does not cause confusion. 

We derive the asymptotic distribution of the LRTS by applying Corollary \ref{cor_appn} to $\ell_n(\psi,\pi,\xi_2)$ and representing $s_{\pi k}$ and $t_\vartheta$ in (\ref{lkx0_expand}) in terms of $\vartheta$, $f(Y_k|\overline{\bf{Y}}_{k-1};\gamma,\theta)$, and its derivatives; $s_{\pi k}$ involves higher-order derivatives, and $t_\vartheta$ consists of the functions of the polynomials of (reparameterized) $\vartheta$. Section \ref{subsec:nonnormal} analyzes the case when the regime-specific distribution of $y_k$ is not normal with unknown variance. Section \ref{subsec:hetero_normal} analyzes the case when the regime-specific distribution $y_k$ is normal with regime-specific and unknown variance. Section \ref{subsec:homo_normal} handles the normal distribution where the variance is unknown and common across regimes. Note that because $\overline{\bf Y}_{-\infty}^\infty$ and ${\bf X}_{-\infty}^\infty$ are independent when $\psi = \psi^*$, we have 
\begin{equation}\label{indep}
\mathbb{P}_{\psi^*\pi}({\bf X}_{-\infty}^\infty|\overline{\bf Y}_{-\infty}^\infty) = \mathbb{P}_{\psi^*\pi}({\bf X}_{-\infty}^\infty).
\end{equation}
Define $q_k := \mathbb{I}\{X_k=1\}$, so that $\alpha = \mathbb{E}_{\psi^*\pi}[q_k]$.

\subsection{Non-normal distribution}\label{subsec:nonnormal}
 
When we apply Corollary \ref{cor_appn} to regime switching models, $s_{\pi k}$ is a function of \\$\nabla^j\overline p_{\psi^*\pi} (Y_k| \overline{{\bf Y}}_{0}^{k-1})/\overline p_{\psi^*\pi} (Y_k| \overline{{\bf Y}}_{0}^{k-1})$'s with $ \overline p_{\psi\pi} (Y^k_{1}| \overline{{\bf Y}}_{0})$ defined in (\ref{p_bar}). In order to express $\nabla^j\overline p_{\psi^* \pi} (Y_k| \overline{{\bf Y}}_{0}^{k-1}) / \overline p_{\psi^* \pi} (Y_k| \overline{{\bf Y}}_{0}^{k-1})$ in terms of $\nabla^j f(y|x;\gamma,\theta)$ via the Louis information principle (Lemma \ref{louis} in the appendix), we first derive the derivatives of the \emph{complete data} conditional density $p_{\vartheta_2} (y_k,x_k| \overline{\bf{y}}_{k-1},x_{k-1}) = g_{\vartheta_{2,y}}(y_k|\overline{\bf{y}}_{k-1}, x_k) q_{\vartheta_{2,x}}(x_{k-1},x_k)=\sum_{j =1}^2 \mathbb{I}\{x_k=j\}f(y_k|\overline{\bf{y}}_{k-1};\gamma,\theta_j)q_{\vartheta_{2,x}}(x_{k-1},x_k)$.

Consider the following reparameterization. Let
\begin{equation}
\left(
\begin{array}{c}
\lambda \\
\nu \\
\end{array}
\right):=
\left(
\begin{array}{c}
\theta_1 - \theta_2 \\
\alpha\theta_1 + (1-\alpha)\theta_2\\
\end{array}
\right),
\quad \text{so that} \quad
\left(
\begin{array}{c}
\theta_1\\
\theta_2 \\
\end{array}
\right)=
\left(
\begin{array}{c}
\nu+ (1-\alpha) \lambda\\
\nu- \alpha\lambda\\
\end{array}
\right). \label{repara}
\end{equation}
Let $\eta := (\gamma',\nu')'$ and $\psi_{\alpha}:=(\eta',\lambda')' \in \Theta_\eta \times \Theta_\lambda$. Under the null hypothesis of one regime, the true value of $\psi_{\alpha}$ is given by $\psi_{\alpha}^*:=(\gamma^*,\theta^*,0)'$. Henceforth, we suppress the subscript $\alpha$ from $\psi_{\alpha}$. Using this definition of $\psi$, let $\vartheta_2 := (\psi',\pi')' \in \Theta_\psi \times \Theta_\pi$. By using reparameterization (\ref{repara}) and noting that $q_k = \mathbb{I}\{x_k=1\}$, we have $p_{\psi\pi} (y_k,x_k| \overline{\bf{y}}_{k-1},x_{k-1}) = g_{\psi}(y_k|\overline{\bf{y}}_{k-1}, x_k) q_{\pi}(x_{k-1},x_k)$ and 
\begin{equation} \label{repara_g}
g_{\psi}(y_k|\overline{\bf{y}}_{k-1},x_k) = f(y_k|\overline{\bf{y}}_{k-1};\gamma,\nu+(q_k-\alpha)\lambda ).
\end{equation}
Henceforth, let $f^*_{k}$, $\nabla f^*_{k}$, $g^*_{k}$, and $\nabla g^*_{k}$ denote $f(Y_k|\overline{\bf{Y}}_{k-1};\gamma^*,\theta^*)$, $\nabla f(Y_k|\overline{\bf{Y}}_{k-1};\gamma^*,\theta^*)$, $g_{\psi^*}(Y_k|\overline{\bf{Y}}_{k-1},X_k) $, and $\nabla g_{\psi^*}(Y_k|\overline{\bf{Y}}_{k-1},X_k)$, respectively, and similarly for $\log f^*_{k}$ ,$\nabla \log f^*_{k}$, $\log g^*_{k}$, and $\nabla \log g^*_{k}$. 
Expanding $g_{\psi}(Y_k|\overline{\bf{Y}}_{k-1},X_k)$ twice with respect to $\psi=(\gamma',\nu',\lambda')'$ and evaluating at $\psi^*$ gives
\begin{equation} \label{dg}
\begin{aligned}
\nabla_{\eta} g_k^* &= \nabla_{(\gamma',\theta')'} f_k^*, \quad \nabla_{\lambda} g_k^* = (q_k -\alpha)\nabla_{\theta} f_k^*, \\
\nabla_{\lambda\eta'} g_k^* &= (q_k -\alpha)\nabla_{\theta(\gamma',\theta')} f_k^*, \quad \nabla_{\lambda\lambda'} g_k^* = (q_k-\alpha)^2 \nabla_{\theta\theta'} f_k^*.
\end{aligned} 
\end{equation}
Recall $\varrho := \text{corr}_{\vartheta_2^*}(q_k,q_{k+1})$. Observe that $q_k$ satisfies
\begin{equation} \label{markov_moments}
\begin{aligned}
\mathbb{E}_{\vartheta_2^*}(q_k-\alpha)^2 &= \alpha(1-\alpha), \quad \mathbb{E}_{\vartheta_2^*}(q_k-\alpha)^3 = \alpha(1-\alpha)(1-2\alpha), \\
\mathbb{E}_{\vartheta_2^*}(q_k-\alpha)^4 &= \alpha(1-\alpha)(3\alpha^2-3\alpha+1), \quad \text{corr}_{\vartheta_2^*}(q_k,q_{k+\ell}) = \varrho^{|\ell|}, 
\end{aligned}
\end{equation}
where the first three results follow from the property of a Bernoulli random variable, and the last result follows from \citet[][p. 684]{hamilton94book}. Then, it follows from (\ref{indep}) and (\ref{markov_moments}) that
\begin{equation} \label{qk_moments}
\begin{aligned}
&\mathbb{E}_{\vartheta^*}[q_k -\alpha |\overline{{\bf Y}}_{-\infty}^n] = 0,\quad \mathbb{E}_{\vartheta^*}[(q_{t_1}-\alpha) (q_{t_2}-\alpha) | \overline{\bf Y}_{-\infty}^n ] = \alpha(1-\alpha)\varrho^{t_2-t_1}, \quad t_2 \geq t_1.
\end{aligned}
\end{equation} 
From Lemma \ref{louis}, $\log p_{\psi\pi}(y_k, x_k|\overline{{\bf y}}_{k-1},x_{k-1}) = \log g_{\psi}(y_k|\overline{\bf{y}}_{k-1},x_k) + \log q_\pi(x_{k-1},x_k)$, and the definition of $\overline p_{\psi \pi} ( Y_1^k| \overline{{\bf Y}}_{0} )$ in (\ref{p_bar}), we obtain
\[
\begin{aligned}
&\frac{\nabla_\psi \overline p_{\psi^* \pi} (Y_k| \overline{{\bf Y}}_{0}^{k-1})}{\overline p_{\psi^* \pi} (Y_k| \overline{{\bf Y}}_{0}^{k-1})} =\nabla_\psi \log \overline p_{\psi^* \pi} (Y_k| \overline{{\bf Y}}_{0}^{k-1}) = \sum_{t=1}^k \mathbb{E}_{\vartheta^*} \left[\nabla_\psi \log g^*_{t} \middle|\overline{{\bf Y}}_{0}^k \right] - \sum_{t=1}^{k-1} \mathbb{E}_{\vartheta^*} \left[\nabla_\psi \log g^*_{t} \middle|\overline{{\bf Y}}_{0}^{k-1} \right].
\end{aligned}
\]
Applying (\ref{dg}), (\ref{qk_moments}), and $g_k^* =f_k^*$ to the right-hand side gives
\begin{equation} \label{d1p}
\begin{aligned}
\frac{\nabla_\eta \overline p_{\psi^* \pi} (Y_k| \overline{{\bf Y}}_{0}^{k-1})}{\overline p_{\psi^* \pi} (Y_k| \overline{{\bf Y}}_{0}^{k-1})} =\nabla_{(\gamma',\theta')'} \log f_k^*, \quad
\frac{\nabla_\lambda \overline p_{\psi^* \pi} (Y_k| \overline{{\bf Y}}_{0}^{k-1})}{ \overline p_{\psi^* \pi} (Y_k| \overline{{\bf Y}}_{0}^{k-1})}=0.
\end{aligned}
\end{equation}
Similarly, it follows from Lemma \ref{louis}, (\ref{dg}), (\ref{qk_moments}), (\ref{d1p}), and $g_k^* =f_k^*$ that
\begin{align}
& \nabla_{\lambda\eta'} \overline p_{\psi^* \pi} (Y_k| \overline{{\bf Y}}_{0}^{k-1})/\overline p_{\psi^* \pi} (Y_k| \overline{{\bf Y}}_{0}^{k-1}) =0, \label{d2p0}\\
& \nabla_{\lambda\lambda'} \overline p_{\psi^* \pi} (Y_k| \overline{{\bf Y}}_{0}^{k-1})/\overline p_{\psi^* \pi} (Y_k| \overline{{\bf Y}}_{0}^{k-1}) \nonumber \\
& = \nabla_{\lambda\lambda'} \log \overline p_{\psi^* \pi} (Y_k| \overline{{\bf Y}}_{0}^{k-1}) \nonumber \\
& = \sum_{t=1}^k \mathbb{E}_{\vartheta^*} \left[\nabla_{\lambda\lambda'} \log g^*_{t} \middle|\overline{{\bf Y}}_{0}^k \right] - \sum_{t=1}^{k-1} \mathbb{E}_{\vartheta^*} \left[\nabla_{\lambda\lambda'} \log g^*_{t} \middle|\overline{{\bf Y}}_{0}^{k-1} \right] \nonumber \\
& \quad + \sum_{t_1=1}^k \sum_{t_2=1}^k \mathbb{E}_{\vartheta^*} \left[ \frac{\nabla_{\lambda} g^*_{t_1}}{g^*_{t_1}} \frac{\nabla_{\lambda'} g^*_{t_2}}{g^*_{t_2}} \middle| \overline{{\bf Y}}_{0}^k \right] - \sum_{t_1=1}^{k-1} \sum_{t_2=1}^{k-1} \mathbb{E}_{\vartheta^*} \left[ \frac{\nabla_{\lambda} g^*_{t_1}}{g^*_{t_1}} \frac{\nabla_{\lambda'} g^*_{t_2}}{g^*_{t_2}} \middle| \overline{{\bf Y}}_{0}^{k-1} \right] \nonumber \\
& = \alpha(1-\alpha) \left[ \frac{ \nabla_{\theta\theta'}f_k^*}{f_k^*} + \sum_{t=1}^{k-1} \varrho^{k-t} \left( \frac{\nabla_{\theta} f_{t}^*}{f_{t}^*} \frac{\nabla_{\theta'} f_{k}^*}{f_{k}^*} + \frac{\nabla_{\theta} f_{k}^*}{f_{k}^*} \frac{\nabla_{\theta'} f_{t}^*}{f_{t}^*} \right)\right].\label{d2p}
\end{align} 
Because the first-order derivative with respect to $\lambda$ is identically equal to zero in (\ref{d1p}), 
the unique elements of $\nabla_{\eta} \overline p_{\psi^* \pi} (Y_k| \overline{{\bf Y}}_{0}^{k-1}) /\overline p_{\psi^* \pi} (Y_k| \overline{{\bf Y}}_{0}^{k-1}) $ and $\nabla_{\lambda\lambda'} \overline p_{\psi^* \pi} (Y_k| \overline{{\bf Y}}_{0}^{k-1})/\overline p_{\psi^* \pi} (Y_k| \overline{{\bf Y}}_{0}^{k-1})$ constitute the generalized score $s_{\pi k}$ in Corollary \ref{cor_appn}. This score is approximated by a stationary martingale difference sequence, where the approximation error satisfies Assumption \ref{assn_a3}.

We collect some notations. Recall $\psi = (\eta',\lambda')'$ and $\eta = (\gamma',\nu')'$. For a $q \times 1$ vector $\lambda$ and a $q \times q$ matrix $s$, define the $q_\lambda \times 1$ vectors $v(\lambda)$ and $V(s)$ as
\begin{equation} \label{v_lambda}
\begin{aligned}
v(\lambda) &:= ( \lambda_1^2,\ldots,\lambda_q^2,\lambda_1\lambda_2,\ldots,\lambda_1\lambda_q,\lambda_2 \lambda_3, \ldots, \lambda_2 \lambda_q,\ldots,\lambda_{q-1}\lambda_q)',\\
V(s) &:= (s_{11}/2,\ldots,s_{qq}/2,s_{12},\ldots,s_{1q},s_{23},\ldots,s_{2q},\ldots,s_{q,q-1})'.
\end{aligned}
\end{equation}
Noting that $\alpha(1-\alpha)>0$ for $\alpha\in\Theta_{\alpha}$, define, with $t_{\lambda}(\lambda,\pi):=\alpha(1-\alpha)v(\lambda)$, 
\begin{equation}\label{score}
\begin{aligned}
t(\psi,\pi) &:= 
\begin{pmatrix}
\eta - \eta^* \\
t_{\lambda}(\lambda,\pi)
\end{pmatrix}, \ 
s_{ \varrho k} : = \begin{pmatrix}
s_{\eta k} \\
s_{\lambda \varrho k}
\end{pmatrix}, \ \text{where} \
s_{\eta k} : = \frac{\nabla_\eta \overline p_{\psi^* \pi} (Y_k| \overline{{\bf Y}}_{0}^{k-1})}{\overline p_{\psi^* \pi} (Y_k| \overline{{\bf Y}}_{0}^{k-1})} = \begin{pmatrix}
\nabla_{\gamma} f_k^* / f_k^* \\
\nabla_{\theta} f_k^* / f_k^*
\end{pmatrix},
\end{aligned}
\end{equation}
and $s_{\lambda \varrho k} := V(s_{\lambda\lambda \varrho k})$ with 
\begin{equation} \label{score_lambda}
s_{\lambda\lambda \varrho k} := \frac{\nabla_{\lambda\lambda'} \overline p_{\psi^* \pi} (Y_k| \overline{{\bf Y}}_{0}^{k-1})}{\alpha(1-\alpha) \overline p_{\psi^* \pi} (Y_k| \overline{{\bf Y}}_{0}^{k-1})} = \frac{\nabla_{\theta \theta'} f_k^*}{f_k^*} + \sum_{t= 1}^{k-1} \varrho^{k-t} \left( \frac{\nabla_{\theta} f_{t}^*}{f_{t}^*} \frac{\nabla_{\theta'} f_{k}^*}{f_{k}^*} + \frac{\nabla_{\theta} f_{k}^*}{f_{k}^*} \frac{\nabla_{\theta'} f_{t}^*}{f_{t}^*} \right).
\end{equation}
Here, $s_{\varrho k}$ in (\ref{score}) depends on $\varrho$ but not on $\alpha$ and corresponds to $s_{\pi k}$ in Corollary \ref{cor_appn}.
The following proposition shows that the log-likelihood function is approximated by a quadratic function of $\sqrt{n}t(\psi,\pi)$. Let $\mathcal{N}_{\varepsilon } := \{ \vartheta_2 \in \Theta_{2 }: |t(\psi,\pi)|< \varepsilon \}$. Let $A_{n \varepsilon }(\xi) := \{\vartheta \in \mathcal{N}_{\varepsilon}: \ell_n(\psi,\pi,\xi) - \ell_n(\psi^*,\pi,\xi) \geq 0\} $ and $A_{n\varepsilon c}(\xi) := A_{n\varepsilon }(\xi)\cup \mathcal{N}_{c/\sqrt{n}}$, where we suppress the subscript $2$ from $\xi_2$.
We use this definition of $A_{n\varepsilon c}(\xi)$ through Sections \ref{subsec:nonnormal}--\ref{subsec:homo_normal}. As shown in Sections \ref{subsec:hetero_normal} and \ref{subsec:homo_normal}, Assumption \ref{A-nonsing1} does not hold for regime switching models with a normal distribution. 
\begin{assumption} \label{A-nonsing1} 
$0< \inf_{\varrho\in \Theta_{\varrho}} \lambda_{\min}(\mathcal{I}_{\varrho}) \leq \sup_{\varrho\in\Theta_{\varrho}} \lambda_{\max}(\mathcal{I}_\varrho) < \infty$
 for $\mathcal{I}_{\varrho}= \lim_{k\rightarrow \infty} \mathbb{E}_{\vartheta^*}(s_{\varrho k}s_{\varrho k}')$, where $s_{\varrho k}$ is given in (\ref{score}).
\end{assumption} 
\begin{proposition} \label{P-quadratic} Suppose Assumptions \ref{assn_a1}, \ref{assn_a2}, \ref{assn_a4}, \ref{A-consist}, and \ref{A-nonsing1} hold. Then, under the null hypothesis of $M=1$, (a) $\sup_{\xi}\sup_{\vartheta \in A_{n \varepsilon }(\xi)} |t(\psi,\pi)| = O_{p \varepsilon }(n^{-1/2})$; and (b) for any $c>0$,
\begin{equation} \label{ln_appn}
\sup_{\xi \in \Xi}\sup_{\vartheta \in A_{n \varepsilon c}(\xi) } \left| \ell_n(\psi,\pi,\xi) - \ell_n(\psi^*,\pi,\xi) - \sqrt{n}t(\psi,\pi)' \nu_n (s_{\varrho k}) + n t(\psi,\pi)' \mathcal{I}_\varrho t(\psi,\pi)/2 \right| = o_{p \varepsilon} (1).
\end{equation}
\end{proposition}
We proceed to derive the asymptotic distribution of the LRTS. With $s_{\varrho k}$ defined in (\ref{score}), define 
\begin{equation} \label{I_lambda}
\begin{aligned}
&\mathcal{I}_{\eta} := \mathbb{E}_{\vartheta^*}(s_{\eta k} s_{\eta k}'), \quad \mathcal{I}_{\lambda\varrho_1\varrho_2} := \lim_{k\rightarrow \infty}\mathbb{E}_{\vartheta^*}(s_{\lambda\varrho_1 k} s_{\lambda\varrho_2 k}'), \quad \mathcal{I}_{\lambda\eta \varrho } := \lim_{k\rightarrow \infty}\mathbb{E}_{\vartheta^*}(s_{\lambda\varrho k} s_{\eta k}'),\\
&\mathcal{I}_{\eta \lambda \varrho} := \mathcal{I}_{\lambda \eta \varrho}', \quad \mathcal{I}_{\lambda.\eta \varrho_1 \varrho_2}:=\mathcal{I}_{\lambda \varrho_1 \varrho_2}-\mathcal{I}_{\lambda\eta \varrho_1}\mathcal{I}_{\eta}^{-1}\mathcal{I}_{\eta\lambda \varrho_2}, \quad \mathcal{I}_{\lambda.\eta \varrho}:=\mathcal{I}_{\lambda.\eta \varrho\varrho},\quad\ Z_{\lambda \varrho}:=(\mathcal{I}_{\lambda.\eta \varrho})^{-1}G_{\lambda.\eta \varrho},　
\end{aligned}
\end{equation} 
where
$G_{\lambda. \eta \varrho}$ is a $q_\lambda$-vector mean zero Gaussian process indexed by $\varrho$ with $\text{cov}(G_{\lambda. \eta\varrho_1},G_{\lambda. \eta\varrho_2}) = \mathcal{I}_{\lambda.\eta \varrho_1 \varrho_2}$. Define the set of admissible values of $\sqrt{n}\alpha(1-\alpha)v(\lambda)$ when $n\rightarrow \infty$ by $v(\mathbb{R}^q):= \{ x \in \mathbb{R}^{q_\lambda}: x = v(\lambda) \text{ for some } \lambda \in \mathbb{R}^q\}$. Define $\tilde t_{\lambda \varrho}$ by
\begin{equation} \label{t-lambda}
r_{\lambda \varrho}(\tilde t_{\lambda \varrho}) = \inf_{t_\lambda \in v(\mathbb{R}^q)}r_{\lambda \varrho}(t_{\lambda }), \quad r_{\lambda \varrho}(t_{\lambda}) := (t_{\lambda} -Z_{\lambda \varrho})' \mathcal{I}_{\lambda.\eta \varrho} (t_{\lambda} -Z_{\lambda \varrho}).
\end{equation}
The following proposition establishes the asymptotic null distribution of the LRTS. 
\begin{proposition} \label{P-LR}
Suppose Assumptions \ref{assn_a1}, \ref{assn_a2}, \ref{assn_a4}, \ref{A-consist}, and \ref{A-nonsing1} hold. Then, under the null hypothesis of $M=1$, $LR_n \overset{d}{\rightarrow} \sup_{ \varrho \in\Theta_{\varrho}} \left(\tilde t_{\lambda \varrho}' \mathcal{I}_{\lambda.\eta \varrho} \tilde t_{\lambda \varrho} \right)$.
\end{proposition} 
In Proposition \ref{P-LR}, the LRTS and its asymptotic distribution depend on the choice of $\epsilon$ because $\Theta_\varrho=[-1+2\epsilon,1-2\epsilon]$. It is possible to develop a version of the EM test 
\citep[][]{chenli09as, chenlifu12jasa, kasaharashimotsu15jasa} in this context that does not impose an explicit restriction on the parameter space for $p_{11}$ and $p_{22}$; however, we leave such an extension to future research.

\begin{remark}
When applied to the Markov regime switching model, the tests of \citet{carrasco14em} use the residuals from projecting $\nabla_{\theta \theta'} f_k / f_k + 2 \sum_{t=1}^{k-1} \varrho^{k-t} (\nabla_{\theta} f_{t} / f_{t}) (\nabla_{\theta'} f_{k} / f_{k})$ on $\nabla_{\theta} f_k / f_k$, where both are evaluated at the one-regime MLE. Therefore, in the non-normal case, the LRT and tests of \citet{carrasco14em} are based on the same score function.
\end{remark}

\subsection{Heteroscedastic normal distribution}\label{subsec:hetero_normal}

Suppose that $Y_k\in \mathbb{R}$ in the $j$-th regime follows a normal distribution with regime-specific intercept $\mu_j$ and variance $\sigma_j^2$. We split $\theta_j$ into $\theta_j = (\zeta_j,\sigma_j^2)'= (\mu_j,\beta_j',\sigma_j^2)'$, and write the density of the $j$-th regime as
\begin{equation}\label{normal-density}
f(y_k|\overline{\bf{y}}_{k-1};\gamma,\theta_j) =f(y_k|\overline{\bf{y}}_{k-1};\gamma,\zeta_j,\sigma^2_j) = \frac{1}{\sigma_j}\phi\left( \frac{y_k - \mu_j - \varpi(\overline{\bf{y}}_{k-1};\gamma,\beta_j ) }{\sigma_j}\right),
\end{equation}
for some function $\varpi$. In many applications, $\varpi$ is a linear function of $\gamma$ and $\beta_j$, e.g., $\varpi(\overline{\bf{y}}_{k-1},w_k;\gamma,\beta_j)= (\overline{\bf{y}}_{k-1})'\beta_j + w_k'\gamma$. Consider the following reparameterization introduced in \citet{kasaharashimotsu15jasa} ($\theta$ in Kasahara and Shimotsu corresponds to $\zeta$ here):
\begin{equation}
\left(
\begin{array}{c}
\zeta_1\\
\zeta_2\\
\sigma_1^2\\
\sigma_2^2
\end{array}
\right) =\left(
\begin{array}{c}
\nu_{\zeta} + (1-\alpha)\lambda_{\zeta} \\
\nu_{\zeta} -\alpha\lambda_{\zeta}\\
\nu_\sigma + (1- \alpha)(2\lambda_\sigma+ C_1 \lambda_\mu^2)\\
\nu_\sigma - \alpha(2\lambda_\sigma+ C_2 \lambda_\mu^2)
\end{array}
\right), \label{repara2}
\end{equation}
where $\nu_{\zeta}=(\nu_\mu,\nu_{\beta}')'$, $\lambda_{\zeta}=(\lambda_\mu,\lambda_{\beta}')'$, $C_1 := -(1/3)(1 + \alpha)$, and $C_2 := (1/3)(2 - \alpha)$, so that $C_1=C_2-1$. Collect the reparameterized parameters, except for $\alpha$, into one vector $\psi_{\alpha}$. As in Section \ref{subsec:nonnormal}, we suppress the subscript $\alpha$ from $\psi_{\alpha}$. Let the reparameterized density be 
\begin{equation} \label{repara_g_hetero}
g_{\psi}(y_k|\overline{\bf{y}}_{k-1},x_k) = f\left(y_k|\overline{\bf{y}}_{k-1};\gamma,\nu_\zeta + (q_k-\alpha)\lambda_\zeta, \nu_\sigma + (q_k - \alpha)(2\lambda_\sigma + (C_2 -q_k) \lambda_\mu^2 )\right).
\end{equation}
Let $\psi := (\eta',\lambda')' \in \Theta_\psi = \Theta_\eta \times \Theta_\lambda$, where $\eta := (\gamma',\nu_\zeta',\nu_\sigma)'$ and $\lambda := (\lambda_\zeta',\lambda_\sigma)'$. Because the likelihood function of a normal mixture model is unbounded when $\sigma_j \rightarrow 0$ \citep{hartigan85book}, we impose $\sigma_j \geq \epsilon_\sigma$ for a small $\epsilon_\sigma >0$ in $\Theta_\psi$. We proceed to derive the derivatives of $g_{\psi}(Y_k|\overline{\bf{Y}}_{k-1},X_k)$ evaluated at $\psi^*$. $\nabla_\psi g_k^*$, $\nabla_{\lambda\eta'}g_k^*$, and $\nabla_{\lambda\lambda'}g_k^*$ are the same as those given in (\ref{dg}) except for $\nabla_{\lambda_\mu^2}g_k^*$ and that those with respect to $\lambda_\sigma^j$ are multiplied by $2^j$. The higher-order derivatives of $g_{\psi}(Y_k|\overline{\bf{Y}}_{k-1},X_k)$ with respect to $\lambda_\mu$ are derived by following \citet{kasaharashimotsu15jasa}. 
From Lemma \ref{lemma_normal_der} and the fact that normal density $f(\mu,\sigma^2)$ satisfies
\begin{equation} \label{norma_der}
\begin{aligned}
\nabla_{\mu^2} f(\mu,\sigma^2) &= 2\nabla_{\sigma^2} f(\mu,\sigma^2), \quad \nabla_{\mu^3} f(\mu,\sigma^2) = 2\nabla_{\mu\sigma^2} f(\mu,\sigma^2), \ \ \text{and}\\
\nabla_{\mu^4} f(\mu,\sigma^2) &= 2 \nabla_{\mu^2\sigma^2} f(\mu,\sigma^2)= 4 \nabla_{ \sigma^2 \sigma^2} f(\mu,\sigma^2),
\end{aligned}
\end{equation}
we have
\begin{equation} \label{d3g}
\nabla_{\lambda_\mu^i} g_k^* = d_{ik} \nabla_{\mu^i} f_k^*, \quad i = 1, \ldots, 4,
\end{equation}
where 
\begin{align*} 
 d_{0k} & :=1, \quad d_{1k} := q_k - \alpha, \quad d_{2k} := (q_k - \alpha)(C_2 - \alpha), \quad d_{3k} := 2 (q_k - \alpha)^2(1-\alpha-q_k), \\
d_{4k} & := -2(q_k - \alpha)^4 + 3(q_k-\alpha)^2(\alpha - C_2)^2. 
\end{align*} 
It follows from $\mathbb{E}_{\vartheta^*}[ q_k|\overline{{\bf Y}}_{-\infty}^n]=\alpha$, (\ref{markov_moments}), and elementary calculation that
\begin{equation} \label{d3g2}
\begin{aligned}
\mathbb{E}_{\vartheta^*}[ d_{ik}|\overline{{\bf Y}}_{-\infty}^n] &= 0, \quad \mathbb{E}_{\vartheta^*}[ \nabla_{\lambda_\mu^i} g_k^*|\overline{{\bf Y}}_{-\infty}^k] = 0, \quad i = 1,2,3, \\
\mathbb{E}_{\vartheta^*}[ d_{4k}|\overline{{\bf Y}}_{-\infty}^n] &= \alpha(1 - \alpha) b(\alpha), \\
\mathbb{E}_{\vartheta^*}[ \nabla_{\lambda_\mu^4} g_k^*|\overline{{\bf Y}}_{-\infty}^k] &= \alpha(1-\alpha)b(\alpha)\nabla_{\mu^4} f_k^* = \alpha(1-\alpha) b(\alpha) 4\nabla_{\sigma^2\sigma^2}f_k^* = b(\alpha) \mathbb{E}_{\vartheta^*}[ \nabla_{\lambda_\sigma^2} g_k^*|\overline{{\bf Y}}_{-\infty}^k],
\end{aligned}
\end{equation}
with $b(\alpha): = -(2/3) (\alpha^2 - \alpha + 1) < 0$. Hence, $\mathbb{E}_{\vartheta^*}[ \nabla_{\lambda_\sigma^2} g_k^*|\overline{{\bf Y}}_{-\infty}^k]$ and $\mathbb{E}_{\vartheta^*}[ \nabla_{\lambda_\mu^4} g_k^*|\overline{{\bf Y}}_{-\infty}^k]$ are linearly dependent.

We proceed to derive a representation of $\nabla^j \overline p_{\psi^* \pi} (Y_k| \overline{{\bf Y}}_{0}^{k-1})/\overline p_{\psi^* \pi} (Y_k| \overline{{\bf Y}}_{0}^{k-1})$ in terms of $\nabla^j f_k^*$. Repeating the calculation leading to (\ref{d1p})--(\ref{d2p}) and using (\ref{d3g2}) gives the following. First, (\ref{d1p}) and (\ref{d2p0}) still hold; second, the elements of $\nabla_{\lambda\lambda'} \overline p_{\psi^* \pi} (Y_k| \overline{{\bf Y}}_{0}^{k-1})/\overline p_{\psi^* \pi} (Y_k| \overline{{\bf Y}}_{0}^{k-1})$ except for the $(1,1)$-th element are given by (\ref{d2p}) after adjusting that the derivative with respect to $\lambda_{\sigma}$ must be multiplied by 2 (e.g., $\mathbb{E}_{\vartheta^*}[\nabla_{\lambda_\sigma} g_k^* |\overline{{\bf Y}}_{-\infty}^n] = 2 \nabla_{\sigma^2}f_k^*$ and $\mathbb{E}_{\vartheta^*}[\nabla_{\lambda_\sigma\lambda_\mu} g_k^* |\overline{{\bf Y}}_{-\infty}^n] = 2 \nabla_{\sigma^2\mu}f_k^*$); third,
\begin{equation} \label{d2p-lambda}
\frac{\nabla_{\lambda_\mu^2} \overline p_{\psi^* \pi} (Y_k| \overline{{\bf Y}}_{0}^{k-1})}{\overline p_{\psi^* \pi} (Y_k| \overline{{\bf Y}}_{0}^{k-1})} = \alpha(1-\alpha) \sum_{t=1}^{k-1} \varrho^{k-t} \left( 2 \frac{\nabla_{\mu} f_{t}^*}{f_{t}^*} \frac{\nabla_{\mu} f_{k}^*}{f_{k}^*}\right).
\end{equation} 
When $\varrho \neq 0$, $\nabla_{\lambda_\mu^2} \overline p_{\psi^* \pi} (Y_k| \overline{{\bf Y}}_{0}^{k-1})/\overline p_{\psi^* \pi} (Y_k| \overline{{\bf Y}}_{0}^{k-1})$ is a non-degenerate random variable as in the non-normal case. When $\varrho=0$, however, $\nabla_{\lambda_\mu^2} \overline p_{\psi^* \pi} (Y_k| \overline{{\bf Y}}_{0}^{k-1})/\overline p_{\psi^* \pi} (Y_k| \overline{{\bf Y}}_{0}^{k-1})$ becomes identically equal to 0, and indeed the first non-zero derivative with respect to $\lambda_\mu$ is the fourth derivative.

Because of this degeneracy, we derive the asymptotic distribution of the LRTS by expanding $\ell_n(\psi,\pi,\xi)-\ell_n(\psi^*,\pi,\xi)$ four times. It is not correct, however, to simply approximate $\ell_n(\psi,\pi,\xi)-\ell_n(\psi^*,\pi,\xi)$ by a quadratic function of $\lambda_\mu^2$ (and other terms) when $\varrho \neq 0$ and a quadratic function of $\lambda_\mu^4$ when $\varrho=0$. This results in discontinuity at $\varrho=0$ and fails to provide a valid uniform approximation. We establish a uniform approximation by expanding $\ell_n(\psi,\pi,\xi)$ four times but expressing $\ell_n(\psi,\pi,\xi)$ in terms of $\varrho\lambda_\mu^2$, $\lambda_\mu^4$, and other terms.
 
For $m \geq 0$, define $\zeta_{k,m}(\varrho):= \sum_{t=-m+1}^{k-1} \varrho^{k-t-1} 2 \nabla_{\mu} f_{t}^*\nabla_{\mu} f_{k}^*/ f_{t}^* f_{k}^*$. Then, we can write (\ref{d2p-lambda}) as
\begin{equation} \label{d2p2}
\frac{\nabla_{\lambda_\mu^2} \overline p_{\psi^* \pi} (Y_k| \overline{{\bf Y}}_{0}^{k-1})}{\alpha(1-\alpha) \overline p_{\psi^* \pi} (Y_k| \overline{{\bf Y}}_{0}^{k-1})} = \sum_{t=1}^{k-1} \varrho^{k-t}\left( 2 \frac{\nabla_{\mu} f_{t}^*}{f_{t}^*} \frac{\nabla_{\mu} f_{k}^*}{f_{k}^*}\right) = \varrho \zeta_{k,0}(\varrho).
\end{equation}
Note that $\zeta_{k,m}(\varrho)$ satisfies $\mathbb{E}_{\vartheta^*}[\zeta_{k,m}(\varrho)| \overline{{\bf Y}}_{-m}^{k-1}]=0$ and is non-degenerate even when $\varrho=0$. Define $v(\lambda_\beta)$ as $v(\lambda)$ in (\ref{v_lambda}) but replacing $\lambda$ with $\lambda_\beta$. Collect the relevant parameters as 
\begin{equation} \label{t-psi2}
t(\psi,\pi) :=
\begin{pmatrix}
\eta - \eta^* \\
t_{\lambda}(\lambda,\pi)
\end{pmatrix},
\end{equation}
where 
\begin{equation} \label{t_lambda_rho}
t_{\lambda}(\lambda,\pi)
:=\alpha(1-\alpha)
\begin{pmatrix}
\varrho \lambda_\mu^2 \\
 \lambda_\mu\lambda_\sigma\\
\lambda_\sigma^2 + b(\alpha)\lambda_\mu^4/12\\
\lambda_\beta \lambda_\mu\\
\lambda_\beta \lambda_\sigma\\
v(\lambda_\beta)
\end{pmatrix}, 
\end{equation}
with $b(\alpha)= -(2/3)(\alpha^2-\alpha+1)<0$. Recall $\theta_j = (\zeta_j',\sigma_j^2)'= (\mu_j,\beta_j',\sigma_j^2)'$. Similarly to (\ref{score_lambda}), define the elements of the generalized score by
\begin{equation} \label{score_lambda_beta}
\begin{pmatrix}
* & s_{\lambda_{\mu\beta} \varrho k} & s_{\lambda_{\mu\sigma} \varrho k} \\
s_{\lambda_{\beta\mu} \varrho k} & s_{\lambda_{\beta\beta} \varrho k} & s_{\lambda_{\beta\sigma} \varrho k} \\
s_{\lambda_{\sigma\mu} \varrho k} & s_{\lambda_{\beta\sigma} \varrho k}& s_{\lambda_{\sigma\sigma} \varrho k} 
\end{pmatrix}
= \frac{\nabla_{\theta \theta'} f_k^*}{f_k^*} + \sum_{t= 1 }^{k-1} \varrho^{k-t} \left( \frac{\nabla_{\theta} f_{t}^*}{f_{t}^*} \frac{\nabla_{\theta'} f_{k}^*}{f_{k}^*} + \frac{\nabla_{\theta} f_{k}^*}{f_{k}^*} \frac{\nabla_{\theta'} f_{t}^*}{f_{t}^*} \right).
\end{equation}
Define the generalized score as 
\begin{equation}\label{score_normal}
s_{\varrho k} : = 
\begin{pmatrix}
s_{\eta k} \\
s_{\lambda \varrho k}
\end{pmatrix},
\quad \text{where}\quad
s_{\eta k} : = \begin{pmatrix}
\nabla_{\gamma} f_k^* / f_k^* \\
\nabla_{\theta} f_k^* / f_k^* 
\end{pmatrix} \
\text{ and }\ 
s_{\lambda \varrho k}
:= 
\begin{pmatrix}
\zeta_{k,0}(\varrho)/2 \\
 2 s_{\lambda_{\mu\sigma} \varrho k} \\
 2 s_{\lambda_{\sigma\sigma} \varrho k}\\
s_{\lambda_{\beta \mu} \varrho k}\\
 2 s_{\lambda_{\beta \sigma} \varrho k}\\
 V(s_{\lambda_{\beta\beta} \varrho k}) 
\end{pmatrix}.
\end{equation} 
The following proposition establishes a uniform approximation of the log-likelihood ratio. 
\begin{assumption} \label{A-nonsing2}
(a) $0< \inf_{\varrho\in\Theta_{\varrho}} \lambda_{\min}(\mathcal{I}_\varrho) \leq \sup_{\varrho\in\Theta_{\varrho}}\lambda_{\max}(\mathcal{I}_\varrho) < \infty$ for $\mathcal{I}_{\varrho}= \lim_{k\rightarrow \infty } \mathbb{E}_{\vartheta^*}(s_{\varrho k}s_{\varrho k}')$, where $s_{\varrho k}$ is given in (\ref{score_normal}). (b) $\sigma_1^*,\sigma_2^* >\epsilon_\sigma$. 
\end{assumption} 

\begin{proposition} \label{P-quadratic-N1} Suppose Assumptions \ref{assn_a1}, \ref{assn_a2}, \ref{assn_a4}, \ref{A-consist}, and \ref{A-nonsing2} hold and the density of the $j$-th regime is given by (\ref{normal-density}). Then, under the null hypothesis of $M=1$, (a) $\sup_{\vartheta \in A_{n \varepsilon}(\xi)} |t(\psi,\pi)| = O_{p \varepsilon }(n^{-1/2})$; and (b) for any $c>0$,
\begin{equation} \label{ln_appn_N1}
\sup_{\xi \in \Xi} \sup_{\vartheta \in A_{n\varepsilon c}(\xi) } \left| \ell_n(\psi,\pi,\xi) - \ell_n(\psi^*,\pi,\xi) - \sqrt{n}t(\psi,\pi)' \nu_n (s_{\varrho k}) + n t(\psi,\pi)' \mathcal{I}_\varrho t(\psi,\pi)/2 \right| = o_{p \varepsilon }(1).
\end{equation}
\end{proposition}
The asymptotic null distribution of the LRTS is characterized by the supremum of $ 2 t_{\lambda}' G_{\lambda.\eta\varrho}- t_{\lambda}' \mathcal{I}_{\lambda.\eta \varrho} t_{\lambda}$, where $G_{\lambda. \eta \varrho}$ and $\mathcal{I}_{\lambda.\eta \varrho}$ are defined analogously to those in (\ref{I_lambda}) but with $s_{\varrho k}$ defined in (\ref{score_normal}), and the supremum is taken with respect to $t_{\lambda}$ and $\varrho\in\Theta_{\varrho}$ under the constraint implied by the limit of the set of possible values of $\sqrt{n}t_{\lambda}(\lambda,\pi)$ as $n\rightarrow\infty$. This constraint is given by the union of $\Lambda_{\lambda}^1$ and $\Lambda_{\lambda \varrho}^2$, where $q_\beta := \dim(\beta)$, $q_{\lambda}:= 3 +2q_\beta+q_\beta(q_\beta+1)/2$, and
\begin{equation}
\begin{aligned}\label{Lambda-lambda}
&\Lambda_{\lambda }^1:=\{ t_{\lambda}=( t_{\varrho\mu^2}, t_{\mu\sigma },t_{\sigma^2 },t_{\beta\mu}',t_{\beta\sigma}',t_{v(\beta)}')' \in \mathbb{R}^{q_{\lambda} } : \\
& \qquad \qquad (t_{\varrho\mu^2},t_{\mu\sigma },t_{\sigma^2 },t_{\beta\mu}')'\in \mathbb{R}\times\mathbb{R}\times\mathbb{R}_{-}\times\mathbb{R}^{q_\beta}, t_{\beta\sigma}=0, t_{v(\beta)}=0\}, \quad \\
&\Lambda_{\lambda \varrho}^2:=\{ t_{\lambda }= (t_{\varrho\mu^2},t_{\mu\sigma },t_{\sigma^2},t_{\beta\mu}',t_{\beta\sigma}',t_{v(\beta)}')' \in \mathbb{R}^{q_{\lambda} }: t_{\varrho\mu^2} = \varrho \lambda_{\mu}^2, t_{\mu\sigma }= \lambda_\mu\lambda_\sigma, \\
& \qquad \qquad 
t_{\sigma^2}= \lambda_\sigma^2, t_{\beta\mu }= \lambda_{\beta} \lambda_\mu, t_{\beta\sigma }= \lambda_{\beta} \lambda_\sigma, t_{v(\beta)} =v_{\beta}(\lambda_{\beta})\ \text{for some }\lambda\in \mathbb{R}^{2+q_\beta} \}. 
\end{aligned}
\end{equation}
Note that $\Lambda_{\lambda\varrho}^2$ depends on $\varrho$, whereas $\Lambda_{\lambda}^1$ does not depend on $\varrho$. Heuristically, $\Lambda_{\lambda}^1$ and $\Lambda_{\lambda\varrho}^2$ correspond to the limits of the set of possible values of $\sqrt{n}t_{\lambda}(\lambda,\pi)$ when $\liminf_{n\to\infty}n^{1/8}|\lambda_{\mu}| > 0$ and $\lambda_{\mu}=o(n^{-1/8})$, respectively. When $\liminf_{n\to\infty}n^{1/8}|\lambda_{\mu}| > 0$, we have $(\hat\lambda_\sigma, \hat\lambda_\beta) =O_p(n^{-3/8})$ because $t_{\lambda}(\hat\lambda,\pi)=O_p(n^{-1/2})$. Further, the set of possible values of $\sqrt{n} \varrho\lambda_{\mu}^2$ converges to $\mathbb{R}$ because $\varrho$ can be arbitrarily small. Consequently, the limit of $\sqrt{n}t_{\lambda}(\lambda,\pi)$ is characterized by $\Lambda_{\lambda}^1$.

Define $Z_{\lambda \varrho}$ and $\mathcal{I}_{\lambda.\eta \varrho}$ as in (\ref{I_lambda}) but with $s_{\pi k}$ defined in (\ref{score_normal}). Let $Z_{\lambda 0}$ and $\mathcal{I}_{\lambda.\eta 0}$ denote $Z_{\lambda \varrho}$ and $\mathcal{I}_{\lambda.\eta \varrho}$ evaluated at $\varrho=0$. Define $\tilde{t}_{\lambda}^1$ and $\tilde{t}_{\lambda\varrho}^2$ by
\begin{equation} \label{t-lambda-N1}
\begin{aligned}
r_{\lambda}(\tilde{t}_{\lambda }^1) & = \inf_{t_{\lambda} \in \Lambda_{\lambda}^1}r_{\lambda}(t_{\lambda}), \quad r_{\lambda}(t_{\lambda}) := (t_{\lambda} -Z_{\lambda 0})' \mathcal{I}_{\lambda.\eta 0} (t_{\lambda} -Z_{\lambda 0}) \quad \\
r_{\lambda\varrho}(\tilde{t}_{\lambda\varrho}^2) & = \inf_{t_{\lambda} \in \Lambda_{\lambda\varrho}^2}r_{\lambda\varrho}(t_{\lambda}), \quad r_{\lambda\varrho}(t_{\lambda}) := 
(t_{\lambda} -Z_{\lambda \varrho})' \mathcal{I}_{\lambda.\eta \varrho} (t_{\lambda} -Z_{\lambda \varrho}). 
\end{aligned}
\end{equation}
The following proposition establishes the asymptotic null distribution of the LRTS. 
\begin{proposition} \label{P-LR-N1}
Suppose that the assumptions in Proposition \ref{P-quadratic-N1} hold. Then, under the null hypothesis of $M=1$, $LR_n \overset{d}{\rightarrow} \max\{ \mathbb{I}\{\varrho=0\} (\tilde t_{\lambda }^1)' \mathcal{I}_{\lambda.\eta 0} \tilde t_{\lambda }^1, \sup_{\varrho\in\Theta_{\varrho}} (\tilde t_{\lambda \varrho}^2)' \mathcal{I}_{\lambda.\eta \varrho} \tilde t_{\lambda \varrho}^2 \}$. 
\end{proposition} 
\begin{remark}
\citet{quzhuo17wp} derive the asymptotic distribution of the LRTS under the restriction that $\varrho\geq \epsilon>0$. 
\end{remark}

\begin{remark}
It is possible to extend our analysis to the exponential-LR type tests studied by \citet{andrewsploberger94em} and \citet{carrasco14em}.
\end{remark}

\subsection{Homoscedastic normal distribution}\label{subsec:homo_normal}

Suppose that $Y_k\in \mathbb{R}$ in the $j$-th regime follows a normal distribution with the regime-specific intercept $\mu_j$ but with common variance $\sigma^2$. We split $\gamma$ and $\theta_j$ into $\gamma=(\tilde \gamma',\sigma^2)'$ and $\theta_j=(\mu_j,\beta_j')'$, and write the density of the $j$-th regime as
\begin{equation}\label{normal-density-homo}
f(y_k|\overline{\bf{y}}_{k-1};\gamma,\theta_j) =f(y_k|\overline{\bf{y}}_{k-1};\tilde \gamma,\theta_j,\sigma^2) = \frac{1}{\sigma}\phi\left( \frac{y_k - \mu_j- \varpi(\overline{\bf{y}}_{k-1};\tilde\gamma,\beta_j ) }{\sigma}\right),
\end{equation}
for some function $\varpi$. Consider the following reparameterization:
\begin{equation}
\left(
\begin{array}{c}
\theta_1\\
\theta_2\\
\sigma^2 
\end{array}
\right) =\left(
\begin{array}{c}
\nu_{\theta} + (1-\alpha)\lambda \\
\nu_{\theta} -\alpha\lambda \\ 
\nu_\sigma - \alpha(1-\alpha) \lambda_\mu^2
\end{array}
\right), \label{repara-homo}
\end{equation}
where $\nu_{\theta}=(\nu_\mu,\nu_\beta')'$ and $\lambda=(\lambda_\mu,\lambda_\beta')'$. Collect the reparameterized parameters, except for $\alpha$, into one vector $\psi_{\alpha}$. Suppressing $\alpha$ from $\psi_{\alpha}$, let the reparameterized density be
\begin{equation}\label{repara_g_homo}
g_{\psi}(y_k|\overline{\bf{y}}_{k-1},x_k) = f\left(y_k|\overline{\bf{y}}_{k-1};\tilde\gamma,\nu_\theta+(q_k-\alpha)\lambda, \nu_\sigma -\alpha(1-\alpha) \lambda_\mu^2 \right).
\end{equation}
Let $\eta = (\tilde\gamma',\nu_\theta',\nu_\sigma)'$; then, the first and second derivatives of $g_{\psi}(y_k|\overline{\bf{y}}_{k-1},x_k)$ with respect to $\eta$ and $\lambda$ are the same as those given in (\ref{dg}) except for $\nabla_{\lambda_\mu^2}g_{\psi}(y_k|\overline{\bf{y}}_{k-1},x_k)$. We derive the higher-order derivatives of $g_{\psi}(y_k|\overline{\bf{y}}_{k-1},x_k)$ with respect to $\lambda_\mu$. From Lemmas \ref{lemma_normal_der} and (\ref{norma_der}), we obtain
\begin{equation} \label{d3g-homo}
\begin{aligned}
\nabla_{\lambda \eta^i} g_k^* &= d_{1k} \nabla_{\theta \eta^i} f_k^* \quad \text{for } i=0,1,\ldots,\\
\nabla_{\lambda_\mu^i} g_k^* &= d_{ik} \nabla_{\mu^i} f_k^* \quad \text{for }i = 0,1, \ldots, 4,
\end{aligned}
\end{equation}
where $d_{0k} :=1$, $d_{1k} := q_k - \alpha$, $d_{2k} := (q_k-\alpha)^2-\alpha(1-\alpha)$, $d_{3k} := (q_k-\alpha)^3 - 3(q_k-\alpha)\alpha(1-\alpha)$, and
$d_{4k} := (q_k - \alpha)^4 -6 (q_k - \alpha)^2 \alpha(1-\alpha) + 3\alpha^2(1 - \alpha)^2$.
It follows from $\mathbb{E}_{\vartheta^*}[ q_k|\overline{{\bf Y}}_{-\infty}^n]=\alpha$, (\ref{markov_moments}), and elementary calculation that 
\begin{equation} \label{d3g2-homo}
\begin{aligned}
\mathbb{E}_{\vartheta^*}[ \nabla_{\lambda_\mu^i} g_k^*|\overline{{\bf Y}}_{0}^k] &= 0, \quad\mathbb{E}_{\vartheta^*}[ d_{ik}|\overline{{\bf Y}}_{0}^k] = 0,\quad i = 1,2, \\
 \mathbb{E}_{\vartheta^*}[ d_{3k}|\overline{{\bf Y}}_{0}^k] &= \alpha(1 - \alpha) (1-2\alpha),
\quad \mathbb{E}_{\vartheta^*}[ d_{4k}|\overline{{\bf Y}}_{0}^k] = \alpha(1 - \alpha) (1-6\alpha+6\alpha^2).
\end{aligned}
\end{equation}
Repeating the calculation leading to (\ref{d1p})--(\ref{d2p}) and using (\ref{d3g2-homo}) gives the following. First, (\ref{d1p}) and (\ref{d2p0}) still hold; second, the elements of $\nabla_{\lambda\lambda'} \overline p_{\psi^* \pi} (Y_k| \overline{{\bf Y}}_{0}^{k-1})/\overline p_{\psi^* \pi} (Y_k| \overline{{\bf Y}}_{0}^{k-1})$ are given by (\ref{d2p}) except for the $(1,1)$-th element; third, $\nabla_{\lambda_\mu^2} \overline p_{\psi^* \pi} (Y_k| \overline{{\bf Y}}_{0}^{k-1})/\overline p_{\psi^* \pi} (Y_k| \overline{{\bf Y}}_{0}^{k-1})$ is given by (\ref{d2p-lambda}). Further, Lemma \ref{lemma_d34_homo} in the appendix shows that when $\varrho=0$, $\nabla_{\lambda_\mu^3} \overline p_{\psi^* \pi} (Y_k| \overline{{\bf Y}}_{0}^{k-1})/\overline p_{\psi^* \pi} (Y_k| \overline{{\bf Y}}_{0}^{k-1}) = \alpha(1-\alpha)(1-2\alpha) \nabla_{\mu^3}f_k^*/f_k^*$ and $\nabla_{\lambda_\mu^4} \overline p_{\psi^* \pi} (Y_k| \overline{{\bf Y}}_{0}^{k-1})/\overline p_{\psi^* \pi} (Y_k| \overline{{\bf Y}}_{0}^{k-1}) = \alpha(1-\alpha)(1-6\alpha+6\alpha^2) \nabla_{\mu^4}f_k^*/f_k^*$. Because $\nabla_{\lambda_\mu^3} \overline p_{\psi^* \pi} (Y_k| \overline{{\bf Y}}_{0}^{k-1})/\overline p_{\psi^* \pi} (Y_k| \overline{{\bf Y}}_{0}^{k-1})=0$ when $\alpha=1/2$ and $\varrho=0$, we expand $\ell_n(\psi,\pi,\xi)$ four times and express it in terms of $\varrho\lambda_\mu^2$, $(1-2\alpha)\lambda_\mu^3$, $\lambda_\mu^4$, and other terms to establish a uniform approximation.

Collect the relevant parameters as
\begin{equation} \label{t-psi2-homo}
t(\psi,\pi) :=
\begin{pmatrix}
\eta - \eta^* \\
t_{\lambda}(\lambda,\pi)
\end{pmatrix} \quad \text{and}\quad
t_{\lambda}(\lambda,\pi)
:=\alpha(1-\alpha)
\begin{pmatrix}
 \varrho \lambda_\mu^2 \\
 (1-2\alpha) \lambda_\mu^3 \\
 (1-6\alpha+6\alpha^2)\lambda_\mu^4 \\
 \lambda_\beta \lambda_\mu\\
 v(\lambda_\beta)
\end{pmatrix}.
\end{equation}
Define the generalized score as
\begin{equation}\label{score_normal_homo}
s_{\varrho k} : = 
\begin{pmatrix}
s_{\eta k} \\
s_{\lambda \varrho k}
\end{pmatrix},
\quad \text{where}\quad
s_{\eta k} : = \begin{pmatrix}
\nabla_{\gamma} f_k^* / f_k^* \\
\nabla_{\theta} f_k^* / f_k^* 
\end{pmatrix} \
\text{ and }\ 
s_{\lambda \varrho k}
:= 
\begin{pmatrix}
\zeta_{k,0}(\varrho)/2 \\
s_{\lambda_{\mu}^3 k}/3!\\
s_{\lambda_{\mu}^4 k}/4!\\
s_{\lambda_{\beta \mu} \varrho k}\\
V(s_{\lambda_{\beta\beta} \varrho k})
\end{pmatrix},
\end{equation}
where $\zeta_{k,m}(\varrho)$ is defined as in (\ref{d2p2}), $s_{\lambda_{\mu}^i k}:=\nabla_{\mu^i}f_k^*/f_k^*$ for $i=3,4$, and
$s_{\lambda_{\beta \mu} \varrho k}$ and $s_{\lambda_{\beta\beta} \varrho k}$ are defined as in (\ref{score_lambda_beta}) but using the density (\ref{normal-density-homo}) in place of (\ref{normal-density}).
Define, with $q_\beta := \dim(\beta)$ and $q_{\lambda}:=3 + q_\beta+q_\beta(q_\beta+1)/2$,
\begin{equation}
\begin{aligned}\label{Lambda-lambda-homo}
&\Lambda_{\lambda }^1:=\{ t_{\lambda}= ( t_{\varrho\mu^2}, t_{\mu^3}, t_{\mu^4}, t_{\beta\mu}', t_{v(\beta)}')' \in \mathbb{R}^{q_{\lambda}} : (t_{\varrho\mu^2}, t_{\mu^3}, t_{\mu^4}, t_{\beta\mu}')'\in \mathbb{R}\times\mathbb{R}\times\mathbb{R}_{-}\times\mathbb{R}^{q_\beta}, t_{v(\beta)}=0\}, \quad \\
&\Lambda_{\lambda \varrho}^2:=\{ t_{\lambda}= ( t_{\varrho\mu^2}, t_{\mu^3}, t_{\mu^4}, t_{\beta\mu}', t_{v(\beta)}')' \in \mathbb{R}^{q_{\lambda}}: t_{\varrho\mu^2} = \varrho \lambda_{\mu}^2, t_{\mu^3}=t_{\mu^4}=0, t_{\beta\mu }= \lambda_{\beta} \lambda_\mu, \\
& \qquad \quad 
 t_{v(\beta)} =v_{\beta}(\lambda_{\beta})\ \text{for some }\lambda\in \mathbb{R}^{1+q_\beta} \}. 
\end{aligned}
\end{equation}

The following two propositions correspond to Propositions \ref{P-quadratic-N1} and \ref{P-LR-N1}, establishing a uniform approximation of the log-likelihood ratio and asymptotic distribution of the LRTS.
\begin{assumption} \label{A-nonsing2-homo}
$0< \inf_{\varrho \in \Theta_{\varrho}} \lambda_{\min}(\mathcal{I}_\varrho) \leq \sup_{\varrho \in \Theta_{\varrho}}\lambda_{\max}(\mathcal{I}_\varrho) < \infty$ for $\mathcal{I}_{\varrho}= \lim_{k\rightarrow \infty} \mathbb{E}_{\vartheta^*}(s_{\varrho k}s_{\varrho k}')$, where $s_{\varrho k}$ is given in (\ref{score_normal_homo}).
\end{assumption}
\begin{proposition} \label{P-quadratic-N1-homo}
Suppose Assumptions \ref{assn_a1}, \ref{assn_a2}, \ref{assn_a4}, \ref{A-consist}, and \ref{A-nonsing2-homo} hold and the density of the $j$-th regime is given by (\ref{normal-density-homo}). Then, statements (a) and (b) of Proposition \ref{P-quadratic-N1} hold.
\end{proposition} 
\begin{proposition} \label{P-LR-N1-homo} Suppose that the assumptions in Proposition \ref{P-quadratic-N1-homo} hold. Then, under the null hypothesis of $M=1$, $LR_n \overset{d}{\rightarrow} \max\{ \mathbb{I}\{\varrho=0\} (\tilde t_{\lambda }^1)' \mathcal{I}_{\lambda.\eta 0} \tilde t_{\lambda }^1, \sup_{\varrho\in\Theta_\varrho} (\tilde t_{\lambda \varrho}^2)' \mathcal{I}_{\lambda.\eta \varrho} \tilde t_{\lambda \varrho}^2 \}$, where $\tilde{t}_{\lambda}^1$ and $\tilde{t}_{\lambda\varrho}^2$ are defined as in (\ref{t-lambda-N1}) but in terms of $(Z_{\lambda \varrho}, \mathcal{I}_{\lambda.\eta \varrho}, Z_{\lambda 0},\mathcal{I}_{\lambda.\eta 0})$ constructed with $s_{\varrho k}$ defined in (\ref{score_normal_homo}) and $\Lambda_{\lambda }^1$ and $\Lambda_{\lambda \varrho}^2$ defined in (\ref{Lambda-lambda-homo}).
\end{proposition} 

\section{Testing $H_0:M=M_0$ against $H_A:M=M_0+1$ for $M_0 \geq 2$}\label{sec-general}
In this section, we derive the asymptotic distribution of the LRTS for testing the null hypothesis of $M_0$ regimes against the alternative of $M_0+1$ regimes for general $M_0 \geq 2$. We suppress the covariate ${\bf W}_{a}^b$ unless confusion might arise.

Let $\vartheta_{M_0}^*=((\vartheta_{M_0,x}^*)',(\vartheta_{M_0,y}^*)')'$ denote the parameter of the $M_0$-regime model, where $\vartheta_{M_0,x}^*$ contains $p_{ij}^*= q_{\vartheta^*_{M_0,x}}(i,j)>0$ for $i= 1,\ldots,M_0$ and $j=1,\ldots,M_0-1$, and $\vartheta_{M_0,y}^* = ((\theta_1^*)',\ldots,(\theta_{M_0}^*)',(\gamma^*)')'$. We assume $\max_{i}\sum_{j=1}^{M_0-1}p_{ij}^*<1$ and $\theta_{1}^*<\ldots< \theta_{M_0}^*$ for identification. 
The true $M_0$-regime conditional density of ${\bf Y}_1^n$ given $\overline{\bf Y}_{0}$ and $x_0$ is 
\begin{equation} \label{true_model}
p_{\vartheta_{M_0}^*}({\bf Y}_1^n| \overline{\bf Y}_{0},x_0) = \sum_{{\bf x}_1^n\in \mathcal{X}_{M_0}^n}\prod_{k=1}^n p_{\vartheta_{M_0}^*}(Y_k,x_k|\overline{\bf Y}_{k-1},x_{k-1}),
\end{equation} 
where $p_{\vartheta_{M_0}^*} (y_k,x_k| \overline{\bf{y}}_{k-1},x_{k-1}) = g_{\vartheta_{M_0,y}^*}(y_k|\overline{\bf{y}}_{k-1}, x_k) q_{\vartheta_{M_0,x}^*}(x_{k-1},x_k)$ with $g_{\vartheta_{M_0,y}^*}(y_k|\overline{\bf{y}}_{k-1},x_k) = \sum_{j=1,\ldots,M_0} \mathbb{I}\{x_k=j\}f(y_k|\overline{\bf{y}}_{k-1};\gamma,\theta_j^*)$.
 
Let the conditional density of ${\bf Y}_{1}^n$ of an $(M_0+1)$-regime model be
\begin{equation} \label{fitted_model}
p_{\vartheta_{M_0+1}}({\bf Y}_1^n| \overline{\bf Y}_{0},x_0) := \sum_{{\bf x}_1^n\in \mathcal{X}_{M_0+1}^n}\prod_{k=1}^n p_{\vartheta_{M_0+1}}(Y_k,x_k|\overline{\bf Y}_{k-1},x_{k-1}),
\end{equation}
where $p_{\vartheta_{M_0+1}}(y_k,x_k|\overline{\bf y}_{k-1},x_{k-1})$ is defined similarly to $p_{\vartheta_{M_0}^*}(y_k,x_k|\overline{\bf y}_{k-1},x_{k-1})$ with $\vartheta_{M_0+1,x}:= \{p_{ij}\}_{i=1,\ldots,M_0+1,j=1,\ldots,M_0}$ and $\vartheta_{M_0+1,y} := (\theta_1',\ldots,\theta_{M_0+1}',\gamma')'$. We assume that $\min_{i,j}p_{ij} \geq \epsilon$ for some $\epsilon\in (0,1/2)$.

Write the null hypothesis as $H_0 = \cup_{m=1}^{M_0} H_{0m}$ with
\[
H_{0m} : \theta_1 < \cdots < \theta_{m} = \theta_{m + 1} < \cdots < \theta_{M_0 + 1}.
\] 
Define the set of values of $\vartheta_{M_0+1}$ that yields the true density (\ref{true_model}) under $\mathbb{P}_{\vartheta^*_{M_0}}$ as $\Upsilon^*:=\{\vartheta_{M_0 + 1} \in \Theta_{M_0+1,\epsilon}: p_{\vartheta_{M_0+1}}({\bf Y}_1^n| \overline{\bf Y}_{0}, x_0) = p_{\vartheta_{M_0}^*}({\bf Y}_1^n| \overline{\bf Y}_{0}, x_0)\ \mathbb{P}_{\vartheta^*_{M_0}}\text{-a.s.}\}$. Under $H_{0m}$, the $(M_0 + 1)$-regime model (\ref{fitted_model}) generates the true $M_0$-regime density (\ref{true_model}) if $\theta_m = \theta_{m + 1} = \theta_{m}^{*}$ and the transition matrix of $X_k$ reduces to that of the true $M_0$-regime model.

We reparameterize the transition probability of $X_k$ by writing $\vartheta_{M_0+1,x}$ as $\vartheta_{M_0+1,x} = (\vartheta_{xm}',\pi_{xm}')'$, where $\vartheta_{xm}$ is point identified under $H_{0m}$, while $\pi_{xm}$ is not point identified under $H_{0m}$. The transition probability of $X_k$ under $\vartheta_{M_0+1,x}$ equals the transition probability of $X_k$ under $\vartheta_{M_0,x}^*$ if and only if $\vartheta_{xm} = \vartheta_{xm}^*$. The detailed derivation including the definition of $\vartheta_{xm}^*$ is provided in Section \ref{subsec:p_m_repara} in the appendix. 
Define the subset of $\Upsilon^*$ that corresponds to $H_{0m}$ as 
\begin{align*} 
\Upsilon_{m}^*& := \left\{\vartheta_{M_0 + 1} \in \Theta_{M_0 + 1}: 
\theta_j=\theta_{j}^*\ \text{for}\ 1 \leq j < m;\ \theta_m = \theta_{m + 1} = \theta_{m}^{*}; \right.\\
& \qquad \left. \theta_{j} = \theta_{j - 1}^*\ \text{for}\ h+1 < j \leq M_0+1;\ \gamma=\gamma^*;\ \vartheta_{xm}=\vartheta_{xm}^*\right\};
\end{align*} 
then, $\Upsilon^*= \Upsilon_{1}^* \cup \cdots \cup \Upsilon_{M_0}^*$ holds. 

For $M=M_0,M_0+1$, let $\ell_n(\vartheta_M,\xi_M) := \log \left( \sum_{x_0=1}^M p_{\vartheta_M}({\bf Y}_1^n| \overline{\bf Y}_{0},x_0) \xi_M(x_0) \right)$ denote the $M$-regime log-likelihood for a given initial distribution $\xi_M(x_0) \in \Xi_M$. We treat $\xi_M(x_0)$ as fixed. Let $\hat\vartheta_{M_0}:= \arg\max_{\vartheta_{M_0} \in \Theta_{{M_0}}} \ell_n(\vartheta_{M_0},\xi_{M_0})$ and $\hat\vartheta_{M_0+1}:= \arg\max_{\vartheta_{M_0+1} \in \Theta_{M_0+1}} \ell_n(\vartheta_{M_0+1},\xi_{M_0+1})$. The following proposition shows that the MLE is consistent in the sense that the distance between $\hat{\vartheta}_{M_0+1}$ and $\Upsilon^*$ tends to 0 in probability. The proof of Proposition \ref{P-consist_M} is essentially the same as the proof of Proposition \ref{P-consist} and hence is omitted. 
\begin{assumption} \label{A-consist_M}
(a) $\Theta_{M_0}$ and $\Theta_{M_0+1}$ are compact, and $\vartheta_{M_0}^*$ is in the interior of $\Theta_{M_0}$. (b) For all $(x,x') \in \mathcal{X}$ and all $(\overline{{\bf y}},y',w)\in \mathcal{Y}^s\times \mathcal{Y}\times \mathcal{W}$, $f(y'|\overline{\bf{y}}_0,w;\gamma,\theta)$ is continuous in $(\gamma,\theta)$. (c) $\mathbb{E}_{\vartheta^*_{M_0}}[ \log ( p_{\vartheta_{M_0}}(Y_1 |\overline{\bf{Y}}_{-m}^0,{\bf W}_{-m}^1)] = \mathbb{E}_{\vartheta^*_{M_0}}[ \log p_{\vartheta_{M_0}^*}(Y_1 |\overline{\bf{Y}}_{-m}^0,{\bf W}_{-m}^1) ]$ for all $m \geq 0$ if and only if $\vartheta_{M_0} = \vartheta_{M_0}^*$. (d) $\mathbb{E}_{\vartheta^*_{M_0}}[ \log ( p_{\vartheta_{M_0+1}}(Y_1 |\overline{\bf{Y}}_{-m}^0,{\bf W}_{-m}^0) ] = \mathbb{E}_{\vartheta^*_{M_0}} [\log p_{\vartheta_{M_0}^*}(Y_1 |\overline{\bf{Y}}_{-m}^0,{\bf W}_{-m}^1) ]$ for all $m \geq 0$ if and only if $\vartheta_{M_0+1}\in \Upsilon^*$. \end{assumption}

\begin{proposition} \label{P-consist_M} 
Suppose Assumptions \ref{assn_a1}, \ref{assn_a2}, and \ref{A-consist_M} hold. Then, under the null hypothesis of $M=M_0$, $\hat\vartheta_{M_0} \overset{p}{\rightarrow} \vartheta_{M_0}^*$ and $\inf_{\vartheta_{M_0+1} \in \Upsilon^*} |\hat{\vartheta}_{M_0+1}-\vartheta_{M_0+1}|\overset{p}{\rightarrow} 0$.
\end{proposition}

Let $LR_{M_0,n}:=2[\ell_n(\hat\vartheta_{M_0+1},\xi_{M_0+1})-\ell_n(\hat\vartheta_{M_0},\xi_{M_0})]$ denote the LRTS for testing $H_0:M=M_0$ against $H_A:M=M_0+1$.
We proceed to derive the asymptotic distribution of the LRTS by analyzing the behavior of the LRTS when $\vartheta_{M_0+1} \in \Upsilon^*_m$ for each $m$. Define $J_m :=\{m,m+1\}$. Observe that if ${\bf X}_1^k \in J_m^k$, then ${\bf X}_1^k$ follows a two-state Markov chain on $J_m$ whose transition probability is characterized by $\alpha_m := \mathbb{P}_{\vartheta_{M_0+1}}(X_k=m|X_k \in J_m)$ and $\varrho_m := \text{corr}_{\vartheta_{M_0+1}}(X_{k-1},X_{k}|(X_{k-1},X_{k}) \in J_m^2)$. Collect reparameterized $\pi_{xm}$ into $\pi_{xm} := (\varrho_m,\alpha_m, \phi_{m}')'$, where $\phi_m$ does not affect the transition probability of ${\bf X}_1^k$ when ${\bf X}_1^k \in J_m^{k}$. See Section \ref{subsec:p_m_repara} in the appendix for the detailed derivation.
 
Define $q_{kj} := \mathbb{I}\{X_k=j\}$; then, we can write $\alpha_m$ and $\varrho_m$ as $\alpha_m = \mathbb{E}_{\vartheta_{M_0+1}}(q_{km}|X_k \in J_m)$ and $\varrho_m =\text{corr}_{\vartheta_{M_0+1}}(q_{k-1,m},q_{km}|(X_{k-1},X_{k}) \in J_m^2)$. Because $\overline{{\bf Y}}_{-\infty}^\infty$ provides no information for distinguishing between $X_k=m$ and $X_k={m+1}$ if $\theta_m = \theta_{m+1}$, we can write $\alpha_m$ and $\varrho_m$ as
\begin{equation} \label{alpha_rho_h}
\alpha_m = \mathbb{E}_{\vartheta_{M_0+1}}(q_{km}|X_k \in J_m,\overline{{\bf Y}}_{-\infty}^\infty )\quad \text{and}\quad
\varrho_m = \text{corr}_{\vartheta_{M_0+1}}(q_{k-1,m},q_{km}|(X_{k-1},X_{k}) \in J_m^2,\overline{{\bf Y}}_{-\infty}^\infty). 
\end{equation}

\subsection{Non-normal distribution} \label{sec:M-nonnormal}

For non-normal component distributions, consider the following reparameterization similar to (\ref{repara}):
\begin{equation*}
\begin{pmatrix}
\theta_m\\
\theta_{m+1} \\
\end{pmatrix}
=
\begin{pmatrix}
\nu_m + (1-\alpha_m) \lambda_m\\
\nu_m- \alpha_m\lambda_m\\
\end{pmatrix}. 
\end{equation*} 
Collect the reparameterized identified parameters into one vector $\psi_m: = (\eta_m',\lambda_m')'$, where $\eta_m=(\gamma',\{\theta_j'\}_{j=1}^{m-1}, \nu_m',\{\theta_j'\}_{j=m+2}^{M_0+1},\vartheta_{xm}')'$, so that the reparameterized $(M_0+1)$-regime log-likelihood function is $\ell_n(\psi_m,\pi_{xm},\xi_{M_0+1})$. Let $\psi_m^{*}=(\eta_m^*,\lambda_m^*)=((\vartheta_{M_0}^*)',0')'$ denote the value of $\psi_{m}$ under $H_{0m}$. Define the reparameterized conditional density of $y_k$ as 
\begin{align*}
g_{\psi_m}(y_k|\overline{\bf{y}}_{k-1}, x_k) & : = \mathbb{I}\{x_k \in J_m\} f(y_k|\overline{\bf{y}}_{k-1};\gamma,\nu_m + (q_{km}-\alpha_m) \lambda_m) + \sum_{j \in \overline J_m} q_{kj} f(y_k|\overline{\bf{y}}_{k-1}; \gamma,\theta_j),
\end{align*}
where $\overline J_m := \{1,\ldots,M_0+1\}\setminus J_m$. Let $f_{mk}^{*}$ denote $f(Y_k|\overline{\bf{Y}}_{k-1};\gamma^*,\theta_m^*)$. It follows from (\ref{alpha_rho_h}) and the law of iterated expectations that
\begin{equation} \label{qk_moments_M}
\begin{aligned}
&\mathbb{E}_{\vartheta^*_{M_0}}\left[\frac{\mathbb{I}\{X_k \in J_m\}(q_{km} -\alpha_m)}{g_{\psi_m^*}(Y_k|\overline{\bf{Y}}_{k-1}, X_k)} \middle|\overline{{\bf Y}}_{-\infty}^n\right] \\ 
&=\mathbb{E}_{\vartheta^*_{M_0}}\left[\mathbb{E}_{\vartheta^*_{M_0}}\left[\frac{ q_{km} -\alpha_m }{f_{mk}^*}\middle|X_k \in J_m,\overline{{\bf Y}}_{-\infty}^n\right] \mathbb{I}\{X_k \in J_m\} \middle|\overline{{\bf Y}}_{-\infty}^n\right] = 0, \\
&\mathbb{E}_{\vartheta^*_{M_0}}\left[\frac{\mathbb{I}\{X_{t_1} \in J_m\}\mathbb{I}\{X_{t_2} \in J_m\}(q_{t_1h}-\alpha_m) (q_{t_2h}-\alpha_m)}{g_{\psi_m^*}(Y_{t_1}|\overline{\bf{Y}}_{{t_1}-1}, X_{t_1})g_{\psi_m^*}(Y_{t_2}|\overline{\bf{Y}}_{{t_2}-1}, X_{t_2})} \middle| \overline{\bf Y}_{-\infty}^n \right] \\
& = \mathbb{E}_{\vartheta^*_{M_0}}\left[\mathbb{E}_{\vartheta^*_{M_0}}\left[\frac{(q_{t_1h}-\alpha_m) (q_{t_2h}-\alpha_m)}{f_{m t_1}^* f_{m t_2}^*} \middle|{\bf X}_{t_1}^{t_2} \in J_m^{t_2-t_1+1},\overline{{\bf Y}}_{-\infty}^n\right] \mathbb{I}\{(X_{t_1},X_{t_2}) \in J_m^2\} \middle|\overline{{\bf Y}}_{-\infty}^n \right] \\
& = \frac{\alpha_m(1-\alpha_m)\varrho_m^{t_2-t_1}}{f_{m t_1}^* f_{m t_2}^*} \mathbb{P}_{\vartheta^*_{M_0}}((X_{t_1},X_{t_2}) \in J_m^2|\overline{{\bf Y}}_{-\infty}^n), \quad t_2 \geq t_1,
\end{aligned}
\end{equation} 
where the second equality holds because $g_{\psi_m^*}(Y_k|\overline{\bf{Y}}_{k-1}, X_k)=f_{mk}^*$ if $X_k \in J_m$, and the last equality holds because, conditional on $\{{\bf X}_{t_1}^{t_2} \in J_m^{t_2-t_1+1},\overline{\bf Y}_{-\infty}^n\} $, ${\bf X}_{t_1}^{t_2}$ is a two-state stationary Markov process with parameter $(\alpha_m,\varrho_m)$.

Let $g_{0k}^*$, $q_{0k}^*$, and $\overline p_{0k}^*$ denote $g_{\vartheta_{M_0,y}^*}(Y_k,X_k| \overline{{\bf Y}}_{k-1},X_{k-1})$, $q_{\vartheta_{M_0,x}^*}(X_{k-1},X_k)$, and $\overline p_{\vartheta_{M_0}^*}(Y_k| \overline{{\bf Y}}_{0}^{k-1})$. Let $\nabla g_{0k}^*$ denote the derivative of $g_{\vartheta_{M_0,y}}(Y_k,X_k| \overline{{\bf Y}}_{k-1},X_{k-1})$ evaluated at $\vartheta_{M_0,y}^*$, and define $\nabla q_{0k}^*$ and $\nabla \overline p_{0k}^*$ similarly. Repeating a derivation similar to (\ref{dg})--(\ref{d2p}) but using (\ref{qk_moments_M}) in place of (\ref{qk_moments}), we obtain 
\begin{equation} \label{d1p_M}
\begin{aligned}
& \nabla_{\eta_m} \overline p_{\psi^*_m \pi} (Y_k| \overline{{\bf Y}}_{0}^{k-1})/\overline p_{\psi^*_m \pi} (Y_k| \overline{{\bf Y}}_{0}^{k-1}) \\
& =
 \sum_{t=1}^k \mathbb{E}_{\vartheta^*} \left[\nabla_{\vartheta_{M_0}} \log (g_{0t}^*q_{0t}^*) \middle|\overline{{\bf Y}}_{0}^k \right] - \sum_{t=1}^{k-1} \mathbb{E}_{\vartheta^*} \left[\nabla_{\vartheta_{M_0}} \log (g_{0t}^* q_{0t}^*)\middle|\overline{{\bf Y}}_{0}^{k-1} \right] \\
& = \nabla_{\vartheta_{M_0}} \overline p_{\vartheta_{M_0}^*} (Y_k| \overline{{\bf Y}}_{0}^{k-1})/\overline p_{\vartheta_{M_0}^*} (Y_k| \overline{{\bf Y}}_{0}^{k-1}), 
\end{aligned}
\end{equation} 
\begin{align}
& \nabla_{\lambda_m} \overline p_{\psi^*_m \pi} (Y_k| \overline{{\bf Y}}_{0}^{k-1}) / \overline p_{\psi^*_m \pi} (Y_k| \overline{{\bf Y}}_{0}^{k-1}) =0, \quad \nabla_{{\lambda_m} \eta_m'} \overline p_{\psi^*_m \pi} (Y_k| \overline{{\bf Y}}_{0}^{k-1})/\overline p_{\psi^*_m \pi} (Y_k| \overline{{\bf Y}}_{0}^{k-1})=0, \label{d2p_M0} \\
& \frac{\nabla_{{\lambda_m} {\lambda_m} '} \overline p_{\psi^*_m \pi} (Y_k| \overline{{\bf Y}}_{0}^{k-1})}{\overline p_{\psi^*_m \pi} (Y_k| \overline{{\bf Y}}_{0}^{k-1})}
 = \alpha_m(1-\alpha_m) \frac{ \nabla_{\theta\theta'}f_{mk}^{*}}{f_{mk}^{*}} \mathbb{P}_{\vartheta^*_{M_0}}(X_k \in J_m|\overline{{\bf Y}}_{0}^k) \nonumber \\
& \quad + \alpha_m(1-\alpha_m) \sum_{t=1}^{k-1} \varrho_{m}^{k-t} \left( \frac{\nabla_{\theta} f_{mt}^{*}}{f_{mt}^{*}} \frac{\nabla_{\theta'} f_{mk}^{*}}{f_{mk}^{*}} + \frac{\nabla_{\theta} f_{mk}^{*}}{f_{mk}^{*}} \frac{\nabla_{\theta'} f_{mt}^{*}}{f_{mt}^{*}} \right) \mathbb{P}_{\vartheta^*_{M_0}}((X_t,X_k) \in J_m^2|\overline{{\bf Y}}_{0}^k).\label{d2p_M}
\end{align} 

Define $\tilde \varrho := (\varrho_1,\ldots,\varrho_{M_0})'$, define $t_{\lambda}(\lambda_m,\pi_m)$ as $t_{\lambda}(\lambda,\pi)$ in (\ref{score}) by replacing $(\lambda,\pi)$ with $(\lambda_m,\pi_m)$, and let
\begin{equation} \label{stilde}
t(\psi_m,\pi_m) := 
\begin{pmatrix}
\eta_m - \eta^* \\
 t_{\lambda}(\lambda_m,\pi_m) 
\end{pmatrix}, \ 
\tilde{s}_{\tilde\varrho k} :=
\begin{pmatrix}
\tilde {s}_{\eta k} \\
\tilde {s}_{\lambda \tilde\varrho k}
\end{pmatrix},\ \text{ where }\
\tilde s_{\eta k} : = 
\frac{\nabla_{\eta_m} \overline p_{\psi^*_m \pi} (Y_k| \overline{{\bf Y}}_{0}^{k-1}) }{\overline p_{\psi^*_m \pi} (Y_k| \overline{{\bf Y}}_{0}^{k-1})},
\ 
\tilde{s}_{\lambda \tilde\varrho k} :=
\begin{pmatrix}
{s}_{\lambda \varrho_1 k}^1 \\
\vdots \\
{s}_{\lambda\varrho_{M_0} k}^{M_0}
\end{pmatrix},
\end{equation}
and $s_{\lambda \varrho_m k}^m := V(s_{\lambda\lambda \varrho_m k}^m)$, where $s_{\lambda\lambda \varrho_m k}^m $ is defined similarly to (\ref{score_lambda}) as
\begin{equation} \label{score_lambda_M}
\begin{aligned}
s_{\lambda\lambda\varrho_m k}^m &: = \frac{\nabla_{\theta \theta'} f_{mk}^{*}}{f_{mk}^{*}} \mathbb{P}_{\vartheta^*_{M_0}} (X_k \in J_m|\overline{{\bf Y}}_{0}^k) \\
& + \sum_{t=1}^{k-1} \varrho_m ^{k-t} \left( \frac{\nabla_{\theta} f_{mt}^{*}}{f_{mt}^{*}} \frac{\nabla_{\theta'} f_{mk}^{*}}{f_{mk}^{*}} + \frac{\nabla_{\theta} f_{mk}^{*}}{f_{mk}^{*}} \frac{\nabla_{\theta'} f_{mt}^{*}}{f_{mt}^{*}} \right) \mathbb{P}_{\vartheta^*_{M_0}} ((X_t,X_k) \in J_m^2|\overline{{\bf Y}}_{0}^k).
\end{aligned}
\end{equation}
Similarly to (\ref{I_lambda}), define 
\begin{equation} \label{I_tilde_lambda}
\begin{aligned}
\tilde {\mathcal{I}}_{\eta} &:= \mathbb{E}_{\vartheta^*_{M_0}}(\tilde s_{\eta k} \tilde s_{\eta k}'), \quad \tilde{\mathcal{I}}_{\lambda\tilde \varrho_1\tilde \varrho_2} := \lim_{k\rightarrow\infty} \mathbb{E}_{\vartheta^*_{M_0}}(\tilde s_{\lambda\tilde \varrho_1 k} \tilde s_{\lambda\tilde \varrho_2 k}'), \quad
\tilde{\mathcal{I}}_{\lambda\eta \tilde \varrho } := \lim_{k\rightarrow\infty} \mathbb{E}_{\vartheta^*_{M_0}}(\tilde s_{\lambda \tilde \varrho k} \tilde s_{\eta k}'), \\ \tilde{\mathcal{I}}_{\eta \lambda \tilde \varrho} &:= \tilde{\mathcal{I}}_{\lambda \eta \tilde \varrho}', \quad \tilde{\mathcal{I}}_{\lambda.\eta \tilde\varrho_1\tilde\varrho_2}:= \tilde{\mathcal{I}}_{\lambda \tilde\varrho_1\tilde\varrho_2}-\tilde{\mathcal{I}}_{\lambda\eta \tilde\varrho_1}\tilde{\mathcal{I}}_{\eta}^{-1}\tilde{\mathcal{I}}_{\eta\lambda \tilde\varrho_2},\quad \tilde{\mathcal{I}}_{\lambda.\eta \varrho_m}^m := \mathbb{E}_{\vartheta^*_{M_0}}[G_{\lambda.\eta \varrho_m}^m (G_{\lambda.\eta \varrho_m}^m)'],\\
Z^m_{\lambda \varrho_m} &:= (\tilde{\mathcal{I}}_{\lambda.\eta \varrho_m}^m)^{-1}G_{\lambda.\eta \varrho_m}^m,
\end{aligned}
\end{equation}
where ${G}_{\lambda.\eta \tilde \varrho}=((G_{\lambda.\eta \varrho_1}^1)',\ldots,(G_{\lambda.\eta\varrho_{M_0}}^{M_0})')'$ is an $M_0 q_\lambda$-vector mean zero Gaussian process with \\$\text{cov}({G}_{\lambda.\eta \tilde\varrho_1},{G}_{\lambda.\eta \tilde \varrho_2}) = \tilde{\mathcal{I}}_{\lambda.\eta \tilde\varrho_1\tilde\varrho_2}$. Note that $G_{\lambda.\eta \tilde \varrho}$ corresponds to the residuals from projecting $\tilde s_{\lambda \tilde \varrho k}$ on $\tilde s_{\eta k}$. Define $\tilde t_{\lambda \varrho_m}^m$ by
\begin{equation*} 
g_{\lambda \varrho_m}^m(\tilde t_{\lambda \varrho_m}^m) = \inf_{t_\lambda \in v(\mathbb{R}^q)}g_{\lambda \varrho_m}^m(t_{\lambda }), \quad g_{\lambda \varrho_m}^m(t_{\lambda}) := (t_{\lambda} -Z_{\lambda \varrho_m}^m)' \tilde{\mathcal{I}}_{\lambda.\eta \varrho_m}^m (t_{\lambda} -Z_{\lambda \varrho_m}^m). 
\end{equation*}
The following proposition gives the asymptotic null distribution of the LRTS. Under the stated assumptions, the log-likelihood function permits a quadratic approximation in the neighborhood of $\Upsilon_{m}^*$ similar to the one in Proposition \ref{P-quadratic}. Define $A_{n\varepsilon c}^m(\xi) := \{\vartheta_{M_0+1} \in \Theta_{M_0+1} : \{ \ell_n(\psi_m,\pi_m,\xi) - \ell_n(\psi_m^*,\pi_m,\xi) \geq 0 \} \land |t(\psi_m,\pi_m)|< \varepsilon \} \cup \mathcal{N}_{c/\sqrt{n}}$. Under $H_0: M=M_0$, for any $c>0$, for $m=1,\ldots,M_0$, and uniformly in $\xi \in \Xi$ and $\vartheta_{M_0+1} \in A_{n\varepsilon c}^m(\xi)$,
\[
\ell_n(\psi_m,\pi_m,\xi) - \ell_n(\psi_m^*,\pi_m,\xi) \ - \sqrt{n}t(\psi_m,\pi_m)' \nu_n (s_{\varrho_m k}) + n t(\psi_m,\pi_m)' \mathcal{I}_{\varrho_m} t(\psi_m,\pi_m)/2 = o_{p \varepsilon} (1),
\]
where $s_{ \varrho_m k} := (\tilde{s}_{\eta k}',({s}_{\lambda \varrho_m k}^m)')'$ and $\mathcal{I}_{\varrho_m} = \lim_{k\rightarrow\infty} \mathbb{E}_{\vartheta^*_{M_0}}(s_{ \varrho_m k}s_{ \varrho_m k}')$. Consequently, the LRTS is asymptotically distributed as the maximum of the $M_0$ random variables, each of which represents the asymptotic distribution of the LRTS that tests $H_{0m}$. Denote the parameter space for $\varrho_m$ by $\Theta_{\varrho_m}$, and let $\tilde\Theta_{\varrho}:=\Theta_{\varrho_1}\times\ldots\times\Theta_{\varrho_{M_0}}$.
\begin{assumption} \label{A-nonsing_M}
$0< \inf_{\tilde \varrho \in \tilde\Theta_{\varrho}} \lambda_{\min}(\tilde{\mathcal{I}}_{\tilde\varrho}) \leq \sup_{\tilde \varrho \in \tilde\Theta_{\varrho}} \lambda_{\max}(\tilde{\mathcal{I}}_{\tilde\varrho}) < \infty$ for $\tilde{\mathcal{I}}_{\tilde\varrho}:= \lim_{k\rightarrow\infty}\mathbb{E}_{\vartheta^*_{M_0}}(\tilde{s}_{\tilde\varrho k}\tilde{s}_{\tilde\varrho k}')$, where $\tilde{s}_{\tilde\varrho k}$ is given in (\ref{stilde}).
\end{assumption}
\begin{proposition} \label{P-LR_M}
Suppose Assumptions \ref{assn_a1}, \ref{assn_a2}, \ref{assn_a4}, \ref{A-consist_M}, and \ref{A-nonsing_M} hold. Then, under $H_0: M=M_0$, $LR_{M_0,n}\overset{d}{\rightarrow} \max_{m=1,\ldots,M_0}\left\{ \sup_{\varrho_m \in \Theta_{\varrho}^m} \left( (\tilde t_{\lambda \varrho_m}^m)' \tilde{\mathcal{I}}_{\lambda.\eta \varrho_m}^m \tilde t_{\lambda \varrho_m}^m \right) \right\} $.
\end{proposition}

\subsection{Heteroscedastic normal distribution}\label{subsec:hetero_normal_M}

As in Section \ref{subsec:hetero_normal}, we assume that $Y_k \in \mathbb{R}$ in the $j$-th regime follows a normal distribution with the regime-specific intercept and variance of which density is given by (\ref{normal-density}). Consider the following reparameterization similar to (\ref{repara2}):
\begin{equation*}
\left(
\begin{array}{c}
\zeta_m\\
\zeta_{m+1}\\
\sigma_m^2\\
\sigma_{m+1}^2
\end{array}
\right) =\left(
\begin{array}{c}
\nu_{\zeta m} + (1-\alpha_m)\lambda_{\zeta m} \\
\nu_{\zeta m} -\alpha_m\lambda_{\zeta m}\\
\nu_{\sigma m} + (1- \alpha_m)(2\lambda_{\sigma m}+ C_1 \lambda_{\mu m}^2)\\
\nu_{\sigma m} - \alpha_m(2\lambda_{\sigma m}+ C_2 \lambda_{\mu m}^2)
\end{array}
\right), 
\end{equation*}
where $\nu_{\zeta m}=(\nu_\mu,\nu_{\beta}')'$, $\lambda_{\zeta m}=(\lambda_{\mu m},\lambda_{\beta m}')'$, $C_1 := -(1/3)(1 + \alpha_m)$, and $C_2 := (1/3)(2 - \alpha_m)$. As in Section \ref{sec:M-nonnormal}, we collect the reparameterized identified parameters into $\psi_m: = (\eta_m',\lambda_m')'$, where $\eta_m=(\gamma',\{\theta_j'\}_{j=1}^{m-1}, \nu_{\zeta m}', \nu_{\sigma m},\{\theta_j'\}_{j=m+2}^{M_0+1},\vartheta_{xm}')'$ and $\lambda_m:=(\lambda_{\zeta m}',\lambda_{\sigma m})'$. Similar to (\ref{repara_g_hetero}), define the reparameterized conditional density of $y_k$ as
\begin{equation*}
\begin{aligned}
& g_{\psi_m}(y_k|\overline{\bf{y}}_{k-1}, x_k) = \sum_{j \in \overline J_m} q_{kj} f(y_k|\overline{\bf{y}}_{k-1}; \gamma,\theta_j) \\
& \quad + \mathbb{I}\{x_k \in J_m\} f\left(y_k|\overline{\bf{y}}_{k-1};\gamma,\nu_{\zeta m}+(q_{km}-\alpha_m)\lambda_{\zeta m}, \nu_{\sigma m} + (q_{km}-\alpha_m)(2\lambda_{\sigma m} + (C_2-q_{km}) \lambda_{\mu m}^2 )\right).
\end{aligned} 
\end{equation*} 
Let $g_{mk}^*$, $f_{mk}^*$, $\nabla g_{mk}^*$, and $\nabla f_{mk}^*$ denote $g_{\psi_m^{*}}(Y_k|\overline{\bf{Y}}_{k-1}, X_k)$, $f(Y_k|\overline{\bf Y}_{k-1};\gamma^*,\theta_m^*)$, $\nabla g_{\psi_m^{*}}(Y_k|\overline{\bf{Y}}_{k-1}, X_k)$, and $\nabla f(Y_k|\overline{\bf Y}_{k-1};\gamma^*,\theta_m^*)$. From (\ref{d3g}) and a derivation similar to (\ref{qk_moments_M}), we obtain the following result that corresponds to (\ref{d3g2}) in testing homogeneity:
\begin{equation} \label{d3g2_M}
\begin{aligned}
\mathbb{E}_{\vartheta^*_{M_0}} \left[ \nabla_{\lambda_{\mu m}^i} g_{mk}^{*} / g_{mk}^* \middle|\overline{{\bf Y}}_{-\infty}^k \right] &= 0,\quad i = 1,2,3,\\
\mathbb{E}_{\vartheta^*_{M_0}} \left[ \nabla_{\lambda_{\mu m}^4} g_{mk}^{*} / g_{mk}^{*} \middle| \overline{{\bf Y}}_{-\infty}^k \right] &= \alpha_m(1 - \alpha_m) b(\alpha_m)( \nabla_{\mu^4}f_{mk}^* / f_{mk}^* )\mathbb{P}_{\vartheta^*_{M_0}}(X_k \in J_m|\overline{{\bf Y}}_{-\infty}^k) \\
&= b(\alpha_m) \mathbb{E}_{\vartheta^*_{M_0}}\left[ \nabla_{\lambda_{\sigma m}^2} g_{mk}^{*} / g_{mk}^{*} \middle|\overline{{\bf Y}}_{-\infty}^k\right].
\end{aligned}
\end{equation}
Repeating the calculation leading to (\ref{d1p_M})--(\ref{d2p_M}) and using (\ref{d3g2_M}) gives the following. First, (\ref{d1p_M}) and (\ref{d2p_M0}) still hold; second, the elements of $\nabla_{\lambda_m\lambda_m'} \overline p_{\psi^* \pi} (Y_k| \overline{{\bf Y}}_{0}^{k-1})/\overline p_{\psi^* \pi} (Y_k| \overline{{\bf Y}}_{0}^{k-1})$ except for the $(1,1)$-th element are given by (\ref{d2p_M}) while adjusting the derivative with respect to $\lambda_{\sigma m}$ by multiplying by 2; third,
\begin{equation*} 
\frac{\nabla_{\lambda_{\mu m}^2} \overline p_{\psi^*_m \pi}(Y_k| \overline{{\bf Y}}_{0}^{k-1})}{\overline p_{\psi^*_m \pi} (Y_k| \overline{{\bf Y}}_{0}^{k-1})} = \alpha_m(1-\alpha_m) \sum_{t=1}^{k-1} \varrho_m^{k-t} \left( 2 \frac{\nabla_{\mu} f_{mt}^*}{f_{mt}^*} \frac{\nabla_{\mu} f_{mk}^*}{f_{mk}^*}\right)\mathbb{P}_{\vartheta^*_{M_0}}((X_t,X_k) \in J_m^2|\overline{{\bf Y}}_{0}^k).
\end{equation*} 

For $m \geq 0$, define $\zeta_{k,m}^m(\varrho_m):= \sum_{t=-m+1}^{k-1} \varrho_m^{k-t-1} 2 (\nabla_{\mu} f_{mt}^* \nabla_{\mu} f_{mk}^* / f_{mt}^* f_{mk}^*) \mathbb{P}_{\vartheta^*_{M_0}}((X_t,X_k) \in J_m^2|\overline{{\bf Y}}_{0}^k)$. Similarly to (\ref{score_lambda_beta}), define the elements of the generalized score as 
\begin{equation} 
\begin{aligned}
&
\begin{pmatrix}
 * & s_{ \lambda_{\mu\beta}\varrho_m k}^m & s_{\lambda_{\mu\sigma} \varrho_m k}^m \\
s_{ \lambda_{\beta\mu}\varrho_m k}^m & s_{ \lambda_{\beta\beta}\varrho_m k}^m & s_{ \lambda_{\beta\sigma}\varrho_m k}^m \\
s_{ \lambda_{\sigma\mu}\varrho_m k}^m & s_{ \lambda_{\beta\sigma}\varrho_m k}^m & s_{ \lambda_{\sigma\sigma}\varrho_m k}^m 
\end{pmatrix} : = \frac{\nabla_{\theta \theta'} f_{mk}^{*}}{f_{mk}^{*}} \mathbb{P}_{\vartheta^*_{M_0}} (X_k \in J_m|\overline{{\bf Y}}_{0}^k) \\
& \qquad + \sum_{t=1}^{k-1} \varrho_m ^{k-t} \left( \frac{\nabla_{\theta} f_{mt}^{*}}{f_{mt}^{*}} \frac{\nabla_{\theta'} f_{mk}^{*}}{f_{mk}^{*}} + \frac{\nabla_{\theta} f_{mk}^{*}}{f_{mk}^{*}} \frac{\nabla_{\theta'} f_{mt}^{*}}{f_{mt}^{*}} \right) \mathbb{P}_{\vartheta^*_{M_0}} ((X_t,X_k) \in J_m^2|\overline{{\bf Y}}_{0}^k).
\end{aligned}
\end{equation}
Similarly to (\ref{score_normal}), define $\tilde{s}_{\tilde \varrho k}$ as in (\ref{stilde}) by redefining $s_{\lambda \varrho_m k}^m$ in (\ref{stilde}) as
\begin{equation} \label{stilde_normal}
s_{\lambda \varrho_m k}^m
:= \begin{pmatrix}
\zeta_{k,0}^m(\varrho_m)/2 &
2 s_{\lambda_{\mu\sigma} \varrho_m k}^m &
2 s_{\lambda_{\sigma\sigma} \varrho_m k}^m&
(s_{\lambda_{\beta \mu} \varrho_m k}^m)' &
2 (s_{\lambda_{\beta \sigma} \varrho_m k}^m)' &
V(s_{\lambda_{\beta\beta} \varrho _h k}^m)'
\end{pmatrix} '. 
\end{equation}
Define ${\mathcal{I}}_{\lambda.\eta \varrho_m}^m$ and $Z^m_{\lambda \varrho_m}$ as in (\ref{I_tilde_lambda}) with $s_{\lambda \varrho_m k}^m$ defined in (\ref{stilde_normal}). Let $Z_{\lambda 0}^m$ and $\mathcal{I}_{\lambda.\eta 0}^m$ denote $Z_{\lambda \varrho_m}^m$ and $\mathcal{I}_{\lambda.\eta \varrho_m}^m$ evaluated at $\varrho_m=0$. Define $\Lambda_\lambda^1$ as in (\ref{Lambda-lambda}), and define $\Lambda_{\lambda \varrho_m}^2$ as in (\ref{Lambda-lambda}) by replacing $\varrho$ with $\varrho_m$. Similar to (\ref{t-lambda-N1}), define $\tilde{t}_{\lambda}^{m1}$ and $\tilde{t}_{\lambda \varrho_m}^{m2}$ by $r_{\lambda}(\tilde{t}_{\lambda}^{m1}) = \inf_{t_{\lambda} \in \Lambda_{\lambda}^1}r_{\lambda}^m(t_{\lambda})$ and $r_{\lambda\varrho_m}(\tilde{t}_{\lambda \varrho_m}^{m2}) = \inf_{t_{\lambda} \in \Lambda_{\lambda\varrho_m}^2}r_{\lambda\varrho_m}^m(t_{\lambda})$, where $r_{\lambda}^m(t_{\lambda}) := (t_{\lambda} -Z_{\lambda 0}^m)' \mathcal{I}_{\lambda.\eta 0}^m (t_{\lambda} -Z_{\lambda 0}^m)$ and $r_{\lambda\varrho_m}^m(t_{\lambda}) := (t_{\lambda} -Z_{\lambda \varrho_m}^m)' \mathcal{I}_{\lambda.\eta \varrho_m}^m (t_{\lambda} -Z_{\lambda \varrho_m}^m)$.

The following proposition establishes the asymptotic null distribution of the LRTS. As in the non-normal case, the LRTS is asymptotically distributed as the maximum of the $M_0$ random variables. 
\begin{assumption} \label{A-nonsing_M_normal}
Assumption \ref{A-nonsing_M} holds when $\tilde{s}_{\tilde \varrho_m k}$ is given in (\ref{stilde_normal}).
\end{assumption}
\begin{proposition} \label{P-LR_M_normal}
Suppose Assumptions \ref{assn_a1}, \ref{assn_a2}, \ref{assn_a4}, \ref{A-consist_M}, and \ref{A-nonsing_M_normal} hold and the component density of the $j$-th regime is given by (\ref{normal-density}). Then, under $H_0:m=M_0$, 
$LR_{M_0,n} \overset{d}{\rightarrow} \max_{m=1,\ldots,M_0}\{ \max\{ \mathbb{I}\{\varrho_m=0\} (\tilde t_{\lambda }^{m1})' \mathcal{I}_{\lambda.\eta 0}^m \tilde t_{\lambda }^{m1}, \sup_{\varrho_m \in \Theta_{\varrho}^m} (\tilde t_{\lambda \varrho_m}^{m2})' \mathcal{I}_{\lambda.\eta \varrho_m}^m \tilde t_{\lambda \varrho_m}^{m2} \}\}$.
\end{proposition} 

\subsection{Homoscedastic normal distribution}

As in Section \ref{subsec:homo_normal}, we assume that $Y_k\in \mathbb{R}$ in the $j$-th regime follows a normal distribution with the regime-specific intercept and common variance whose density is given by (\ref{normal-density-homo}).

The asymptotic distribution of the LRTS is derived by using a reparameterization
\begin{equation*}
\left(
\begin{array}{c}
\theta_m\\
\theta_{m+1}\\
\sigma^2
\end{array}
\right) =\left(
\begin{array}{c}
\nu_{\theta m} + (1-\alpha_m)\lambda_m \\
\nu_{\theta m} -\alpha_m\lambda_m \\ 
\nu_{\sigma m} - \alpha_m(1-\alpha_m) \lambda_{\mu m}^2
\end{array}
\right), 
\end{equation*}
similar to (\ref{repara-homo}) and following the derivation in Sections \ref{subsec:homo_normal} and \ref{subsec:hetero_normal_M}. For brevity, we omit the details of the derivation. Define $s_{\lambda\lambda \varrho_m k}^m$ as in (\ref{score_lambda_M}), and denote each element of $s_{\lambda\lambda \varrho_m k}^m$ as
\begin{equation*}
s_{\lambda\lambda \varrho_m k}^m =
\begin{pmatrix}
* & s_{\lambda_{\mu\beta} \varrho_m k}^m \\
s_{\lambda_{\beta\mu} \varrho_m k}^m & s_{\lambda_{\beta\beta} \varrho_m k}^m \\
\end{pmatrix}.
\end{equation*}
Similarly to (\ref{score_normal_homo}), define $\tilde{s}_{\tilde \varrho k}$ as in (\ref{stilde}) by redefining $s_{\lambda \varrho_m k}^m$ in (\ref{stilde}) as
\begin{equation}\label{stilde_normal_homo}
s_{\lambda \varrho_m k}^m
:= 
\begin{pmatrix}
\zeta_{k,0}^m(\varrho_m)/2 & 
s_{\lambda_{\mu}^3 k}^m/3! & 
s_{\lambda_{\mu}^4 k}^m/4! &
(s_{\lambda_{\beta \mu} \varrho k}^m)' & 
V(s_{\lambda_{\beta\beta} \varrho k}^m)'
\end{pmatrix} ',
\end{equation}
where $s_{\lambda_{\mu}^i k}^m:= \mathbb{P}_{\vartheta^*_{M_0} (X_k \in J_m|\overline{{\bf Y}}_{0}^k)} \nabla_{\mu^i}f(Y_k|\overline{\bf Y}_{k-1};\gamma^*,\theta_m^*)/f(Y_k|\overline{\bf Y}_{k-1};\gamma^*,\theta_m^*)$ for $i=3,4$.

The following proposition establishes the asymptotic null distribution of the LRTS.
\begin{assumption} \label{A-nonsing_M_normal_homo}
Assumption \ref{A-nonsing_M} holds when $\tilde{s}_{\tilde \varrho_m k}$ is given in (\ref{stilde_normal_homo}).
\end{assumption}
\begin{proposition} \label{P-LR_M_normal_homo}
Suppose Assumptions \ref{assn_a1}, \ref{assn_a2}, \ref{assn_a4}, \ref{A-consist_M}, and \ref{A-nonsing_M_normal_homo} hold and the component density of the $j$-th regime is given by (\ref{normal-density-homo}). Then, under $H_0:m=M_0$, $ LR_{M_0,n} \overset{d}{\rightarrow} \max_{m=1,\ldots,M_0}\{ \max\{ \mathbb{I}\{\varrho_m=0\} (\tilde t_{\lambda }^{m1})' \mathcal{I}_{\lambda.\eta 0}^m \tilde t_{\lambda }^{m1}, \sup_{\varrho_m \in \Theta_{\varrho_m\epsilon}} (\tilde t_{\lambda \varrho_m}^{m2})' \mathcal{I}_{\lambda.\eta \varrho_m}^m \tilde t_{\lambda \varrho_m}^{m2} \}\}$, where $\tilde{t}_{\lambda}^{m1}$ and $\tilde{t}_{\lambda \varrho_m}^{m2}$ are defined as in Proposition \ref{P-LR_M_normal} but in terms of $(Z_{\lambda \varrho_m}^m, \mathcal{I}_{\lambda.\eta \varrho_m}^m, Z_{\lambda 0}^m, \mathcal{I}_{\lambda.\eta 0}^m)$ constructed with $s_{\lambda \varrho_m k}^m$ given in (\ref{stilde_normal_homo}) and $\Lambda_\lambda^1$ and $\Lambda_{\lambda \varrho_m}^2$ defined as in (\ref{Lambda-lambda-homo}) but replacing $\varrho$ with $\varrho_m$. 
\end{proposition}

\section{Asymptotic distribution under local alternatives} 
In this section, we derive the asymptotic distribution of our LRTS under local alternatives. While we focus on the case of testing $H_0: M=1$ against $H_A: M=2$, it is straightforward to extend the analysis to the case of testing $H_0: M=M_0$ against $H_A: M=M_0+1$ for $M_0\geq 2$.

Given $\pi\in \Theta_{\pi}$, we define a local parameter $h:=\sqrt{n}t(\psi,\pi)$, so that
\[
h =
\begin{pmatrix}
h_\eta\\
h_\lambda 
\end{pmatrix}=
\begin{pmatrix}
\sqrt{n} (\eta-\eta^*)\\
\sqrt{n} t_{\lambda}(\lambda,\pi)
\end{pmatrix},
\] 
where $t_{\lambda}(\lambda,\pi)$ differs across the different models and is given by (\ref{score_lambda}), (\ref{t_lambda_rho}), and (\ref{t-psi2-homo}). Given $h=(h_\eta',h_{\lambda}')'$ and $\pi\in\Theta_{\pi}$, we consider the sequence of contiguous local alternatives $\vartheta_n = (\psi_n',\pi_n')' = (\eta_n',\lambda_n',\pi_n')'\in\Theta_{\eta}\times\Theta_\lambda\times\Theta_{\pi}$ such that
\begin{equation}\label{local-alternative}
h_{\eta} = \sqrt{n}(\eta_n-\eta^*),\quad h_\lambda =\sqrt{n} t_{\lambda}(\lambda_{n},\pi_n)+o(1),\quad \text{and}\ \pi_n - \pi =o(1).
\end{equation} 

Let $\mathbb{P}_{\vartheta,x_0}^n$ be the probability measure on $\{Y_k\}_{k=1}^n$ under $\vartheta$ conditional on the value of $\overline{\bf Y}_0$, $X_0$, and ${\bf W}_1^n$. Then, the log-likelihood ratio is given by
\[
\log \frac{d\mathbb{P}_{\vartheta_n,x_0}^n}{d \mathbb{P}_{\vartheta^*,x_0}^n} = \ell_n(\psi_n,\pi_n,x_0)-\ell_n(\psi^*,\pi,x_0)= \log \left( \frac{\sum_{{\bf x}_1^{n}} \prod_{k=1}^n f_k(\eta_n, \lambda_n) q_{\pi_n}(x_{k-1},x_k) }{ \prod_{k=1}^n f_k(\eta^*,0) } \right),
\] 
where $f_k(\eta,\lambda)$ is defined by the right-hand side of (\ref{repara_g}), (\ref{repara_g_hetero}), and (\ref{repara_g_homo}) for the models of the non-normal distribution, heteroscedastic normal distribution, and homoscedastic normal distribution, respectively. The following result follows from Le Cam's first and third lemmas and facilitates the derivation of the asymptotic distribution of the LRTS under $\mathbb{P}_{\vartheta_n,x_0}^n$.

\begin{proposition} \label{P-LAN} Suppose that the assumptions of Propositions \ref{P-quadratic}, \ref{P-quadratic-N1}, and \ref{P-quadratic-N1-homo} hold for the models of the non-normal, heteroscedastic normal, and homoscedastic normal distributions, respectively. Then, uniformly in $x_0 \in \mathcal{X}$, (a) $\mathbb{P}_{\vartheta_n,x_0}^n$ is mutually contiguous with respect to $\mathbb{P}_{\vartheta^*,x_0}^n$, and (b) under $\mathbb{P}_{\vartheta_n,x_0}^n$, we have $\log (d\mathbb{P}_{\vartheta_n,x_0}^n / d \mathbb{P}_{\vartheta^*,x_0}^n) = h' \nu_n(s_{\varrho_n k}) - \frac{1}{2}h' \mathcal{I}_{\varrho} h + o_{p}(1)$ with $\nu_n(s_{\varrho_n k}) \overset{d}{\rightarrow} N(\mathcal{I}_{ \varrho}h, \mathcal{I}_{ \varrho})$. 
\end{proposition} 

\subsection{Non-normal distribution}
 
For the non-normal distribution, the sequence of contiguous local alternatives is given by $\lambda_n = \bar \lambda/n^{1/4}$ because then $h_\lambda = \sqrt{n}\alpha(1-\alpha)v(\lambda_n)=\alpha(1-\alpha)v(\bar\lambda)$ holds. The following proposition derives the asymptotic distribution of the LRTS for the non-normal distribution under $H_{1 n}: (\pi_n,\eta_n,\lambda_n) = (\bar\pi,\eta^*,\bar \lambda/n^{1/4})$. 
 
\begin{proposition} \label{P-LAN2} Suppose that the assumptions of Proposition \ref{P-LR} hold. 
For $\bar \pi\in\Theta_{\pi}$ and $\bar \lambda\neq 0$, define $h_{\lambda}:= \bar\alpha(1-\bar\alpha) v(\bar \lambda)$. Then, under $H_{1 n}: (\pi_n,\eta_n,\lambda_n) = (\bar\pi,\eta^*,\bar \lambda/n^{1/4})$, we have $LR_n \overset{d}{\rightarrow} \sup_{\varrho \in \Theta_{\varrho}} (\tilde t_{\lambda\varrho h})' \mathcal{I}_{\lambda.\eta\varrho} \tilde t_{\lambda\varrho h}$, where $\tilde t_{\lambda \varrho h}$ is defined as in (\ref{t-lambda}) but replacing $Z_{\lambda\varrho}$ in (\ref{t-lambda}) with $(\mathcal{I}_{\lambda.\eta\varrho})^{-1} G_{\lambda.\eta\varrho}+ h_{\lambda}$.
\end{proposition}

\subsection{Heteroscedastic normal distribution}

For the model with the heteroscedastic normal distribution, the sequences of contiguous local alternatives characterized by (\ref{local-alternative}) include the local alternatives of order $n^{-1/8}$. 
\begin{proposition} \label{P-LAN3} Suppose that the assumptions of Proposition \ref{P-LR-N1} hold for model (\ref{normal-density}). For $\bar \varrho \in (-1,1)$, $\bar \alpha \in (0,1)$, and $\bar \lambda := (\bar \lambda_\mu,\bar \lambda_\sigma,\bar \lambda_\beta')'\neq (0,0,0)'$, let
\begin{align*}
& H_{1 n}^a: (\varrho_n,\alpha_n,\eta_n,\lambda_{\mu n}, \lambda_{\sigma n},\lambda_{\beta n}) = (\bar \varrho/n^{1/4}, \bar\alpha,\eta^*, \bar \lambda_\mu/ n^{1/8}, \bar \lambda_{\sigma}/n^{3/8},\bar \lambda_{\beta}/n^{3/8}), \\
& H_{1 n}^b: (\varrho_n,\alpha_n,\eta_n,\lambda_{\mu n}, \lambda_{\sigma n},\lambda_{\beta n}) = (\bar \varrho, \bar\alpha,\eta^*, \bar \lambda_\mu/ n^{1/4}, \bar \lambda_{\sigma}/n^{1/4},\bar \lambda_{\beta}/n^{1/4}),
\end{align*}
and define
\begin{align*}
h_{\lambda}^a: &=
\bar\alpha(1-\bar\alpha) \times (\bar \varrho \bar \lambda_\mu^2,
 \bar \lambda_\mu \bar \lambda_\sigma,
 b(\bar\alpha) \bar \lambda_\mu^4/12,
\bar \lambda_{\beta}' \bar \lambda_\mu, 0, 0)', \\
h_{\lambda}^b: &= 
 \bar\alpha(1-\bar\alpha) \times (\bar \varrho\bar \lambda_\mu^2,
 \bar \lambda_\mu \bar \lambda_\sigma,
 \bar \lambda_\sigma^2,
 \bar \lambda_{\beta}' \bar \lambda_\mu,
\bar \lambda_\beta' \bar \lambda_{\sigma},
 v( \bar \lambda_\beta)')'.
\end{align*}
Then, for $j \in \{a,b\}$, under $H_{1 n}^j$, we have $LR_n \overset{d}{\rightarrow} \max\{ \mathbb{I}\{\varrho=0\} (\tilde t_{\lambda h}^{1j})' \mathcal{I}_{\lambda.\eta 0} \tilde t_{\lambda h}^{1j}, \sup_{\varrho \in \Theta_{\varrho}} (\tilde t_{\lambda \varrho h}^{2j})' \mathcal{I}_{\lambda.\eta \varrho} \tilde t_{\lambda \varrho h}^{2j} \}$, where $\tilde t_{\lambda h}^{1j}$ and $\tilde t_{\lambda \varrho h}^{2j}$ are defined as in (\ref{t-lambda-N1}) but replacing $Z_{\lambda\varrho}$ with $(\mathcal{I}_{\lambda.\eta\varrho})^{-1} G_{\lambda.\eta\varrho}+ h_{\lambda}^j$.
\end{proposition} 

In the local alternative $H_{1 n}^a$, $\varrho_n$ converges to $0$, and $\lambda_{\mu n}$ converges to 0 at a slower rate than $n^{-1/4}$. Our test has non-trivial power against these local alternatives in the neighborhood of $\varrho=0$. By contrast, the test of \citet{carrasco14em} does not have power against the local alternatives in the neighborhood of $\varrho=0$, as discussed in Section 5 of \citet{carrasco14em}. The test proposed by \citet{quzhuo17wp} assumes that $\varrho$ is bounded away from zero and hence their test rules out $H_{1n}^a$.

\subsection{Homoscedastic normal distribution} 
 
The local alternatives for the model with the homoscedastic distribution also include those of order $n^{-1/8}$ in the neighborhood of $\varrho=0$.

\begin{proposition} \label{P-LAN4} Suppose that the assumptions of Proposition \ref{P-quadratic-N1-homo} hold for model (\ref{normal-density-homo}). For $\bar \varrho \in (-1,1)$, $\bar \alpha \in (0,1)$, $\Delta_\alpha \neq 0$, and $\bar \lambda := (\bar \lambda_\mu,\bar \lambda_\beta')'\neq (0,0)'$, let
\begin{align*}
& H_{1 n}^a: (\varrho_n,\alpha_n,\eta_n,\lambda_{\mu n}, \lambda_{\beta n}) = (\bar \varrho/n^{1/4}, 1/2+\Delta_{\alpha}/n^{1/8}, \eta^*, \bar \lambda_\mu/ n^{1/8}, \bar \lambda_{\beta}/n^{3/8}), \\
& H_{1 n}^b: (\varrho_n,\alpha_n,\eta_n,\lambda_{\mu n}, \lambda_{\beta n}) = (\bar \varrho, \bar\alpha,\eta^*, \bar \lambda_\mu/ n^{1/4}, \bar \lambda_{\beta}/n^{1/4}),
\end{align*}
and define $h_{\lambda}^a := (1/4) \times (\bar\varrho \bar\lambda_{\mu}^2,\Delta_{\alpha} \bar\lambda_{\mu}^3,-\bar\lambda_\mu^4/2, \bar\lambda_\beta'\bar\lambda_\mu, 0)'$ and $h_{\lambda}^b := \bar\alpha(1-\bar\alpha) \times (\bar \varrho\bar \lambda_\mu^2, 0, 0, \bar \lambda_{\beta}' \bar \lambda_\mu, v( \bar \lambda_\beta)')'$. For $j=\{a,b\}$, define $\tilde t_{\lambda h}^{1j}$ and $\tilde t_{\lambda \varrho h}^{2j}$ as in (\ref{t-lambda-N1}) but replacing $Z_{\lambda\varrho}$ with $(\mathcal{I}_{\lambda.\eta\varrho})^{-1} G_{\lambda.\eta\varrho}+ h_{\lambda}^j$, where $\mathcal{I}_{\lambda.\eta\varrho}$ and $G_{\lambda.\eta\varrho}$ are constructed with $s_{\varrho k}$ defined in (\ref{score_normal_homo}), and $\Lambda_{\lambda }^1$ and $\Lambda_{\lambda \varrho}^2$ are defined in (\ref{Lambda-lambda-homo}). Then, under $H_{1 n}^j$, we have $LR_n \overset{d}{\rightarrow} \max\{ \mathbb{I}\{\varrho=0\} (\tilde t_{\lambda h}^{1j})' \mathcal{I}_{\lambda.\eta 0} \tilde t_{\lambda h}^{1j}, \sup_{\varrho \in \Theta_{\varrho}} (\tilde t_{\lambda \varrho h}^{2j})' \mathcal{I}_{\lambda.\eta \varrho} \tilde t_{\lambda \varrho h}^{2j} \}$.
\end{proposition} 

\section{Parametric bootstrap}
We consider the following parametric bootstrap to obtain the bootstrap critical value $c_{\alpha,B}$ and bootstrap $p$-value of our LRTS for testing $H_0: M=M_0$ against $H_A: M=M_0+1$.
\begin{enumerate}
\item Using the observed data, compute $\hat \vartheta_{M_0}$, $\hat\vartheta_{M_0+1}$, and $LR_{M_0,n}$.
\item Given $\hat \vartheta_{M_0}$ and $\xi_{M_0}$, generate $B$ independent samples $\{Y_1^b,\ldots,Y_n^b\}_{b=1}^B$ under $H_0$ with $\vartheta_{M_0}=\hat\vartheta_{M_0}$ conditional on the observed value of $\overline{\bf Y}_{0}$ and ${\bf W}_{1}^n$.
\item For each simulated sample $\{Y_k^b\}_{k=1}^n$ with $(\overline{\bf Y}_0,{\bf W}_1^n)$, estimate $\hat \vartheta_{M_0}^b$ and $\hat\vartheta_{M_0+1}^b$ as in Step 1, and let $LR_{M_0,n}^b := 2[\ell_n(\hat\vartheta_{M_0+1}^b ,\xi_{M_0+1})-\ell_n(\hat\vartheta_{M_0}^b ,\xi_{M_0})]$ for $b=1,\ldots,B$.
\item Let $c_{\alpha,B}$ be the $(1-\alpha)$ quantile of $\{LR_{M_0,n}^b\}_{b=1}^B$, and define the bootstrap $p$-value as $B^{-1}\sum_{b=1}^B \mathbb{I}\{ LR_{M_0,n}^b > LR_{M_0, n}\}$.
\end{enumerate}

The following proposition shows the consistency of the bootstrap critical values $c_{\alpha,B}$ for testing $H_0: M_0=1$. We omit the result for testing $H_0: M_0\geq 2$; it is straightforward to extend the analysis to the case for $M_0\geq 2$ with more tedious notations. 
 
\begin{proposition} \label{P-bootstrap} 
Suppose that the assumptions of Propositions \ref{P-quadratic}, \ref{P-quadratic-N1}, and \ref{P-quadratic-N1-homo} hold for the models of the non-normal, heteroscedastic normal, and homoscedastic normal distributions, respectively. Then, the bootstrap critical values $c_{\alpha,B}$ converge to the asymptotic critical values in probability as $n$ and $B$ go to infinity under $H_0$ and under the local alternatives described in Propositions \ref{P-LAN2}, \ref{P-LAN3}, and \ref{P-LAN4}.
\end{proposition}

\section{Simulations and empirical application}

\subsection{Simulations}

We consider the following two models: 
\begin{align} 
\text{Model 1}:&\quad Y_{k} = \mu_{X_k} + \beta Y_{k-1} +\varepsilon_{k},\quad \varepsilon_k\sim N(0,\sigma^2),\label{model1}\\
\text{Model 2}:&\quad Y_{k} = \mu_{X_k} + \beta Y_{k-1} +\varepsilon_{k},\quad \varepsilon_k\sim N(0,\sigma_{X_k}^2),\label{model2}
\end{align}
where $X_k\in\{1,\ldots,M\}$ with $p_{ij} = p(X_k=i|X_{k-1}=j)$. Model 1 in (\ref{model1}) is similar to the model used in \citet{chowhite07em}. This model has a switching intercept but the variance parameter $\sigma^2$ does not switch across regimes. In Model 2 in (\ref{model2}), both the intercept and the variance parameters switch across regimes.

We compare the size and power property of our bootstrap LRT and that of the QLR test of \citet{chowhite07em} with $\Theta_\mu = [-2,2]$ and the supTS test of \citet{carrasco14em} with $\rho \in [-0.9. 0.9]$. The critical values are computed by bootstrap with $B=199$ bootstrap samples. Note that this comparison favors the LRT over the supTS test because the supTS test is designed to detect general parameter variation including Markov chain.

In Model 2, we set the lower bound of $\sigma$ as $\epsilon_\sigma = 0.01\hat \sigma$, where $\hat \sigma$ is the estimate of $\sigma$ from one-regime model. We also find that adding a penalty term to the log-likelihood function improves the finite sample property of the LRT. The penalty term prevents $\sigma_j$ from taking an extremely small value and takes the form 
$-a_n \sum_{j = 1}^{M} \{\hat\sigma^2/\sigma_j^2+\log(\sigma_j^2/\hat\sigma^2) - 1\}$. 
We set $a_n = 20n^{-1/2}$ and compute the test statistic using the log-likelihood function without the penalty term. Because the penalty term and its derivatives are $o_p(1)$ from the compactness of $\Theta_M$, adding this penalty term does not affect the consistency of the MLE or the asymptotic distribution of the LRTS.

We first examine the rejection frequency of $H_0:M=1$ against $H_A:M=2$ when the data are generated by $H_0: M=1$ with $(\beta,\mu,\sigma)=(0.5,0,1)$. The first panel in Table \ref{table:bootstrap-size} reports the rejection frequency of the bootstrap tests at the nominal 10\%, 5\%, and 1\% levels over $3000$ replications with $n=200$ and $500$. Overall, all the tests have good sizes. 
 
Table \ref{table:bootstrap-power-model1} reports the power of the three tests for testing the null hypothesis of $M=1$ at the nominal level of 5\%. We generate $3000$ data sets for $n= 500$ under the alternative hypothesis of $M=2$ by setting $\mu_1=0.2, 0.6$, and $1.0$ and $\mu_2=-\mu_1$, while $(p_{11},p_{22})=(0.25,0.25)$, $(0.50,0.50)$, $(0.70,0.70)$, and $(0.90,0.90)$. We set $\sigma=1$ for Model 1 and $(\sigma_1,\sigma_2)=(1.1,0.9)$ for Model 2. For Model 1, the power of all the tests increases as $\mu_1$ increases except for the supTS test with $(p_{11},p_{22})=(0.9,0.9)$. As $(p_{11},p_{22})$ moves away from $(0.5,0.5)$, the power of the LRT increases, whereas the QLRT has decreasing power. The LRT performs better than the supTS and QLR tests except for the case with $(p_{11},p_{22})=(0.25,0.25)$, where the supTS test performs very well, and the case with $(p_{11},p_{22})=(0.5,0.5)$, where the QLRT outperforms the LRT, because the true model is a finite mixture in this case. The last three columns of Table \ref{table:bootstrap-power-model1} report the power of the LRT and supTS test to detect alternative models with switching variances (i.e., Model 2 with $M=2$). We did not examine the power of the QLRT because this test assumes non-switching variance. The LRT has stronger power than the supTS test in most cases.
 
The second panel in Table \ref{table:bootstrap-size} reports the rejection frequency of the LRT for testing $H_0:M=2$ against $H_A:M=3$ when the data are generated under the null hypothesis, showing its good size property, while neither the QLRT nor the supTS test is applicable for testing $H_0:M=2$ against $H_A:M=3$.

Table \ref{table:bootstrap-power-model2} reports the power of our LRT for testing the null hypothesis of $M=2$ at the nominal level of 5\%. We generate $3000$ data sets for $n= 500$ under the alternative hypothesis of $M=3$ across different values of $(\mu_1,\mu_2,\mu_3)$ and $(p_{11},p_{22},p_{33})$ with $p_{ij} = (1-p_{ii})/2$ for $j\neq i$, where we set $(\beta,\sigma)=(0.5,1.0)$ for Model 1 and $(\beta,\sigma_1,\sigma_2) = (0.5,0.9,1.2)$ for Model 2. Similar to the case of $H_0:M=1$, the power of the LRT for testing $H_0: M=2$ against $H_A: M=3$ increases when the alternative is further away from $H_0$ or when latent regimes become more persistent.

\subsection{Empirical example} 
Using U.S. GDP per capita quarterly growth rate data from 1960Q1 to 2014Q4, we estimate the regime switching models with common variance (i.e., Model 1 in (\ref{model1})) and with switching variances (i.e., Model 2 in (\ref{model2})) for $M=1$, $2$, $3$, and $4$ and sequentially test the null hypothesis of $M=M_0$ against the alternative hypothesis $M=M_0+1$ for $M_0=1$, $2$, $3$, and $4$.\footnote{For both models, we restrict the parameter values for the transition probabilities by setting $\epsilon=0.05$ to prevent the issue of unbounded likelihood.} 
We also report the Akaike information criteria (AIC) and Bayesian information criteria (BIC) for reference, although, to our best knowledge,
the consistency of the AIC and BIC for selecting the number of regimes has not been established in the literature.

Table \ref{table:select} reports the selected number of regimes by the AIC, BIC, and LRT. For Model (\ref{model1}) with common variance, our LRT selects $M=4$, while the AIC and BIC select $M=3$ and $M=1$, respectively. For Model (\ref{model2}) with switching variance, both the LRT and the AIC select $M=3$, while the BIC selects $M=2$.

Panel A of Table \ref{table:parameter} and Figure \ref{fig:common} report the parameter estimates and posterior probabilities of being in each regime for the model with common variance for $M=2$, $3$, and $4$. Across the different specifications in $M$, the estimated values of $\mu_1$, $\mu_2, \ldots, \mu_M$ are well separated in the common variance model, indicating that each regime represents a booming or recession period with different degrees. In Figure \ref{fig:common}, when the number of regimes is specified as $M=2$, the posterior probability of the ``recession'' regime (Regime 1) against that of the ``booming'' regime (Regime 2) sharply rises during the collapse of Lehman Brothers in 2008 and then declines after 2009. When the number of regimes is specified as $M=3$, in addition to the ``recession'' and ``booming'' regimes corresponding to Regimes 1 and 2, respectively, the regime with a rapid change in the growth rates from low to high is captured by Regime 3; for the model with $M=3$ in Figure \ref{fig:common}, the posterior probability of Regime 3 rises in late 2009 when the U.S. economy started to recover from the Lehman collapse. When the number of regimes is specified as $M=4$, Regime 1 now captures a rapid change in the growth rates from high to low, where the posterior probability of Regime 1 becomes high when the growth rate of the U.S. economy rapidly declined in the middle of the Lehman collapse. The LRT selects the model with four regimes, which capture the rapid changes in the growth rates of U.S. GDP per capita during the Lehman collapse period.
 
Panel B of Table \ref{table:parameter} and Figure \ref{fig:switching} report the parameter estimates and posterior probabilities of being in each regime for the model with switching variance, respectively. When the number of regimes is specified as $M=2$, the variance parameter estimates are very different between the two regimes, while the intercept estimates are similar, indicating that Regime 1 is the ``low volatility'' regime, while Regime 2 is the ``high volatility'' regime. When the number of regimes is specified as $M=3$, different regimes capture different states of the U.S. economy in terms of both growth rates and volatilities.\footnote{We may test the null hypothesis of $\sigma_1=\sigma_2=\sigma_3$ in the model with switching variance given $M=3$ by the standard LRT with the critical value obtained from the chi-square distribution with two degrees of freedom. With $LRTS = 2\times (-297.01+307.39)=20.76$, the null hypothesis of $\sigma_1=\sigma_2=\sigma_3$ is rejected at the 1\% significance level, suggesting that the model with switching variance is more appropriate than the model with common variance.} Regime 1 is characterized by the negative value of the intercept with high volatility, capturing a recession period. Regime 2 is characterized by the positive value of the intercept with low volatility, capturing a booming/stable economy. Regime 3 is characterized by a high value of the intercept and high variance, capturing both a rapid recovery in the growth rates and high volatility in the aftermath of the Lehman collapse in 2009. The LRT selects the model with three regimes when the model is specified with switching variance.

\section{Appendix}

Henceforth, for notational brevity, we suppress ${\bf W}_{a}^{b}$ from the conditioning variables and conditional densities when doing so does not cause confusion.

\subsection{Proof of propositions and corollaries}
 
\begin{proof}[Proof of Proposition \ref{uconlike}]
The proof is essentially identical to the proof of Lemma 2 in DMR. Therefore, the details are omitted. The only difference from DMR is (i) we do not impose Assumption (A2) of DMR, but this does not affect the proof because Assumption (A2) is not used in the proof of Lemma 2 in DMR, and (ii) we have ${\bf W}_1^n$, but our Lemma \ref{x_diff}(a) extends Corollary 1 of DMR to accommodate the $W_k$'s. Consequently, the argument of the proof of DMR goes through.
\end{proof}
 
\begin{proof}[Proof of Proposition \ref{Ln_thm1}] 

Define $h_{\vartheta k x_0} := \sqrt{l_{\vartheta k x_0}}-1$. By using the Taylor expansion of $2 \log(1+x) = 2x - x^2(1+o(1))$ for small $x$, we have, uniformly in $x_0 \in \mathcal{X}$ and $\vartheta \in \mathcal{N}_{c/\sqrt{n}}$,
\begin{equation}\label{ell_appn}
\ell_n(\psi,\pi,x_0) - \ell_n(\psi^*, \pi, x_0) = 2 \sum_{k=1}^n \log(1+h_{\vartheta k x_0}) = n P_n(2h_{\vartheta k x_0} - [1+o_p(1)] h_{\vartheta k x_0}^2).
\end{equation}
The stated result holds if we show that
\begin{align} 
&\sup_{x_0 \in \mathcal{X}} \sup_{\vartheta \in \mathcal{N}_{c/\sqrt{n}}} \left| nP_n(h_{\vartheta k x_0}^2) - n t_\vartheta' \mathcal{I}_\pi t_\vartheta/4 \right| = o_p(1)\quad\text{and}\label{hk_appn}\\
&\sup_{x_0 \in \mathcal{X}}\sup_{\vartheta \in \mathcal{N}_{c/\sqrt{n}}} | nP_n(h_{\vartheta k x_0}) - \sqrt{n} t_\vartheta' \nu_n(s_{\pi k})/2 + nt_\vartheta \mathcal{I}_\pi t_\vartheta'/8 | = o_p(1),\label{hk_appn2}
\end{align}
because then the right-hand side of (\ref{ell_appn}) is equal to $\sqrt{n} t_\vartheta' \nu_n (s_{\pi k}) - t_\vartheta \mathcal{I}_\pi t_\vartheta'/2 + o_p(1)$ uniformly in $x_0 \in \mathcal{X}$ and $\vartheta \in \mathcal{N}_{c/\sqrt{n}}$.

We first show (\ref{hk_appn}). Let $m_{\vartheta k} := t_\vartheta' s_{\pi k} + r_{\vartheta k}$, so that $l_{\vartheta k x_0} -1 = m_{\vartheta k} + u_{\vartheta k x_0}$. Observe that
\begin{equation} \label{B2}
\max_{1\leq k \leq n} \sup_{\vartheta \in \mathcal{N}_{c/\sqrt{n}}} |m_{\vartheta k}| = \max_{1\leq k \leq n} \sup_{\vartheta \in \mathcal{N}_{c/\sqrt{n}}} |t_\vartheta' s_{\pi k} + r_{\vartheta k}| = o_p(1), 
\end{equation}
from Assumptions \ref{assn_a3}(a) and (c) and Lemma \ref{max_bound}. Write $4P_n ( h_{\vartheta k x_0}^2)$ as
\begin{equation}\label{B0}
4P_n ( h_{\vartheta k x_0}^2) = P_n \left(\frac{4(l_{\vartheta k x_0}-1)^2}{(\sqrt{l_{\vartheta k x_0}} + 1)^2}\right) = P_n(l_{\vartheta k x_0}-1)^2 - P_n \left( (l_{\vartheta k x_0}-1)^3 \frac{(\sqrt{l_{\vartheta k x_0}}+3)}{(\sqrt{l_{\vartheta k x_0}} + 1)^3}\right).
\end{equation}
It follows from Assumptions \ref{assn_a3}(a)(b)(c)(e)(f) and $(E|XY|)^2 \leq E|X|^2E|Y|^2$ that, uniformly in $\vartheta \in \mathcal{N}_\varepsilon$, 
\begin{equation} \label{lk_lln}
P_n(l_{\vartheta k x_0}-1)^2 = t_\vartheta'P_n(s_{\pi k}s_{\pi k}')t_\vartheta + 2 t_\vartheta' P_n[s_{\pi k} (r_{\vartheta k} + u_{\vartheta k x_0})] + P_n(r_{\vartheta k} + u_{\vartheta k x_0})^2 = t_\vartheta'P_n(s_{\pi k}s_{\pi k}')t_\vartheta + \zeta_{\vartheta n x_0}, 
\end{equation} 
where $\zeta_{\vartheta n x_0}$ satisfies $\sup_{x_0 \in \mathcal{X}}|\zeta_{\vartheta n x_0}|=O_p(|t_\vartheta|^2|\psi-\psi^*|)+O_p(n^{-1}|t_\vartheta||\psi-\psi^*|) + O_p(n^{-1}|\psi-\psi^*|^2)$. Then, (\ref{hk_appn}) holds because $\sup_{\pi \in \Theta_\pi}| P_n(s_{\pi k}s_{\pi k}') - \mathcal{I}_\pi| = o_p(1)$ and the second term on the right-hand side of (\ref{B0}) is bounded by, from (\ref{B2}), $P_n( m_{\vartheta k}^2) = t_\vartheta'\mathcal{I}_\pi t_\vartheta + o_p(|t_\vartheta|^2)$, and Assumption \ref{assn_a3}(e),
\begin{align*}
& \mathcal{C} \sup_{x_0 \in \mathcal{X}}\sup_{\vartheta \in \mathcal{N}_{c/\sqrt{n}}} P_n \left[ |m_{\vartheta k}|^3 + 3 m_{\vartheta k}^2 |u_{\vartheta k x_0}| + 3|m_{\vartheta k}| u_{\vartheta k x_0}^2 \right] + \mathcal{C} \sup_{x_0 \in \mathcal{X}}\sup_{\vartheta \in \mathcal{N}_{c/\sqrt{n}}} P_n (|u_{\vartheta k x_0}|^3) \\
& \leq o_p(1) \sup_{x_0 \in \mathcal{X}}\sup_{\vartheta \in \mathcal{N}_{c/\sqrt{n}}} P_n \left[ m_{\vartheta k}^2+u_{\vartheta k x_0}^2 \right] + \mathcal{C} \sup_{x_0 \in \mathcal{X}}\sup_{\vartheta \in \mathcal{N}_{c/\sqrt{n}}} P_n (|u_{\vartheta k x_0}|^3) = o_p(n^{-1}).
\end{align*}

We proceed to show (\ref{hk_appn2}). Consider the following expansion of $h_{\vartheta k x_0}$:
\begin{equation} \label{hk_1}
h_{\vartheta k x_0} = (l_{\vartheta k x_0}-1)/2 - h_{\vartheta k x_0}^2/2 = (t_\vartheta' s_{\pi k} + r_{\vartheta k}+ u_{\vartheta k x_0})/2 - h_{\vartheta k x_0}^2/2.
\end{equation}
Then, (\ref{hk_appn2}) follows from (\ref{hk_appn}), (\ref{hk_1}), and Assumptions \ref{assn_a3}(d) and (e), and the stated result follows.
\end{proof}

\begin{proof}[Proof of Proposition \ref{Ln_thm2}] 
For part (a), it follows from $\log(1+x) \leq x$ and $h_{\vartheta k x_0} = (l_{\vartheta k x_0}-1)/2 - h_{\vartheta k x_0}^2/2$ (see (\ref{hk_1})) that
\begin{equation} \label{Ln_ineq}
\ell_n(\psi,\pi, x_0)-\ell_n(\psi^*,\pi, x_0) = 2 \sum_{k=1}^n \log(1+h_{\vartheta k x_0}) \leq 2 n P_n(h_{\vartheta k x_0}) = \sqrt{n} \nu_n(l_{\vartheta k x_0}-1) - n P_n(h_{\vartheta k x_0}^2).
\end{equation} 
Observe that $h_{\vartheta k x_0}^2 =(l_{\vartheta k x_0}-1)^2/(\sqrt{l_{\vartheta k x_0}} + 1)^2 \geq \mathbb{I}\{l_{\vartheta k x_0} \leq \kappa \} (l_{\vartheta k x_0}-1)^2/ (\sqrt{\kappa}+1)^2$ for any $\kappa>0$. Therefore, \begin{equation}\label{pn_lower}
P_n (h_{\vartheta k x_0}^2)\geq (\sqrt{\kappa}+1)^{-2}P_n \left( \mathbb{I}\{l_{\vartheta k x_0} \leq \kappa\} (l_{\vartheta k x_0}-1)^2\right). 
\end{equation} 
Substituting (\ref{lk_lln}) into the right-hand side of (\ref{pn_lower}) gives
\begin{equation} \label{pn_lower_2}
P_n (h_{\vartheta k x_0}^2) \geq (\sqrt{\kappa}+1)^{-2} t_\vartheta' \left[ P_n(s_{\pi k}s_{\pi k}') - P_n(\mathbb{I}\{l_{\vartheta k x_0} > \kappa\}s_{\pi k}s_{\pi k}') \right] t_\vartheta + \zeta_{\vartheta n x_0}.
\end{equation}
From H\"older's inequality, we have $P_n(\mathbb{I}\{l_{\vartheta k x_0} > \kappa\}|s_{\pi k}|^2 ) \leq [P_n(\mathbb{I}\{l_{\vartheta k x_0} > \kappa\} )]^{\delta/(2+\delta)} [ P_n(|s_{\pi k}|^{2+\delta}) ]^{2/(2+\delta)}$. The right-hand side is no larger than $\kappa^{-\delta/(2+\delta)} O_p(1)$ uniformly in $x_0 \in \mathcal{X}$ and $\vartheta \in \mathcal{N}_{\varepsilon }$ because (i) it follows from $\kappa \mathbb{I}\{l_{\vartheta k x_0} > \kappa\} \leq l_{\vartheta k x_0}$ that $P_n(\mathbb{I}\{l_{\vartheta k x_0} > \kappa\} ) \leq \kappa^{-1} P_n(l_{\vartheta k x_0})$ and $\sup_{x_0 \in \mathcal{X}} \sup_{\vartheta \in \mathcal{N}_{\varepsilon }}|P_n(l_{\vartheta k x_0})-1| = o_p(1)$ from Assumptions \ref{assn_a3}(d)--(g), and (ii) $P_n(\sup_{\pi \in \Theta_\pi}|s_{\pi k}|^{2+\delta})=O_p(1)$ from Assumption \ref{assn_a3}(a). Consequently, $\mathbb{P}( \sup_{x_0 \in \mathcal{X}} \sup_{\vartheta \in \mathcal{N}_{\varepsilon}}P_n(\mathbb{I}\{l_{\vartheta k x_0} > \kappa\}|s_{\pi k}|^2 ) \geq \lambda_{\min}/4 ) \to 0$ as $\kappa \to \infty$, and hence we can write (\ref{pn_lower_2}) as $P_n (h_{\vartheta k x_0}^2) \geq \tau (1+o_p(1)) t_\vartheta' \mathcal{I}_{\pi} t_\vartheta + O_p(|t_\vartheta|^2|\psi-\psi^*|) + O_p(n^{-1})$ for $\tau:=(\sqrt{\kappa}+1)^{-2}/2>0$ by taking $\kappa$ sufficiently large. Because $\sqrt{n} \nu_n(l_{\vartheta k x_0}-1) = \sqrt{n} t_\vartheta'\nu_n (s_{\pi k}) + O_p(1)$ from Assumptions \ref{assn_a3}(d) and (e), it follows from (\ref{Ln_ineq}) that, uniformly in $x_0 \in \mathcal{X}$ and $\vartheta \in \mathcal{N}_{\varepsilon}$,
\begin{equation}\label{rk_lower2}
-\eta \leq \ell_n(\psi,\pi, x_0)-\ell_n(\psi^*,\pi, x_0) \leq \sqrt{n} t_\vartheta' \nu_n (s_{\pi k}) - \tau (1+o_p(1)) n t_\vartheta' \mathcal{I}_\pi t_\vartheta + O_p(n |t_\vartheta|^2 |\psi-\psi^*| ) + O_p(1).
\end{equation}
Let $T_{n}:= \mathcal{I}_{\pi}^{1/2}\sqrt{n} t_\vartheta$. From (\ref{rk_lower2}), Assumptions \ref{assn_a3}(b) and (g), and the fact $\psi-\psi^* \to 0$ if $t_\vartheta \to 0$, we obtain the following result: for any $\delta>0$, there exist $\varepsilon>0$ and $M,n_0<\infty$ such that
\begin{equation}
\mathbb{P}\left(\inf_{x_0 \in \mathcal{X}} \inf_{\vartheta \in \mathcal{N}_{\varepsilon}} \left( |T_n| M - (\tau/2) |T_n|^2 + M \right) \geq 0\right) \geq 1-\delta,\ \text{ for all }\ n > n_0.
\end{equation}
Rearranging the terms inside $\mathbb{P}(\cdot)$ gives $\sup_{x_0 \in \mathcal{X}} \sup_{\vartheta \in \mathcal{N}_{\varepsilon}} (|T_{n}|-(M/\tau))^2 \leq 2M/\tau+(M/\tau)^2$. Taking its square root gives $\mathbb{P}(\sup_{x_0 \in \mathcal{X}} \sup_{\vartheta \in \mathcal{N}_{\varepsilon}}|T_{n}| \leq M_1) \geq 1-\delta$ for a constant $M_1$, and part (a) follows. Part (b) follows from part (a) and Proposition \ref{Ln_thm1}.
\end{proof}

\begin{proof}[Proof of Corollary \ref{cor_appn}]
Because the logarithm is monotone, we have $\inf_{x_0 \in \mathcal{X}} \ell_n(\psi,\pi,x_0) \leq \ell_n(\psi,\pi,\xi) \leq \sup_{x_0 \in \mathcal{X}} \ell_n(\psi,\pi,x_0)$. Part (a) then follows from Proposition \ref{Ln_thm1}. For part (b), note that we have $\vartheta \in A_{n\varepsilon }(\xi,\eta)$ only if $\vartheta \in A_{n\varepsilon }(x_0,\eta)$ for some $x_0$. Consequently, part (b) follows from Proposition \ref{Ln_thm2}.
\end{proof}

\begin{proof}[Proof of Proposition \ref{proposition-ell}]
First, observe that parts (a) and (b) hold when the right-hand side is replaced with $K_j(k+m)^7\rho^{\lfloor(k+m-1)/24\rfloor}$ and $K_j(k+m)^7\rho^{\lfloor (k+m-1)/1340 \rfloor}$ by using Lemmas \ref{ell_lambda} and \ref{lemma-bound-1} and noting that $q_1=6q_0, q_2=5q_0, q_3=4q_0,\ldots, q_6=q_0$. For example, when $j=2$, we can bound $\sup_{x\in\mathcal{X}}\sup_{\vartheta\in \mathcal{N}^*}|\nabla^2\ell_{k,m,x}(\vartheta) - \nabla^j2\overline\ell_{k,m}(\vartheta)|$ from $\nabla^2 \ell_{k,m,x}(\vartheta)=\Delta^2_{1,k,m,x}(\vartheta)+\Delta^{1,1}_{2,k,m,x}(\vartheta)$, $\sup_{x\in \mathcal{X}} \sup_{\vartheta\in \mathcal{N}^*} |\Delta^{\mathcal{I}(j)}_{j,k,m,x}(\vartheta)-\overline\Delta^{\mathcal{I}(j)}_{j,k,m}(\vartheta)|\leq {K}_{\mathcal{I}(j)}(k+m)^{ 7 } \rho^{ \lfloor (k+m-1)/24 \rfloor}$, ${K}_{\mathcal{I}(j)} \in L^{r_{\mathcal{I}(j)}}(\mathbb{P}_{\vartheta^*})$, $r_{(2)} = q_{2} = 5q_0$, and $r_{(1,1)} = q_{1}/2 = 3q_0$. Second, letting $\rho_* = \rho^{1/1340}\mathbb{I}\{\rho>0\}$ and redefining $K_j$ gives parts (a) and (b). Parts (c) and (d) follow from Lemmas \ref{ell_lambda} and \ref{lemma-bound-1}. 
\end{proof}

\begin{proof}[Proof of Proposition \ref{lemma-omega}]
First, we prove part (a). The proof of part (b) is essentially the same as that of part (a) and hence omitted. Observe that
\begin{equation*} 
\begin{aligned}
\nabla l^j_{k,m,x}(\vartheta)- \nabla \overline l^j_{k,m}(\vartheta)& = \Psi^j_{k,m,x}(\vartheta)\left(\frac{ p_{\vartheta}(Y_k|\overline{\bf{Y}}_{-m}^{k-1},X_{-m}=x)}{ p_{\vartheta^*}(Y_k|\overline{\bf{Y}}_{-m}^{k-1},X_{-m}=x)}-\frac{ \overline p_{\vartheta}(Y_k|\overline{\bf{Y}}_{-m}^{k-1})}
{\overline p_{\vartheta^*}(Y_k|\overline{\bf{Y}}_{-m}^{k-1})}\right) \\
&\quad +\frac{\overline p_{\vartheta}(Y_k|\overline{\bf{Y}}_{-m}^{k-1})}
{\overline p_{\vartheta^*}(Y_k|\overline{\bf{Y}}_{-m}^{k-1})}\left(\Psi^j_{k,m,x}(\vartheta)- \overline{\Psi}^j_{k,m}(\vartheta)
\right),
\end{aligned}
\end{equation*}
where
\begin{equation*} 
\Psi^j_{k,m,x}(\vartheta) := \frac{\nabla^j p_{\vartheta}(Y_k|\overline{\bf{Y}}_{-m}^{k-1},X_{-m}=x)}{ p_{\vartheta}(Y_k|\overline{\bf{Y}}_{-m}^{k-1},X_{-m}=x)}, \quad \overline\Psi^j_{k,m}(\vartheta) := \frac{\nabla^j \overline p_{\vartheta}(Y_k|\overline{\bf{Y}}_{-m}^{k-1})}{\overline p_{\vartheta}(Y_k|\overline{\bf{Y}}_{-m}^{k-1})}.
\end{equation*}
In view of Lemma \ref{lemma-ratio} and H\"older's inequality, part (a) holds if, for $j=1,\ldots,6$, there exist random variables $(\{A_{j,k}\}_{k=1}^n, B_j )\in L^{q_0}(\mathbb{P}_{\vartheta^*})$ and $\rho_* \in (0,1)$ such that, for all $1 \leq k \leq n$ and $m \geq 0$, 
\begin{equation} \label{Psi_bound}
(A) \quad \sup_{m \geq 0}\sup_{x\in\mathcal{X}}\sup_{\vartheta\in \mathcal{N}^*} | \Psi^j_{k,m,x}(\vartheta)| \leq A_{j,k}, \quad (B) \quad \sup_{x\in\mathcal{X}}\sup_{\vartheta\in \mathcal{N}^*} |\Psi^j_{k,m,x}(\vartheta) - \overline\Psi^j_{k,m}(\vartheta)| \leq B_j (k+m)^{ 7 }\rho_*^{k+m-1}.
\end{equation}
We show (A) and (B). From (\ref{der_formula}) we have, suppressing $(\vartheta)$ and superscript $1$ from $\nabla^1 \ell_{k,m,x}$,
\begin{equation*} 
\begin{aligned}
\Psi^1_{k,m,x}&= \nabla\ell_{k,m,x},\quad \Psi^2_{k,m,x}= \nabla^2\ell_{k,m,x} +(\nabla\ell_{k,m,x} )^2,\\
\Psi^3_{k,m,x}&=\nabla^3\ell_{k,m,x} +3\nabla^2\ell_{k,m,x}
\nabla \ell_{k,m,x}
+(\nabla \ell_{k,m,x})^3,\\
\Psi^4_{k,m,x}&=\nabla^4\ell_{k,m,x}+4 \nabla^3 \ell_{k,m,x} \nabla \ell_{k,m,x}+3(\nabla^2\ell_{k,m,x})^2 +6 \nabla^2 \ell_{k,m,x} (\nabla\ell_{k,m,x})^2+ (\nabla\ell_{k,m,x})^4,\\
\Psi^5_{k,m,x} &= \nabla^5 \ell_{k,m,x}+5 \nabla^4 \ell_{k,m,x} \nabla \ell_{k,m,x}+10\nabla^3 \ell_{k,m,x} \nabla^2\ell_{k,m,x}+10\nabla^3 \ell_{k,m,x}(\nabla\ell_{k,m,x})^2 \\
& \quad +15(\nabla^2 \ell_{k,m,x})^2 \nabla \ell_{k,m,x} +10\nabla^2\ell_{k,m,x}(\nabla \ell_{k,m,x})^3+ (\nabla\ell_{k,m,x})^5,\\
\Psi^6_{k,m,x} & = \nabla^6 \ell_{k,m,x} + 6 \nabla^5 \ell_{k,m,x} \nabla \ell_{k,m,x} + 15 \nabla^4 \ell_{k,m,x} \nabla^2 \ell_{k,m,x} + 15\nabla^4 \ell_{k,m,x} (\nabla \ell_{k,m,x})^2 \\
& \quad +10( \nabla^3 \ell_{k,m,x} )^2 +60 \nabla^3 \ell_{k,m,x} \nabla^2 \ell_{k,m,x} \nabla \ell_{k,m,x} +20\nabla^3\ell_{k,m,x} (\nabla \ell_{k,m,x})^3 \\
& \quad +15( \nabla^2 \ell_{k,m,x} )^3 +45(\nabla^2 \ell_{k,m,x})^2( \nabla \ell_{k,m,x})^2 +15\nabla^2 \ell_{k,m,x}(\nabla \ell_{k,m,x})^4 + (\nabla\ell_{k,m,x})^6,
\end{aligned}
\end{equation*}
and $\overline\Psi^j_{k,m}$ is written analogously with $\nabla^j \overline\ell_{k,m}$ replacing $\nabla^j \ell_{k,m,x}$. Therefore, (A) of (\ref{Psi_bound}) follows from Proposition \ref{proposition-ell}(c) and H\"older's inequality. (B) of (\ref{Psi_bound}) follows from Proposition \ref{proposition-ell}(a)(c), the relation $ab-cd = a(b-c)-c(a-d)$, $a^n - b^n = (a-b)\sum_{i=0}^{n-1} (a^{n-1-i}b^i)$, and H\"older's inequality.

For part (c), the bound of $\nabla l^j_{k,m,x}(\vartheta)$ follows from writing $\nabla l^j_{k,m,x}(\vartheta) = [\overline p_{\vartheta}(Y_k|\overline{\bf{Y}}_{-m}^{k-1},X_{-m}=x)/ \overline p_{\vartheta^*}(Y_k|\overline{\bf{Y}}_{-m}^{k-1},X_{-m}=x) ] \Psi^j_{k,m,x}(\vartheta)$ and using (\ref{Psi_bound}) and Lemma \ref{lemma-ratio}. $\overline\nabla^j l_{k,m}(\vartheta)$ is bounded by a similar argument. Part (d) follows from parts (a)--(c), the completeness of $L^{q}(\mathbb{P}_{\vartheta^*})$, Markov's inequality, and the Borel--Cantelli Lemma. Part (e) follows from combining parts (a) and (b) and letting $m' \to \infty$ in part (b).
\end{proof}

\begin{proof}[Proof of Proposition \ref{P-consist}]
The consistency of $\hat \vartheta_1$ follows from Theorem 2.1 of \citet{neweymcfadden94hdbk} because (i) $\vartheta_1^*$ uniquely maximizes $\mathbb{E}_{\vartheta_1^*} \log f(Y_1|\overline{\bf Y}_0,W_1;\gamma,\theta)$ from Assumption \ref{A-consist}(c), and (ii) $\sup_{\vartheta_1 \in \Theta_1}|n^{-1}\ell_{0,n}(\vartheta_1) - \mathbb{E}_{\vartheta_1^*} \log f(Y_1|\overline{\bf Y}_0,W_1;\gamma,\theta)| \overset{p}{\rightarrow} 0$ and $\mathbb{E}_{\vartheta_1^*} \log f(Y_1|\overline{\bf Y}_0,W_1;\gamma,\theta)$ is continuous because $(Y_k,W_k)$ is strictly stationary and ergodic from Assumption \ref{assn_a1}(e) and $\mathbb{E}_{\vartheta_1^*} \sup_{\vartheta_1 \in \Theta_1}|\log f(Y_1|\overline{\bf Y}_0,W_1;\gamma,\theta)| <\infty$ from Assumption \ref{assn_a2}(c).
 
We proceed to prove the consistency of $\hat\vartheta_2$. Define, similarly to pp.\ 2265--6 in DMR, $\Delta_{k,m,x}(\vartheta_2):= \log p_{\vartheta_{2}}(Y_k |\overline{\bf{Y}}_{-m}^{k-1},{\bf W}_{-m}^k,X_{-m}=x)$, $\Delta_{k,m}(\vartheta_2):= \log p_{\vartheta_{2}}(Y_k |\overline{\bf{Y}}_{-m}^{k-1},{\bf W}_{-m}^k)$, $\Delta_{k,\infty}(\vartheta_2):=\lim_{m\to \infty} \Delta_{k,m}(\vartheta_2)$, and $\ell(\vartheta_2):=\mathbb{E}_{\vartheta_1^*}[\Delta_{0,\infty}(\vartheta)]$. Observe that Lemmas 3 and 4 as well as Proposition 2 of DMR hold for our $\{\Delta_{k,m,x}(\vartheta_2),\Delta_{k,m}(\vartheta_2),\Delta_{k,\infty}(\vartheta_2), \ell_n(\vartheta_2,x_0),\ell(\vartheta_2)\}$ under our assumptions because (i) our Assumption \ref{assn_a1}(e) can replace their Assumption (A2) in the proof of their Lemmas 3 and 4 and Proposition 2, and (ii) our Lemma \ref{x_diff}(a) extends Corollary 1 of DMR to accommodate the $W_k$'s. It follows that (i) $\ell(\vartheta_2)$ is maximized if and only if $\vartheta_2 \in \Gamma^*$ from Assumption \ref{A-consist}(d) because $\mathbb{E}_{\vartheta^*_{1}}[ \log p_{\vartheta_{2}}(Y_1 |\overline{\bf{Y}}_{-m}^0,{\bf W}_{-m}^1)]$ converges to $\ell(\vartheta_2)$ uniformly in $\vartheta_2$ as $m \to \infty$ from Lemma 3 of DMR and the dominated convergence theorem, (ii) $\ell(\vartheta_2)$ is continuous from Lemma 4 of DMR, and (iii) $\sup_{\xi_2}\sup_{\vartheta_2 \in \Theta_2} |n^{-1} \ell_n(\vartheta_2,\xi_2) - \ell(\vartheta_2)| \overset{p}{\rightarrow} 0$ holds from Proposition 2 of DMR and $\ell_n(\vartheta_2,\xi_2) \in [\min_{x_0} \ell_n(\vartheta_2,x_0), \max_{x_0} \ell_n(\vartheta_2,x_0)]$. Consequently, $\inf_{\vartheta_2\in\Gamma^*}|\hat\vartheta_2-\vartheta_2|\overset{p}{\rightarrow} 0$ follows from Theorem 2.1 of \citet{neweymcfadden94hdbk} with an adjustment for the fact that the maximizer of $\ell(\vartheta_2)$ is a set, not a singleton.
\end{proof} 
 
\begin{proof}[Proof of Proposition \ref{P-quadratic}]

We prove the stated result by applying Corollary \ref{cor_appn} to $l_{\vartheta k x_0} -1$ with $l_{\vartheta k x_0}$ defined in (\ref{density_ratio}). Because the first and second derivatives of $l_{\vartheta k x_0} -1$ play the role of the score, we expand $l_{\vartheta k x_0} -1$ with respect to $\psi$ up to the third order. Let $q = \dim(\psi)$. For a $k\times 1$ vector $a$, define $a^{\otimes p} := a \otimes a\otimes \cdots \otimes a$ ($p$ times) and $\nabla_{a^{\otimes p}} :=\nabla_{a} \otimes \nabla_{a}\otimes \cdots \otimes \nabla_{a}$ ($p$ times). Recall that the $(p+1)$-th order Taylor expansion of $f(x)$ with $x \in \mathbb{R}^q$ around $x=x^*$ is given by
\begin{align*}
f(x) & = f(x^*) + \sum_{j=1}^p \frac{1}{j!} \nabla_{(x^{\otimes j})'} f(x^*) (x-x^*)^{\otimes j} + \frac{1}{(p+1)!} \nabla_{(x^{\otimes (p+1)})'} f(\overline{x}) (x-x^*)^{\otimes (p+1)},
\end{align*}
where $\overline{x}$ lies between $x$ and $x^*$, and $\overline{x}$ may differ from element to element of $\nabla_{x^{\otimes (p+1)}} f(\overline{x})$.

Choose $\epsilon>0$ sufficiently small so that $\mathcal{N}_{\epsilon}$ is a subset of $\mathcal{N}^*$ in Assumption \ref{assn_a4}. For $m \geq 0$ and $j =1,2,\ldots$, let 
\[
\Lambda^j _{k,m,x_{-m}}(\psi,\pi):= \frac{\nabla_{\psi^{\otimes j}} p_{\psi \pi} (Y_k| \overline{{\bf Y}}_{-m}^{k-1}, x_{-m})}{ j! p_{\psi^* \pi} (Y_k| \overline{{\bf Y}}_{-m}^{k-1}, x_{-m})}, \quad \Lambda^j_{k,m}(\psi, \pi):= \frac{\nabla_{\psi^{\otimes j}} p_{\psi \pi} (Y_k| \overline{{\bf Y}}_{-m}^{k-1})}{ j! p_{\psi^* \pi} (Y_k| \overline{{\bf Y}}_{-m}^{k-1})},
\]
and $\Delta \psi:= \psi -\psi^*$. With this notation, expanding $l_{\vartheta k x_0} -1$ three times around $\psi^*$ while fixing $\pi$ gives, with $\overline \psi \in [\psi, \psi^*]$, 
\begin{align} 
l_{\vartheta k x_0} -1 &= \Lambda^1_{k,0,x_0}(\psi^*,\pi)' \Delta \psi + \Lambda^2_{k,0,x_0}(\psi^*,\pi)' (\Delta \psi)^{\otimes 2} + \Lambda^3_{k,0,x_0}(\overline \psi,\pi)' (\Delta \psi)^{\otimes 3} \nonumber \\
 & = \Lambda^1_{k,0}(\psi^*,\pi)' \Delta \psi + \Lambda^2_{k,0}(\psi^*,\pi)' (\Delta \psi)^{\otimes 2} + \Lambda^3_{k,0}(\overline \psi,\pi)' (\Delta \psi)^{\otimes 3} + u_{kx_0}(\psi,\pi), \label{lk_expansion}
\end{align}
where $\overline \psi$ may differ from element to element of $\Lambda^3_{k,0,x_0}(\overline \psi,\pi)$, and $u_{kx_0}(\psi,\pi) := \sum_{j=1}^2 [\Lambda^j_{k,0,x_0}(\psi^*,\pi)- \Lambda^j_{k,0}(\psi^*,\pi)]' (\Delta \psi)^{\otimes j} + [\Lambda^3_{k,0,x_0}(\overline \psi,\pi) -\Lambda^3_{k,0}(\overline \psi,\pi)]' (\Delta \psi)^{\otimes 3}$.

Noting that $\nabla_\lambda \overline p_{\psi^* \pi} (Y_k| \overline{{\bf Y}}_{0}^{k-1})=0$ and $\nabla_{\lambda\eta'} \overline p_{\psi^* \pi} (Y_k| \overline{{\bf Y}}_{0}^{k-1})=0$ from (\ref{d1p}), we may rewrite (\ref{lk_expansion}) as
\begin{equation} \label{lk_2}
l_{k\vartheta x_0} -1 = t(\psi,\pi)' s_{\varrho k} + r_{k,0}(\psi,\pi) + u_{kx_0}(\psi,\pi), 
\end{equation}
where $s_{\varrho k}$ is defined in (\ref{score}), $r_{k,0}(\psi,\pi) := \widetilde \Lambda_{k,0}(\pi)' (\Delta \eta)^{\otimes 2} + \Lambda^3_{k,0}(\overline \psi,\pi)' (\Delta \psi)^{\otimes 3}$,
where $\widetilde \Lambda_{k,0}(\pi)$ denotes the part of $\Lambda^2_{k,0}(\psi^*,\pi)$ corresponding to $(\Delta \eta)^{\otimes 2}$.

For $m \geq 0$, define $v_{k, m}(\vartheta):= (\Lambda^1_{k, m}( \psi,\pi)', \Lambda^2_{k, m}( \psi,\pi)',\Lambda^3_{k, m}( \psi,\pi)')'$, and define $v_{k,\infty}(\vartheta):= \lim_{m\rightarrow \infty} v_{k,m}(\vartheta)$. In order to apply Corollary \ref{cor_appn} to $l_{\vartheta k x_0} -1$, we first show
\begin{align} 
&\sup_{\vartheta \in \mathcal{N}_\epsilon}\left|P_n [v_{k,0}(\vartheta)v_{k,0}(\vartheta)'] - \mathbb{E}_{\vartheta^*}[v_{k,\infty}(\vartheta)v_{k,\infty}(\vartheta)']\right| = o_p(1), \label{uniform_lln} \\
&\nu_n(v_{k,0}(\vartheta)) \Rightarrow W(\vartheta), \label{weak_cgce}
\end{align}
where $W(\vartheta)$ is a mean-zero continuous Gaussian process with $\mathbb{E}_{\vartheta^*}[W(\vartheta_1)W(\vartheta_2)'] = \mathbb{E}_{\vartheta^*}[v_{k,\infty}(\vartheta_1)v_{k,\infty}(\vartheta_2)']$. (\ref{uniform_lln}) holds because $\sup_{\vartheta \in \mathcal{N}_\epsilon}P_n[v_{k,0}(\vartheta)v_{k,0}(\vartheta)' - v_{k,\infty}(\vartheta)v_{k,\infty}(\vartheta)'] = o_p(1)$ from Proposition \ref{lemma-omega}, and $v_{k,\infty}(\vartheta)v_{k,\infty}(\vartheta)'$ satisfies a uniform law of large numbers (Lemma 2.4 and footnote 18 of \citet{neweymcfadden94hdbk}) because $v_{k,\infty}(\vartheta)$ is continuous in $\vartheta$ from the continuity of $\nabla^j l_{k,m,x}(\vartheta)$ and Proposition \ref{lemma-omega}, and $\mathbb{E}_{\vartheta^*}\sup_{\vartheta \in \mathcal{N}_{\epsilon}}|v_{k,\infty}(\vartheta)|^{2} < \infty$ from Proposition \ref{lemma-omega}. (\ref{weak_cgce}) holds because $\sup_{\vartheta \in \mathcal{N}_\epsilon}\nu_n(v_{k,0}(\vartheta) - v_{k,\infty}(\vartheta)) = o_p(1)$ from Proposition \ref{lemma-omega} and $\nu_n(v_{k,\infty}(\vartheta)) \Rightarrow W(\vartheta)$ from Theorem 10.2 of \citet{pollard90book} because (i) the space of $\vartheta$ is totally bounded, (ii) the finite dimensional distributions of $\nu_n(v_{k,\infty}(\cdot))$ converge to those of $W(\cdot)$ from a martingale CLT because $v_{k,\infty}(\vartheta)$ is a stationary $L^2(\mathbb{P}_{\vartheta^*})$ martingale difference sequence for all $\vartheta\in\mathcal{N}_{\epsilon}$ from Proposition \ref{lemma-omega}, and (iii) $\{\nu_n(v_{k,\infty}(\cdot)): n\geq 1\}$ is stochastically equicontinuous from Theorem 2 of \citet{hansen96et} because $v_{k,\infty}(\vartheta)$ is Lipschitz continuous in $\vartheta$ and both $v_{k,\infty}(\vartheta)$ and the Lipschitz coefficient are in $L^q(\mathbb{P}_{\vartheta^*})$ with $q > \text{dim}(\vartheta)$ from Proposition \ref{lemma-omega}.

We proceed to show that the terms on the right-hand side of (\ref{lk_2}) satisfy Assumptions \ref{assn_a3}(a)--(g). Observe that $t(\psi,\pi) = 0$ if and only if $\psi=\psi^*$. First, $s_{\varrho k}$ satisfies Assumptions \ref{assn_a3}(a), (b), and (g) by Proposition \ref{lemma-omega}, (\ref{uniform_lln}), (\ref{weak_cgce}), and Assumption \ref{A-nonsing1}. Second, $r_{k,0}(\psi,\pi)$ satisfies Assumptions \ref{assn_a3}(c) and (d) from Proposition \ref{lemma-omega} and (\ref{weak_cgce}). Third, $u_{k x_0}(\psi,\pi)$ satisfies Assumptions \ref{assn_a3}(e) and (f) from Proposition \ref{lemma-omega}(c). Therefore, the stated result follows from Corollary \ref{cor_appn}(b). 
\end{proof}

\begin{proof}[Proof of Proposition \ref{P-LR}] 
The proof is similar to that of Proposition 3 of \citet{kasaharashimotsu15jasa}. Let $t_{\eta}: = \eta - \eta^*$ and $t_{\lambda } := \alpha(1-\alpha)v(\lambda)$, so that $t(\psi,\pi) = (t_{\eta}',t_{\lambda}')'$. Let $\hat \psi_\pi := \arg\max_{\psi \in \Theta_\psi} \ell_n(\psi,\pi,\xi)$ denote the MLE of $\psi$, and split $t(\hat\psi_\pi,\pi)$ as $t(\hat\psi_\pi,\pi)= (\hat t_{\eta}',\hat t_{\lambda}')'$, where we suppress the dependence of $\hat t_{\eta}$ and $\hat t_{\lambda}$ on $\pi$. Define $G_{\varrho n} := \nu_n (s_{\varrho k})$. Let
\[
G_{\varrho n} = \begin{bmatrix}
G_{\eta n} \\
G_{\lambda \varrho n}
\end{bmatrix}, \quad
\begin{aligned}
G_{\lambda. \eta \varrho n} &:= G_{\lambda \varrho n} - \mathcal{I}_{\lambda \eta\varrho}\mathcal{I}_{\eta }^{-1}G_{\eta n}, \quad Z_{\lambda. \eta \varrho n} := \mathcal{I}_{\lambda.\eta \varrho}^{-1}G_{\lambda.\eta \varrho n},\\
t_{\eta.\lambda \varrho} &:= t_{\eta} + \mathcal{I}_{\eta }^{-1}\mathcal{I}_{\eta\lambda \varrho} t_{\lambda}.
\end{aligned}
\]
Then, we can write (\ref{ln_appn}) as
\begin{equation} \label{LR_appn}
\sup_{\xi \in \Xi}\sup_{\vartheta \in A_{n\varepsilon c}(\xi) } \left| 2 \left[\ell_n(\psi,\pi,\xi) - \ell_n(\psi^*,\pi,\xi) \right] - A_n(\sqrt{n} t_{\eta.\lambda \varrho}) - B_{\varrho n}(\sqrt{n} t_{\lambda}) \right| =o_p(1),
\end{equation}
where 
\begin{equation} \label{B_pi}
\begin{aligned}
A_n(t_{\eta.\lambda \varrho}) & = 2t_{\eta.\lambda \varrho}'G_{\eta n} - t_{\eta.\lambda \varrho}'\mathcal{I}_{\eta}t_{\eta.\lambda \varrho}, \\
B_{\varrho n}(t_{\lambda}) &= 2t_{\lambda}' G_{\lambda.\eta \varrho n} - t_{\lambda}' \mathcal{I}_{\lambda.\eta \varrho} t_{\lambda} = Z_{\lambda \varrho n}' \mathcal{I}_{\lambda.\eta \varrho} Z_{\lambda \varrho n}- (t_{\lambda} - Z_{\lambda \varrho n})'\mathcal{I}_{\lambda.\eta \varrho}(t_{\lambda} - Z_{\lambda \varrho n}).
\end{aligned}
\end{equation}
Observe that $2[\ell_{0n}(\hat\vartheta_0) - \ell_{0n}(\vartheta_0^*) ] = \max_{t_\eta}[2 \sqrt{n} t_{\eta}' G_{\eta n} - n t_{\eta}' \mathcal{I}_\eta t_{\eta}] + o_p(1) = \max_{t_{\eta.\lambda \varrho}} A_n(\sqrt{n} t_{\eta.\lambda \varrho}) + o_p(1)$ from applying Corollary \ref{cor_appn} to $\ell_{0n}(\vartheta_0)$ and noting that the set of possible values of both $\sqrt{n} t_\eta$ and $\sqrt{n}{t}_{\eta.\lambda \varrho}$ approaches $\mathbb{R}^{\dim(\eta)}$. In conjunction with (\ref{LR_appn}), we obtain, uniformly in $\pi\in \Theta_{\pi}$,
\begin{equation} \label{LR_appn2}
2[\ell_n( \hat \psi_\pi , \pi,\xi) - \ell_{0n}(\hat\vartheta_0)] = B_{\varrho n}(\sqrt{n} \hat{t}_{\lambda}) + o_p(1). 
\end{equation}
Define $\tilde t_\lambda$ by $B_{\varrho n}(\sqrt{n} \tilde{t}_{\lambda}) = \max_{t_\lambda \in \alpha(1-\alpha)v(\Theta_\lambda)} B_{\varrho n}(\sqrt{n} t_{\lambda})$. Then, we have
\[
2[\ell_n(\hat \psi_\pi, \pi,\xi) - \ell_{0n}(\hat\vartheta_0)] = B_{\varrho n}(\sqrt{n} \tilde{t}_{\lambda}) + o_p(1),
\] 
uniformly in $\pi \in \Theta_{\pi}$ because (i) $B_{\varrho n}(\sqrt{n} \tilde{t}_{\lambda}) \geq 2[\ell_n(\hat \psi_\pi, \pi,\xi) - \ell_{0n}(\hat\vartheta_0)] + o_p(1)$ from the definition of $\tilde{t}_{\lambda}$ and (\ref{LR_appn2}), and (ii) $2[\ell_n(\hat \psi_\pi, \pi,\xi) - \ell_{0n}(\hat\vartheta_0)] \geq B_{\varrho n}(\sqrt{n} \tilde{t}_{\lambda}) + o_p(1)$ from the definition of $\hat \psi$, (\ref{LR_appn}), and $\tilde{t}_{\lambda}=O_p(n^{-1/2})$.

Finally, the asymptotic distribution of $\sup_{\varrho}B_{\varrho n}(\sqrt{n} \tilde{t}_{\lambda})$ follows from applying Theorem 1(c) of \citet{andrews01em} to $B_{\varrho n}(\sqrt{n} \tilde{t}_{\lambda})$. First, Assumption 2 of \citet{andrews01em} holds trivially for $B_{\varrho n}(\sqrt{n} \tilde{t}_{\lambda})$. Second, Assumption 3 of \citet{andrews01em} is satisfied by (\ref{weak_cgce}) and Assumption \ref{A-nonsing1}. Assumption 4 of \citet{andrews01em} is satisfied by Proposition \ref{P-quadratic}. Assumption $5^*$ of \citet{andrews01em} holds with $B_T=n^{1/2}$ because $\alpha(1-\alpha)v(\Theta_\lambda)$ is locally equal to the cone $v(\mathbb{R}^q)$ given that $\alpha(1-\alpha)>0$ for all $\alpha\in\Theta_{\alpha}$. Therefore, $\sup_{\varrho\in \Theta_{\varrho}}B_{\varrho n}(\sqrt{n}\tilde{t}_{\lambda}) \overset{d}{\rightarrow} \sup_{\varrho\in \Theta_{\varrho}}(\tilde{t}_{\lambda \varrho}'{\mathcal{I}}_{\lambda.\eta \varrho}\tilde{t}_{\lambda \varrho})$ follows from Theorem 1(c) of \citet{andrews01em}.
\end{proof}

\begin{proof}[Proof of Proposition \ref{P-quadratic-N1}]
The proof is similar to that of Proposition \ref{P-quadratic}. Define $\Lambda^j_{k,m,x_{-m}}(\psi,\pi)$ and $\Lambda^j_{k,m}(\psi, \pi)$ as in the proof of Proposition \ref{P-quadratic}. Expanding $l_{k\vartheta x_0} -1$ five times around $\psi^*$ similarly to (\ref{lk_expansion}) while fixing $\pi$ gives, with $\overline \psi \in [\psi, \psi^*]$, 
\begin{align} 
l_{k\vartheta x_0} -1 & = \sum_{j=1}^4 \Lambda^j_{k,0}(\psi^*,\pi)' (\Delta \psi)^{\otimes j} + \Lambda^5_{k,0}(\overline \psi,\pi)' (\Delta \psi)^{\otimes 5} + u_{kx_0}(\psi,\pi), \label{lk_expansion_normal}
\end{align}
where $u_{kx_0}(\psi,\pi) := \sum_{j=1}^4 [\Lambda^j_{k,0,x_0}(\psi^*,\pi)- \Lambda^j_{k,0}(\psi^*,\pi)]'(\Delta \psi)^{\otimes j}+ [\Lambda^5_{k,0,x_0}(\overline \psi,\pi) - \Lambda^5_{k,0}(\overline \psi,\pi)]' (\Delta \psi)^{\otimes 5}$.

Define $\overline p_{\psi \pi k,0} := \overline p_{\psi \pi} (Y_k| \overline{{\bf Y}}_{0}^{k-1})$. Observe that $s_{\varrho k}$ defined in (\ref{score_normal}) satisfies
\begin{equation*} 
s_{\varrho k}
:=
\begin{pmatrix}
\nabla_{\eta} \overline p_{\psi^* \pi k, 0} / \overline p_{\psi^* \pi k, 0} \\
\zeta_{k, 0}(\varrho)/2 \\
\nabla_{\lambda_\mu\lambda_\sigma} \overline p_{\psi^* \pi k, 0} / \alpha(1-\alpha)\overline p_{\psi^* \pi k, 0} \\
\nabla_{\lambda_\sigma^2} \overline p_{\psi^* \pi k, 0} / 2\alpha(1-\alpha) \overline p_{\psi^* \pi k, 0} \\
\nabla_{\lambda_\beta\lambda_\mu} \overline p_{\psi^* \pi k, 0} / \alpha(1-\alpha)\overline p_{\psi^* \pi k, 0} \\
\nabla_{\lambda_\beta\lambda_\sigma} \overline p_{\psi^* \pi k, 0} / \alpha(1-\alpha)\overline p_{\psi^* \pi k, 0} \\
V(\nabla_{\lambda_\beta \lambda_\beta} \overline p_{\psi^* \pi k, 0})/ \alpha(1-\alpha) \overline p_{\psi^* \pi k, 0} 
\end{pmatrix}.
\end{equation*}
Noting that $\nabla_\lambda \overline p_{\psi^* \pi} (Y_k| \overline{{\bf Y}}_{0}^{k-1})=0$ and $\nabla_{\lambda\eta'} \overline p_{\psi^* \pi} (Y_k| \overline{{\bf Y}}_{0}^{k-1})=0$ from (\ref{d1p}) and (\ref{d2p0}), we may rewrite (\ref{lk_expansion_normal}) as, with $t(\psi,\pi)$ and $s_{\varrho k}$ defined in (\ref{t-psi2}) and (\ref{score_normal}),
\begin{equation} \label{lk_2_normal}
l_{\vartheta k x_0} -1 = t(\psi,\pi)' s_{\varrho k} + r_{k,0}(\pi) + u_{kx_0}(\psi,\pi),
\end{equation} 
where $r_{k,0}(\pi):= \widetilde \Lambda_{k,0}(\pi) '\tau(\psi) + \Lambda^5_{k,0}(\overline \psi,\pi) ' (\Delta \psi)^{\otimes 5 } + \lambda_\mu^4[\nabla_{\lambda_\mu^4} \overline p_{\psi^* \pi k,0}- b(\alpha)\nabla_{\lambda_\sigma^2} \overline p_{\psi^* \pi k,0}]/4! \overline p_{\psi^* \pi k,0}$, $\tau(\psi)$ is the vector that collects the elements of $\{(\Delta \psi)^{\otimes j} \}_{j=2}^4$ that are not in $t(\psi,\pi)$, and $\widetilde \Lambda_{k,0}(\pi)$ denotes the vector of the corresponding elements of $\{\Lambda^j_{k,0}(\psi^*,\pi)\}_{j=2}^4$. 
 
The stated result follows from Corollary \ref{cor_appn} if the terms on the right-hand side of (\ref{lk_2_normal}) satisfy Assumption \ref{assn_a3}. Similarly to the proof of Proposition \ref{P-LR}, define $v_{k, m}(\vartheta):=(\zeta_{k,m}(\varrho), \Lambda^1_{k, m}( \psi,\pi)', \ldots, \Lambda^5_{k, m}( \psi,\pi)' )'$. 
Note that $\zeta_{k,m}(\varrho)$ satisfies Proposition \ref{lemma-omega} because the mean value theorem and $\nabla_{\lambda_\mu^2} \overline p_{\psi^* 0 \alpha} (Y_k| \overline{{\bf Y}}_{-m}^{k-1})=0$ gives $\zeta_{k,m}(\varrho) = [\nabla_{\lambda_\mu^2}\overline p_{\psi^* \varrho \alpha} (Y_k| \overline{{\bf Y}}_{0}^{k-1}) - \nabla_{\lambda_\mu^2}\overline p_{\psi^* 0 \alpha} (Y_k| \overline{{\bf Y}}_{0}^{k-1})]/[\varrho \alpha(1-\alpha) \overline p_{\psi^* \varrho \alpha} (Y_k| \overline{{\bf Y}}_{0}^{k-1})] =\nabla_\varrho \nabla_{\lambda_\mu^2}\overline p_{\psi^* \alpha \bar \varrho} (Y_k| \overline{{\bf Y}}_{-m}^{k-1}) /[\alpha(1-\alpha) \overline p_{\psi^* \bar\varrho \alpha} (Y_k| \overline{{\bf Y}}_{-m}^{k-1})]$ for $\bar \varrho \in [0,\varrho]$.
Therefore, $v_{k,\infty}(\vartheta):=\lim_{m\rightarrow \infty} v_{k,m}(\vartheta)$ is well-defined, and $v_{k,0}(\vartheta)$ and $v_{k,\infty}(\vartheta)$ satisfy (\ref{uniform_lln})--(\ref{weak_cgce}) from repeating the argument in the proof of Proposition \ref{P-LR}.

We proceed to show that the terms on the right-hand side of (\ref{lk_2_normal}) satisfy Assumption \ref{assn_a3}. Observe that $t(\psi,\pi) = 0$ if and only if $\psi=\psi^*$. $s_{\varrho k}$ and $u_{k x_0}(\psi,\pi)$ satisfy Assumption \ref{assn_a3} by noting that $s_{\varrho k}$ is a linear function of $v_{k,0}(\vartheta)$ and using the argument in the proof of Proposition \ref{P-quadratic} by replacing Assumption \ref{A-nonsing1} with Assumption \ref{A-nonsing2}. We show that each component of $r_{k,0}(\pi)$ satisfies Assumptions \ref{assn_a3}(c) and (d). First, $ \Lambda^5_{k,0}(\overline \psi,\pi)' (\Delta \psi)^{\otimes 5}$ satisfies Assumptions \ref{assn_a3}(c) and (d) from Proposition \ref{lemma-omega}, (\ref{weak_cgce}) and $\lambda_\mu^5 = (12\lambda_\mu/b(\alpha)) [\lambda_\sigma^2 + b(\alpha)\lambda_\mu^4/12] - 12(\lambda_\sigma/b(\alpha))\lambda_\mu \lambda_\sigma = O(|\psi| |t(\psi,\pi)|)$. Second, $\lambda_\mu^4[\nabla_{\lambda_\mu^4} \overline p_{\psi^* \pi k,0}- b(\alpha)\nabla_{\lambda_\sigma^2} \overline p_{\psi^* \pi k,0}]/ \overline p_{\psi^* \pi k,0}$ satisfies Assumptions \ref{assn_a3}(c) and (d) from Lemma \ref{lemma_d34}(b). Third, for $\widetilde \Delta_{k,0}(\pi)'\tau(\psi)$, observe that $\nabla_{\lambda \eta^j}\overline p_{\psi^* \pi k,0}=0$ for any $j \geq 1$ in view of (\ref{repara_g_hetero})--(\ref{d3g2}). Therefore, $\widetilde \Delta_{k,0}(\pi)'\tau(\psi)$ is written as, with $\Delta \eta:= \eta- \eta^*$,
\begin{equation} \label{d4r}
\widetilde \Delta_{k,0}(\pi)'\tau(\psi) = \nabla_{(\eta^{\otimes 2})'} \overline p_{\psi^* \pi k,0} (\Delta \eta)^{\otimes 2} / 2!\overline p_{\psi^* \pi k,0} + R_{3k\vartheta} + R_{4k\vartheta}, 
\end{equation}
where $R_{3k\vartheta}:= \nabla_{(\psi^{\otimes 3})'} \overline p_{\psi^* \pi k,0} (\Delta \psi)^{\otimes 3} /3!\overline p_{\psi^* \pi k,0}$ and 
\begin{equation}\label{r_jk}
R_{4k\vartheta}:= [\nabla_{(\psi^{\otimes 4})'} \overline p_{\psi^* \pi k,0} (\Delta \psi)^{\otimes 4} - \nabla_{\lambda_\mu^4} \overline p_{\psi^* \pi k,0} \lambda_\mu^4]/4!\overline p_{\psi^* \pi k,0}.
\end{equation}
The first term in (\ref{d4r}) clearly satisfies Assumptions \ref{assn_a3}(c) and (d). The terms in $R_{3k\vartheta}$ belong to one of the following three groups: (i) the term associated with $\lambda_\sigma^3$, (ii) the term associated with $\lambda_\mu^3$, or (iii) the other terms. These terms satisfy Assumptions \ref{assn_a3}(c) and (d) because the term (i) is bounded by $|\psi||t(\psi,\pi)|$ because $\lambda_\sigma^3 = \lambda_\sigma [\lambda_\sigma^2+b(\alpha)\lambda_\mu^4/12] - (\lambda_\mu^3b(\alpha))\lambda_\mu \lambda_\sigma/12$, the term (ii) is bounded by $\varrho \lambda_\mu^3$ from Lemma \ref{lemma_d34}(a), and the terms in (iii) are bounded by $|\psi||t(\psi,\pi)|$ because they either contain $\Delta \eta$ or a term of the form $\lambda_\mu^i\lambda_\sigma^j\lambda_\beta^k$ with $i+j+k=3$ and $i,j \neq 3$. Similarly, the terms in $R_{4k\vartheta}$ satisfy Assumptions \ref{assn_a3}(c) and (d) because they either contain $\Delta \eta$ or a term of the form $\lambda_\mu^i\lambda_\sigma^j\lambda_\beta^k$ with $i+j+k=4$ and $i \neq 4$. This proves that $r_{k,0}(\pi)$ satisfies Assumptions \ref{assn_a3}(c) and (d), and the stated result is proven. 
\end{proof}

\begin{proof}[Proof of Proposition \ref{P-LR-N1}]

The proof is similar to the proof of Proposition 3(c) of \citet{kasaharashimotsu15jasa}. Let $(\hat{\psi}_\alpha, \hat \varrho_\alpha) : = {\arg\max}_{(\psi, \varrho) \in \Theta_{\psi}\times\Theta_{\varrho}} \ell_n(\psi, \varrho,\alpha,\xi)$ denote the MLE of $(\psi,\varrho)$ for a given $\alpha$. Consider the sets $\Theta_{\lambda}^1:=\{\lambda \in \Theta_{\lambda}:|\lambda_\mu| \geq n^{-1/8}(\log n)^{-1} \}$ and $\Theta_{\lambda}^2:=\{\lambda \in \Theta_{\lambda}:|\lambda_\mu| \leq n^{-1/8}(\log n)^{-1}\}$, so that $\Theta_{\lambda} = \Theta_{\lambda}^1 \cup \Theta_{\lambda}^2$. For $j=1,2$, define $(\hat{\psi}^j_\alpha, \hat \varrho^j_\alpha) : = {\arg\max}_{(\psi,\varrho) \in \Theta_{\psi}\times \Theta_{\varrho}, \lambda \in \Theta_{\lambda}^j} \ell_n(\psi, \varrho,\alpha,\xi)$. Then, uniformly in $\alpha$,
\[
\ell_n(\hat \psi_\alpha, \hat \varrho_\alpha, \alpha,\xi) = \max\left\{\ell_n(\hat\psi_\alpha^1, \hat\varrho_\alpha^1,\alpha,\xi),\ell_n(\hat\psi_\alpha^2,\hat\varrho_\alpha^2,\alpha,\xi)\right\}.
\]
Henceforth, we suppress the dependence of $\hat \psi_\alpha$, $\hat \varrho_\alpha$, etc. on $\alpha$.

Define $B_{\varrho n}(t_{\lambda}(\lambda, \varrho,\alpha ))$ as in (\ref{B_pi}) in the proof of Proposition \ref{P-LR} but using $t(\psi,\pi)$ and $s_{\varrho k}$ defined in (\ref{t-psi2}) and (\ref{score_normal}) and replacing $t_{\lambda}$ in (\ref{B_pi}) with $t_{\lambda}(\lambda, \varrho,\alpha )$. Observe that the proof of Proposition \ref{P-LR} goes through up to (\ref{LR_appn2}) with the current notation and that $G_{\varrho n}$ and $\mathcal{I}_\varrho$ are continuous in $\varrho$. Further, $\hat{\varrho}^1 = O_p(n^{-1/4}(\log n)^2)$ because $\hat \varrho^1 (\hat \lambda_\mu^1)^2 = O_p(n^{-1/2})$ from Proposition \ref{P-quadratic-N1}(a) and $|\hat{\lambda}_\mu^1| \geq n^{-1/8}(\log n)^{-1}$. Consequently, $B_{\hat \varrho^1 n}(\sqrt{n}t_{\lambda}(\hat\lambda^1, \hat\varrho^1,\alpha)) = B_{0 n}(\sqrt{n}t_{\lambda}(\hat\lambda^1, \hat\varrho^1,\alpha)) + o_p(1)$, and, uniformly in $\alpha$,
\begin{equation} \label{ln_max}
2[\ell_n(\hat \psi,\hat \varrho,\alpha,\xi) - \ell_{0n}(\hat\vartheta_0)] = \max \{ B_ { 0 n}(\sqrt{n} t_{\lambda}(\hat\lambda^1, \hat\varrho^1,\alpha)), B_{ \hat \varrho^2 n}(\sqrt{n} t_{\lambda}(\hat\lambda^2, \hat\varrho^2,\alpha))\} + o_p(1).
\end{equation}

We proceed to construct parameter spaces $\tilde\Lambda_{\lambda\alpha}^1$ and $\tilde\Lambda_{\lambda \alpha\varrho}^2$ that are locally equal to the cones $\Lambda_\lambda^1$ and $\Lambda_{\lambda\varrho}^2$ defined in (\ref{Lambda-lambda}). Define $c(\alpha) := \alpha(1-\alpha)$, and denote the elements of $t_{\lambda}(\hat\lambda^j,\hat\varrho^j,\alpha)$ corresponding to (\ref{t_lambda_rho}) by
\begin{equation*} 
t_{\lambda}(\hat\lambda^j,\hat\varrho^j,\alpha)
=
\begin{pmatrix}
\hat t_{\varrho\mu^2}^j\\
\hat t_{\mu\sigma}^j \\
\hat t_{\sigma^2}^j\\
\hat t_{\beta \mu}^j\\
\hat t_{\beta \sigma}^j\\
\hat t_{v(\beta)}^j\\
\end{pmatrix}
:= 
 c(\alpha) 
\begin{pmatrix}
 \hat\varrho^j (\hat \lambda_\mu^j)^2 \\
 \hat\lambda_\mu \hat \lambda_\sigma\\
 (\hat\lambda_\sigma^j)^2 + b(\alpha)(\hat \lambda_\mu^j)^4/12 \\
 \hat \lambda_\beta^j \hat \lambda_\mu^j\\
 \hat \lambda_\beta^j \hat \lambda_\sigma^j\\
 v(\hat \lambda_\beta^j)
\end{pmatrix}. 
\end{equation*}
Note that $\hat{\lambda}_{\sigma}^1 = O_p(n^{-3/8} \log n)$ and $\hat{\lambda}_\beta^1 = O_p(n^{-3/8} \log n)$ because $(\hat{t} _{\mu \sigma}^1, \hat{t} _{\beta \mu}^1) = O_p(n^{-1/2})$ from Proposition \ref{P-quadratic-N1}(a) and $|\hat{\lambda}_\mu^1| \geq n^{-1/8}(\log n)^{-1}$. Furthermore, $\hat t_{\sigma^2}^2 = c(\alpha) (\hat \lambda_\sigma^2)^2 + o_p(n^{-1/2})$ because $|\hat\lambda_\mu^2|\leq n^{-1/8} (\log n)^{-1}$. Consequently,
\begin{equation} \label{t_hat_1}
\begin{aligned}
\hat{t}_{\beta \sigma}^1 &= o_p(n^{-1/2}), \quad \hat{t}_{v(\beta)}^1 = o_p(n^{-1/2}), \quad \hat t_{\sigma^2}^1 = c(\alpha) b(\alpha)(\hat{\lambda}_\mu^1)^4/12 + o_p(n^{-1/2}), \\
\hat t_{\sigma^2}^2 &= c(\alpha) (\hat \lambda_\sigma^2)^2 + o_p(n^{-1/2}).
\end{aligned}
\end{equation}
In view of this, let $t_{\lambda}(\lambda, \varrho,\alpha):=( t_{\varrho\mu^2}, t_{\mu\sigma },t_{\sigma^2 },t_{\beta\mu}',t_{\beta\sigma}',t_{v(\beta)}')' \in \mathbb{R}^{q_{\lambda}}$, and consider the following sets:
\begin{equation*}
\begin{aligned}
&\tilde\Lambda_{\lambda \alpha}^1:=\{ t_{\lambda}(\lambda, \varrho, \alpha) : t_{\varrho\mu^2} = c(\alpha) \varrho \lambda_\mu^2, t_{\mu\sigma }= c(\alpha)\lambda_\mu \lambda_\sigma, t_{\sigma^2} = c(\alpha) b(\alpha)\lambda_\mu^4/12, \\
& \qquad \qquad t_{\beta\mu} = c(\alpha)\lambda_{\beta}\lambda_{\mu}, t_{\beta \sigma}=0, t_{v(\beta)} = 0 \ \text{for some } (\lambda,\varrho) \in \Theta_\lambda\times\Theta_{\varrho}\}, \\
&\tilde\Lambda_{\lambda\alpha\varrho}^2:=\{ t_{\lambda}(\lambda, \varrho, \alpha) : t_{\varrho\mu^2} = c(\alpha)\varrho \lambda_{\mu}^2, t_{\mu\sigma } = c(\alpha)\lambda_\mu\lambda_\sigma, t_{\sigma^2}= c(\alpha)\lambda_\sigma^2, \\
& \qquad \qquad t_{\beta\mu }= c(\alpha)\lambda_{\beta} \lambda_\mu, t_{\beta\sigma }= c(\alpha)\lambda_{\beta} \lambda_\sigma, t_{v(\beta)} =c(\alpha)v(\lambda_{\beta})\ \text{for some }\lambda\in \Theta_\lambda \}. 
\end{aligned}
\end{equation*}
$\tilde\Lambda_{\lambda \alpha}^1$ is indexed by $\alpha$ but does not depend on $\varrho$ because $B_ {0n}(\cdot)$ in (\ref{ln_max}) does not depend on $\varrho$, whereas $\tilde\Lambda_{\lambda \alpha\varrho}^2$ is indexed by both $\alpha$ and $\varrho$ because $B_ {\hat \varrho^2 n}(\cdot)$ in (\ref{ln_max}) depends on $\hat \varrho^2$. Define $(\tilde \lambda_{\alpha}^1, \tilde \varrho_{\alpha}^1)$ and $\tilde \lambda_{\alpha\varrho}^2$ by $B_{0 n}(\sqrt{n} t_{\lambda}(\tilde\lambda_{\alpha}^1,\tilde\varrho_{\alpha}^1,\alpha)) = {\max}_{t_{\lambda}(\lambda,\varrho,\alpha) \in \tilde\Lambda_{\lambda \alpha}^1} B_{0 n}(\sqrt{n} t_{\lambda}(\lambda,\varrho,\alpha))$ and $B_{\varrho n}(\sqrt{n} t_{\lambda}(\tilde\lambda_{\alpha\varrho}^2,\varrho,\alpha)) = {\max}_{t_{\lambda}(\lambda,\varrho,\alpha) \in \tilde \Lambda_{\lambda \alpha\varrho}^2} B_{\varrho n}(\sqrt{n} t_{\lambda}(\lambda,\varrho,\alpha))$.

Define $W_n(\alpha):=\max\{B_{0 n}(\sqrt{n} t_{\lambda}(\tilde\lambda_{\alpha}^1,\tilde\varrho_{\alpha}^1,\alpha)), \sup_{\varrho \in \Theta_{\varrho}}B_{\varrho n}(\sqrt{n} t_{\lambda}(\tilde\lambda^2_{\alpha\varrho},\varrho,\alpha)) \}$, then we have 
\begin{equation} \label{LR_W}
2[\ell_n(\hat \psi, \hat \varrho, \alpha,\xi) - \ell_{0n}(\hat\vartheta_0)] = W_n(\alpha) + o_p(1), 
\end{equation}
 uniformly in $\alpha \in \Theta_\alpha$ because (i) $W_n(\alpha) \geq 2[\ell_n(\hat \psi, \hat \varrho, \alpha,\xi) - \ell_{0n}(\hat\vartheta_0)] + o_p(1)$ in view of the definition of $(\tilde \lambda_{\alpha}^1,\tilde \varrho_{\alpha}^1,\tilde \lambda_{\alpha\varrho}^2)$, (\ref{ln_max}), and (\ref{t_hat_1}), and (ii) $2[\ell_n(\hat \psi,\hat \varrho, \alpha, \xi) - \ell_{0n}(\hat\vartheta_0)] \geq \\ \max\{ 2[\max_\eta\ell_n(\eta,\tilde \lambda^1_\alpha, \tilde \varrho^1_\alpha,\alpha,\xi),\sup_{\varrho \in \Theta_\varrho} \max_\eta \ell_n(\eta,\tilde \lambda^2_{\alpha\varrho}, \varrho,\alpha,\xi) \} - 2\ell_{0n}(\hat\vartheta_0)+ o_p(1) = W_n(\alpha)+ o_p(1)$
from the definition of $(\hat \psi,\hat\varrho)$.

The asymptotic distribution of the LRTS follows from applying Theorem 1(c) of \citet{andrews01em} to $(B_{0 n}(\sqrt{n} t_{\lambda}(\tilde\lambda_{\alpha}^1,\tilde\varrho_{\alpha}^1,\alpha)), B_{\varrho n}(\sqrt{n} t_{\lambda}(\tilde\lambda^2_{\alpha\varrho},\varrho,\alpha)))$. First, Assumption 2 of \citet{andrews01em} holds trivially for $B_{\varrho n}(\sqrt{n} t(\lambda,\varrho,\alpha))$. Second, Assumption 3 of \citet{andrews01em} is satisfied by (\ref{weak_cgce}) and Assumption \ref{A-nonsing2}. Assumption 4 of \citet{andrews01em} is satisfied by Proposition \ref{P-quadratic-N1}. Assumption $5^*$ of \citet{andrews01em} holds with $B_T=n^{1/2}$ because $\tilde\Lambda_{\lambda \alpha}^1$ is locally (in a neighborhood of $\varrho=0, \lambda=0$) equal to the cone $\Lambda_{\lambda}^1$ and $\tilde\Lambda_{\lambda \alpha\varrho}^2$ is locally equal to the cone $\Lambda_{\lambda \varrho}^2$ uniformly in $\varrho\in\Theta_{\varrho\epsilon}$. Consequently, $W_n(\alpha) \overset{d}{\rightarrow} \max\{ \mathbb{I}\{\varrho=0\} (\tilde t_{\lambda }^1)' \mathcal{I}_{\lambda.\eta 0} \tilde t_{\lambda }^1, \sup_{\varrho\in\Theta_{\varrho}} (\tilde t_{\lambda \varrho}^2)' \mathcal{I}_{\lambda.\eta \varrho} \tilde t_{\lambda \varrho}^2 \}$ uniformly in $\alpha$ from Theorem 1(c) of \citet{andrews01em}, and the stated result follows from (\ref{LR_W}).
\end{proof}

\begin{proof}[Proof of Proposition \ref{P-quadratic-N1-homo}]
The proof is similar to that of Proposition \ref{P-quadratic-N1}. Expanding $l_{k\vartheta x_0} -1$ five times around $\psi^*$ and proceeding as in the proof of Proposition \ref{P-quadratic-N1} gives
\begin{equation} \label{lk_2_normal_homo}
l_{\vartheta k x_0} -1 = t(\psi,\pi)' s_{\varrho k} + r_{k,0}(\pi) + u_{kx_0}(\psi,\pi), 
\end{equation} 
where $t(\psi,\pi)$ is defined in (\ref{t-psi2-homo}), $s_{\varrho k}$ is defined in (\ref{score_normal_homo}) and satisfies 
\begin{equation*} 
 s_{\varrho k} 
:=
\begin{pmatrix}
\nabla_{\eta} \overline p_{\psi^* \pi k,0} / \overline p_{\psi^* \pi k,0} \\
\zeta_{k,0}(\varrho)/2 \\
\nabla_{\mu^3} f_k^* / 3! f_k^* \\
\nabla_{\mu^4} f_k^* / 4! f_k^* \\
\nabla_{\lambda_\beta\lambda_\mu} \overline p_{\psi^* \pi k,0} / \alpha(1-\alpha)\overline p_{\psi^* \pi k,0} \\
\widetilde{\nabla}_{v(\lambda_\beta)} \overline p_{\psi^* \pi k,0} / \alpha(1-\alpha) \overline p_{\psi^* \pi k,0} 
\end{pmatrix},
\end{equation*}
and
\begin{align*}
r_{k,0}(\pi):&= \widetilde \Lambda_{k,0}(\pi) '\tau(\psi) + \Lambda^5_{k,0}(\overline \psi,\pi) ' (\Delta \psi)^{\otimes 5} \\
&\quad + \lambda_\mu^3[\nabla_{\lambda_\mu^3} \overline p_{\psi^* \pi k,0}/\overline p_{\psi^* \pi k,0} - \alpha(1-\alpha)(1-2\alpha)\nabla_{\mu^3} f_k^*/f_k^*]/3!\\
& \quad+ \lambda_\mu^4[\nabla_{\lambda_\mu^4} \overline p_{\psi^* \pi k,0}/\overline p_{\psi^* \pi k,0} - \alpha(1-\alpha)(1-6\alpha+6 \alpha^2)\nabla_{\mu^4} f_k^*/f_k^*]/4!,
\end{align*}
where $u_{kx_0}(\psi,\pi)$, $\overline p_{\psi \pi k,m}$, and the terms in the definition of $r_{k,0}(\pi)$ are defined similarly to those in the proof of Proposition \ref{P-quadratic-N1}.

The stated result is proven if the terms on the right-hand side of (\ref{lk_2_normal_homo}) satisfy Assumption \ref{assn_a3}. $t(\psi,\pi) = 0$ if and only if $\psi=\psi^*$. $s_{\varrho k}$ and $u_{kx_0}(\psi,\pi)$ satisfy Assumption \ref{assn_a3} by the same argument as the proof of Proposition \ref{P-quadratic-N1}. For $r_{k,0}(\pi)$, first, $\Lambda^5_{k,0}(\overline \psi,\pi)' (\Delta \psi)^{\otimes 5}$ satisfies Assumptions \ref{assn_a3}(c) and (d) from a similar argument to the proof of Proposition \ref{P-quadratic-N1}; $\lambda_\mu^5$ is dominated by $\lambda_\mu^3$ or $\lambda_\mu^4$ because $\inf_{0\leq \alpha \leq 1}\max\{|1-2\alpha|,|1-6\alpha + 6\alpha^2|\}>0$. Second, similar to (\ref{d4r}) in the proof of Proposition \ref{P-quadratic-N1}, write $\widetilde \Lambda_{k,0}(\pi) '\tau(\psi)=\nabla_{(\eta^{\otimes 2})'} \overline p_{\psi^* \pi k,0} (\Delta \eta)^{\otimes 2} / 2!\overline p_{\psi^* \pi k,0}+ \tilde R_{3k\vartheta} + R_{4k\vartheta}$, where $\tilde R_{3k\vartheta}:= [\nabla_{(\psi^{\otimes 3})'} \overline p_{\psi^* \pi k,0} (\Delta \psi)^{\otimes 3}- \nabla_{\lambda_\mu^3} \overline p_{\psi^* \pi k,0} \lambda_\mu^3]/3!\overline p_{\psi^* \pi k,0}$, and $R_{4k\vartheta}$ is defined as $R_{4k\vartheta}$ in (\ref{r_jk}). The term $\nabla_{(\eta^{\otimes 2})'} \overline p_{\psi^* \pi k,0} (\Delta \eta)^{\otimes 2} / 2!\overline p_{\psi^* \pi k,0}$ clearly satisfies Assumptions \ref{assn_a3}(c) and (d). The terms in $\tilde R_{3k\vartheta}$ satisfy Assumptions \ref{assn_a3}(c) and (d) because they contain either $\Delta \eta$ or $\lambda_\mu^2\lambda_\beta$ or $\lambda_\mu\lambda_\beta^2$ or $\lambda_\beta^3$. The terms in $R_{4k\vartheta}$ satisfy Assumptions \ref{assn_a3}(c) and (d) because they either contain $\Delta \eta$ or a term of the form $\lambda_\mu^i\lambda_\beta^{4-i}$ with $ 1 \leq i \leq 3$. The last two terms in $r_{k,0}(\pi)$ satisfy Assumptions \ref{assn_a3}(c) and (d) from Lemma \ref{lemma_d34_homo}. Therefore, $r_{k,0}(\pi)$ satisfies Assumptions \ref{assn_a3}(c) and (d), and the stated result is proven. 
\end{proof}

\begin{proof}[Proof of Proposition \ref{P-LR-N1-homo}] 
The proof is similar to the proof of Proposition \ref{P-LR-N1}. Let $(\hat{\psi}, \hat \varrho, \hat \alpha) : = {\arg\max}_{(\psi, \varrho,\alpha) \in \Theta_{\psi}\times\Theta_{\varrho}\times\Theta_{\alpha}} \ell_n(\psi,\varrho,\alpha,\xi)$ denote the MLE of $(\psi,\varrho,\alpha)$. Consider the sets $\Theta_{\lambda}^1:=\{\lambda \in \Theta_{\lambda}:|\lambda_\mu| \geq n^{-1/6}(\log n)^{-1} \}$ and $\Theta_{\lambda}^2:=\{\lambda \in \Theta_{\lambda}:|\lambda_\mu| \leq n^{-1/6}(\log n)^{-1} \}$, so that $\Theta_{\lambda} = \Theta_{\lambda}^1 \cup \Theta_{\lambda}^2$. For $j=1,2$, define $(\hat{\psi}^j, \hat \varrho^j,\hat\alpha^j) : = {\arg\max}_{(\psi,\varrho,\alpha) \in \Theta_{\psi}\times \Theta_{\varrho }\times \Theta_{\alpha }, \lambda \in \Theta_{\lambda}^j} \ell_n(\psi,\varrho,\alpha,\xi)$, so that $\ell_n(\hat \psi, \hat \varrho,\hat\alpha,\xi) = \max_{j \in \{1,2\}} \ell_n(\hat\psi^j, \hat\varrho^j,\hat\alpha^j,\xi)$.

Define $B_{\varrho n}(t_\lambda(\lambda,\varrho,\alpha))$ as in (\ref{B_pi}) in the proof of Proposition \ref{P-LR} but using $t(\psi,\pi)$ and $s_{\varrho k}$ defined in (\ref{t-psi2-homo}) and (\ref{score_normal_homo}) and replacing $t_{\lambda}$ in (\ref{B_pi}) with $t_\lambda(\lambda,\varrho,\alpha)$. Observe that $\hat{\varrho}^1 = O_p(n^{-1/6}(\log n)^2)$ because $\hat \varrho^1 (\hat \lambda_\mu^1)^2 = O_p(n^{-1/2})$ from Proposition \ref{P-quadratic-N1-homo}(a) and $|\hat{\lambda}_\mu^1| \geq n^{-1/6}(\log n)^{-1}$. Using the argument of the proof of Proposition \ref{P-LR-N1} leading to (\ref{ln_max}), we obtain 
\begin{equation*}
2[\ell_n(\hat \psi, \hat \varrho, \hat \alpha,\xi) - \ell_{0n}(\hat\vartheta_0)]
 = \max \{ B_ { 0 n}(\sqrt{n} t_\lambda(\hat\lambda^1,\hat\varrho^1,\hat \alpha^1)), B_{ \hat \varrho^2 n}(\sqrt{n} t_\lambda(\hat\lambda^2, \hat\varrho^2,\hat \alpha^2))\} + o_p(1).
\end{equation*}

We proceed to construct parameter spaces that are locally equal to the cones $\Lambda_\lambda^1$ and $\Lambda_{\lambda\varrho}^2$ defined in (\ref{Lambda-lambda-homo}). Define $c(\alpha) := \alpha(1-\alpha)$, and denote the elements of $t_\lambda(\hat\lambda^j,\hat\varrho^j,\hat\alpha^j)$ corresponding to (\ref{t-psi2-homo}) by
\begin{equation*}
t_{\lambda}(\hat\lambda^j,\hat\varrho^j,\hat\alpha^j)
=
\begin{pmatrix}
\hat t_{\varrho\mu^2}^j\\
\hat t_{\mu^3}^j \\
\hat t_{\mu^4}^j \\
\hat t_{\beta \mu}^j\\
\hat t_{v(\beta)}^j\\
\end{pmatrix}
:= 
c(\hat\alpha^j)
\begin{pmatrix}
 \hat\varrho^j (\hat \lambda_\mu^j)^2 \\
 (1-2\hat\alpha^j) (\hat\lambda_\mu^j)^3\\
 (1-6\hat\alpha^j+6(\hat\alpha^j)^2) (\hat\lambda_\mu^j)^4\\
 \hat \lambda_\beta^j \hat \lambda_\mu^j\\
 v(\hat \lambda_\beta^j)
\end{pmatrix}. 
\end{equation*}
Note that $\hat{\lambda}_\beta^1 = O_p(n^{-1/3} \log n)$ because $\hat{t} _{\beta \mu}^1 = O_p(n^{-1/2})$ from Proposition \ref{P-quadratic-N1-homo}(a) and $|\hat{\lambda}_\mu^1| \geq n^{-1/6}(\log n)^{-1}$. Furthermore, $|\hat\lambda_\mu^2| \leq n^{-1/6}(\log n)^{-1}$. Therefore,
\begin{equation*}
\hat{t}_{v(\beta)}^1 = o_p(n^{-1/2}), \quad \hat t_{\mu^3}^{2} = o_p(n^{-1/2}), \quad \hat t_{\mu^4}^{2}=o_p(n^{-1/2}).
\end{equation*}
In view of this, let $t_\lambda(\lambda,\varrho,\alpha):=( t_{\varrho\mu^2}, t_{\mu^3},t_{\mu^4},t_{\beta\mu}',t_{v(\beta)}')' \in \mathbb{R}^{q_{\lambda}}$, and consider the following sets:
\begin{align*}
&\tilde\Lambda_{\lambda }^{1}:=\{ t_\lambda(\lambda,\varrho,\alpha) : t_{\varrho\mu^2} = c(\alpha) \varrho \lambda_\mu^2, t_{\mu^3}= c(\alpha)(1-2\alpha)\lambda_\mu^3, t_{\mu^4} = c(\alpha)(1-6\alpha+6\alpha^2) \lambda_\mu^4, \\
& \qquad \qquad t_{\beta\mu} = c(\alpha)\lambda_{\beta}\lambda_{\mu}, t_{v(\beta)} = 0 \ \text{for some }(\lambda,\varrho,\alpha) \in \Theta_\lambda\times \Theta_{\varrho}\times \Theta_{\alpha}\}, \\
&\tilde\Lambda_{\lambda \alpha\varrho}^2:=\{ t_\lambda(\lambda,\varrho,\alpha) : t_{\varrho\mu^2} = c(\alpha)\varrho \lambda_{\mu}^2, t_{\mu^3} = t_{\mu^4}= 0, \\
& \qquad \qquad t_{\beta\mu }= c(\alpha)\lambda_{\beta} \lambda_\mu, t_{v(\beta)} =c(\alpha)v(\lambda_{\beta})\ \text{for some }\lambda\in \Theta_\lambda \}. 
\end{align*}
Define $(\tilde \lambda^1, \tilde \varrho^1, \tilde \alpha^1)$ and $\tilde \lambda_{\alpha\varrho}^2$ by $B_{0 n}(\sqrt{n} t_\lambda(\tilde\lambda^1,\tilde \varrho^1, \tilde \alpha^1)) = {\max}_{t_\lambda(\lambda,\varrho,\alpha) \in \tilde\Lambda_{\lambda}^1} B_{0 n}(\sqrt{n} t_\lambda(\lambda,\varrho,\alpha))$ and \\ $B_{\varrho n}(\sqrt{n} t_\lambda(\tilde\lambda_{\alpha\varrho}^2,\varrho,\alpha)) = {\max}_{t_\lambda(\lambda,\varrho,\alpha) \in \tilde \Lambda_{\lambda\alpha\varrho}^2} B_{\varrho n}(\sqrt{n} t_\lambda(\lambda,\varrho,\alpha))$. $\tilde\Lambda_{\lambda}^{1}$ is locally (in the neighborhood of $\varrho=0$, $\lambda=0$) equal to the cone $\Lambda_{\lambda}^{1}$ because, when $|1-2\alpha|\geq \epsilon>0$ for some positive constant $\epsilon$, we have $t_{\mu^4}/t_{\mu^3} \to 0$ as $\lambda_\mu \to 0$, and when $\alpha$ is in the neighborhood of $1/2$, we have $1-6\alpha+6\alpha^2 <0$. $\tilde\Lambda_{\lambda \alpha \varrho}^2$ is locally equal to the cone $\Lambda_{\lambda\varrho}^2$ uniformly in $\varrho \in \Theta_{\varrho}$.

Define $W_n:=\max\{B_{0 n}(\sqrt{n} t_\lambda(\tilde\lambda^1,\tilde\varrho^1,\tilde\alpha^1)), \sup_{(\alpha,\varrho) \in\Theta_{\alpha} \times \Theta_{\varrho} } B_{\varrho n}(\sqrt{n} t_\lambda(\tilde\lambda^2_{\alpha\varrho},\varrho,\alpha)) \}$. Proceeding as in the proof of Proposition \ref{P-LR-N1} gives $2[\ell_n(\hat \psi, \hat \varrho,\hat \alpha,\xi) - \ell_{0n}(\hat\vartheta_0)] = W_n + o_p(1)$, and the asymptotic distribution of the LRTS follows from applying Theorem 1(c) of \citet{andrews01em} to $(B_{0 n}(\sqrt{n} t_\lambda(\tilde\lambda^1, \tilde\varrho^1,\tilde\alpha^1 )), B_{\varrho n}(\sqrt{n} t_\lambda(\tilde\lambda^2_{\alpha\varrho},\varrho,\alpha))$. 
\end{proof}

\begin{proof}[Proof of Propositions \ref{P-LR_M}, \ref{P-LR_M_normal}, and \ref{P-LR_M_normal_homo}] 
Let $\mathcal{N}_m^*$ denote an arbitrarily small neighborhood of $\Upsilon_m^*$, and let $\hat{\psi}_m$ denote a local MLE that maximizes $\ell_n(\psi_m,\pi_m,\xi_{M_0+1})$ subject to $\psi_m \in \mathcal{N}_m^*$. Proposition \ref{P-consist_M} and $\Upsilon^*=\cup_{m=1}^{M_0}\Upsilon_m^*$ imply that $\ell_n(\hat\vartheta_{M_0+1},\xi_{M_0+1}) = \max_{m=1,\ldots,M_0} \ell_n(\hat\psi_m,\pi_m,\xi_{M_0+1})$ with probability approaching one. Because $\psi_\ell^{*}\notin \mathcal{N}_m^*$ for any $\ell \neq m$, it follows from Proposition \ref{P-consist_M} that $\hat{\psi}_m-\psi_m^{*}=o_p(1)$.

Next, $\ell_n(\psi_{m},\pi_m,\xi_{M_0+1}) - \ell_n(\psi_m^{*},\pi_m,\xi_{M_0+1})$ admits the same expansion as $\ell_n(\psi,\pi,\xi) - \ell_n(\psi^{*},\pi,\xi)$ in (\ref{ln_appn}) or (\ref{ln_appn_N1}). 
Therefore, the stated result follows from applying the proof of Propositions \ref{P-LR}, \ref{P-LR-N1}, and \ref{P-LR-N1-homo} to $\ell_n(\hat \psi_m, \pi_m,\xi_{M_0+1}) - \ell_{n}(\hat\vartheta_{M_0},\xi_{M_0})$ for each $m$ and combining the results to derive the joint asymptotic distribution of $\{\ell_n(\hat \psi_m, \pi_m,\xi_{M_0+1}) - \ell_{n}(\hat\vartheta_{M_0},\xi_{M_0})\}_{m=1}^{M_0}$.
\end{proof}

\begin{proof}[Proof of Proposition \ref{P-LAN}] 
Observe that Proposition \ref{Ln_thm1} holds under $\mathbb{P}_{\vartheta^*,x_0}^n$ under the assumptions of Propositions \ref{P-quadratic}, \ref{P-quadratic-N1}, and \ref{P-quadratic-N1-homo}. Because $\vartheta_n=(\eta_n',\lambda_n',\pi_n')'\in \mathcal{N}_{c/\sqrt{n}}$ by choosing $c> |h|$, it follows from Proposition \ref{Ln_thm1} that
\begin{equation}\label{expansion}
\sup_{x_0 \in \mathcal{X}} \left|\log \frac{d\mathbb{P}_{\vartheta_n,x_0}^n}{d \mathbb{P}_{\vartheta^*,x_0}^n} - h' \nu_n(s_{\varrho_n k}) + \frac{1}{2}h' \mathcal{I}_{\varrho_n} h \right|=o_{\mathbb{P}_{\vartheta^*,x_0}^n}(1),
\end{equation}
where $s_{\varrho k}$ is given by (\ref{score}), (\ref{score_normal}), and (\ref{score_normal_homo}) for the models of the non-normal distribution, heteroscedastic normal distribution, and homoscedastic normal distribution, respectively. Furthermore, $\nu_n(s_{\varrho_n k}) \Rightarrow G_\varrho$ under $\mathbb{P}_{\vartheta^*,x_0}^n$, where $G_\varrho$ is a mean zero Gaussian process with $cov(G_{\varrho_1},G_{\varrho_2})=\mathcal{I}_{\varrho_1\varrho_2}:= \lim_{k\rightarrow \infty} \mathbb{E}_{\vartheta^*} (s_{\varrho_1 k}s_{\varrho_2 k}')$. Therefore, $d\mathbb{P}_{\vartheta_n,x_0}^n / d \mathbb{P}_{\vartheta^*,x_0}^n$ converges in distribution under $\mathbb{P}_{\vartheta^*,x_0}^n$ to $\exp\left( N( \mu,\sigma^2) \right)$ with $\mu=-(1/2) h' \mathcal{I}_{ \varrho} h$ and $\sigma^2= h' \mathcal{I}_{ \varrho} h$, so that $E(\exp\left( N( \mu,\sigma^2) \right))=1$. Consequently, part (a) follows from Le Cam's first lemma (see, e.g., Corollary 12.3.1 of \citet{lehmannromano05book}). Part (b) follows from Le Cam's third lemma (see, e.g., Corollary 12.3.2 of \citet{lehmannromano05book}) because part (a) and (\ref{expansion}) imply that
\[
\begin{pmatrix}
\nu_n(s_{\varrho_n k})\\
\log\frac{d\mathbb{P}_{\vartheta_n,x_0}^n}{d \mathbb{P}_{\vartheta^*,x_0}^n} 
\end{pmatrix}
 \overset{d}{\rightarrow} 
N\left(
\begin{pmatrix}
0\\
-\frac{1}{2} h' \mathcal{I}_{\varrho} h
\end{pmatrix}, 
\begin{pmatrix}
\mathcal{I}_{\varrho}&\mathcal{I}_{\varrho}h\\
h'\mathcal{I}_{\varrho}&h'\mathcal{I}_{\varrho}h
\end{pmatrix}
\right)\quad\text{under $\mathbb{P}_{\vartheta^*,x_0}^n$}.
\] 
\end{proof} 

\begin{proof}[Proof of Proposition \ref{P-LAN2}] 
The proof follows the argument in the proof of Proposition \ref{P-LR}. Observe that $h_\eta=0$ and $h_{\lambda} =\sqrt{n}t_{\lambda}(\lambda_n,\pi_n)$ hold under $H_{1n}$. Therefore, Proposition \ref{P-LAN} holds under $\mathbb{P}_{\vartheta_n,x_0}^n$ implied by $H_{1n}$, and, in conjunction with Theorem 12.3.2(a) of \citet{lehmannromano05book}, Propositions \ref{lemma-omega} and \ref{P-quadratic} hold under $\mathbb{P}_{\vartheta_n,x_0}^n$. Consequently, the proof of Proposition \ref{P-LR} goes through if we replace $G_{\lambda.\eta\varrho n}\Rightarrow G_{\lambda.\eta\varrho}$ with $G_{\lambda.\eta\varrho n} \Rightarrow G_{\lambda.\eta\varrho} + ( \mathcal{I}_{\lambda \varrho \varrho} - \mathcal{I}_{\lambda\eta\varrho} \mathcal{I}_{\eta}^{-1} \mathcal{I}_{\eta\lambda\varrho} ) h_\lambda = G_{\lambda.\eta\varrho} + \mathcal{I}_{\lambda.\eta\varrho} h_{\lambda}$, and the stated result follows. \end{proof}

\begin{proof}[Proof of Propositions \ref{P-LAN3} and \ref{P-LAN4}] 
The proof is similar to the proof of Proposition \ref{P-LAN2}. Observe that, for $j \in \{a,b\}$, $h_\eta^j=0$ and $h_{\lambda}^j =\sqrt{n}t_{\lambda}(\lambda_n,\pi_n)+o(1)$ hold under $H_{1 n}^j$. Therefore, Proposition \ref{P-LAN} holds under $\mathbb{P}_{\vartheta_n,x_0}^n$ implied by $H_{1n}^j$, and the stated result follows from repeating the argument of proof of Proposition \ref{P-LAN2}. 
\end{proof}

\begin{proof}[Proof of Proposition \ref{P-bootstrap}] 
We only provide the proof for the models of the non-normal distribution with $M_0=1$ because the proof for the other models is similar. The proof follows the argument in the proof of Theorem 15.4.2 in \citet{lehmannromano05book}. Define $\bf{C}_\eta$ as the set of sequences $\{\eta_n\}$ satisfying $\sqrt{n}(\eta_n - \eta^*) \to h_\eta$ for some finite $h_\eta$. 
Denote the MLE of the one-regime model parameter by $\hat \eta_n$. For the MLE under $H_0$, $\sqrt{n}(\hat \eta_n - \eta^*)$ converges in distribution to a $\mathbb{P}_{\vartheta^*}$-a.s. finite random variable by the standard argument. Then, by the Almost Sure Representation Theorem (e.g., Theorem 11.2.19 of \citet{lehmannromano05book}), there exist random variables $\tilde \eta_n$ and $\tilde h_\eta$ defined on a common probability space such that $\hat \eta_n$ and $\tilde \eta_n$ have the same distribution and $\sqrt{n}(\tilde \eta_n - \eta^*)\rightarrow \tilde h_\eta$ almost surely. Therefore, $\{ \tilde \eta_n \}\in \bf{C}_\eta$ with probability one, and the stated result under $H_0$ follows from Lemma \ref{lemma_btsp} because $\hat \eta_n$ and $\tilde \eta_n$ have the same distribution.

For the MLE under $H_{1n}$, note that the proof of Proposition \ref{P-LAN2} goes through when $h_\eta$ is finite even if $h_\eta \neq 0$. Therefore, $\sqrt{n}(\hat \eta_n - \eta^*)$ converges in distribution to a $\mathbb{P}_{\vartheta_n}$-a.s. finite random variable under $H_{1n}$. Hence, the stated result follows from Lemma \ref{lemma_btsp} and repeating the argument in the case of $H_0$.
\end{proof}

\subsection{Auxiliary results}

\subsubsection{Missing information principle}

The following lemma extends equations (3.1) and (3.2) in Louis (1982), expressing the higher-order derivatives of the log-likelihood function in terms of the conditional expectation of the derivatives of the complete data log-likelihood function. For notational brevity, assume $\vartheta$ is scalar. Adaptations to vector-valued $\vartheta$ are straightforward but need more tedious notation. Let $\nabla^j \ell(Y):= \nabla_\vartheta^j \log P(Y;\vartheta)$ and $\nabla^j \ell(Y,X):= \nabla_\vartheta^j \log P(Y,X;\vartheta)$. For random variables $V_1,\ldots,V_q$ and $Y$, define the central conditional moment of $(V_1^{r_1}\cdots V_q^{r_q})$ as $\mathbb{E}^c [V_1^{r_1} \cdots V_q^{r_q} |Y ] := \mathbb{E} [ (V_1-\mathbb{E}[V_1|Y])^{r_1} \cdots (V_q-\mathbb{E}[V_q|Y])^{r_q} |Y] $.
\begin{lemma} \label{louis}
For any random variables $X$ and $Y$ with density $P(Y,X;\theta)$ and $P(Y;\theta)$,
\begin{align*}
&\nabla \ell(Y) = \mathbb{E}\left[ \nabla \ell(Y,X) \middle| Y \right], \quad \nabla^2 \ell(Y) = \mathbb{E}\left[ \nabla^2 \ell(Y,X) \middle| Y \right] + \mathbb{E}^c\left[ (\nabla \ell(Y,X))^2\middle| Y \right], \\
&\nabla^3 \ell(Y) = \mathbb{E}\left[ \nabla^3 \ell(Y,X) \middle| Y \right]+ 3 \mathbb{E}^c\left[ \nabla^2 \ell(Y,X) \nabla \ell(Y,X)\middle| Y \right] + \mathbb{E}^c\left[(\nabla \ell(Y,X))^3\middle| Y \right], \\
&\nabla^4 \ell(Y)
= \mathbb{E}\left[ \nabla^4 \ell(Y,X)\middle| Y \right] +4 \mathbb{E}^c \left[\nabla^3 \ell(Y,X) \nabla \ell(Y,X)\middle| Y \right] + 3 \mathbb{E}^c\left[ (\nabla^2 \ell(Y,X))^2 \middle| Y \right] \\
& \quad + 6 \mathbb{E}^c\left[\nabla^2 \ell(Y,X)(\nabla \ell(Y,X))^2\middle| Y \right] +\mathbb{E}^c\left[(\nabla \ell(Y,X))^4 \middle| Y \right] -3 \left\{ \mathbb{E}^c\left[ (\nabla \ell(Y,X))^2 \middle| Y \right]\right\}^2, \\
&\nabla^5 \ell(Y) = \mathbb{E}\left[ \nabla^5 \ell(Y,X)\middle| Y\right] +5 \mathbb{E}^c\left[ \nabla^4 \ell(Y,X) \nabla \ell(Y,X) \middle| Y\right] +10 \mathbb{E}^c\left[ \nabla^3 \ell(Y,X) \nabla^2 \ell(Y,X) \middle| Y\right] \\
&\quad+10\mathbb{E}^c\left[ \nabla^3 \ell(Y,X) (\nabla \ell(Y,X))^2 \middle| Y\right]+15\mathbb{E}^c\left[ (\nabla^2 \ell(Y,X))^2 \nabla \ell(Y,X) \middle| Y\right] \\
&\quad +10\mathbb{E}^c\left[ \nabla^2 \ell(Y,X) (\ell(Y,X))^3 \middle| Y\right] -30 \mathbb{E}^c\left[ \nabla^2 \ell(Y,X) \nabla \ell(Y,X) \middle| Y\right] \mathbb{E}^c\left[( \nabla \ell(Y,X))^2 \middle| Y\right] \\
&\quad+\mathbb{E}^c\left[ (\nabla \ell(Y,X))^5 \middle| Y\right] - 10\mathbb{E}^c\left[ (\nabla \ell(Y,X))^3 \middle| Y\right] \mathbb{E}^c\left[(\nabla \ell(Y,X))^2 \middle| Y\right], \\
&\nabla^6 \ell(Y) = \mathbb{E}\left[ \nabla^6 \ell(Y,X)\middle|Y \right] \\
& \quad +6\mathbb{E}^c\left[\nabla^5 \ell(Y,X)\nabla \ell(Y,X) \middle| Y\right] +15 \mathbb{E}^c\left[\nabla^4 \ell(Y,X)\nabla^2 \ell(Y,X) \middle| Y\right] \\
& \quad+15 \mathbb{E}^c\left[\nabla^4 \ell(Y,X) (\nabla \ell(Y,X))^2\middle| Y\right] + 60\mathbb{E}^c\left[\nabla^3 \ell(Y,X)\nabla^2 \ell(Y,X)\nabla \ell(Y,X) \middle| Y\right] \\
& \quad+10\mathbb{E}^c\left[(\nabla^3 \ell(Y,X))^2\middle| Y\right]+ 15 \mathbb{E}^c\left[ (\nabla^2 \ell(Y,X))^3 \middle| Y\right] \\
& \quad+20\mathbb{E}^c\left[\nabla^3 \ell(Y,X)(\nabla \ell(Y,X))^3\middle| Y\right] - 60 \mathbb{E}^c\left[\nabla^3 \ell(Y,X) \nabla \ell(Y,X)\middle| Y\right]\mathbb{E}\left[(\nabla \ell(Y,X))^2\middle| Y\right]\\
& \quad+45 \mathbb{E}^c\left[ (\nabla^2 \ell(Y,X))^2 (\nabla \ell(Y,X))^2 \middle| Y\right] -90 \left\{\mathbb{E}^c\left[ \nabla^2 \ell(Y,X) \nabla \ell(Y,X) \middle| Y\right]\right\}^2\\
& \quad-45 \mathbb{E}^c\left[ (\nabla^2 \ell(Y,X))^2 \middle| Y\right] \mathbb{E}^c\left[ (\nabla \ell(Y,X))^2 \middle| Y\right] \\
& \quad +15 \mathbb{E}^c\left[ \nabla^2 \ell(Y,X) (\nabla \ell(Y,X))^4\middle| Y\right] -90 \mathbb{E}^c\left[ \nabla^2 \ell(Y,X) (\nabla \ell(Y,X))^2 \middle| Y\right] \mathbb{E}^c\left[ (\nabla \ell(Y,X))^2 \middle| Y\right]\label{5-1}\\
& \quad -60 \mathbb{E}^c\left[ \nabla^2 \ell(Y,X) \nabla \ell(Y,X) \middle| Y\right] \mathbb{E}^c\left[ (\nabla \ell(Y,X))^3 \middle| Y\right] \\
& \quad +\mathbb{E}^c\left[ (\nabla \ell(Y,X))^6 \middle| Y\right] -15 \mathbb{E}^c\left[ (\nabla \ell(Y,X))^4 \middle| Y\right] \mathbb{E}^c\left[ (\nabla \ell(Y,X))^2 \middle| Y\right] \\
& \quad - 10\left\{\mathbb{E}^c\left[ (\nabla \ell(Y,X))^3 \middle| Y\right]\right\}^2 + 30\left\{\mathbb{E}^c\left[ (\nabla \ell(Y,X))^2 \middle| Y\right]\right\}^3,
\end{align*}
provided that the conditional expectation on the right-hand side exists. When $P(Y;\theta)$ on the left-hand side is replaced with $P(Y|Z;\theta)$, the stated result holds with $P(Y,X;\theta)$ and $\mathbb{E}[\cdot |Y]$ on the right-hand side replaced with $P(Y,X|Z;\theta)$ and $\mathbb{E}[\cdot |Y,Z]$. 
\end{lemma} 
\begin{proof}[Proof of Lemma \ref{louis}]
The stated result follows from a direct calculation and relations such as\\ $\nabla_\vartheta^j P(Y;\vartheta)/P(Y;\vartheta) = \mathbb{E}[ \nabla_\vartheta^j P(Y,X;\vartheta)/P(Y,X;\vartheta)|Y]$ and
\begin{equation} \label{der_formula}
\begin{aligned}
\nabla \log f &= \nabla f/f, \quad \nabla^2 \log f = \nabla^2 f/f - (\nabla \log f )^2,\\
\nabla^3 \log f &= \nabla^3 f/f -3 \nabla^2 f \nabla f/f^2 + 2 (\nabla f/f)^3, \\
\nabla^4 \log f &= \nabla^4 f/f -4 \nabla^3 f \nabla f/f^2 - 3 (\nabla^2 f/f)^2 + 12 \nabla^2 f (\nabla f)^2/ f^3 -6 (\nabla f/f)^4, \\
\nabla^5 \log f &= \nabla^5 f/f -5 \nabla^4 f \nabla f/f^2 - 10 \nabla^3 f \nabla^2 f/f^2 + 20 \nabla^3 f (\nabla f)^2/f^3 \\
& \quad + 30 (\nabla^2 f)^2 \nabla f /f^3 - 60 \nabla^2 f(\nabla f)^3 /f^4 + 24 (\nabla f/f)^5, \\
\nabla^6 \log f & = \nabla ^6 f/f -6 \nabla^5 f \nabla f/f^2 - 15 \nabla ^4 f \nabla^2 f/f^2 + 30 \nabla^4 f (\nabla f)^2 / f^3 -10 (\nabla^3 f)^2/f^2 \\
& \quad +120 \nabla^3 f \nabla^2 f \nabla f/f^3 -120 \nabla^3 f (\nabla f)^3/f^4 + 30 (\nabla^2 f)^3/f^3 \\
& \quad - 270 (\nabla^2 f)^2(\nabla f)^2 /f^4 +360 \nabla^2 f (\nabla f)^4 /f^5 - 120 (\nabla f)^6/f^6, \\
\nabla^3 f/f &= \nabla^3 \log f + 3 \nabla^2 \log f \nabla \log f +\left(\nabla \log f\right)^3,\\
\nabla^4 f/f &= \nabla^4 \log f + 4 \nabla^3 \log f \nabla \log f + 3(\nabla^2 \log f )^2 +6 \nabla^2 \log f (\nabla\log f)^2+ (\nabla\log f)^4,\\
\nabla^5 f/f & = \nabla^5 \log f +5 \nabla^4 \log f \nabla \log f +10 \nabla^3 \log f \nabla^2 \log f +10 \nabla^3 \log f (\nabla \log f)^2 \\
& \quad +15 (\nabla^2 \log f)^2 \nabla \log f +10 \nabla^2 \log f (\nabla \log f)^3 + (\nabla\log f)^5, \\
\nabla^6 f/f & = \nabla^6 \log f + 6 \nabla^5 \log f \nabla \log f + 15 \nabla^4 \log f \nabla^2 \log f + 15\nabla^4 \log f (\nabla \log f)^2 \\
& \qquad +10( \nabla^3 \log f )^2 +60 \nabla^3 \log f \nabla^2 \log f \nabla \log f +20\nabla^3\log f (\nabla \log f)^3 \\
& \qquad +15( \nabla^2 \log f )^3 +45(\nabla^2 \log f)^2( \nabla \log f)^2 +15\nabla^2 \log f(\nabla \log f)^4 + (\nabla\log f)^6.
\end{aligned}
\end{equation}
For example, $\nabla^3 \ell(Y)$ is derived by writing $\nabla^3 \ell(Y)$ as, with suppressing $\vartheta$, 
\begin{align*}
& \nabla^3 \ell(Y) \\
& = \frac{\nabla^3 P(Y)}{P(Y)} - 3 \frac{\nabla^2 P(Y)}{P(Y)}\frac{\nabla P(Y)}{P(Y)} + 2\left(\frac{\nabla P(Y)}{P(Y)}\right)^3 \\
& = \mathbb{E}\left[ \frac{\nabla^3 P(Y,X)}{P(Y,X)} \middle| Y\right] -3 \mathbb{E}\left[ \frac{\nabla^2 P(Y,X)}{P(Y,X)} \middle| Y\right] \mathbb{E}\left[ \frac{\nabla P(Y,X)}{P(Y,X)} \middle| Y\right] + 2 \left\{ \mathbb{E}\left[ \frac{\nabla P(Y,X)}{P(Y,X)} \middle| Y\right]\right\}^3\\
& = \mathbb{E}\left[ \nabla^3 \ell(Y,X) + 3 \nabla^2 \ell(Y,X) \nabla \ell(Y,X) + (\nabla \ell(Y,X))^3 \middle| Y\right]\\
&\quad -3 \mathbb{E}\left[ \nabla^2 \ell(Y,X) + (\nabla \ell(Y,X))^2 \middle| Y\right] \mathbb{E}\left[ \nabla \ell(Y,X) \middle| Y\right]+ 2 \left\{ \mathbb{E}\left[ \nabla \ell(Y,X) \middle| Y\right]\right\}^3,
\end{align*}
and collecting terms. $\nabla^4 \ell(Y)$, $\nabla^5 \ell(Y)$, and $\nabla^6 \ell(Y)$ are derived similarly.
\end{proof}

\subsubsection{Auxiliary lemmas}

We first collect the notations. Define $\overline{\bf Z}_{k-1}^k:=(X_{k-1},\overline{\bf Y}_{k-1},W_k,X_k,Y_k)$ and denote the derivative of the complete data log-density by
\begin{equation} \label{phi_i}
\phi^i(\vartheta,\overline{\bf Z}_{k-1}^k) :=\nabla^i \log p_{\vartheta}(Y_k,X_k|\overline {\bf Y}_{k-1},X_{k-1},W_k), \quad i\geq 1.
\end{equation}
We use short-handed notation $\phi^i_{\vartheta k}:=\phi^i(\vartheta,\overline{\bf Z}_{k-1}^k)$. We also suppress the superscript $1$ from $\phi^1_{\vartheta k}$, so that $\phi_{\vartheta k}=\phi^1_{\vartheta k}$. 
For random variables $V_1,\ldots,V_q$ and a conditioning set $\mathcal{F}$, define the central conditional moment of $(V_1,\ldots,V_q)$ as
\begin{align*}
\mathbb{E}_\vartheta^c \left[V_1,\ldots,V_q \middle|\mathcal{F} \right] & := \mathbb{E}_{\vartheta} \left[ \left(V_1-\mathbb{E}_{\vartheta}[V_1|\mathcal{F}]\right) \cdots \left(V_q-\mathbb{E}_{\vartheta}[V_q|\mathcal{F}]\right) \middle| \mathcal{F}\right].
\end{align*}
For example, $\mathbb{E}_{\vartheta}^c \left[\phi_{\vartheta k_1}\phi_{\vartheta k_2} \middle|\mathcal{F} \right] := \mathbb{E}_{\vartheta} \left[ \left( \phi_{\vartheta k_1}-\mathbb{E}_{\vartheta}[\phi_{\vartheta k_1}|\mathcal{F}]\right) \left( \phi_{\vartheta k_2}-\mathbb{E}_{\vartheta}[\phi_{\vartheta k_2}|\mathcal{F}]\right) \middle| \mathcal{F}\right]$.
 
Let $\mathcal{I}(j) = (i_1,\ldots,i_j)$ denote a sequence of positive integers with $j$ elements, let $\sigma(\mathcal{I}(j))$ denote the set of all the unique permutations of $(i_1,\ldots,i_j)$, and let $|\sigma(\mathcal{I}(j))|$ denote its cardinality. For example, if $\mathcal{I}(3) = (2,1,1)$, then $\sigma(\mathcal{I}(3)) = \{(2,1,1), (1,2,1), (1,1,2)\}$ and $|\sigma(\mathcal{I}(3))|=3$; if $\mathcal{I}(3) = (1,1,1)$, then $\sigma(\mathcal{I}(3)) = (1,1,1)$ and $|\mathcal{I}(3)|=1$. Let $\mathcal{T}(j) = (t_1,\ldots,t_j)$ for $j=1,\ldots,6$. For a conditioning set $\mathcal{F}$, define the symmetrized central conditional moments as 
\begin{equation} \label{Phi_defn}
\begin{aligned}
\Phi^{\mathcal{I}(1)}_{\vartheta \mathcal{T}(1)}[\mathcal{F}]& :=
\mathbb{E}_{\vartheta}\left[ \phi^{i_1}_{\vartheta t_1} \middle| \mathcal{F}\right] ,
\quad
\Phi^{\mathcal{I}(2)}_{\vartheta \mathcal{T}(2)}[\mathcal{F}] :=
\frac{1}{|\sigma(\mathcal{I}(2))|} \sum_{(\ell_1,\ell_2) \in \sigma(\mathcal{I}(2))} \mathbb{E}_{\vartheta}^c\left[ \phi^{\ell_1}_{\vartheta t_1} \phi^{\ell_2}_{\vartheta t_2} \middle| \mathcal{F}\right],
\\
\Phi^{\mathcal{I}(3)}_{\vartheta \mathcal{T}(3)}[\mathcal{F}] & := 
\frac{1}{|\sigma(\mathcal{I}(3))|} \sum_{(\ell_1,\ell_2,\ell_3) \in \sigma(\mathcal{I}(3))} \mathbb{E}_{\vartheta}^c\left[ \phi^{\ell_1}_{\vartheta t_1} \phi^{\ell_2}_{\vartheta t_2} \phi^{\ell_3}_{\vartheta t_3} \middle| \mathcal{F}\right], \\
\Phi^{\mathcal{I}(4)}_{\vartheta \mathcal{T}(4)}[\mathcal{F}] & := \frac{1}{|\sigma(\mathcal{I}(4))|} \sum_{(\ell_1,\ldots,\ell_4) \in \sigma(\mathcal{I}(4))} \tilde \Phi_{\vartheta \mathcal{T}(4)}^{\ell_1 \ell_2 \ell_3 \ell_4}, 
\end{aligned}
\end{equation}
where $\tilde \Phi_{\vartheta \mathcal{T}(4)}^{\ell_1 \ell_2 \ell_3 \ell_4} :=\mathbb{E}_{\vartheta}^c[ \phi^{\ell_1}_{\vartheta t_1} \phi^{\ell_2}_{\vartheta t_2} \phi^{\ell_3}_{\vartheta t_3} \phi^{\ell_4}_{\vartheta t_4} | \mathcal{F}] - \mathbb{E}_{\vartheta}^c[\phi^{\ell_1}_{\vartheta t_1} \phi^{\ell_2}_{\vartheta t_2}|\mathcal{F}] \mathbb{E}_{\vartheta}^c[\phi^{\ell_3}_{\vartheta t_3} \phi^{\ell_4}_{\vartheta t_4}|\mathcal{F}] \\- \mathbb{E}_{\vartheta}^c[\phi^{\ell_1}_{\vartheta t_1} \phi^{\ell_3}_{\vartheta t_3}|\mathcal{F}] \mathbb{E}_{\vartheta}^c[\phi^{\ell_2}_{\vartheta t_2} \phi^{\ell_4}_{\vartheta t_4}|\mathcal{F}] - \mathbb{E}_{\vartheta}^c[\phi^{\ell_1}_{\vartheta t_1} \phi^{\ell_4}_{\vartheta t_4}|\mathcal{F}]\mathbb{E}_{\vartheta}^c[\phi^{\ell_2}_{\vartheta t_2} \phi^{\ell_3}_{\vartheta t_3}|\mathcal{F}] $, and
\begin{equation} \label{Phi_defn_2}
\begin{aligned} 
\Phi^{\mathcal{I}(5)}_{\vartheta \mathcal{T}(5)}[\mathcal{F}] & := \frac{1}{|\sigma(\mathcal{I}(5))|} \sum_{(\ell_1,\ldots,\ell_5) \in \sigma(\mathcal{I}(5))} \left( \mathbb{E}_\vartheta^c\left[ \phi^{\ell_1}_{\theta t_1} \phi^{\ell_2}_{\theta t_2} \phi^{\ell_3}_{\theta t_3} \phi^{\ell_4}_{\theta t_4} \phi^{\ell_5}_{\theta t_5} \middle| \mathcal{F}\right] \right. \\
& \left. \qquad - \sum_{(\{a,b,c\},\{d,e\}) \in \sigma_5} \mathbb{E}_\vartheta^c\left[\phi^{\ell_{a}}_{\theta t_{a}} \phi^{\ell_b}_{\theta t_b}\phi^{\ell_c}_{\theta t_c} \middle|\mathcal{F}\right] \mathbb{E}_\vartheta^c\left[ \phi^{\ell_d}_{\theta t_d} \phi^{\ell_e}_{\theta t_e} \middle| \mathcal{F}\right] \right), \\
\Phi^{\mathcal{I}(6)}_{\vartheta \mathcal{T}(6)}[\mathcal{F}]
 &:= \mathbb{E}_\vartheta^c\left[ \phi_{\theta t_1} \phi_{\theta t_2} \phi_{\theta t_3} \phi_{\theta t_4} \phi_{\theta t_5} \phi_{\theta t_6} \middle| \mathcal{F}\right] - \sum_{(\{a,b,c,d\},\{e,f\}) \in \sigma_{61}} \mathbb{E}_\vartheta^c\left[ \phi_{\theta t_a} \phi_{\theta t_b} \phi_{\theta t_c} \phi_{\theta t_d} \middle| \mathcal{F}\right] \mathbb{E}_\vartheta^c\left[ \phi_{\theta t_e} \phi_{\theta t_f} \middle| \mathcal{F}\right] \\
&\qquad - \sum_{(\{a,b,c\},\{d,e,f\}) \in \sigma_{62}} \mathbb{E}_\vartheta^c\left[ \phi_{\theta t_a} \phi_{\theta t_b} \phi_{\theta t_c} \middle| \mathcal{F}\right] \mathbb{E}_\vartheta^c\left[ \phi_{\theta t_d} \phi_{\theta t_e} \phi_{\theta t_f} \middle| \mathcal{F}\right] \\
&\qquad + 2 \sum_{(\{a,b\},\{c,d\},\{e,f\}) \in \sigma_{63}} \mathbb{E}_\vartheta^c\left[ \phi_{\theta t_a} \phi_{\theta t_b} \middle| \mathcal{F}\right] \mathbb{E}_\vartheta^c\left[ \phi_{\theta t_c} \phi_{\theta t_d} \middle| \mathcal{F}\right] \mathbb{E}_\vartheta^c\left[ \phi_{\theta t_e} \phi_{\theta t_f} \middle| \mathcal{F}\right],
\end{aligned}
\end{equation}
where
\begin{equation}\label{partitions}
\begin{aligned}
\sigma_5 &:= \text{the set of } {\textstyle{5 \choose 3} =10} \text{ partitions of } \{1,2,3,4,5\} \text{ of the form } \{a,b,c\},\{d,e\},\\
\sigma_{61} &:= \text{the set of } {\textstyle{6 \choose 4} = 15} \text{ partitions of } \{1,2,3,4,5,6\} \text{ of the form } \{a,b,c,d\},\{e,f\}, \\
\sigma_{62} &:= \text{the set of } {\textstyle{6 \choose 3}/2 = 10} \text{ partitions of } \{1,2,3,4,5,6\} \text{ of the form } \{a,b,c\},\{d,e,f\}, \\
\sigma_{63} &:= \text{the set of } {\textstyle{6 \choose 2}\textstyle{4 \choose 2}/6 = 15} \text{ partitions of } \{1,2,3,4,5,6\} \text{ of the form } \{a,b\},\{c,d\},\{e,f\}.
\end{aligned}
\end{equation} 
Note that these moments are symmetric with respect to $(t_1,\ldots,t_j)$. For $j=1,2,\ldots,6$, $k \geq 1$, $m \geq 0$, and $x \in \mathcal{X}$, define the difference between the sums of the $\Phi^{\mathcal{I}(j)}_{\vartheta \mathcal{T}(j)}$'s over different time indices and conditioning sets as 
\begin{equation}\label{delta-tau}
\begin{aligned}
\Delta^{\mathcal{I}(j)}_{j,k,m,x}(\vartheta) &:= \sum_{\mathcal{T}(j)\in \{-m+1,\ldots,k\}^j} {\Phi}^{\mathcal{I}(j)}_{\vartheta \mathcal{T}(j)}\left[\overline{\bf Y}_{-m}^k,{\bf W}_{-m}^{k},X_{-m}=x\right] \\
& \quad - \sum_{\mathcal{T}(j)\in \{-m+1,\ldots,k-1\}^j} {\Phi}^{\mathcal{I}(j)}_{\vartheta \mathcal{T}(j)}\left[\overline{\bf Y}_{-m}^{k-1},{\bf W}_{-m}^{k-1},X_{-m}=x\right], 
\end{aligned}
\end{equation} 
where $\sum_{\mathcal{T}(j)\in \{-m+1,\ldots,k\}^j}$ denotes $\sum_{t_1=-m+1}^k\sum_{t_2=-m+1}^k \cdots \sum_{t_j=-m+1}^k$, and $\sum_{\mathcal{T}(j) \in \{-m+1,\ldots,k-1\}^j}$ is defined similarly. Define $\overline{\Delta}^{\mathcal{I}(j)}_{j,k,m}(\theta)$ analogously to $\Delta^{\mathcal{I}(j)}_{j,k,m,x}(\vartheta)$ by dropping $X_{-m}=x$ from the conditioning variable.

Henceforth, we suppress the conditioning variable ${\bf W}_{-m}^n$ from the conditioning sets and conditional densities unless confusion might arise. The following lemma expresses the derivatives of the log-densities, $ \nabla^j \ell_{k,m,x}(\vartheta)$'s, in terms of the $\Delta^{\mathcal{I}(j)}_{j,k,m,x}(\vartheta)$'s. The first two equations are also given in DMR (p. 2272 and pp. 2276--7).
\begin{lemma} \label{ell_lambda}
For all $1 \leq k \leq n$, $m \geq 0$, and $x \in \mathcal{X}$,
\begin{align*}
&\nabla^1 \ell_{k,m,x}(\vartheta)=\Delta^1_{1,k,m,x}(\vartheta),\quad 
\nabla^2 \ell_{k,m,x}(\vartheta)=\Delta^2_{1,k,m,x}(\vartheta)+\Delta^{1,1}_{2,k,m,x}(\vartheta),\\ 
&\nabla^3 \ell_{k,m,x}(\vartheta)= \Delta^{3}_{1,k,m,x}(\vartheta)+3\Delta^{2,1}_{2,k,m,x}(\vartheta)+\Delta^{1,1,1}_{3,k,m,x}(\vartheta),\\
&\nabla^4 \ell_{k,m,x}(\vartheta)= \Delta^4_{1,k,m,x}(\vartheta) + 4 \Delta^{3,1}_{2,k,m,x}(\vartheta) + 3 \Delta^{2,2}_{2,k,m,x}(\vartheta) + 6\Delta^{2,1,1}_{3,k,m,x}(\vartheta) + \Delta^{1,1,1,1}_{4,k,m,x}(\vartheta),\\
&\nabla^5 \ell_{k,m,x}(\vartheta)= \Delta^5_{1,k,m,x}(\vartheta)+ 5\Delta^{4,1}_{2,k,m,x}(\vartheta) + 10 \Delta^{3,2}_{2,k,m,x}(\vartheta) +10\Delta^{3,1,1}_{3,k,m,x}(\vartheta)+15 \Delta^{2,2,1}_{3,k,m,x}(\vartheta)\\
& \quad +10\Delta^{2,1,1,1}_{4,k,m,x}(\vartheta) +\Delta^{1,1,1,1,1}_{5,k,m,x}(\vartheta), \\
&\nabla^6 \ell_{k,m,x}(\vartheta)= \Delta^6_{1,k,m,x}(\vartheta) + 6\Delta^{5,1}_{2,k,m,x}(\vartheta) + 15 \Delta^{4,2}_{2,k,m,x}(\vartheta) + 10 \Delta^{3,3}_{2,k,m,x}(\vartheta)+ 15\Delta^{4,1,1}_{3,k,m,x}(\vartheta) \\
& \quad + 60\Delta^{3,2,1}_{3,k,m,x}(\vartheta) + 15\Delta^{2,2,2}_{3,k,m,x}(\vartheta) + 20 \Delta^{3,1,1,1}_{4,k,m,x}(\vartheta) + 45 \Delta^{2,2,1,1}_{4,k,m,x}(\vartheta) + 15 \Delta^{2,1,1,1,1}_{5,k,m,x}(\vartheta) +\Delta^{1,1,1,1,1}_{6,k,m,x}(\vartheta). 
\end{align*} 
Further, the above holds when $\nabla^j \ell_{k,m,x}(\vartheta)$ and $\Delta^{\mathcal{I}(j)}_{j,k,m,x}(\vartheta)$ are replaced with $\nabla^j \overline \ell_{k,m}(\vartheta)$ and $\overline\Delta^{\mathcal{I}(j)}_{j,k,m}(\vartheta)$.
\end{lemma}

\begin{proof}[Proof of Lemma \ref{ell_lambda}]
The stated result follows from writing $\nabla^j \ell_{k,m,x}(\vartheta) = \\ \nabla^j \log p_{\vartheta}({\bf Y}_{-m+1}^{k}|\overline{\bf{Y}}_{-m},X_{-m}=x) - \nabla^j \log p_{\vartheta}({\bf{Y}}_{-m+1}^{k-1}|\overline{\bf{Y}}_{-m},X_{-m}=x)$, applying Lemma \ref{louis} to the right-hand side, and noting that $\nabla^j \log p_{\vartheta}({\bf Y}_{-m+1}^k,{\bf X}_{-m+1}^k|\overline{\bf{Y}}_{-m},X_{-m}) = \sum_{t=-m+1}^k \phi^j(\vartheta,\overline{\bf Z}_{t-1}^t)$ (see (\ref{cond_density}) and (\ref{phi_i})). The result for $\nabla^j \ell_{k,m,x}(\vartheta)$ with $j=1,2$ is also given in DMR (p. 2272 and pp. 2276--7). For $j=3$, the term $\Delta^{2,1}_{2,k,m,x}(\vartheta)$ follows from $\sum_{t_1=-m+1}^k \sum_{t_2=-m+1}^k \mathbb{E}_\vartheta^c[ \phi^{2}_{\vartheta t_1} \phi^{1}_{\vartheta t_2} |\overline{\bf Y}_{-m}^k,X_{-m}=x] = \sum_{t_1=-m+1}^k \sum_{t_2=-m+1}^k\Phi^{2,1}_{\vartheta t_1t_2}[\overline{\bf Y}_{-m}^k,X_{-m}=x]$. For $j=4$, note that when we apply Lemma \ref{louis} to $\nabla^4 \log p_{\vartheta}({\bf Y}_{-m+1}^{k}|\overline{\bf{Y}}_{-m},X_{-m}=x)$, the last two terms on the right-hand side of Lemma \ref{louis} can be written as $\sum_{\mathcal{T}(4)\in\{-m+1,\ldots,k\}^4 } {\Phi}^{1,1,1,1}_{\vartheta \mathcal{T}(4)}[\overline{\bf Y}_{-m}^k,X_{-m}=x]$. The result for $j=5$ follows from a similar argument. For $j=6$, note that when we apply Lemma \ref{louis} to $\nabla^6 \log p_{\vartheta}({\bf Y}_{-m+1}^{k}|\overline{\bf{Y}}_{-m},X_{-m}=x)$, the last four terms on the right-hand side of Lemma \ref{louis} can be written as $\sum_{\mathcal{T}(6)\in\{-m+1,\ldots,k\}^6} {\Phi}^{\mathcal{I}(6)}_{\vartheta \mathcal{T}(6)}[\overline{\bf Y}_{-m}^k,X_{-m}=x]$.
\end{proof}

The following lemma provides bounds on $\Phi^{\mathcal{I}(j)}_{\vartheta \mathcal{T}(j)}[\mathcal{F}]$ defined in (\ref{Phi_defn}) and (\ref{Phi_defn_2}) and is used in the proof of Lemma \ref{lemma-bound-1}. For $j =2,\ldots,6$, define $\|\phi^i_t\|_{\infty}:=\sup_{\vartheta\in \mathcal{N}^*} \sup_{x,x'}|\phi^i(\vartheta,Y_t,x,\overline{\bf{Y}}_{t-1},x')|$ and $\|\phi^{\mathcal{I}(j)}_{\mathcal{T}(j)}\|_{\infty}:=\sum_{(\ell_1,\ldots,\ell_j)\in\sigma(\mathcal{I}(j)) }\|\phi_{t_1}^{\ell_1}\|_{\infty} \cdots \|\phi_{t_j}^{\ell_j}\|_{\infty}$.
\begin{lemma}\label{lemma_ijl}
Under Assumptions \ref{assn_a1}, \ref{assn_a2}, and \ref{assn_a4}, there exists a finite non-stochastic constant $C$ that does not depend on $\rho$ such that, for all $m'\geq m\geq 0$, all $-m<t_1\leq t_2 \leq \cdots \leq t_j\leq n$, all $\vartheta\in \mathcal{N}^*$ and all $x\in \mathcal{X}$, and $j=2,\ldots,6$,
\begin{equation*}
\begin{aligned} 
(a) & \quad |\Phi^{\mathcal{I}(j)}_{\vartheta \mathcal{T}(j)}[\overline{\bf{Y}}_{-m}^n]| \leq C \rho^{(t_2-t_1-1)_+\vee(t_3-t_2-1)_+\vee \cdots \vee(t_j-t_{j-1}-1)_+} \|\phi^{\mathcal{I}(j)}_{\mathcal{T}(j)}\|_{\infty},\\
(b) & \quad |\Phi^{\mathcal{I}(j)}_{\vartheta \mathcal{T}(j)}[\overline{\bf{Y}}_{-m}^n,X_{-m}=x]|\leq C \rho^{(t_2-t_1-1)_+\vee(t_3-t_2-1)_+\vee \cdots \vee(t_j-t_{j-1}-1)_+}\|\phi^{\mathcal{I}(j)}_{\mathcal{T}(j)}\|_{\infty},\\
(c) & \quad |\Phi^{\mathcal{I}(j)}_{\vartheta \mathcal{T}(j)}[\overline{\bf{Y}}_{-m}^n,X_{-m}=x]
- \Phi^{\mathcal{I}(j)}_{\vartheta \mathcal{T}(j)}[\overline{\bf{Y}}_{-m}^n]|\leq C\rho^{(m+t_1-1)_+}\|\phi^{\mathcal{I}(j)}_{\mathcal{T}(j)}\|_{\infty},\\
(d) & \quad |\Phi^{\mathcal{I}(j)}_{\vartheta \mathcal{T}(j)}[\overline{\bf{Y}}_{-m}^n,X_{-m}=x]
- \Phi^{\mathcal{I}(j)}_{\vartheta \mathcal{T}(j)}[\overline{\bf{Y}}_{-m'}^n,X_{-m'}=x]|\leq C\rho^{(m+t_1-1)_+}\|\phi^{\mathcal{I}(j)}_{\mathcal{T}(j)}\|_{\infty},\\
(e) & \quad |\Phi^{\mathcal{I}(j)}_{\vartheta \mathcal{T}(j)}[\overline{\bf{Y}}_{-m}^n] - \Phi^{\mathcal{I}(j)}_{\vartheta \mathcal{T}(j)}[\overline{\bf{Y}}_{-m}^{n -1}]| \leq C\rho^{(n -1 -t_j)_+}\|\phi^{\mathcal{I}(j)}_{\mathcal{T}(j)}\|_{\infty},\\
(f) & \quad |\Phi^{\mathcal{I}(j)}_{\vartheta \mathcal{T}(j)}[\overline{\bf{Y}}_{-m}^n,X_{-m}=x] - \Phi^{\mathcal{I}(j)}_{\vartheta \mathcal{T}(j)}[\overline{\bf{Y}}_{-m}^{n -1},X_{-m}=x]| \leq C\rho^{(n -1 -t_j)_+}\|\phi^{\mathcal{I}(j)}_{\mathcal{T}(j)}\|_{\infty}.
\end{aligned}
\end{equation*}
\end{lemma}
\begin{proof}[Proof of Lemma \ref{lemma_ijl}]

Recall $\sup_{\vartheta \in \mathcal{N}^*}\sup_{x,x'}|\phi^i(\vartheta,Y_t,x,\overline{\bf{Y}}_{t-1},x')- \mathbb{E}_\vartheta[\phi^i(\vartheta,Y_t,x,\overline{\bf{Y}}_{t-1},x')|\mathcal{F}]| \\\leq 2 \sup_{\vartheta \in \mathcal{N}^*}\sup_{x,x'}|\phi^i(\vartheta,Y_t,x,\overline{\bf{Y}}_{t-1},x')|$ for the conditioning sets $\mathcal{F}$ that appear in the lemma. Define $\tilde \phi^i_{\vartheta t}:=\phi^i(\vartheta,\overline{\bf Z}_{t-1}^t)- \mathbb{E}_\vartheta[\phi^i(\vartheta,\overline{\bf Z}_{t-1}^t)| \overline{\bf Y}_{-m}^{n}]$, so that $\mathbb{E}_\vartheta^c[\phi^{\ell_1}_{\vartheta t_1} \cdots \phi^{\ell_j}_{\vartheta t_j}|\overline{\bf{Y}}_{-m}^n] = \mathbb{E}_\vartheta[ \tilde \phi^{\ell_1}_{\vartheta t_1} \cdots \tilde \phi^{\ell_j}_{\vartheta t_j}| \overline{\bf Y}_{-m}^n]$. Henceforth, we suppress the subscript $\vartheta$ from $\phi^{i}_{\vartheta t}$ and $\tilde \phi^{i}_{\vartheta t}$.
 
Recall that $\phi^i(\vartheta,\overline{\bf Z}_{t-1}^t)$ depends on $X_t$ and $X_{t-1}$. Parts (c) and (d) follow from Lemma \ref{x_diff}(a) and the fact that, for any two probability measures $\mu_1$ and $\mu_2$, $\sup_{f(x):\max_x|f(x)| \leq 1}|\int f(x) d\mu_1(x)- \int f(x) d\mu_2(x)| = 2 \|\mu_1-\mu_2\|_{TV}$ (see, e.g., \citet[][Proposition 4.5]{levin09book}). Similarly, parts (e) and (f) for $t_j \leq n-1$ follow from Lemma \ref{x_diff}(b), and parts (e) and (f) for $t_j=n$ follow from $|\Phi^{\mathcal{I}(j)}_{\vartheta \mathcal{T}(j)}[\cdot]| \leq 2^j \|\phi^{\mathcal{I}(j)}_{\mathcal{T}(j)}\|_{\infty}$.
 
We proceed to show parts (a) and (b). The results for $j=2$ and $j=3$ follow from Lemma \ref{x_diff}(c) and 
\begin{equation} \label{prod_moments}
\begin{aligned}
E(X_{t_1} - EX_{t_1})\cdots(X_{t_j} - EX_{t_j}) & = \text{cov}[X_{t_1},(X_{t_2} - EX_{t_2})\cdots(X_{t_j} - EX_{t_j})]\\
& = \text{cov}[(X_{t_1} - EX_{t_1})\cdots(X_{t_{j-1}} - EX_{t_{j-1}}),X_{t_j}].
\end{aligned}
\end{equation}
Before proving the results for $j\geq 4$, we collect some results. For a conditioning set $\mathcal{F}= \overline{\bf{Y}}_{-m}^n$ or $\{\overline{\bf{Y}}_{-m}^n,X_m=x \}$, Lemmas \ref{x_diff}(c) and (\ref{prod_moments}) imply that
\begin{align}
&\qquad |\mathbb{E}_\vartheta^c[\phi^{\ell_1}_{ t_1} \cdots \phi^{\ell_j}_{ t_j}|\mathcal{F}]| \leq \mathcal{C} \rho^{(t_2-t_1-1)_+\vee(t_j-t_{j-1}-1)_+} \|\phi^{\mathcal{I}(j)}_{\mathcal{T}(j)}\|_{\infty}, \label{cov1} \\
& |\mathbb{E}_\vartheta^c[\phi^{\ell_1}_{ t_1} \cdots \phi^{\ell_j}_{ t_j}|\mathcal{F}] - \mathbb{E}_\vartheta^c[\phi^{\ell_1}_{ t_1}\cdots\phi^{\ell_k}_{t_k}|\mathcal{F}]\mathbb{E}_\vartheta^c[\phi^{\ell_{k+1}}_{ t_{k+1}}\cdots \phi^{\ell_j}_{ t_j}|\mathcal{F}]| \nonumber\\
&\quad = | \text{cov}_\vartheta[\tilde \phi^{\ell_1}_{ t_1}\cdots \tilde \phi^{\ell_k}_{t_k}, \tilde \phi^{\ell_{k+1}}_{ t_{k+1}}\cdots \tilde \phi^{\ell_j}_{ t_j} |\mathcal{F}] | \leq \mathcal{C} \rho^{(t_{k+1}-t_{k}-1)_+}\|\phi^{\mathcal{I}(j)}_{\mathcal{T}(j)}\|_{\infty}\quad \text{for any $2 \leq k \leq j-2$}. \label{cov_k} 
\end{align}

Parts (a) and (b) hold for $j=4$ because $\Phi^{\mathcal{I}(4)}_{\vartheta \mathcal{T}(4)}[\mathcal{F}]\leq \mathcal{C} \rho^{(t_{2}-t_{1}-1)_+\vee (t_4 - t_3-1)_+}\|\phi^{\mathcal{I}(4)}_{\mathcal{T}(4)}\|_{\infty}$ from (\ref{cov1}) and we have $\Phi^{\mathcal{I}(4)}_{\vartheta \mathcal{T}(4)}[\mathcal{F}] \leq \mathcal{C} \rho^{(t_{3}-t_{2}-1)_+}\|\phi^{\mathcal{I}(4)}_{\mathcal{T}(4)}\|_{\infty}$ from writing $\tilde \Phi_{\vartheta \mathcal{T}(4)}^{\ell_1 \ell_2 \ell_3 \ell_4}$ defined in (\ref{Phi_defn}) as $\tilde \Phi_{\vartheta \mathcal{T}(4)}^{\ell_1 \ell_2 \ell_3 \ell_4} = \text{cov}_{\vartheta}[ \tilde\phi^{\ell_1}_{ t_1} \tilde\phi^{\ell_2}_{ t_2}, \tilde \phi^{\ell_3}_{ t_3} \tilde \phi^{\ell_4}_{ t_4} | \mathcal{F}] - \mathbb{E}_{\vartheta}^c[\phi^{\ell_1}_{ t_1} \phi^{\ell_3}_{ t_3}|\mathcal{F}] \mathbb{E}_{\vartheta}^c[\phi^{\ell_2}_{ t_2} \phi^{\ell_4}_{ t_4}|\mathcal{F}] - \mathbb{E}_{\vartheta}^c[\phi^{\ell_1}_{ t_1} \phi^{\ell_4}_{ t_4}|\mathcal{F}]\mathbb{E}_{\vartheta}^c[\phi^{\ell_2}_{ t_2} \phi^{\ell_3}_{ t_3}|\mathcal{F}]$ and applying (\ref{cov_k}). Parts (a)--(b) for $j=5$ follow from a similar argument.

For $j=6$, first, $\Phi^{\mathcal{I}(6)}_{\vartheta \mathcal{T}(6)}[\mathcal{F}]$ is bounded by $\mathcal{C} \rho^{(t_2-t_1-1)_+\vee(t_6-t_{5}-1)_+}\|\phi^{\mathcal{I}(6)}_{\mathcal{T}(6)}\|_{\infty}$ from (\ref{cov1}). Second, write $\Phi^{\mathcal{I}(6)}_{\vartheta \mathcal{T}(6)}[\mathcal{F}] = A_1 + A_2$, where $A_1 = \mathbb{E}_\vartheta^c\left[ \phi_{ t_1} \phi_{ t_2} \phi_{ t_3} \phi_{ t_4} \phi_{ t_5} \phi_{ t_6} \middle| \mathcal{F}\right] - \mathbb{E}_\vartheta^c\left[ \phi_{ t_1} \phi_{ t_2} \phi_{ t_3}\middle| \mathcal{F}\right] \mathbb{E}_\vartheta^c\left[ \phi_{ t_4} \phi_{ t_5} \phi_{ t_6} \middle| \mathcal{F}\right]$ and $A_2$ denotes all the terms on the right-hand side of $\Phi^{\mathcal{I}(6)}_{\vartheta \mathcal{T}(6)}[\mathcal{F}]$ in (\ref{Phi_defn_2}) except for $A_1$. $A_1$ is bounded by $\mathcal{C}\rho^{(t_4-t_3-1)_+}\|\phi^{\mathcal{I}(6)}_{\mathcal{T}(6)}\|_{\infty}$ from (\ref{cov_k}), and $A_2$ is bounded by $\mathcal{C}\rho^{(t_4-t_3-1)_+}\|\phi^{\mathcal{I}(6)}_{\mathcal{T}(6)}\|_{\infty}$ from (\ref{cov1}). Therefore, $\Phi^{\mathcal{I}(6)}_{\vartheta \mathcal{T}(6)}[\mathcal{F}]$ is bounded by $\mathcal{C} \rho^{(t_4-t_3-1)_+}\|\phi^{\mathcal{I}(6)}_{\mathcal{T}(6)}\|_{\infty}$. Third, write $\Phi^{\mathcal{I}(6)}_{\vartheta \mathcal{T}(6)}[\mathcal{F}] = B_1 + B_2+B_3$, where $B_1 = \mathbb{E}_\vartheta^c [ \phi_{ t_1} \phi_{ t_2} \phi_{ t_3} \phi_{ t_4} \phi_{ t_5} \phi_{ t_6} | \mathcal{F}] - \mathbb{E}_\vartheta^c[ \phi_{ t_3} \phi_{ t_4} \phi_{ t_5}\phi_{ t_6}| \mathcal{F}] \mathbb{E}_\vartheta^c[ \phi_{ t_1} \phi_{ t_2} | \mathcal{F}]$, $B_2 = - \sum_{(\{1,2,c,d\},\{e,f\}) \in X_{61}} \mathbb{E}_\vartheta^c[ \phi_{1} \phi_{ 2} \phi_{ t_c} \phi_{ t_d} | \mathcal{F}] \mathbb{E}_\vartheta^c[ \phi_{ t_e} \phi_{ t_f} | \mathcal{F}] + 2 \sum_{(\{a,b\},\{c,d\},\{e,f\}) \in X_{63}} \mathbb{E}_\vartheta^c[ \phi_{ t_a} \phi_{ t_b} | \mathcal{F}] \mathbb{E}_\vartheta^c[ \phi_{ t_c} \phi_{ t_d} | \mathcal{F}] \mathbb{E}_\vartheta^c[ \phi_{ t_e} \phi_{ t_f} | \mathcal{F}]$, where $X_{61}$ is the set of ${4 \choose 2} = 6$ partitions of $\{1,2,3,4,5,6\}$ of the form of $\{1,2,c,d\},\{e,f\}$ and \\ $X_{63}:= \{(\{1,2\},\{3,4\},\{5,6\}), (\{1,2\},\{3,5\},\{4,6\}), (\{1,2\},\{3,6\},\{4,5\}) \}$, and $B_3$ denotes all the terms on the right-hand side of $\Phi^{\mathcal{I}(6)}_{\vartheta \mathcal{T}(6)}[\mathcal{F}]$ except for $B_1+B_2$. $B_1$ is bounded by $\mathcal{C} \rho^{(t_3-t_2-1)_+}\|\phi^{\mathcal{I}(6)}_{\mathcal{T}(6)}\|_{\infty}$ from (\ref{cov_k}). We can write $B_2$ as \\ $\sum_{(\{1,2,c,d\},\{e,f\}) \in X_{61}} \{ - \mathbb{E}_\vartheta^c[ \phi_{t_1} \phi_{ t_2} \phi_{ t_c} \phi_{ t_d} | \mathcal{F}] \mathbb{E}_\vartheta^c[ \phi_{ t_e} \phi_{ t_f} | \mathcal{F}] + \mathbb{E}_\vartheta^c[ \phi_{ t_1} \phi_{ t_2} | \mathcal{F}] \mathbb{E}_\vartheta^c[ \phi_{ t_c} \phi_{ t_d} | \mathcal{F}] \mathbb{E}_\vartheta^c[ \phi_{ t_e} \phi_{ t_f} | \mathcal{F}] \} = - \sum_{(\{1,2,c,d\},\{e,f\}) \in X_{61}} \mathbb{E}_\vartheta^c[ \phi_{ t_e} \phi_{ t_f} | \mathcal{F}] \text{cov}_{\vartheta}[\tilde \phi_{\theta t_1} \tilde \phi_{\theta t_2},\tilde \phi_{\theta t_c} \tilde \phi_{\theta t_d}|\mathcal{F}]$, then this is bounded by \\ $\mathcal{C}\rho^{(t_3-t_2-1)_+}\|\phi^{\mathcal{I}(6)}_{\mathcal{T}(6)}\|_{\infty}$ from (\ref{cov_k}). Finally, $B_3$ is bounded by $\mathcal{C} \rho^{(t_3-t_2-1)_+}\|\phi^{\mathcal{I}(6)}_{\mathcal{T}(6)}\|_{\infty}$ from (\ref{cov1}). Therefore, $\Phi^{\mathcal{I}(6)}_{\vartheta \mathcal{T}(6)}[\mathcal{F}]$ is bounded by $\mathcal{C} \rho^{(t_3-t_2-1)_+}\|\phi^{\mathcal{I}(6)}_{\mathcal{T}(6)}\|_{\infty}$. From a similar argument, $\Phi^{\mathcal{I}(6)}_{\vartheta \mathcal{T}(6)}[\mathcal{F}]$ is also bounded by $\mathcal{C}\rho^{(t_5-t_4-1)_+}\|\phi^{\mathcal{I}(6)}_{\mathcal{T}(6)}\|_{\infty}$, and parts (a) and (b) follow. 
\end{proof}

We next present the result that bound the difference between $\Delta^{\mathcal{I}(j)}_{j,k,m,x}(\vartheta)$ and $\overline{\Delta}^{\mathcal{I}(j)}_{j,k,m}(\theta)$ that appear on the right-hand side in Lemma \ref{ell_lambda}. This lemma extends Lemmas 13 and 17 of DMR. Let $r_{\mathcal{I}(1)} = q_{i_1}$; $r_{\mathcal{I}(2)} = q_{i_1}/2$ if $i_1=i_2$ and $(q_{i_1} \wedge q_{i_2})/2$ if $i_1 \neq i_2$; $r_{\mathcal{I}(3)} = q_{i_1}/3$ if $i_1=i_2=i_3$, $(q_{i_1}/2 \wedge q_{i_2}/4)$ if $i_1 \neq i_2=i_3$, $(q_{i_1}\wedge q_{i_2}\wedge q_{i_3})/3$ if $i_1$, $i_2$, $i_3$ are distinct; $r_{\mathcal{I}(4)} = q_{i_1}/4$ if $i_1=i_2=i_3=i_4$, $(q_{i_1} \wedge q_{i_3})/4$ if $i_1 \neq i_2=i_3=i_4$ or $i_1 = i_2 \neq i_3=i_4$; $r_{\mathcal{I}(5)} = q_{i_1}/5$ if $i_1=i_2=i_3=i_4=i_5$; $(q_{i_1}/3\wedge q_{i_2}/6)$ if $i_1 \neq i_2=i_3=i_4=i_5$; $r_{\mathcal{I}(6)} = q_1/6$. Part (d) of this lemma establishes the uniform convergence of $\{\Delta^{\mathcal{I}(j)}_{j,k,m,x}(\vartheta)\}_{m\geq 0}$ to a random variable that does not depend on $x$.
\begin{lemma} \label{lemma-bound-1} 
Under Assumptions \ref{assn_a1}, \ref{assn_a2}, and \ref{assn_a4}, for $j=1,\ldots,6$, there exist random variables ${K}_{\mathcal{I}(j)}, \{M_{\mathcal{I}(j), k}\}_{k=1}^n \in L^{r_{\mathcal{I}(j)}}(\mathbb{P}_{\vartheta^*})$ such that, for all $1 \leq k \leq n$ and $m' \geq m \geq 0$, 
\begin{align*} 
\text{(a)}\qquad & \sup_{x\in \mathcal{X}} \sup_{\vartheta\in \mathcal{N}^*} |\Delta^{\mathcal{I}(j)}_{j,k,m,x}(\vartheta)-\overline\Delta^{\mathcal{I}(j)}_{j,k,m}(\vartheta)|\leq {K}_{\mathcal{I}(j)}(k+m)^{ 7 } \rho^{ \lfloor (k+m-1)/24 \rfloor } \quad \text{$\mathbb{P}_{\vartheta^*}$-a.s.},\\
\text{(b)}\qquad &\sup_{x\in \mathcal{X}} \sup_{\vartheta\in \mathcal{N}^*} |\Delta^{\mathcal{I}(j)}_{j,k,m,x}(\vartheta)-\Delta^{\mathcal{I}(j)}_{j,k,m',x}(\vartheta)|\leq {K}_{\mathcal{I}(j)}(k+m)^{ 7 } \rho^{ \lfloor (k+m-1)/ 1340 \rfloor }\quad \text{$\mathbb{P}_{\vartheta^*}$-a.s.},
\end{align*}
(c) $\sup_{m\geq 0} \sup_{x\in \mathcal{X}} \sup_{\vartheta\in \mathcal{N}^*} |\Delta^{\mathcal{I}(j)}_{j,k,m,x}(\vartheta)| + \sup_{m\geq 0} \sup_{\vartheta\in \mathcal{N}^*} |\overline\Delta^{\mathcal{I}(j)}_{j,k,m}(\vartheta)| \leq M_{\mathcal{I}(j), k}$ $\mathbb{P}_{\vartheta^*}$-a.s., (d) Uniformly in $\vartheta\in \mathcal{N}^*$ and $x \in \mathcal{X}$, $\Delta^{\mathcal{I}(j)}_{j,k,m,x}(\vartheta)$ and $\overline\Delta^{\mathcal{I}(j)}_{j,k,m}(\vartheta)$ converge $\mathbb{P}_{\vartheta^*}$-a.s.\ and in $L^{r_{\mathcal{I}(j)}}(\mathbb{P}_{\vartheta^*})$ to $\Delta^{\mathcal{I}(j)}_{j,k,\infty}(\vartheta) \in L^{r_{\mathcal{I}(j)}}(\mathbb{P}_{\vartheta^*})$ as $m\rightarrow \infty$. 
\end{lemma} 

\begin{proof}[Proof of Lemma \ref{lemma-bound-1}] First, we prove parts (a) and (b). Recall $\mathcal{T}(j) = (t_1,\ldots,t_j)$. For part (a), define, suppressing the dependence of $A_{\mathcal{T}(j)}$ on $\vartheta$ and $\mathcal{I}(j)$,
\begin{align*}
A_{\mathcal{T}(j)} & :=
\begin{cases}
\begin{aligned}
&\Phi^{\mathcal{I}(j)}_{\vartheta\mathcal{T}(j)}\left[\overline{\bf Y}_{-m}^k,X_{-m}=x\right] -\Phi^{\mathcal{I}(j)}_{\vartheta\mathcal{T}(j)}\left[ \overline{\bf Y}_{-m}^{k-1},X_{-m}=x\right] \\
& \quad -\Phi^{\mathcal{I}(j)}_{\vartheta\mathcal{T}(j)}\left[\overline{\bf Y}_{-m}^k\right] +\Phi^{\mathcal{I}(j)}_{\vartheta\mathcal{T}(j)}\left[\overline{\bf Y}_{-m}^{k-1}\right],
\end{aligned}
& \text{ if } \max\{t_1,\ldots, t_j\} < k, \\
\Phi^{\mathcal{I}(j)}_{\vartheta\mathcal{T}(j)}\left[\overline{\bf Y}_{-m}^k,X_{-m}=x\right] -\Phi^{\mathcal{I}(j)}_{\vartheta\mathcal{T}(j)}\left[\overline{\bf Y}_{-m}^k\right], & \text{ otherwise},
\end{cases}\\
A_{\mathcal{T}(j,\ell,k)} &:= A_{t_1t_2\cdots t_{j-\ell} \underbrace{{\scriptstyle k\cdots k}}_{{\scriptstyle \ell \ \text{times}}}, }\quad \text{where } \mathcal{T}(j,\ell,k):= (\mathcal{T}(j-\ell),\underbrace{ k,\cdots, k}_{\ell \ \text{times}}).
\end{align*}
Then, we can write $\Delta^{\mathcal{I}(j)}_{j,k,m,x}(\vartheta)-\overline\Delta^{\mathcal{I}(j)}_{j,k,m}(\vartheta) = \sum_{\mathcal{T}(j)\in\{-m+1,\ldots,k\}^j} A_{\mathcal{T}(j)} = \Delta_a + \Delta_b + \Delta_c$, where 
\[
\Delta_a := \sum_{\mathcal{T}(j)\in\{-m+1,\ldots,k-1\}^j} A_{ \mathcal{T}(j)},\quad \Delta_b := \sum_{\ell=1}^{j-1} \binom{j}{\ell} \sum_{\mathcal{T}(j-\ell)\in\{-m+1,\ldots,k-1\}^{j-\ell}} A_{ \mathcal{T}(j,\ell,k)}, \quad \Delta_c:= A_{(k,
\ldots,k)},
\]
and $\Delta_b:=0$ when $j=1$.
From Lemma \ref{lemma_ijl} and the symmetry of $A_{ \mathcal{T}(j)}$, $\Delta_a$ is bounded by $\mathcal{C} B_{j,k,m} M^{\mathcal{I}(j)}_{j,k,m}$, where
\begin{align*} 
B_{j,k,m} & := \sum_{-m+1\leq t_1\leq t_2 \leq \cdots \leq t_j\leq k-1} \left(\rho^{(m+t_1-1)_+} \wedge \rho^{(t_2-t_1-1)_+} \wedge \cdots \wedge \rho^{(t_j-t_{j-1}-1)_+} \wedge \rho^{(k-1-t_j-1)_+} \right) \\
& = \sum_{1\leq t_1\leq t_2 \leq \cdots \leq t_j \leq k+m-1} \left(\rho^{(t_1-1)_+} \wedge \rho^{(t_2-t_1-1)_+} \wedge \cdots \wedge \rho^{(t_j-t_{j-1}-1)_+} \wedge \rho^{(k+m-1-t_j-1)_+} \right), \\ 
M^{\mathcal{I}(j)}_{j,k,m} & := \max_{-m+1\leq t_1, \ldots, t_j\leq k-1} \|\phi^{i_1}_{t_1}\|_{\infty}\| \phi^{i_2}_{t_2}\|_{\infty} \cdots \| \phi^{i_j}_{t_j}\|_{\infty}.
\end{align*}
From $(t-1)_+ \geq \lfloor t/2 \rfloor$ and Lemma \ref{rho_sum}, $B_{j,k,m}$ is bounded by $C_{j2}(\rho) \rho^{\lfloor (k+m-1)/4j\rfloor}$.

We proceed to derive a bound of $M^{\mathcal{I}(j)}_{j,k,m}$. Define $\|\phi^{i}\|^\ell_\infty:=\sum_{t=-\infty}^\infty (|t|\vee 1)^{-2}\|\phi^{i}_t\|_{\infty}^\ell$. When $i_1=i_2=\cdots=i_j$, it follows from Lemma \ref{hoelder} that $M^{\mathcal{I}(j)}_{j,k,m} \leq (k+m)^{j+1}\|\phi^{i_1}\|_{\infty}^j$, and $\|\phi^{i_1}\|_{\infty}^j \in L^{r_{\mathcal{I}(j)}}(\mathbb{P}_{\vartheta^*})$ from Assumption \ref{assn_a4}. In the other cases, observe that if $x,y,z \geq 0$, we have $xy \leq x^2+y^2$, $xyz \leq x ^3+ y ^3+z ^3$, $xy \leq x^{4} + y^{4/3}$, and $xy \leq x^{3} + y^{3/2}$ from Young's inequality. By using this result and Lemma \ref{hoelder}, we can bound $M^{\mathcal{I}(j)}_{j,k,m}$ by
\[
\begin{array}{ccc}
j=2 \text{ and } i_1 \neq i_2: & (k+m)^2 (\|\phi^{i_1}\|_{\infty}^2 + \|\phi^{i_2}\|_{\infty}^2), \\
j=3 \text{ and } i_1 \neq i_2=i_3: & (k+m)^3 ( \|\phi^{i_1}\|_{\infty}^2 + \|\phi^{i_2}\|_{\infty}^4), \\
j=3 \text{ and } i_1,i_2,i_3 \text{ are distinct}: &(k+m)^2 ( \|\phi^{i_1}\|_{\infty}^3 + \|\phi^{i_2}\|_{\infty}^3 + \|\phi^{i_3}\|_{\infty}^3), \\
j=4 \text{ and } i_1 \neq i_2=i_3=i_4: & (k+m)^3 ( \|\phi^{i_1}\|_{\infty}^4 + \|\phi^{i_2}\|_{\infty}^4),\\
j=4 \text{ and } i_1 = i_2 \neq i_3=i_4: & (k+m)^3 ( \|\phi^{i_1}\|_{\infty}^4 + \|\phi^{i_3}\|_{\infty}^4), \\
j=5 \text{ and } i_1 \neq i_2 = i_3=i_4=i_5: & (k+m)^3 ( \|\phi^{i_1}\|_{\infty}^{ 3 } + \|\phi^{i_2}\|_{\infty}^6).
\end{array}
\]
Therefore, from Assumption \ref{assn_a4}, $\Delta_a$ is bounded by the right-hand side of part (a). From Lemmas \ref{lemma_ijl} and \ref{rho_sum}, $\Delta_b$ is bounded by $\mathcal{C} \sum_{\ell=1}^{j-1} \sum_{-m+1 \leq t_1\leq \cdots \leq t_{j-\ell} \leq k-1} (\rho^{(m+t_1-1)_+} \wedge \rho^{(t_2-t_1-1)_+} \wedge \cdots \wedge \rho^{(k-t_{j- \ell }-1)_+}) M_{j,k+1,m}^{\mathcal{I}(j)} \leq \mathcal{C} \rho^{\lfloor (k+m-1)/4(j-1) \rfloor} M_{j,k+1,m}^{\mathcal{I}(j)} $. Similarly, $\Delta_c$ is bounded by $\mathcal{C} \rho^{\lfloor (k+m-1)/4(j-1) \rfloor} M_{j,k+1,m}^{\mathcal{I}(j)}$, and part (a) of the lemma follows.

For part (b), define, for $-m'+1 \leq t_1, \ldots, t_j \leq k$,
\begin{equation*}
D_{\mathcal{T}(j),m',x}:=
\begin{cases}
\Phi^{\mathcal{I}(j)}_{\theta \mathcal{T}(j)}[\overline{\bf Y}_{-m'}^k,X_{-m'}=x] -\Phi^{\mathcal{I}(j)}_{\theta \mathcal{T}(j)}[\overline{\bf Y}_{-m'}^{k-1},X_{-m'}=x], & \text{ if } \max\{t_1,\ldots, t_j\} < k, \\
\Phi^{\mathcal{I}(j)}_{\theta \mathcal{T}(j)}[\overline{\bf Y}_{-m'}^k,X_{-m'}=x], & \text{ otherwise},
\end{cases}
\end{equation*}
and define $D_{\mathcal{T}(j),m,x}$ similarly. Then, we can write $\Delta^{\mathcal{I}(j)}_{j,k,m,x}(\theta)= \sum_{\mathcal{T}(j)\in\{-m+1,\ldots,k\}^{j}}D_{\mathcal{T}(j),m,x}$ and $\Delta^{\mathcal{I}(j)}_{j,k,m',x}(\theta) = \sum_{\mathcal{T}(j)\in\{-m'+1,\ldots,k\}^{j}} D_{\mathcal{T}(j),m',x} = \Delta_d +\Delta_e$, where $ \Delta_d := \sum_{\mathcal{T}(j)\in\{-m+1,\ldots,k\}^{j}}D_{\mathcal{T}(j),m',x}$ and
\[
\Delta_e:= \sum_{\ell=1}^{j} \binom{j}{\ell}\sum_{t_1=-m'+1}^{-m} \cdots \sum_{t_\ell=-m'+1}^{-m} \sum_{t_{\ell+1}=-m+1}^k \cdots \sum_{t_{j}=-m+1}^k D_{\mathcal{T}(j),m',x}.
\]
From the same argument as part (a), $\Delta^{\mathcal{I}(j)}_{j,k,m,x}(\theta) - \Delta_d$ is bounded by the right-hand side of part (a). For $\Delta_e$, observe that, with $M_{j} := \max_{1\leq \ell \leq j}\binom{j}{\ell}$,
\begin{align*}
| \Delta_e | & \leq M_j \sum_{\ell=1}^{j} \sum_{t_1=-m'+1}^{-m} \sum_{t_2=-m'+1}^{-m} \cdots \sum_{t_{\ell}=-m'+1}^{-m} \sum_{t_{\ell+1}=-m+1}^k \cdots \sum_{t_{j}=-m+1}^k \left| D_{\mathcal{T}(j),m',x} \right| \nonumber \\
& \leq j M_j \sum_{t_1=-m'+1}^{-m} \sum_{t_2=-m'+1}^{k} \cdots \sum_{t_j=-m'+1}^{k} \left| D_{\mathcal{T}(j),m',x} \right| \nonumber \\
& \leq j M_j j! \sum_{t_1=-m'+1}^{-m} \sum_{t_1 \leq t_2 \leq \cdots \leq t_{j}\leq k }\left| D_{\mathcal{T}(j),m',x} \right|.
\end{align*}
From Lemma \ref{lemma_ijl}, if $t_1 \leq \cdots \leq t_{j}$, we have $|D_{\mathcal{T}(j),m',x}| \leq \mathcal{C} [\mathbb{I}\{t_j < k\} (\rho^{(t_2-t_1-1)_+} \wedge \rho^{(t_j -t_{j-1}-1)_+} \wedge \cdots \wedge\rho^{(k-1-t_j-1)_+} ) + \mathbb{I}\{t_j = k\} (\rho^{(t_2-t_1-1)_+} \wedge \cdots \wedge \rho^{(t_j -t_{j-1}-1)_+})] \|\phi^{\mathcal{I}(j)}_{\mathcal{T}(j)}\|_{\infty}$. Hence, part (b) follows from Lemma \ref{rho_m}.

For part (c), observe that $\sup_{m\geq 0} \sup_{x\in \mathcal{X}} \sup_{\vartheta\in \mathcal{N}^*} |\Delta^{\mathcal{I}(j)}_{j,k,m,x}(\vartheta) | \leq A+B$, where $A:= \sup_{m\geq 0}\sup_{x\in \mathcal{X}} \sup_{\vartheta\in \mathcal{N}^*} |\Delta^{\mathcal{I}(j)}_{j,k,m,x}(\vartheta) - \Delta^{\mathcal{I}(j)}_{j,k,0,x}(\vartheta) |$ and $B:=\sup_{x\in \mathcal{X}} \sup_{\vartheta\in \mathcal{N}^*} |\Delta^{\mathcal{I}(j)}_{j,k,0,x}(\vartheta) |$. $A$ is bounded by $K_{\mathcal{I}(j)} k^7 \rho^{ \lfloor (k-1)/ 1340 \rfloor }$ from part (b). $B$ does not depend on $m$ and is distributionally equivalent to $\sup_{x\in \mathcal{X}} \sup_{\vartheta\in \mathcal{N}^*} |\Delta^{\mathcal{I}(j)}_{j,1,k-1,x}(\vartheta)|$. This is bounded by $\sup_{x\in \mathcal{X}} \sup_{\vartheta\in \mathcal{N}^*} |\Delta^{\mathcal{I}(j)}_{j,1,k-1,x}(\vartheta) - \Delta^{\mathcal{I}(j)}_{j,1,0,x}(\vartheta)| + \sup_{x\in \mathcal{X}} \sup_{\vartheta\in \mathcal{N}^*} |\Delta^{\mathcal{I}(j)}_{j,1,0,x}(\vartheta)|$. The first term is in $L^{r_{\mathcal{I}(j)}}(\mathbb{P}_{\vartheta^*})$ from part (b), and the second term is in $L^{r_{\mathcal{I}(j)}}(\mathbb{P}_{\vartheta^*})$ from the definition of $\Delta^{\mathcal{I}(j)}_{j,k,m,x}(\vartheta)$. Therefore, there exists $M_{\mathcal{I}(j), k} \in L^{r_{\mathcal{I}(j)}}(\mathbb{P}_{\vartheta^*})$ such that $A+B \leq M_{\mathcal{I}(j), k}$, and part (c) holds in view of part (a). Part (d) follows from parts (a)--(c) because parts (a)--(c) imply that $\{\Delta^{\mathcal{I}(j)}_{j,k,m,x}(\vartheta)\}_{m\geq 0}$ and $\{\Delta^{\mathcal{I}(j)}_{j,k,m}(\vartheta)\}_{m\geq 0}$ are uniform $L^{r_{\mathcal{I}(j)}}(\mathbb{P}_{\vartheta^*})$-Cauchy sequences with respect to $\vartheta\in \mathcal{N}^*$ that converge to the same limit and $L^{q}(\mathbb{P}_{\vartheta^*})$ is complete. 
\end{proof}
 
\begin{lemma}\label{lemma-ratio}
Under Assumptions \ref{assn_a1}, \ref{assn_a2}, and \ref{assn_a4}, there exist random variables $\{K_k\}_{k=1}^n\in L^{(1+\varepsilon)q_\vartheta/\varepsilon}(\mathbb{P}_{\vartheta^*})$ and $\rho \in (0,1)$ such that, for all $1 \leq k \leq n$ and $m' \geq m \geq 0$,
\begin{align*}
\sup_{\vartheta \in \mathcal{N}^*} \left|\frac{ \overline p_{\vartheta}(Y_k|\overline{\bf{Y}}_{-m}^{k-1})}{\overline p_{\vartheta^*}(Y_k|\overline{\bf{Y}}_{-m}^{k-1})}\right| \leq K_k, \quad
 \sup_{x\in \mathcal{X}}\sup_{\vartheta \in \mathcal{N}^*} \left|\frac{ p_{\vartheta}(Y_k|\overline{\bf{Y}}_{-m}^{k-1}, X_{-m}=x)}{ p_{\vartheta^*}(Y_k|\overline{\bf{Y}}_{-m}^{k-1}, X_{-m}=x)}-\frac{ \overline p_{\vartheta}(Y_k|\overline{\bf{Y}}_{-m}^{k-1})}
{\overline p_{\vartheta^*}(Y_k|\overline{\bf{Y}}_{-m}^{k-1})}\right| \leq K_k \rho^{k+m-1}.
\end{align*}
Furthermore, these bounds hold uniformly in $x \in \mathcal{X}$ when $\overline p_{\vartheta}(Y_k|\overline{\bf{Y}}_{-m}^{k-1})$ and $\overline p_{\vartheta^*}(Y_k|\overline{\bf{Y}}_{-m}^{k-1})$ are replaced with $ p_{\vartheta}(Y_k|\overline{\bf{Y}}_{-m'}^{k-1}, X_{-m'}=x)$ and $ p_{\vartheta^*}(Y_k|\overline{\bf{Y}}_{-m'}^{k-1}, X_{-m'}=x)$.
\end{lemma}
\begin{proof}[Proof of Lemma \ref{lemma-ratio}]
The first result follows from noting that $\overline p_{\vartheta}(Y_k|\overline{\bf{Y}}_{-m}^{k-1}) \\
= \sum_{(x_{k-1},x_k )\in \mathcal{X}^2} g_\vartheta(Y_k|\overline{\bf{Y}}_{k-1}, x_k) q_{\vartheta_x}(x_{k-1},x_k) \mathbb{P}_\vartheta(x_{k-1}|\overline{\bf{Y}}_{-m}^{k-1}) \in [\sigma_-G_{\vartheta k},\sigma_+G_{\vartheta k}]$ and using Assumption \ref{assn_a4}(b). For the second result, observe that $| p_{\vartheta}(Y_k|\overline{\bf{Y}}_{-m}^{k-1}, X_{-m}=x) - \overline p_{\vartheta}(Y_k|\overline{\bf{Y}}_{-m}^{k-1})| \leq \sum_{(x_{k-1},x_k )\in \mathcal{X}^2} g_\vartheta(Y_k|\overline{\bf{Y}}_{k-1}, x_k) q_{\vartheta_x}(x_{k-1},x_k) |\mathbb{P}_\vartheta(x_{k-1}|\overline{\bf{Y}}_{-m}^{k-1},X_{-m}=x) - \mathbb{P}_\vartheta(x_{k-1}|\overline{\bf{Y}}_{-m}^{k-1})| \leq \rho^{k+m-1}\sigma_+ G_{\vartheta k}/\sigma_-$, where the second inequality follows from Lemma \ref{x_diff}(a). The second result then follows from writing the left-hand side as
\[
\frac{ p_{\vartheta}(Y_k|\overline{\bf{Y}}_{-m}^{k-1},X_{-m}=x)- \overline p_{\vartheta}(Y_k|\overline{\bf{Y}}_{-m}^{k-1})}{ p_{\vartheta^*}(Y_k|\overline{\bf{Y}}_{-m}^{k-1},X_{-m}=x)}
+ \frac{ \overline p_{\vartheta}(Y_k|\overline{\bf{Y}}_{-m}^{k-1})}{\overline p_{\vartheta^*}(Y_k|\overline{\bf{Y}}_{-m}^{k-1})} \frac{\overline p_{\vartheta^*}(Y_k|\overline{\bf{Y}}_{-m}^{k-1}) - p_{\vartheta^*}(Y_k|\overline{\bf{Y}}_{-m}^{k-1},X_{-m}=x)}{ p_{\vartheta^*}(Y_k|\overline{\bf{Y}}_{-m}^{k-1},X_{-m}=x)},
\]
noting that $p_{\vartheta}(Y_k|\overline{\bf{Y}}_{-m}^{k-1},X_{-m}=x) \geq \sigma_-G_{\vartheta k}$, and using the derived bounds. The results with $ p_{\vartheta}(Y_k|\overline{\bf{Y}}_{-m'}^{k-1}, X_{-m'}=x)$ and $ p_{\vartheta^*}(Y_k|\overline{\bf{Y}}_{-m'}^{k-1}, X_{-m'}=x)$ are proven similarly.
\end{proof}

The following result originally appeared in equations (59)--(60) of \citet{kasaharashimotsu15jasa}. We state this as a lemma for ease of reference.
\begin{lemma} \label{lemma_normal_der}
Let $f(\mu,\sigma^2)$ denote the density of $N(\mu,\sigma^2)$. Then
\begin{equation*}
\left.\nabla_{\lambda_\mu^k}f( c_1\lambda_\mu, c_2 \lambda_\mu^2)\right|_{\lambda_\mu=0} =
\begin{cases}
c_1\nabla_{\mu} f(0,0) & \text{if}\ k = 1, \\
c_1^2\nabla_{\mu^2} f(0,0) + 2c_2 \nabla_{\sigma^2} f(0,0) & \text{if}\ k = 2, \\
c_1^3\nabla_{\mu^3} f(0,0) + 6 c_1 c_2 \nabla_{\mu\sigma^2} f(0,0) & \text{if}\ k = 3, \\
c_1^4\nabla_{\mu^4} f(0,0) + 12 c_1^2c_2 \nabla_{\mu^2} f(0,0)\nabla_{\sigma^2} f(0,0) + 12c_2^2 \nabla_{\sigma^4} f(0,0) & \text{if}\ k = 4.
\end{cases}
\end{equation*}
\end{lemma}
\begin{proof}[Proof of Lemma \ref{lemma_normal_der}]
Observe that a composite function $f(\lambda_\mu, h(\lambda_\mu))$ satisfies
$\nabla_{\lambda_\mu^k}f(\lambda_\mu,h(\lambda_\mu)) = (\nabla_{\lambda_\mu} + \nabla_u)^k f(\lambda_\mu,h(u))|_{u = \lambda_\mu} = \sum_{j = 0}^k \binom{k}{j} \nabla_{\lambda_\mu^{k - j}u^{j}} f(\lambda_\mu,h(u))|_{u = \lambda_\mu}$. Further, because $\nabla_{u^j}u^2|_{u = 0} = 0$ except for $j = 2$, it follows from Fa\`a di Bruno's formula that $\nabla_{u^j} f(c_1\lambda_\mu,c_2 u^2)|_{\lambda_\mu = u = 0}$ is $0$ if $j = 1, 3$, is $2c_2\nabla_{h} f(0, h(0))$ if $j = 2$, and is $12c_2^2 \nabla_{h^2} f(0,h(0))$ if $j = 4$. Therefore, the stated result follows.
\end{proof}

\begin{lemma} \label{lemma_d34}
Suppose the assumptions of Proposition \ref{P-quadratic-N1} hold. Then, there exist $\bar\varrho_1, \bar\varrho_2, \bar\varrho_3 \in (0,\varrho)$ such that, for all $k \geq 1$,
\begin{align*}
(a) & \quad \frac{\nabla_{\lambda_\mu^3} \overline p_{\psi^*\pi} (Y_k| \overline{{\bf Y}}_{0}^{k-1})}{\overline p_{\psi^*\pi} (Y_k| \overline{{\bf Y}}_{0}^{k-1})} = \varrho \frac{\nabla_\varrho \nabla_{\lambda_\mu^3} \overline p_{\psi^* \bar \varrho_1 \alpha} (Y_k| \overline{{\bf Y}}_{0}^{k-1})}{\overline p_{\psi^* \bar \varrho_1 \alpha} (Y_k| \overline{{\bf Y}}_{0}^{k-1})}, \\
(b) & \quad \frac{\nabla_{\lambda_\mu^4} \overline p_{\psi^*\pi} (Y_k| \overline{{\bf Y}}_{0}^{k-1})}{\overline p_{\psi^*\pi} (Y_k| \overline{{\bf Y}}_{0}^{k-1})} - b(\alpha)\frac{\nabla_{\lambda_\sigma^2} \overline p_{\psi^*\pi} (Y_k| \overline{{\bf Y}}_{0}^{k-1})}{\overline p_{\psi^*\pi} (Y_k| \overline{{\bf Y}}_{0}^{k-1})} \\
& \quad = \varrho \frac{\nabla_\varrho \nabla_{\lambda_\mu^4} \overline p_{\psi^* \bar \varrho_2 \alpha} (Y_k| \overline{{\bf Y}}_{0}^{k-1})}{\overline p_{\psi^* \bar \varrho_2 \alpha} (Y_k| \overline{{\bf Y}}_{0}^{k-1})} - \varrho \frac{\nabla_\varrho \nabla_{\lambda_\sigma^2} \overline p_{\psi^* \bar \varrho_3 \alpha} (Y_k| \overline{{\bf Y}}_{0}^{k-1})}{\overline p_{\psi^* \bar \varrho_3 \alpha} (Y_k| \overline{{\bf Y}}_{0}^{k-1})}.
\end{align*} 
\end{lemma}

\begin{proof}[Proof of Lemma \ref{lemma_d34}] 
Part (a) holds if
\begin{equation} \label{d3p}
\nabla_{\lambda_\mu^3} \overline p_{\psi^* 0 \alpha} (Y_k| \overline{{\bf Y}}_{0}^{k-1}) / \overline p_{\psi^* 0 \alpha} (Y_k| \overline{{\bf Y}}_{0}^{k-1}) =0,
\end{equation}
because (i) $\nabla_{\lambda_\mu^3}\overline p_{\psi^* \varrho \alpha} (Y_k| \overline{{\bf Y}}_{0}^{k-1}) - \nabla_{\lambda_\mu^3}\overline p_{\psi^* 0 \alpha} (Y_k| \overline{{\bf Y}}_{0}^{k-1}) = \nabla_\varrho \nabla_{\lambda_\mu^3}\overline p_{\psi^* \bar \varrho \alpha } (Y_k| \overline{{\bf Y}}_{0}^{k-1}) \varrho$ for $\bar \varrho \in (0,\varrho)$ from the mean value theorem and (ii) $\overline p_{\psi^* \varrho \alpha} (Y_k| \overline{{\bf Y}}_{0}^{k-1})$ does not depend on the value of $\varrho$.

We proceed to show (\ref{d3p}). Note that$\nabla_{\lambda_\mu^3} \overline p_{\psi^*\pi} (Y_k| \overline{{\bf Y}}_{0}^{k-1}) /\overline p_{\psi^*\pi} (Y_k| \overline{{\bf Y}}_{0}^{k-1}) = \nabla_{\lambda_\mu^3} \log \overline p_{\psi^*\pi} ({\bf Y}^k_1| \overline{{\bf Y}}_{0}) - \nabla_{\lambda_\mu^3} \log \overline p_{\psi^*\pi} ({\bf Y}^{k-1}_1| \overline{{\bf Y}}_{0})$ from (\ref{der_formula}) and $\nabla_\lambda \overline p_{\psi^* \pi} (Y_k| \overline{{\bf Y}}_{0}^{k-1})=0$. Let $\nabla^i \ell_t^* := \nabla_{\lambda_\mu^i} \log g_t^*$ with $\nabla \ell_t^*=\nabla^1 \ell_t^*$. Observe that 
\begin{equation} \label{d3lambda0}
\begin{aligned}
 \nabla_{\lambda_\mu^3} \log \overline p_{\psi^* 0 \alpha} ({\bf Y}^k_1| \overline{{\bf Y}}_{0})
& = \sum_{t=1}^k \mathbb{E}_{\psi^* 0 \alpha} \left[\nabla^3 \ell^*_{t} \middle|\overline{{\bf Y}}_{0}^k \right] + 3 \sum_{t_1=1}^k \sum_{t_2=1}^k \mathbb{E}_{\psi^* 0 \alpha} \left[ \nabla^2 \ell^*_{t_1} \nabla \ell^*_{t_2} \middle| \overline{{\bf Y}}_{0}^k \right] \\
& \qquad + \sum_{t_1=1}^k \sum_{t_2=1}^k \sum_{t_3=1}^k \mathbb{E}_{\psi^* 0 \alpha} \left[ \nabla \ell^*_{t_1} \nabla \ell^*_{t_2} \nabla \ell^*_{t_3} \middle| \overline{{\bf Y}}_{0}^k \right] \\
& = \sum_{t=1}^k \mathbb{E}_{\psi^* 0 \alpha} \left[\nabla^3 \ell^*_{t} + 3 \nabla^2 \ell^*_{t} \nabla \ell^*_{t} + \nabla \ell^*_{t} \nabla \ell^*_{t} \nabla \ell^*_{t} \middle| \overline{{\bf Y}}_{0}^k \right] \\
& = \sum_{t=1}^k \mathbb{E}_{\psi^* 0 \alpha} \left[\nabla_{\lambda_\mu^3} g_t^* /g_t^* \middle| \overline{{\bf Y}}_{0}^k \right],
\end{aligned}
\end{equation}
where the first equality follows from Lemma \ref{louis}, the second equality holds because (i) $X_t$ is serially independent when $\varrho=0$ and (ii) $\nabla\ell^*_{t} = d_{1t} \nabla_{\mu} f^*_{t}/f^*_{t}$ and $\nabla^2 \ell^*_{t} = d_{2t} \nabla_{\mu}^2 f^*_{t}/f^*_{t} - (d_{1t} \nabla_{\mu} f^*_{t}/f^*_{t})^2$, and (iii) $\mathbb{E}_{\psi^* 0 \alpha} [d_{1t} | \overline{{\bf Y}}_{0}^k ] = \mathbb{E}_{\psi^* 0 \alpha} [d_{2t} | \overline{{\bf Y}}_{0}^k ] =0$ from (\ref{d3g2}), and the third equality follows from (\ref{der_formula}) The right-hand side is 0 from (\ref{d3g2}), and hence part (a) is proven.

For part (b), from a similar argument to part (a), the stated result holds if 
\begin{equation} \label{d4p}
\nabla_{\lambda_\mu^4}\overline p_{\psi^* 0 \alpha} (Y_k| \overline{{\bf Y}}_{0}^{k-1}) / \overline p_{\psi^* 0 \alpha} (Y_k| \overline{{\bf Y}}_{0}^{k-1}) = b(\alpha)\nabla_{\lambda_\sigma^2} \overline p_{\psi^* 0 \alpha} (Y_k| \overline{{\bf Y}}_{0}^{k-1}) / \overline p_{\psi^* 0 \alpha} (Y_k| \overline{{\bf Y}}_{0}^{k-1}).
\end{equation}
Observe that $\nabla_{\lambda_\mu^4} \overline p_{\psi^* 0 \alpha} (Y_k| \overline{{\bf Y}}_{0}^{k-1})/\overline p_{\psi^* 0 \alpha} (Y_k| \overline{{\bf Y}}_{0}^{k-1}) = \nabla_{\lambda_\mu^4} \log \overline p_{\psi^* 0 \alpha} ({\bf Y}^k_0| \overline{{\bf Y}}_{0}) - \nabla_{\lambda_\mu^4} \log \overline p_{\psi^* 0 \alpha} ({\bf Y}^{k-1}_0| \overline{{\bf Y}}_{0})$ from (\ref{der_formula}), $\nabla_\lambda \overline p_{\psi^*\pi} (Y_k| \overline{{\bf Y}}_{0}^{k-1})=0$, and $\nabla_{\lambda_\mu^2} \log \overline p_{\psi^* 0 \alpha} (Y_k| \overline{{\bf Y}}_{0}^{k-1})=0$. A similar derivation to (\ref{d3lambda0}) gives 
\begin{align} 
& \nabla_{\lambda_\mu^4} \log \overline p_{\psi^* 0 \alpha} ({\bf Y}^k_0| \overline{{\bf Y}}_{0}) = \sum_{t=1}^k \mathbb{E}_{\psi^* 0 \alpha}\left[ \nabla_{\lambda_{\mu}^4}g^*_{t} / g^*_{t} \middle| \overline{\bf Y}_0^k\right]. \label{d4pp}
\end{align}
(\ref{d4p}) follows from (\ref{d4pp}) because (i) $\nabla_{\lambda_\sigma^2} \overline p_{\psi^* 0 \alpha} (Y_k| \overline{{\bf Y}}_{0}^{k-1})/\overline p_{\psi^* 0 \alpha} (Y_k| \overline{{\bf Y}}_{0}^{k-1}) = \mathbb{E}_{\vartheta^*}[ \nabla_{\lambda_\sigma^2} g_k^*|\overline{{\bf Y}}_{0}^k]$ from a similar argument to (\ref{d2p}) and (ii) $\mathbb{E}_{\psi^* 0 \alpha} [\nabla_{\lambda_\mu^4} g_t^* /g_t^* | \overline{{\bf Y}}_{0}^k ] = b(\alpha) \mathbb{E}_{\vartheta^*}[ \nabla_{\lambda_\sigma^2} g_k^*|\overline{{\bf Y}}_{0}^k]$ from (\ref{d3g2}). Therefore, part (b) is proven.
\end{proof}

\begin{lemma} \label{lemma_d34_homo}
Suppose the assumptions of Proposition \ref{P-quadratic-N1-homo} hold. Then, there exist $\bar \varrho_1, \bar \varrho_2 \in (0,\varrho)$ such that, for all $k \geq 1$,
\begin{align*}
(a) & \quad \frac{\nabla_{\lambda_\mu^3} \overline p_{\psi^*\pi} (Y_k| \overline{{\bf Y}}_{0}^{k-1})}{\overline p_{\psi^*\pi} (Y_k| \overline{{\bf Y}}_{0}^{k-1})} = \alpha(1-\alpha)(1-2\alpha)\frac{\nabla_{\mu^3}f_k^*}{f_k^*} + \varrho \frac{\nabla_\varrho \nabla_{\lambda_\mu^3} \overline p_{\psi^*\bar \varrho_1 \alpha} (Y_k| \overline{{\bf Y}}_{0}^{k-1})}{\overline p_{\psi^*\bar \varrho_1 \alpha} (Y_k| \overline{{\bf Y}}_{0}^{k-1})}, \\
(b) & \quad \frac{\nabla_{\lambda_\mu^4} \overline p_{\psi^*\pi} (Y_k| \overline{{\bf Y}}_{0}^{k-1})}{\overline p_{\psi^* \pi} (Y_k| \overline{{\bf Y}}_{0}^{k-1})} = \alpha(1-\alpha)(1-6\alpha+6\alpha^2)\frac{\nabla_{\mu^4}f_k^*}{f_k^*} + \varrho \frac{\nabla_\varrho \nabla_{\lambda_\mu^4} \overline p_{\psi^*\bar \varrho_2 \alpha} (Y_k| \overline{{\bf Y}}_{0}^{k-1})}{\overline p_{\psi^*\bar \varrho_2 \alpha} (Y_k| \overline{{\bf Y}}_{0}^{k-1})}.
\end{align*}
\end{lemma}

\begin{proof}[Proof of Lemma \ref{lemma_d34_homo}]
The proof is similar to the proof of Lemma \ref{lemma_d34}(a). From an argument similar to the proof of Lemma \ref{lemma_d34}, the stated results hold if 
\begin{align*} 
(A) &\quad \nabla_{\lambda_\mu^3} \overline p_{\psi^* 0 \alpha} (Y_k| \overline{{\bf Y}}_{0}^{k-1}) / \overline p_{\psi^* 0 \alpha} (Y_k| \overline{{\bf Y}}_{0}^{k-1}) = \alpha(1-\alpha)(1-2\alpha)\nabla_{\mu^3}f_k^* / f_k^*, \\
(B) &\quad \nabla_{\lambda_\mu^4} \overline p_{\psi^* 0 \alpha} (Y_k| \overline{{\bf Y}}_{0}^{k-1}) / \overline p_{\psi^* 0 \alpha} (Y_k| \overline{{\bf Y}}_{0}^{k-1}) = \alpha(1-\alpha)(1-6\alpha+6\alpha^2)\nabla_{\mu^4}f_k^* / f_k^*. 
\end{align*}
Observe that equalities (\ref{d3lambda0}) and (\ref{d4pp}) in the proof of Lemma \ref{lemma_d34} still hold under the assumptions of Proposition \ref{P-quadratic-N1-homo} if we use (\ref{d3g2-homo}) in place of (\ref{d3g2}). Consequently, (A) and (B) follow from (\ref{d3g-homo}), (\ref{d3g2-homo}), and the argument of the proof of Lemma \ref{lemma_d34}, and the stated result follows.
\end{proof}
 
\begin{lemma} \label{lemma_btsp} 
Suppose that the assumptions of Propositions \ref{P-LR} hold. Let $\bf{C}_\eta$ be a set of sequences $\{\eta_n\}$ satisfying $\sqrt{n}(\eta_n - \eta^*) \to h_\eta$ for some finite $h_\eta$. Let $\mathbb{P}^n_{\eta_n} := \prod_{k=1}^n f_k(\eta_n,0)$ denote the probability measure under $\eta_n$ with $\lambda_n=0$. Then, for every sequence $\{\eta_n\} \in \bf{C}_\eta$, the LRTS under $\{\mathbb{P}^n_{\eta_n}\}$ converges in distribution $\sup_{ \varrho \in\Theta_{\varrho}} \left(\tilde t_{\lambda \varrho}' \mathcal{I}_{\lambda.\eta \varrho} \tilde t_{\lambda \varrho} \right)$ given in Propositions \ref{P-LR}.
\end{lemma}
\begin{proof}[Proof of Lemma \ref{lemma_btsp}] 
Observe that $\vartheta_n:=(\pi_n,\eta_n,\lambda_n) = (\pi, \eta^*+h_\eta/\sqrt{n},0)$ satisfies the assumptions of Proposition \ref{P-LAN}. Therefore, Proposition \ref{P-LAN} holds under $\vartheta_n$ with $\nu_n(s_{\varrho n k}) \to_d N(\mathcal{I}_\varrho h, \mathcal{I}_\varrho)$ with $h=(h_\eta',0)'$ under $\mathbb{P}_{\vartheta_n}^n$. Furthermore, the log-likelihood function of the one-regime model admits a similar expansion, and $\log (d \mathbb{P}_{\eta_n}^n / d \mathbb{P}_{\eta^*}^n ) = h_\eta' \nu_n (s_{\eta k}) - (1/2)h_\eta' \mathcal{I}_\eta h_\eta + o_p(1)$ holds under $\mathbb{P}_{\eta_n}^n$. Therefore, the proof of Proposition \ref{P-LR} goes through by replacing $G_{\varrho n}$ with $G_{\varrho n}^h = \left[\begin{smallmatrix} G_{\eta n}^h \\ G_{\lambda \varrho n}^h \end{smallmatrix} \right] := G_{\varrho n} + \mathcal{I}_{\varrho}h$. In view of $G_{\eta n}^h = G_{\eta n} + \mathcal{I}_{\eta}h_\eta$ and $G_{\lambda \varrho n}^h = G_{\lambda \varrho n} + \mathcal{I}_{\lambda \eta \varrho} h_\eta$, we have $G_{\lambda. \eta \varrho n}^h := G_{\lambda \varrho n}^h - \mathcal{I}_{\lambda \eta \varrho}\mathcal{I}_{\eta}^{-1} G_{\eta n}^h = G_{\lambda \varrho n} - \mathcal{I}_{\lambda \eta \varrho}\mathcal{I}_{\eta}^{-1}G_{\eta n} = G_{\lambda \varrho n}$. Therefore, the asymptotic distribution of the LRTS under $\mathbb{P}_{\eta^n}^n$ is the same as that under $\mathbb{P}_{\eta^*}^n$, and the stated result follows.
\end{proof}

\subsubsection{Bounds on difference in state probabilities and conditional moments}
\begin{lemma} \label{max_bound}
Suppose $X_1,\ldots,X_n$ are random variables with $\max_{1\leq i \leq n}\mathbb{E}|X_i|^q < C$ for some $q>0$ and $C \in (0,\infty)$. Then, $\max_{1 \leq i \leq n} |X_i|= o_p(n^{1/q})$.
\end{lemma}
\begin{proof}[Proof of Lemma \ref{max_bound}]
For any $\varepsilon>0$, we have $\mathbb{P}(\max_{1\leq i\leq n}|X_i|>\varepsilon n^{1/q}) \leq \sum_{1\leq i\leq n} \mathbb{P}( |X_i|>\varepsilon n^{1/q}) \\ \leq \varepsilon^{-q} n^{-1} \sum_{1\leq i\leq n} \mathbb{E}(|X_i|^q\mathbb{I}{\{|X_i|>\varepsilon n^{1/q}\}})$ by a version of Markov inequality. As $n\rightarrow \infty$, the right-hand side tends to 0 by the dominated convergence theorem. 
\end{proof} 
 
The following lemma extends Corollary 1 and (39) of DMR and an equation on p. 2298 of DMR; DMR derive these results when $t_1=t_2$ and $t_3=t_4$ and ${\bf W}^n_{-m}$ is absent. For the two probability measures $\mu_1$ and $\mu_2$, the total variation distance between $\mu_1$ and $\mu_2$ is defined as $\|\mu_1-\mu_2\|_{TV}:=\sup_{A}|\mu_1(A)-\mu_2(A)|$. $\|\cdot\|_{TV}$ satisfies $\sup_{f(x):0 \leq f(x) \leq 1}|\int f(x) d\mu_1(x)- \int f(x) d\mu_2(x)| = \|\mu_1-\mu_2\|_{TV}$. In the following, we define $\overline{\bf{V}}_{-m}^n := (\overline{\bf{Y}}_{-m}^n,{\bf W}_{-m}^n)$, and we let $\overline{\bf{v}}_{-m}^n$ and $x_{-m}$ denote ``$\overline{\bf{V}}_{-m}^n = \overline{\bf{v}}_{-m}^n$'' and ``$X_{-m}=x_{-m}$.''
\begin{lemma} \label{x_diff}
Suppose Assumptions \ref{assn_a1}--\ref{assn_a2} hold and $\vartheta_x \in \Theta_x$. Then, we have, for all $\overline{\bf{v}}_{-m}^n$,\\
(a) For all $-m \leq t_1 \leq t_2$ with $-m<n$ and all probability measures $\mu_1$ and $\mu_2$ on $\mathcal{B}(\mathcal{X})$,
\begin{align*}
&\left\| \sum_{x_{-m}\in \mathcal{X}} \mathbb{P}_{\vartheta_x}({\bf X}_{t_1}^{t_2} \in \cdot | x_{-m}, \overline{\bf{v}}_{-m}^n)\mu_1(x_{-m})- \sum_{x_{-m}\in\mathcal{X}} \mathbb{P}_{\vartheta_x}({\bf X}_{t_1}^{t_2} \in \cdot | x_{-m}, \overline{\bf{v}}_{-m}^{n})\mu_2(x_{-m}) \right\|_{TV} \leq \rho^{t_1+m}. 
\end{align*}
(b) For all $-m \leq t_1 \leq t_2 \leq n-1$,
\begin{equation*}
\left\| \mathbb{P}_{\vartheta_x}({\bf X}_{t_1}^{t_2} \in \cdot|\overline{\bf{v}}_{-m}^n,x_{-m}) - \mathbb{P}_{\vartheta_x}({\bf X}_{t_1}^{t_2} \in \cdot|\overline{\bf{v}}_{-m}^{n-1},x_{-m}) \right\|_{TV} \leq \rho^{n-1-t_2}. 
\end{equation*}
The same bound holds when $x_{-m}$ is dropped from the conditioning variables.\\
(c) For all $-m \leq t_1 \leq t_2 < t_{3} \leq t_4$ with $-m<n$,
\begin{align*}
& \left\| \mathbb{P}_{\vartheta_x}({\bf X}_{t_1}^{t_2} \in \cdot, {\bf X}_{t_3}^{t_4} \in \cdot |\overline{\bf{v}}_{-m}^n,x_{-m}) - \mathbb{P}_{\vartheta_x}({\bf X}_{t_1}^{t_2} \in \cdot|\overline{\bf{v}}_{-m}^n,x_{-m})\mathbb{P}_{\vartheta_x}({\bf X}_{t_3}^{t_4} \in \cdot|\overline{\bf{v}}_{-m}^n,x_{-m}) \right\|_{TV} \leq \rho^{t_3-t_2}.
\end{align*}
The same bound holds when $x_{-m}$ is dropped from the conditioning variables.
\end{lemma}

\begin{proof}[Proof of Lemma \ref{x_diff}]
We prove part (a) first. We assume $t_1>-m$ because the stated result holds trivially when $t_1=-m$. Observe that Lemma 1 of DMR still holds when ${\bf W}_{-m}^n$ is added to the conditioning variable because Assumption \ref{assn_a1} implies that $\{(X_k,\overline{\bf Y}_{k})\}_{k=0}^{\infty}$ is a Markov chain given $\{W_k\}_{k=0}^{\infty}$. Therefore, $\{X_{t}\}_{t \geq - m}$ is a Markov chain when conditioned on $\{ \overline{\bf{Y}}_{-m}^{n},{\bf W}_{-m}^n \}$, and hence $\mathbb{P}_{\vartheta_x}({\bf X}_{t_1}^{t_2} \in A |\overline{\bf{v}}_{-m}^n,x_{-m}) = \sum_{x_{t_1} \in \mathcal{X}}\mathbb{P}_{\vartheta_x}({\bf X}_{t_1}^{t_2} \in A |X_{t_1}=x_{t_1},\overline{\bf{v}}_{-m}^n) p_{\vartheta_x}(x_{t_1}|\overline{\bf{v}}_{-m}^n,x_{-m})$ holds. From applying this result and the property of the total variation distance, we can bound the left-hand side of the lemma by $\| \sum_{x_{-m}\in\mathcal{X}} p_{\vartheta_x}(X_{t_1} \in \cdot |\overline{\bf{v}}_{-m}^n,x_{-m}) \mu_1(x_{-m}) - \sum_{x_{-m}\in\mathcal{X}} p_{\vartheta_x}(X_{t_1} \in \cdot |\overline{\bf{v}}_{-m}^n,x_{-m})\mu_2(x_{-m}) \|_{TV}$. This is bounded by $\rho^{t_1+m}$ from Corollary 1 of DMR, which holds when ${\bf W}_{-m}^n$ is added to the conditioning variable. Therefore, part (a) is proven.

We proceed to prove part (b). Observe that the time-reversed process $\{Z_{n-k}\}_{0 \leq k\leq n+m}$ is Markov when conditioned on ${\bf W}_{-m}^n$ and that $W_k$ is independent of $({\bf X}_{0}^{k-1},\overline{\bf Y}_{0}^{k-1})$ given ${\bf W}_{0}^{k-1}$. Consequently, for $k=n,n-1$, we have $\mathbb{P}_{\vartheta_x}({\bf X}_{t_1}^{t_2} \in A|\overline{\bf{v}}_{-m}^k,x_{-m}) = \sum_{x_{t_2} \in \mathcal{X}} \mathbb{P}_{\vartheta_x}({\bf X}_{t_1}^{t_2} \in A|X_{t_2}=x_{t_2},\overline{\bf{v}}_{-m}^{t_2},x_{-m}) p_{\vartheta_x}(x_{t_2}|\overline{\bf{v}}_{-m}^{k},x_{-m})$. Therefore, from the property of the total variation distance, the left-hand side of the lemma is bounded by $\|\mathbb{P}_{\vartheta_x}(X_{t_2} \in \cdot |\overline{\bf{v}}_{-m}^{n},x_{-m}) - \mathbb{P}_{\vartheta_x}(X_{t_2} \in \cdot |\overline{\bf{v}}_{-m}^{n-1},x_{-m})\|_{TV}$. This is bounded by $\rho^{n-1-t_2}$ because equation (39) of DMR p.\ 2294 holds when ${\bf W}_{-m}^n$ is added to the conditioning variables, and the stated result follows. When $x_{-m}$ is dropped from the conditioning variables, part (b) follows from a similar argument with using Lemma 9 and an analogue of Corollary 1 of DMR in place of equation (39) of DMR.

Part (c) follows immediately from writing the left-hand side of lemma as $\sup_{A,B}| \mathbb{P}_{\vartheta_x}({\bf X}_{t_1}^{t_2} \in A|\overline{\bf{v}}_{-m}^n,x_{-m}) [ \mathbb{P}_{\vartheta_x}({\bf X}_{t_3}^{t_4} \in B |\overline{\bf{v}}_{-m}^n, {\bf X}_{t_1}^{t_2} \in A) - \mathbb{P}_{\vartheta_x}({\bf X}_{t_3}^{t_4} \in B|\overline{\bf{v}}_{-m}^n,x_{-m})]|$ and applying part (a).
\end{proof}

\subsubsection{The sums of powers of $\rho$} 
\begin{lemma}\label{two_rho_sum}
For all $\rho \in (0,1)$, $c \geq 1$, $q \geq 1$, and $b>a$,
\begin{align*}
& \sum_{t=-\infty}^{\infty} \left(\rho^{\lfloor (t-a)/cq \rfloor}\wedge \rho^{\lfloor(b-t)/q\rfloor}\right) \leq \frac{q(c+1)\rho^{\lfloor (b-a)/(c+1)q \rfloor}}{ 1-\rho}, \\
& \sum_{t=-\infty}^{\infty} \left(\rho^{\lfloor(t-a)/q \rfloor}\wedge \rho^{\lfloor (b-t)/cq \rfloor}\right) \leq \frac{q(c+1)\rho^{\lfloor (b-a)/(c+1)q \rfloor}}{ 1-\rho}.
\end{align*}
\end{lemma}

\begin{proof}[Proof of Lemma \ref{two_rho_sum}]
The first result holds because the left-hand side is bounded by \\$\sum_{t=-\infty}^{\lfloor (a+bc)/(c+1) \rfloor }\rho^{\lfloor(b-t)/q\rfloor} + \sum_{t=\lfloor (a+bc)/(c+1) \rfloor +1}^{\infty}\rho^{\lfloor (t-a)/cq \rfloor} \leq q \rho^{\lfloor \{ b- \lfloor (a+bc)/(c+1) \rfloor \}/q \rfloor }/(1-\rho) + cq\rho^{\lfloor \{ \lfloor (a+bc)/(c+1) \rfloor +1 -a \} / cq \rfloor}/(1-\rho) \leq q (1+c) \rho^{\lfloor (b-a)/(c+1)q \rfloor}/(1-\rho)$. The second result is proven by bounding the left-hand side by $\sum_{t=-\infty}^{\lfloor (ac+b)/(c+1) \rfloor }\rho^{\lfloor(b-t)/cq\rfloor} + \sum_{t=\lfloor (ac+b)/(c+1) \rfloor +1}^{\infty}\rho^{\lfloor (t-a)/q \rfloor}$ and proceeding similarly.
\end{proof}

The following lemma generalizes the result in the last inequality on p. 2299 of DMR. 
\begin{lemma}\label{rho_sum}
For all $\rho \in (0,1)$, $k \geq 1 $, $q \geq 1$, and $n \geq 0$,
\[
\sum_{0\leq t_1\leq t_2 \leq \cdots \leq t_k \leq n} \left( \rho^{\lfloor t_1 /q \rfloor}\wedge \rho^{\lfloor (t_2-t_1)/q \rfloor} \wedge \cdots \wedge \rho^{\lfloor (t_k-t_{k-1})/q \rfloor} \wedge \rho^{\lfloor (n-t_k)/q \rfloor} \right) \leq C_{kq}(\rho) \rho^{\lfloor n/2kq \rfloor},
\]
where $C_{kq}(\rho):= q^{k} k(k+1)! (1-\rho)^{-k}$.
\end{lemma}
\begin{proof}[Proof of Lemma \ref{rho_sum}]
When $k=1$, the stated result follows from Lemma \ref{two_rho_sum} with $c=1$. We first show that the following holds for $k\geq 2$:
\begin{equation} \label{rho_k_bound}
\sum_{t_1\leq t_2 \leq \cdots \leq t_k \leq n} \left(\rho^{\lfloor (t_2-t_1)/q \rfloor} \wedge \cdots \wedge \rho^{\lfloor (t_k-t_{k-1})/q \rfloor} \wedge \rho^{\lfloor (n-t_k)/q \rfloor}\right) \leq \frac{q^{k-1} (k+1)! \rho^{\lfloor (n-t_1)/kq \rfloor }}{(1-\rho)^{k-1}}.
\end{equation}
We prove (\ref{rho_k_bound}) by induction. When $k=2$, it follows from Lemma \ref{two_rho_sum} with $c=1$ that \\$\sum_{t_2=t_1}^{n}(\rho^{\lfloor (t_2-t_1)/q \rfloor} \wedge \rho^{\lfloor (n-t_2)/q \rfloor}) \leq 2q \rho^{\lfloor (n-t_1)/2q \rfloor} /(1-\rho)$, giving (\ref{rho_k_bound}). Suppose (\ref{rho_k_bound}) holds when $k=\ell$. Then (\ref{rho_k_bound}) holds when $k=\ell+1$ because, from Lemma \ref{two_rho_sum},
\begin{align*}
& \sum_{t_1\leq t_2 \leq \cdots \leq t_{\ell} \leq t_{\ell+1} \leq n}\left(\rho^{\lfloor (t_2-t_1)/q \rfloor} \wedge \rho^{\lfloor (t_3-t_2)/q \rfloor} \wedge \cdots \wedge \rho^{\lfloor (t_{\ell+1}-t_{\ell})/q \rfloor} \wedge \rho^{\lfloor (n-t_{\ell+1})/q \rfloor}\right) \\
&\leq \sum_{t_2=t_1}^{n} \left( \rho^{\lfloor (t_2-t_1)/q \rfloor} \wedge \sum_{t_2 \leq \cdots \leq t_{\ell+1} \leq n}\left( \rho^{\lfloor (t_3-t_2)/q \rfloor} \wedge \cdots \wedge \rho^{\lfloor (t_{\ell+1}-t_{\ell})/q \rfloor} \wedge \rho^{\lfloor (n-t_{\ell+1})/q \rfloor} \right) \right)\\
& \leq \frac{q^{\ell-1} \ell!}{ (1-\rho)^{\ell-1}} \sum_{t_2=t_1}^{n} \left(\rho^{\lfloor (t_2-t_1)/q \rfloor}\wedge \rho^{\lfloor (n-t_2)/\ell q \rfloor}\right) \\
& \leq \frac{q^{\ell}(\ell+1)!}{(1-\rho)^{\ell}} \rho^{\lfloor (n-t_1)/(\ell+1)q \rfloor},
\end{align*}
and hence (\ref{rho_k_bound}) holds for all $k \geq 2$. We proceed to show the stated result. Observe that
\begin{align*}
& \sum_{0\leq t_1\leq t_2 \leq \cdots \leq t_k \leq n} \left( \rho^{\lfloor t_1 /q \rfloor}\wedge \rho^{\lfloor (t_2-t_1)/q \rfloor} \wedge \cdots \wedge \rho^{\lfloor (t_k-t_{k-1})/q \rfloor} \wedge \rho^{\lfloor (n-t_k)/q \rfloor} \right) \\
& \leq 2 \sum_{t_1=0}^{n/2}\sum_{t_1\leq t_2 \leq \cdots \leq t_{k-1}\leq t_k} \sum_{t_k = t_1}^{n-t_1} \left( \rho^{\lfloor t_1 /q \rfloor}\wedge \rho^{\lfloor (t_2-t_1)/q \rfloor} \wedge \cdots \wedge \rho^{\lfloor (t_k-t_{k-1})/q \rfloor} \wedge \rho^{\lfloor (n-t_k)/q \rfloor} \right) \\
&= 2 \sum_{t_1=0}^{n/2}\sum_{t_1\leq t_2 \leq \cdots \leq t_{k-1}\leq t_k} \sum_{t_k = t_1}^{n-t_1} \left( \rho^{\lfloor (t_2-t_1)/q \rfloor} \wedge \cdots \wedge \rho^{\lfloor (t_k-t_{k-1})/q \rfloor} \wedge \rho^{\lfloor (n-t_k)/q \rfloor} \right) \\
& \leq 2 \sum_{t_1=0}^{n/2}\sum_{t_1\leq t_2 \leq \cdots \leq t_{k-1}\leq t_k \leq n} \left( \rho^{\lfloor (t_2-t_1)/q \rfloor} \wedge \cdots \wedge \rho^{\lfloor (t_k-t_{k-1})/q \rfloor} \wedge \rho^{\lfloor (n-t_k)/q \rfloor} \right), 
\end{align*}
where the first inequality holds by symmetry, and the subsequent equality follows from $n-t_k \geq t_1$. From (\ref{rho_k_bound}), the right-hand side is no larger than $q^{k-1} (k+1)! (1-\rho)^{(1-k)} \sum_{t_1=0}^{n/2}\rho^{\lfloor (n-t_1)/kq \rfloor } \leq q^{k} k(k+1)! (1-\rho)^{-k}\rho^{\lfloor n/2kq \rfloor}$, giving the stated result.
\end{proof}

The next lemma generalizes equation (46) and p. 2294 of DMR, who derive a similar bound when $\ell=1,2$.
\begin{lemma} \label{hoelder}
Let $a_j>0$ for all $j$. For all positive integer $\ell \geq 1$ and all $k\geq 1$ and $m\geq 0$, we have $\max_{-m+1\leq t_1, \ldots, t_\ell\leq k} a_{t_1} \cdots a_{t_\ell} \leq (k+m)^{\ell+1} A_\ell$, where $A_\ell := \sum_{t=-\infty}^{\infty} (|t|\vee 1)^{-2} a_t^\ell$.
\end{lemma}
\begin{proof}[Proof of Lemma \ref{hoelder}]
When $\ell=1$, the stated result follows from $\max_{-m+1\leq t\leq k} a_{t} \leq \sum_{t=-m+1}^k a_t = \sum_{t=-m+1}^k (|t|\vee 1)^2(|t|\vee 1)^{-2} a_t \leq (k+m)^2\sum_{t=-\infty}^\infty(|t|\vee 1)^{-2} a_t $. When $\ell \geq 2$, from H\"older's inequality, we have $\max_{-m+1\leq t_1\leq \ldots \leq t_\ell\leq k} a_{t_1} a_{t_2} \cdots a_{t_\ell} \leq (\sum_{t=-m+1}^k a_t)^\ell = [ \sum_{t=-m+1}^k (|t|\vee 1)^{2/\ell} (|t|\vee 1)^{-2/\ell}a_t ] ^\ell \leq [ \sum_{t=-m+1}^k (|t|\vee 1)^{2/(\ell-1)} ]^{(\ell-1)} \sum_{t=-m+1}^k (|t|\vee 1)^{-2} a_t^\ell \leq [(k+m)^{1+2/(\ell-1)}]^{\ell-1} A_\ell = (k+m)^{\ell+1} A_\ell$.
\end{proof}

The following lemma generalizes the bound derived on p. 2301 of DMR.
\begin{lemma} \label{rho_m}
For $\alpha>0$, $q >0$, and $c_{jt}\geq0$, define $c_{jq}^\infty(\rho^{\alpha}) := \sum_{t=-\infty}^{\infty} \rho^{\lfloor \alpha |t|/q \rfloor} c_{jt}$. For all $\rho \in (0,1)$, $k\geq 1$, and $0\leq m\leq m'$,
\begin{equation} \label{rho_6}
\begin{aligned}
&\sum_{t_1=-m'+1}^{-m}\sum_{t_1 \leq t_2 \leq t_3 \leq t_4\leq t_5 \leq t_6 \leq k } \left(\rho^{\lfloor (k-1-t_6)/q \rfloor} \wedge \rho^{\lfloor (t_6-t_5)/q \rfloor} \wedge \rho^{\lfloor (t_5-t_4)/q \rfloor} \wedge\rho^{\lfloor (t_4-t_3)/q \rfloor} \wedge \right. \\
& \left. \qquad \rho^{\lfloor (t_3 -t_2)/q \rfloor} \wedge \rho^{\lfloor (t_2-t_1)/q \rfloor}\right)\prod_{j=1}^6 c_{jt_j} \leq \rho^{\lfloor (k-1+m)/2qa_7\rfloor} c_{1q}^{\infty}\left(\rho^{1/2a_7}\right) \prod_{j=2}^6 c_{jq}^{\infty}\left(\rho^{1/4a_j}\right),
\end{aligned}
\end{equation}
where $(a_j,b_j)$ are defined recursively with $(a_2,b_2)=(1,1)$ and, for $j \geq 3$,
$a_{j+1} = 4a_j(a_j+b_j)/(2a_j-1)$ and $b_{j+1} = a_j(4b_j-1)/(2a_j-1)$.
$a_{j}$ and $b_{j}$ satisfy $a_{j}, b_{j} \geq 3/2$ for all $j$. Direct calculations using Matlab produce $a_7 \doteq 334.5406$.
\end{lemma}

\begin{proof}[Proof of Lemma \ref{rho_m}]
First, observe that the following result holds for $a,b>1/4$, $t_1\leq 0$, and $t_j, t_{j+1} \geq t_1$:
\begin{equation} \label{t_ab}
\begin{aligned}
(a) & \text{ if } t_j \leq \frac{at_{j+1}+t_1}{a+b}, \quad \text {then } \frac{|t_j|}{4a} \leq \frac{a(4a+1)t_{j+1}+(2a-1)t_1}{4a(a+b)} - t_j, \\
(b) & \text{ if } t_j \geq \frac{at_{j+1}+t_1}{a+b}, \quad \text {then } \frac{|t_j|}{4a} \leq \frac{b}{a}t_j - \frac{a(4b-1)t_{j+1}+(2a+4b+1)t_1}{4a(a+b)}.
\end{aligned}
\end{equation}
(a) holds because (i) when $t_j \leq 0$, we have $t_j \leq (at_{j+1}+t_1)/(a+b) \Rightarrow (4a-1)t_j/4a \leq [a(4a-1)t_{j+1}+(4a-1)t_1]/4a(a+b) \Rightarrow -t_j/4a \leq [a(4a-1)t_{j+1}+(4a-1)t_1]/4a(a+b)-t_j$ and $a(4a-1)t_{j+1}+(4a-1)t_1 \leq a(4a-1)t_{j+1}+(4a-1)t_1+2a(t_{j+1}-t_1) = a(4a+1)t_{j+1}+(2a-1)t_1$; (ii) when $t_j \geq 0$, we have $t_j \leq (at_{j+1}+t_1)/(a+b) \Rightarrow (4a+1)t_j/4a \leq [a(4a+1)t_{j+1}+(4a+1)t_1]/4a(a+b) \Rightarrow t_j/4a \leq [a(4a+1)t_{j+1}+(4a+1)t_1]/4a(a+b)- t_j$ and $(4a+1)t_1 \leq (2a-1)t_1$.

(b) holds because (i) when $t_j \leq 0$, we have $t_j \geq (at_{j+1}+t_1)/(a+b) \Rightarrow (4b+1)t_j/4a \geq [a(4b+1)t_{j+1}+(4b+1)t_1]/4a(a+b) \Rightarrow -t_j/4a \leq bt_j/a - [a(4b+1)t_{j+1}+(4b+1)t_1]/4a(a+b)$ and $a(4b+1)t_{j+1}+(4b+1)t_1 \geq a(4b+1)t_{j+1}+(4b+1)t_1 -2a(t_{j+1}-t_1) = a(4b-1)t_{j+1}+(2a+4b+1)t_1$; (ii) when $t_j \geq 0$, we have $t_j \geq (at_{j+1}+t_1)/(a+b) \Rightarrow (4b-1)t_j/4a \geq [a(4b-1)t_{j+1}+(4b-1)t_1]/4a(a+b) \Rightarrow t_j/4a \leq bt_j/a - [a(4b-1)t_{j+1}+(4b-1)t_1]/4a(a+b)$ and $a(4b-1)t_{j+1}+(4b-1)t_1 \geq a(4b-1)t_{j+1}+(2a+4b+1)t_1$.

We proceed to derive the stated bound. It follows from (a) and (b) and $\lfloor x + y \rfloor \geq \lfloor x \rfloor +\lfloor y \rfloor$ that, with $\overline t_j = (a_jt_{j+1}+t_1)/(a_j+b_j)$,
\begin{align}
& \sum_{t_j=-m'+1}^{k}\left( \rho^{\lfloor (t_{j+1} -t_j)/q \rfloor} \wedge \rho^{\lfloor (b_j t_j-t_1)/a_jq \rfloor}\right) c_{jt_j} \nonumber \\
& \leq \rho^{\lfloor \frac{a_j(4b_j-1)t_{j+1} -(2a_j-1)t_1}{4a_j(a_j+b_j)q} \rfloor} \left( \sum_{t_j\leq\overline t_j} \rho^{\lfloor\frac{a_j(4a_j+1)t_{j+1}+(2a_j-1)t_1}{4a_j(a_j+b_j)q} - \frac{t_j}{q} \rfloor} + \sum_{t_j\geq\overline t_j} \rho^{\lfloor \frac{b_j}{a_jq}t_j - \frac{a_j(4b_j-1)t_{j+1}+(2a_j+4b_j+1)t_1}{4a_j(a_j+b_j)q} \rfloor}\right) c_{jt_j}\nonumber \\
& \leq \rho^{\lfloor \frac{a_j(4b_j-1)t_{j+1} -(2a_j-1)t_1}{4a_j(a_j+b_j)q}\rfloor} c_{jq}^{\infty}\left(\rho^{1/4a_j}\right) \nonumber \\
& = \rho^{\lfloor\frac{b_{j+1}t_{j+1} -t_1}{a_{j+1}}\rfloor} c_{jq}^{\infty} \left(\rho^{1/4a_j}\right). \label{rho_j_sum}
\end{align}
Observe that $a_{j+1} \geq 2a_j \geq 2$ and $b_{j+1} \geq 2b_j - (1/2) \geq 3/2$ for all $j \geq 2$. Therefore, we can apply (\ref{t_ab}) and (\ref{rho_j_sum}) to the left-hand side of (\ref{rho_6}) sequentially for $j=2,3,\ldots,6$. Consequently, the left-hand side of (\ref{rho_6}) is no larger than
\[
\sum_{t_1=-m'+1}^{-m}\rho^{\lfloor \frac{b_7(k-1) -t_1}{a_7q} \rfloor} c_{1t_1} \prod_{j=2}^6 c_{jq}^{\infty}\left(\rho^{1/4a_j}\right).
\]
Observe that $|t_1|\leq k-1-2t_1-m$ because $t_1 \leq -m \Rightarrow -t_1 \leq -2t_1-m \leq k-1-2t_1-m$. From $b_7(k-1) \geq k-1 $ and $|t_1|\leq k-1-2t_1-m$, the sum is bounded by
\[
\sum_{t_1=-m'+1}^{-m}\rho^{\lfloor \frac{k-1 -t_1}{a_7 q} \rfloor} c_{1t_1} = \rho^{\lfloor\frac{k-1+m}{2a_7q}\rfloor} \sum_{t_1=-m'+1}^{-m}\rho^{\lfloor\frac{k-1-2t_1-m}{2a_7q}\rfloor} c_{1t_1} \leq \rho^{\lfloor\frac{k-1+m}{2a_7q}\rfloor} c_{1q}^{\infty}\left(\rho^{1/2a_7}\right),
\]
and the stated result follows.
\end{proof}

\subsubsection{Derivation of $\vartheta_{M_0+1,x} = (\vartheta_{xm}',\pi_{xm}')'$ and $\pi_{xm}=(\varrho_m,\alpha_m,\phi_m')'$}\label{subsec:p_m_repara}
Define $\overline J_{m0} := \{1,\ldots,M_0\}\setminus J_m$, and let $p_j$ and $p^*_j$ denote $\mathbb{P}_{\vartheta_{M_0+1}}(X_k=j)$ and $\mathbb{P}_{\vartheta_{M_0}^*}(X_k=j)$, respectively.

We parameterize the transition probability of $X_k$ in terms of its stationary distribution and the first to the $(m-1)$-th rows and the $(m+1)$-th to the $(M_0+1)$-th rows of its transition matrix.\footnote{Suppose a Markov process has a transition probability $P$ and stationary distribution $\pi$ whose elements are strictly positive. If $\pi$ and all the rows of $P$ except for one are identified, then the remaining row of $P$ is identified from the relation $\pi P = \pi$.}
For $i \in \overline J_m$, we reparameterize $(p_{im},p_{i,m+1})$ to $p_{iJ} = p_{im}+p_{i,m+1}=\mathbb{P}_{\vartheta_{M_0+1}}(X_k \in J_m|X_{k-1}=i)$ and $p_{im|iJ} = p_{im}/(p_{im}+p_{i,m+1})$. Furthermore, we reparameterize $(p_m,p_{m+1})$ in the stationary distribution to $p_J = p_m + p_{m+1}=\mathbb{P}_{\vartheta_{M_0+1}}(X_k \in J_m)$ and $p_{m|J}=p_m/(p_m + p_{m+1})=\mathbb{P}_{\vartheta_{M_0+1}}(X_k=m|X_k \in J_m)$. Therefore, with $\land$ and $\lor$ denoting ``and'' and ``or,'' the transition probability of $X_k$ is summarized by 
$\vartheta_{M_0+1,x} :=( \{p_{iJ},p_{im|iJ} \}_{i \in \overline J_m}, \{p_{ij}\}_{i \in \overline J_m \land j \in \overline J_{m0}}, \{p_{m+1,j}\}_{j=1}^{M_0}, \{p_{j}\}_{j \in \overline J_{m0}}, p_J, p_{m|J})$.

Split $\vartheta_{M_0+1,x}$ as $\vartheta_{M_0+1,x}=(\vartheta_{xm}',\pi_{xm}')'$, where $\vartheta_{xm} := ( \{p_{ij}\}_{i \in \overline J_m \land j \in \overline J_{m0}}, \{p_{iJ} \}_{i \in \overline J_m}, \{p_{j}\}_{j \in \overline J_{m0}}, p_J)$ and $\pi_{xm} :=( \{p_{im|iJ} \}_{i \in \overline J_m}, \{p_{m+1,j}\}_{j=1}^{M_0}, p_{m|J})$.
When the $m$-th and $(m+1)$-th regimes are combined into one regime, the transition probability of $X_k$ equals the transition probability of $X_k$ under $\vartheta_{M_0,x}^*$ if and only if $\vartheta_{xm} = \vartheta_{xm}^*:= \{p_{ij} = p_{ij}^{*}\ \text{for}\ i \in \overline J_m \land(1 \leq j \leq m-1); p_{ij} = p_{i,j-1}^{*}\ \text{for}\ i \in \overline J_m \land (m+2 \leq j \leq M_0); p_{iJ} = p_{im}^{*}\ \text{for } i\in \bar J_m;p_j = p_{j}^* \ \text{for}\ 1 \leq j \leq m-1; p_j = p_{j-1}^* \ \text{for}\ m+2 \leq j \leq M_0; p_J = p_m^* \}$. $\pi_{xm}$ is the part of $\vartheta_{M_0+1,x}$ that is not identified under $H_{0m}$.

We proceed to derive the reparameterization of some elements of $\pi_{xm}$ in terms of $(\alpha_m,\varrho_m)$. First, map $p_{m+1,m}$ and $p_{m+1,m+1}$ to $p_{m+1,J}:=p_{m+1,m}+p_{m+1,m+1}=\mathbb{P}_{\vartheta_{M_0+1}}(X_k \in J|X_{k-1}=m+1)$ and $p_{m+1,m|J}:=p_{m+1,m}/p_{m+1,J}=\mathbb{P}_{\vartheta_{M_0+1}}(X_k=m|X_k \in J,X_{k-1}=m+1)$. Let $P_J$ and $\pi_J$ denote the transition matrix and stationary distribution of $X_k$ restricted to lie in $J_m$. The second row of $P_J$ is given by $(p_{m+1,m|J},1-p_{m+1,m|J})$, and $\pi_J$ is given by $(p_{m|J},1-p_{m|J})$. From the relation $\pi_J= \pi_J P_J$, we can obtain the first row of $P_J$ as a function of $p_{m+1,m|J}$ and $p_{m|J}$. Finally, the elements of $P_J$ are mapped to $(\varrho_m,\alpha_m)$ as in Section \ref{sec: testing-1}.

\clearpage

\bibliography{markov}

\clearpage

{

\begin{table} 
	\centering
	\caption{Rejection frequencies (\%) under the null hypothesis at the nominal 10\%, 5\%, and 1\% levels}
		\begin{tabular}{l l |c c c| c c c} \hline \hline 
	 \multicolumn{8}{c}{$H_0: M=1$} \\ \hline
		 &    & \multicolumn{3}{c|}{Model 1}& \multicolumn{3}{c}{Model 2} \\ 
		 & Test    & 10\% & 5\% & 1\% & 10\% & 5\% & 1\% \\ \hline
		$n=200$& \text{LRT}   & 10.43 & 4.63 & 1.13 & 10.17 & 5.27 & 1.00 \\
		& \text{supTS}    	 & 9.87 & 5.10 & 0.93 & 9.63 &4.67 & 0.90 \\
		& \text{QLRT}     & 10.03 & 4.97 & 1.03 & --- & --- & --- \\ \hline
		$n=500$& \text{LRT}   & 8.80 & 4.03 & 0.67& 9.13 & 4.30 & 1.23 \\ 
		& \text{supTS}    & 9.50 & 4.57 & 0.60 & 9.23 & 5.07 & 0.90 \\
		& \text{QLRT}     & 9.07 & 4.43 & 0.80 & --- & --- & --- \\ \hline \hline
	 \multicolumn{8}{c}{$H_0: M=2$} \\ \hline 
		 \multicolumn{2}{c|}{LRT}    & \multicolumn{3}{c|}{Model 1}& \multicolumn{3}{c}{Model 2} \\ 
		& $(p_{11},p_{22})$  & 10\% & 5\% & 1\% & 10\% & 5\% & 1\% \\ \hline
		$n=200$& $(0.5,0.5)$ & 12.06 & 7.16 & 1.70 & 10.57 & 4.87 & 0.80 \\ 
			   & $(0.7,0.7)$ & 11.97 & 6.07 & 1.70 &  9.63 & 4.53 & 1.07 \\ \hline
		$n=500$& $(0.5,0.5)$ & 9.77 & 4.43 & 0.70 &  8.37 & 3.90 & 0.73  \\ 
		       & $(0.7,0.7)$ & 8.40 & 4.20 & 0.70 &  9.63 & 4.80 & 0.63\\ \hline \hline
	\end{tabular}\\ 
	\begin{flushleft}
Notes: We use 199 bootstrap samples and 3000 replications. For testing $H_0: M=2$ using Models 1 and 2, we generate the data under $(\beta,\mu_1,\mu_2,\sigma) = (0.5,-1,1,1)$ and $(\beta,\mu_1,\mu_2,\sigma_1,\sigma_2) = (0.5,-1,1,0.9,1.2)$, respectively. 
\end{flushleft}
		\label{table:bootstrap-size}
\end{table} 

\begin{table}[]
	\centering 
	\caption{Rejection frequencies (\%) for testing $H_0: M=1$ under the alternative hypothesis}
	\begin{tabular}{l|l|ccc | ccc}
		\hline \hline
		 & & \multicolumn{3}{c|}{ Model 1 }& \multicolumn{3}{c}{ Model 2} \\
		$(p_{11},p_{22})$  & Test  & $\mu_1=0.20$ & $\mu_1=0.6$ & $\mu_1=1.0$& $\mu_1=0.20$ & $\mu_1=0.6$ & $\mu_1=1.0$ \\\hline 
		$(0.25,0.25)$& \text{LRT}     & 4.87 & 46.90 & 99.63 & 16.40 & 78.00 & 99.97 \\ 
		& \text{supTS}     & 6.23 & 56.43 & 95.90& 16.37 & 70.97 & 95.37 \\ 
		& \text{QLRT}     & 5.10 & 8.00 & 55.27& --- & --- & --- \\\hline
		$(0.50,0.50)$& \text{LRT}  & 3.80 & 7.03 & 67.87 & 13.70 & 43.77 & 92.77 \\ 
		& \text{supTS}     & 4.07 & 4.40 & 4.60 & 14.70 & 35.77 & 35.30 \\ 
		& \text{QLRT}     & 4.90 & 9.40 & 82.50 & --- & --- & --- \\\hline
		$(0.70,0.70)$& \text{LRT}     & 4.10 & 10.23 & 91.07& 14.63 & 51.37 & 98.17 \\ 
		& \text{supTS}     & 4.57 & 7.40 & 26.37 & 14.90 & 36.20 & 43.43 \\ 
		& \text{QLRT}     & 5.13 & 8.53 & 58.73 & --- & --- & --- \\\hline
		$(0.90,0.90)$& \text{LRT}     & 5.33 & 46.87 & 99.97 & 23.27 & 79.87 & 100.00 \\ 
		& \text{supTS}     & 6.77 & 13.90 & 4.40 & 19.10 & 41.17 & 35.30 \\ 
		& \text{QLRT}     & 4.83 & 5.63 & 5.97 & --- & --- & --- \\\hline \hline 
	\end{tabular}\\ 
	\begin{flushleft} 
Notes: Nominal level of 5\% and $n=500$. We use 199 bootstrap samples and 3000 replications. We set $\mu_2=-\mu_1$ for both models, $(\beta,\sigma)=(0.5,1.0)$ for Model 1, and $(\beta,\sigma_1,\sigma_2) = (0.5,1.1,0.9)$ for Model 2. 
\end{flushleft} 
		\label{table:bootstrap-power-model1}
\end{table}

\begin{table}[]
	\centering 
	\caption{Rejection frequencies (\%) for testing $H_0: M=2$ under the alternative hypothesis}
	\begin{tabular}{l|c|c | c|c }
		\hline \hline
		 &  \multicolumn{2}{c|}{ Model 1 }& \multicolumn{2}{c}{ Model 2} \\ \cline{2-5}
		 & $(\mu_1,\mu_2,\mu_3)$& $(\mu_1,\mu_2,\mu_3)$& $(\mu_1,\mu_2,\mu_3)$& $(\mu_1,\mu_2,\mu_3)$\\
$(p_{11},p_{22},p_{33})$ & $=(1,0 ,-1)$& $=(2,0 ,-2)$& $=(1,0 ,-1)$& $=(2,0 ,-2)$ \\\hline 
$(0.5,0.5,0.5)$		     & 5.23 & 30.80 & 10.33 & 60.03 \\ \hline 
$(0.7,0.7,0.7)$          & 8.47 & 94.03 & 23.10 & 99.33 \\ \hline \hline 
	\end{tabular}\\ 
	\begin{flushleft} 
Notes: Nominal level of 5\% and $n=500$. We use 199 bootstrap samples and 3000 replications. 
We set $(\beta,\sigma)=(0.5,1.0)$ for Model 1 and $(\beta,\sigma_1,\sigma_2,\sigma_3) = (0.5,0.6, 0.9, 1.2)$ for Model 2. For both Models 1 and 2, we set $p_{ij} = (1-p_{ii})/2$ for $j\neq i$, so that, for example, $(p_{12},p_{13})=(0.15,0.15)$ when $p_{11}=0.7$. \end{flushleft} 
		\label{table:bootstrap-power-model2}
\end{table}
     
\clearpage

\begin{table} 
\centering
	\caption{Parameter estimates: U.S. GDP per capita growth, 1960Q1--2014Q4 }
\begin{tabular}[t]{c|cc|cc|cc}   \hline\hline
 \multicolumn{7}{c }{Panel A: Model 1 with common variance}\\ \hline
 & \multicolumn{2}{c| }{$M=2$} & \multicolumn{2}{c| }{$M=3$} & \multicolumn{2}{c}{$M=4$} \\ 
 & coeff. & s.e.& coeff. & s.e.& coeff. & s.e. \\ \hline
$\mu_1$&-0.634&0.200&-0.823&0.151&-2.348&0.649\\
$\mu_2$&0.951&0.176&0.692&0.172&-0.330&0.179\\
$\mu_3$& -- & -- &2.023&0.236&0.532&0.161\\
$\mu_4$& -- & -- & -- & -- &2.025&0.184\\
$\sigma$&0.913&0.053&0.752&0.052&0.832&0.040\\
$\beta$&0.787&0.041&0.773&0.046&0.639&0.053\\ \hline \hline
\multicolumn{7}{c }{Panel B: Model 2 with switching variance}\\ \hline
 & \multicolumn{2}{c| }{$M=2$} & \multicolumn{2}{c| }{$M=3$} & \multicolumn{2}{c}{$M=4$} \\ 
        & coeff. & s.e.& coeff. & s.e.& coeff. & s.e. \\ \hline
$\mu_1   $&0.370&0.123&-0.643&0.308&-0.698&0.359\\
$\mu_2   $&0.426&0.178& 0.618&0.179& 0.580&0.192\\
$\mu_3   $& --  & --  & 1.826&0.325& 1.569&0.523\\
$\mu_4   $& --  & --  & --   & --  & 2.218&0.830\\
$\sigma_1$&0.655&0.063& 1.091&0.167& 1.041&0.197\\
$\sigma_2$&1.495&0.138& 0.605&0.058& 0.578&0.073\\
$\sigma_3$& --  & --  & 0.892&0.154& 0.670&0.282\\
$\sigma_4$& --  & --  & --   & --  & 0.879&0.336\\
$\beta$   &0.867&0.036& 0.784&0.050& 0.783&0.050
\\\hline\hline
\end{tabular}\label{table:parameter}
\end{table} 
\begin{table} 
\centering
\caption{Selection of the number of regimes: U.S. GDP per capita growth, 1960Q1--2014Q4 }
\begin{tabular}[t]{c|ccc|cc|ccc|cc}   \hline\hline 
&\multicolumn{5}{c|}{Model 1 with common variance}&\multicolumn{5}{c}{Model 2 with switching variance}\\ \hline
 & &&&\multicolumn{2}{c| }{LRT}&&&&\multicolumn{2}{c }{LRT}\\
$M_0$&log-like.&AIC&BIC&LR$_n$&$p$-val. &log-like.&AIC&BIC&LR$_n$&$p$-val. \\ \hline
1&-331.70&669.39&\textbf{679.58}&20.86&0.000&-331.70&669.39&679.58&47.25&0.000\\
2&-321.27&656.54&680.29&27.77&0.000&-308.07&632.15&\textbf{659.29}&22.14&0.000\\
3&-307.39&\textbf{640.77}&684.89&15.23&0.020&-297.01&\textbf{624.01}&674.91&\textbf{4.87}&\textbf{0.392}\\
4&-299.77&641.54&712.81&\textbf{6.57}&\textbf{0.523}&-294.57&637.14&718.59&3.01&0.397\\
\hline\hline
\end{tabular}\label{table:select}
\end{table} 	
}

 \begin{figure}[tb] %
	\caption{The posterior probabilities of each regime (Model 1 with common variance): U.S. GDP per capita growth, 1960Q1--2014Q4 }
	\centering
   \begin{minipage}{.9\linewidth}
   \centering
   \medskip
   \textbf{M=2}\\ \smallskip
   \includegraphics[width=0.9\linewidth]{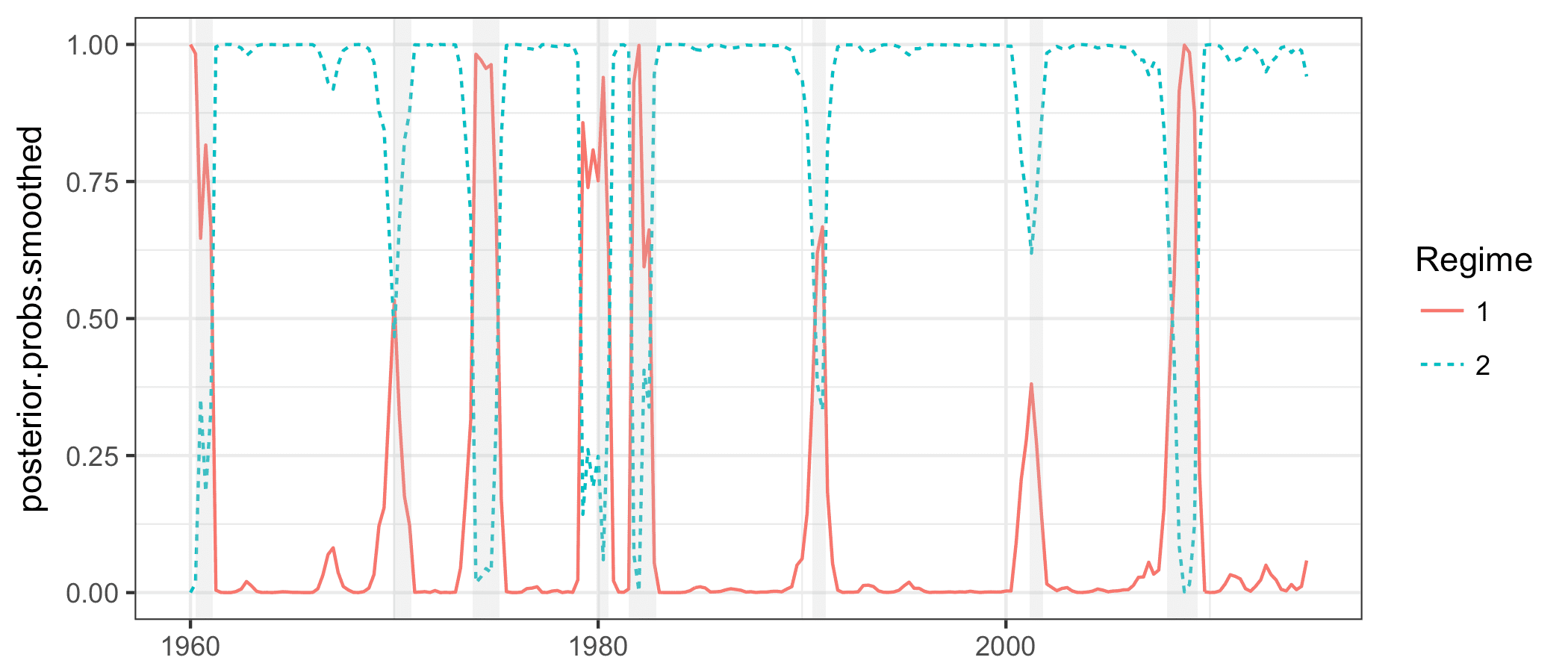}
   \end{minipage}
   \begin{minipage}{.9\linewidth}
   \centering
   \medskip
   \textbf{M=3}\\ \smallskip
   \includegraphics[width=0.9\linewidth]{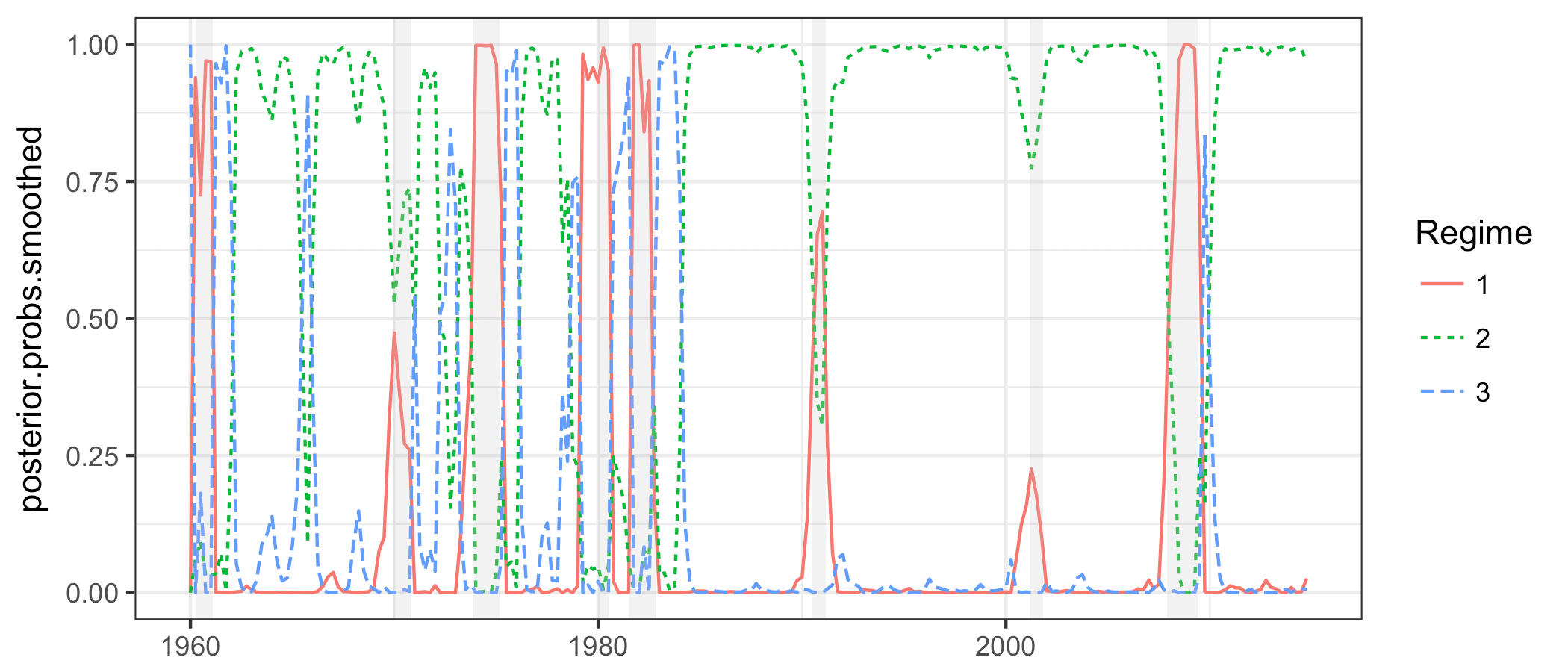}
   \end{minipage} 
   \begin{minipage}{.9\linewidth}
   \centering
   \medskip
   \textbf{M=4}\\ \smallskip
   \includegraphics[width=0.9\linewidth]{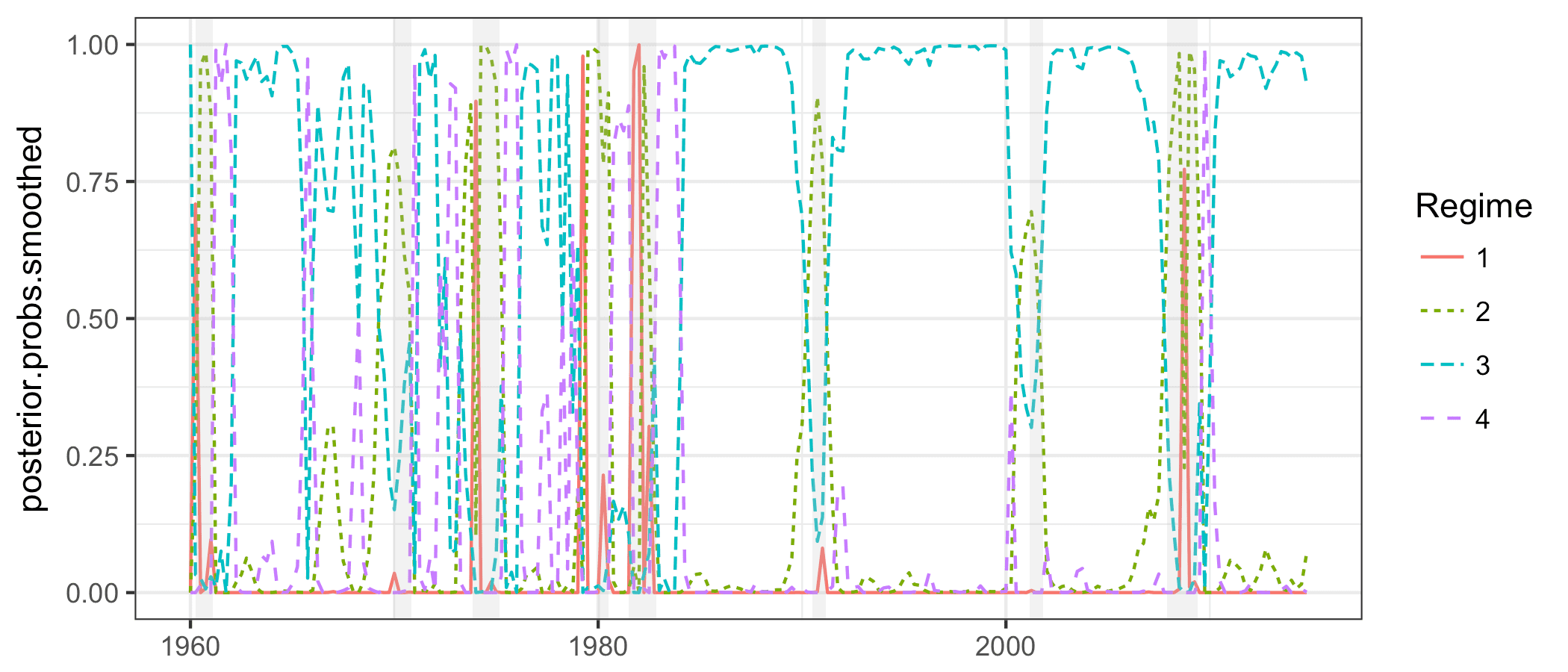}
   \end{minipage} 
\label{fig:common}
 \end{figure}

 \begin{figure}[tb] 
 	\caption{The posterior probabilities of each regime (Model 2 with switching variance): U.S. GDP per capita growth, 1960Q1--2014Q4 }\medskip
  	\centering
  	\begin{minipage}{.9\linewidth}
	 	\centering
   \medskip
   \textbf{M=2}\\ \smallskip
  		\includegraphics[width=0.9\linewidth]{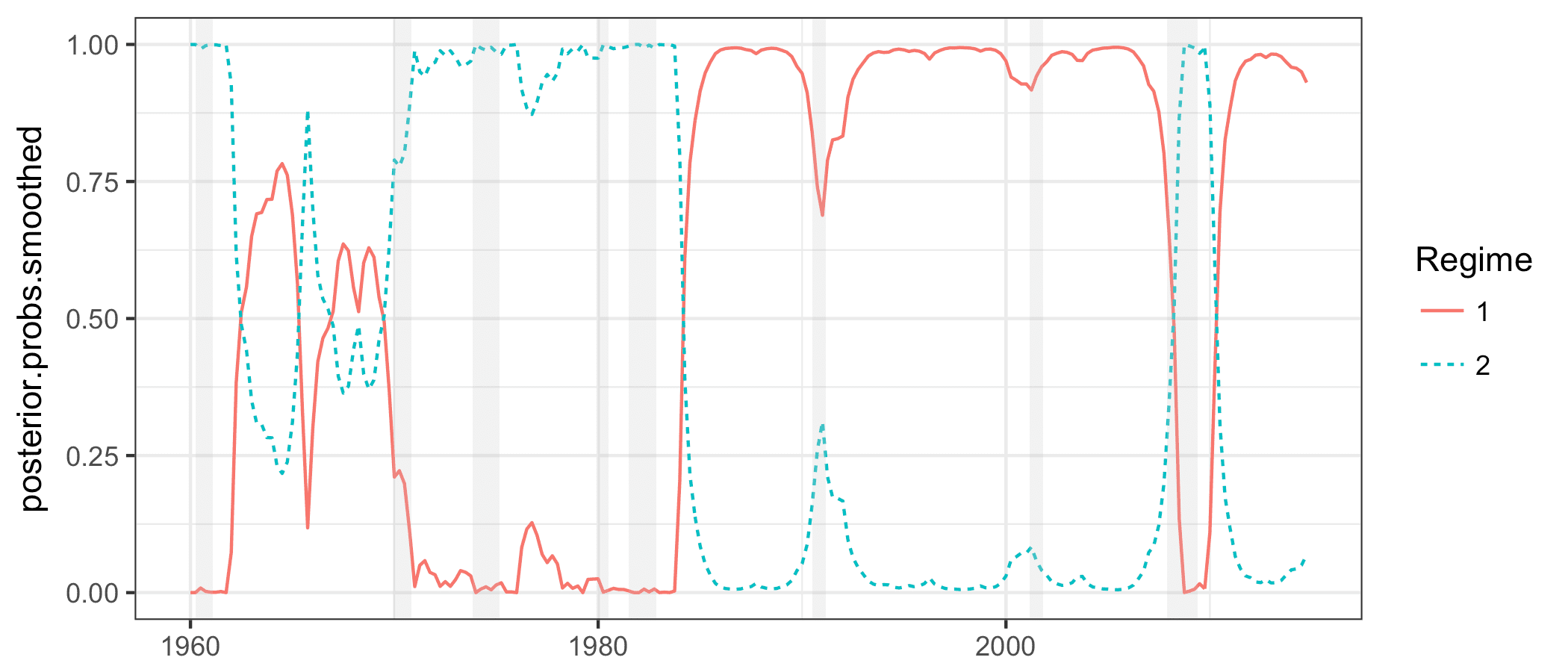}
  	\end{minipage}
  	\begin{minipage}{.9\linewidth}
	 	\centering
   \medskip
   \textbf{M=3}\\ \smallskip
  		\includegraphics[width=0.9\linewidth]{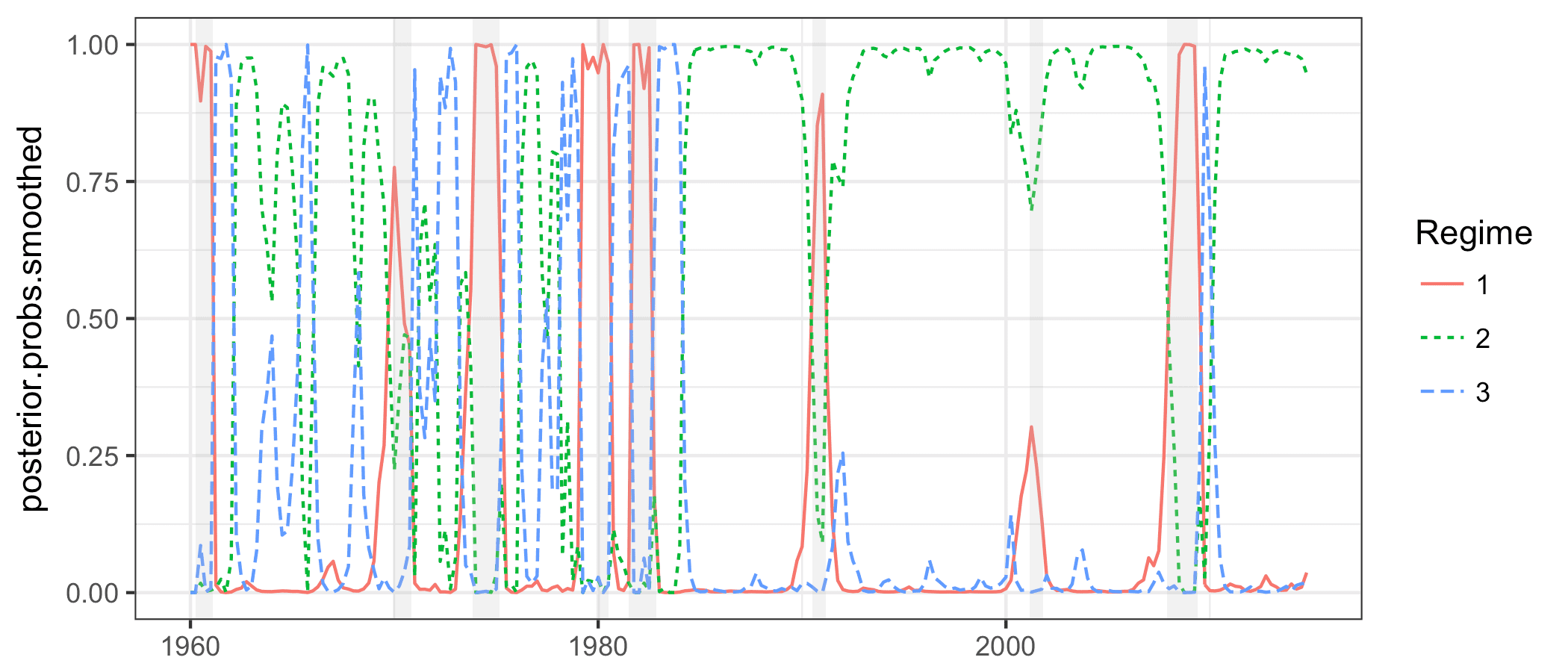}
  	\end{minipage} 
  	\begin{minipage}{.9\linewidth}
	 	\centering
   \medskip
   \textbf{M=4}\\ \smallskip
  		\includegraphics[width=0.9\linewidth]{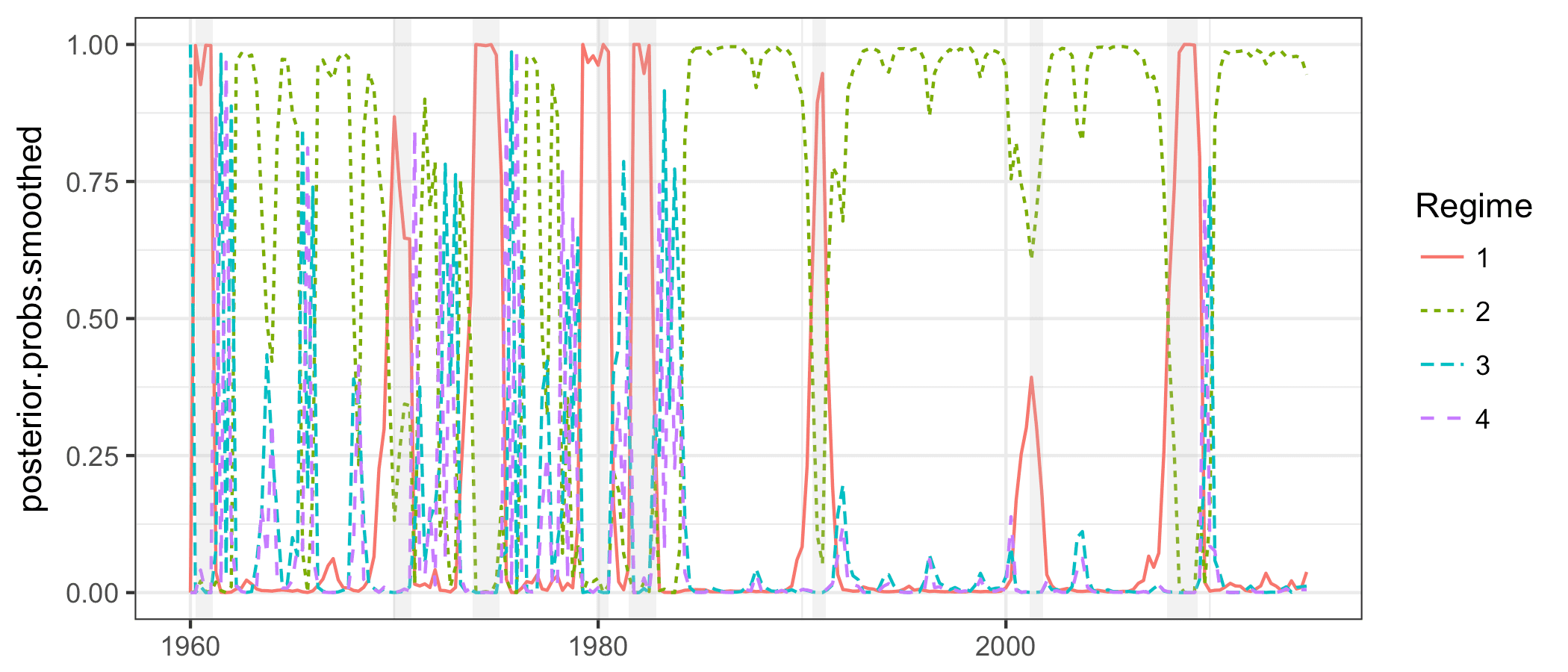}
  	\end{minipage} 
  	\label{fig:switching}
 \end{figure}

\end{document}